\documentclass[11pt]{article} %

\usepackage{url}
\usepackage{amsmath}
\usepackage{amsthm}
\usepackage{amssymb}
\usepackage{mathtools}
\usepackage{thmtools}
\usepackage{dsfont}
\usepackage{graphicx}
\usepackage{subcaption}
\usepackage{caption}
\usepackage[square,semicolon,numbers]{natbib}

\usepackage{booktabs}    %
\usepackage{multirow}    %
\usepackage{array}       %

\usepackage[subfigure]{tocloft}

\newcommand{\tocfooterrule}{\vspace{0.5\baselineskip}\noindent\hrulefill}

\usepackage{etoolbox} %

\usepackage[margin=0.8in]{geometry}
\usepackage{charter}

\usepackage{amsfonts,amssymb}
\usepackage{bbm}
\usepackage{mathrsfs}

\usepackage{amsthm,thmtools,thm-restate}
\usepackage{mathtools} %
\usepackage{amsmath} %
\usepackage[ruled,vlined]{algorithm2e}
\LinesNumbered
\DontPrintSemicolon
\SetKwComment{Comment}{$\triangleright$\ }{}
\usepackage{tabularx}
\usepackage{tabularx}

\usepackage[dvipsnames]{xcolor} %
\definecolor{ptblue}{RGB}{15,76,129} %
\definecolor{ptemerald}{HTML}{009473} %
\definecolor{bluegray}{rgb}{0.4, 0.6, 0.8}
\definecolor{ptilluminating}{HTML}{F5DF4D} %
\definecolor{ptgray}{HTML}{939597}
\usepackage{cprotect}
\definecolor{cobalt}{rgb}{0.0, 0.28, 0.67}
\usepackage{hyperref}
\hypersetup{
linktocpage,
colorlinks=true,
citecolor=cobalt, %
urlcolor=ptblue, %
linkcolor=cobalt, %
}

\usepackage{cleveref}

\usepackage{thm-restate}
\usepackage{enumitem}
\usepackage{verbatim}

\newcommand{\lognsw}{\log \, \mathrm{NSW}}
\newcommand{\col}{\mathrm{atb}}

\newcommand{\NSW}{\mathrm{NSW}}

\DeclareMathOperator*{\argmax}{arg\,max}
\DeclareMathOperator*{\argmin}{arg\,min}
\newcommand{\Alg}{\textsc{Alg}}
\newcommand{\Opt}{\textsc{Opt}}
\newcommand{\ersp}{\texttt{ERSP}}

\theoremstyle{plain}
\newtheorem{theorem}{Theorem}

\newtheorem{lemma}[theorem]{Lemma}

\theoremstyle{definition}
\newtheorem{definition}{Definition}

\theoremstyle{remark}
\newtheorem{remark}{Remark}

\newcommand{\amazon}{\texttt{Amazon}}
\newcommand{\deepC}{\texttt{Deep1b-\!(Clus)}}
\newcommand{\deepP}{\texttt{Deep1b-\!(Prob)}}

\newcommand{\siftC}{\texttt{Sift1m-\!(Clus)}}

\newcommand{\siftP}{\texttt{Sift1m-\!(Prob)}}

\newcommand{\arxiv}{\texttt{ArXiv}}
\newcommand{\pMeanANN}{\texttt{$p$-mean-ANN}}

\newcommand{\diskann}{\texttt{ANN}}
\newcommand{\divann}{\texttt{Div-ANN}}
\newcommand{\multidivann}{\texttt{Multi Div-ANN}}
\newcommand{\multinash}{\texttt{Multi Nash-ANN}}
\newcommand{\multipmean}{\texttt{Multi p-mean-ANN}}

\newcommand{\kANN}{\texttt{ANN}}

\newcommand{\ouralgo}{\texttt{Nash-ANN}}

 \title{Welfarist Formulations for Diverse Similarity Search}

\author{
Siddharth Barman\thanks{Indian Institute of Science. {\tt barman@iisc.ac.in}}
\and
Nirjhar Das\thanks{Indian Institute of Science. {\tt nirjhardas@iisc.ac.in} }
\and
Shivam Gupta\thanks{Indian Institute of Science. {\tt shivamgupta2@iisc.ac.in} }
\and
Kirankumar Shiragur\thanks{Microsoft Research India. {\tt kshiragur@microsoft.com} }
}

\date{}

\begin{document}

\maketitle

\begin{abstract}
Nearest Neighbor Search (NNS) is a fundamental problem in data structures with wide-ranging applications, such as web search, recommendation systems, and, more recently, retrieval-augmented generations (RAG). In such recent applications, in addition to the relevance (similarity) of the returned neighbors, diversity among the neighbors is a central requirement. In this paper, we develop principled welfare-based formulations in NNS for realizing diversity across attributes. Our formulations are based on welfare functions---from mathematical economics---that satisfy central diversity (fairness) and relevance (economic efficiency) axioms. With a particular focus on Nash social welfare, we note that our welfare-based formulations provide objective functions that adaptively balance relevance and diversity in a query-dependent manner. Notably, such a balance was not present in the prior constraint-based approach, which forced a fixed level of diversity and optimized for relevance.  In addition, our formulation provides a parametric way to control the trade-off between relevance and diversity, providing practitioners with flexibility to tailor search results to task-specific requirements. We develop efficient nearest neighbor algorithms with provable guarantees for  the welfare-based objectives. Notably, our algorithm can be applied on top of any standard ANN method (i.e., use standard ANN method as a subroutine) to efficiently find neighbors that approximately maximize our welfare-based objectives. Experimental results demonstrate that our approach is practical and substantially improves diversity while maintaining high relevance of the retrieved neighbors.
\end{abstract}

\newpage
\tableofcontents
\tocfooterrule

\newpage

\section{Introduction}

\label{section:introduction}
 
Nearest Neighbor Search (NNS) is a fundamental problem in computer science with wide-ranging applications in diverse domains, including computer vision~\citep{Wang12}, data mining~\citep{isax2}, information retrieval~\citep{irbook}, classification~\citep{Fix89}, and recommendation systems~\citep{DeepXML21}. The relevance of NNS has grown further in recent years with the advent of retrieval-augmented generation (RAG); see, e.g., \citep{manohar2024parlayann}, \citep{wu2024retrieval}, and references therein. 
Formally, given a set of vectors $P \subset \mathbb{R}^d$, in ambient dimension $d$, and a query vector $q \in \mathbb{R}^d$, the objective in NNS is to find a subset $S$ of $k$ (input) vectors from $P$ that are most similar to $q$ under a similarity function $\sigma: \mathbb{R}^d \times \mathbb{R}^d \to \mathbb{R}_+$. That is, NNS corresponds to the  optimization problem $\argmax_{S \subseteq P: \lvert S \rvert = k} \ \sum_{v \in S} \sigma(q,v)$. Note that, while most prior works in neighbor search express the problem in terms of minimizing distances, we work with the symmetric version of maximizing similarity.\footnote{This enables us to directly apply welfare functions.} 

In practice, the input vectors are high dimensional; in many of the above-mentioned applications the ambient dimension $d$ is close to a thousand. Furthermore, most applications involve a large number of input vectors. This scale makes exact NNS computationally expensive, since applications require, for real-time queries $q$, NNS solutions in time (sub)linear in the number of input vectors $|P|$. To address this challenge, the widely studied framework of Approximate Nearest Neighbor (ANN) search relaxes the requirement of exactness and instead seeks neighbors whose similarities are approximately close to the optimal ones. %

ANN search has received substantial attention over the past three decades. Early techniques relied on space-partitioning methods, including Locality-Sensitive Hashing (LSH) \citep{indyk1998approximate,andoni2008near}, k-d trees \citep{arya1998optimal}, and cover trees \citep{Beygelzimer06}. More recent industry-scale systems adopt clustering-based \citep{Faiss17, BBM18} and graph-based ~\citep{HNSW16,NSG17,NGT2,DiskANN19} approaches, along with other practically-efficient methods~\citep{soar_2023, simhadri2024results}.

While relevance---measured in terms of a similarity function $\sigma(\cdot, \cdot)$---is a primary objective in NNS, prior work has shown that {\it diversity} in the retrieved set of vectors is equally important for user experience, fairness, and reducing redundancy~\citep{carbonell1998use}. For instance, in 2019 Google announced a policy update to limit the number of results from a single domain, thereby reducing redundancy~\citep{google-div}. Similarly, Microsoft recently introduced diversity constraints in ad recommendation systems to ensure that advertisements from a single seller do not dominate the results ~\citep{anand2025graphbased}. Such an adjustment was crucial for improving user experience and promoting fairness for advertisers. These examples highlight how diversity, in addition to enhancing fairness and reducing redundancy, directly contributes to improved search quality for end users.

A natural way to formalize diversity in these settings is to associate each input vector with one or more \emph{attributes}. Diversity can then be measured with respect to these attributes, complementing the similarity-based relevance. Building on this idea, the current work develops a principled framework for diversity in neighbor search by drawing on the theory of collective welfare from mathematical economics~\citep{moulin2004fair}. This perspective enables the design of performance metrics (i.e., optimization criteria) that balance similarity-based relevance and attribute-based diversity in a theoretically grounded manner.

This formulation is based on the perspective that {\it algorithms can be viewed as economic policies.} Indeed, analogous to economic policies, numerous deployed algorithms induce utility (monetary or otherwise) among the participating agents. For instance, an ANN algorithm---deployed to select display advertisements for search queries---impacts the exposure and, hence, the sales of the participating advertisers. Notably, there are numerous other application domains wherein the outputs of the underlying algorithms impact the utilities of individuals; see \cite{angwin2022machine} and \cite{kearns2019ethical} for multiple examples. Hence, in contexts where fairness (diversity) and welfare are important considerations, it is pertinent to evaluate algorithms analogous to how one evaluates economic policies that induce welfare.

In mathematical economics, welfare functions, $f: \mathbb{R}^c \mapsto \mathbb{R}$, provide a principled approach to aggregate the utilities of $c \in \mathbb{Z}_+$ agents into a single measure. Specifically, if an algorithm  induces utilities $u_1, u_2, \ldots, u_c$ among a population of $c$ agents, then the collective welfare is $f(u_1, u_2,\ldots,u_c)$. A utilitarian way of aggregation is by considering the arithmetic mean (average) of the utilities $u_\ell$s. However, note that the arithmetic mean is not an ideal criterion if we are required to be fair among the $c$ agents: the utilitarian welfare (arithmetic mean) can be high even if the utility of only one agent, say $u_1$, is large and all the remaining utilities, $u_2,\ldots, u_c$, are zero. The theory of collective welfare develops meaningful alternatives to the arithmetic mean by identifying welfare functions, $f$s, that satisfy fairness and efficiency axioms.  

Among such alternatives, Nash social welfare (NSW) is an exemplar that upholds multiple fairness axioms, including symmetry, independence of unconcerned agents, scale invariance, and the Pigou-Dalton transfer principle \citep{moulin2004fair}. Nash social welfare is obtained by setting the function $f$ as the geometric mean, $\NSW(u_1, \ldots, u_c) \coloneqq \big( \prod_{\ell=1}^c u_\ell \big)^{1/c}$. The fact that NSW strikes a balance between fairness and economic efficiency is supported by the observation that it sits between egalitarian and utilitarian welfare: the geometric mean is at least as large as the minimum value, $\min_{1 \leq \ell \leq c} \ u_\ell$, and it is also at most the arithmetic mean $\frac{1}{c} \sum_{\ell=1}^c u_\ell$ (the AM-GM inequality). 

The overarching goal of this work is to realize diversity (fairness) across attributes in nearest neighbor search while maintaining relevance of the returned $k$ vectors. Our modeling insight here is to equate attributes with agents and apply Nash social welfare. 

In particular, consider a setting where we have $c \in \mathbb{Z}_+$ different attributes (across the input vectors), and let $S$ be any subset of $k$ vectors (neighbors) among the input set $P$. In our model, each included vector $v \in S$, with attribute $\ell \in [c]$, contributes to the utility $u_\ell$ (see Section \ref{section:our-results}), and the Nash social welfare (NSW) induced by $S$ is the geometric mean of these utilities, $u_1, u_2, \ldots, u_c$. Our objective is to find a size-$k$ subset, $S^* \subseteq P$, of input vectors with as large NSW as possible.  

The following two instantiations highlight the applicability of our model in NNS settings: In a display-advertising context with $c$ sellers, each selected advertisement $v \in S$ of a seller $\ell \in [c]$ contributes to $\ell$'s exposure (utility) $u_\ell$. Similarly, in an apparel-search setup with $c$ colors in total, each displayed product $v \in S$ with color $\ell \in [c]$ contributes to the utility $u_\ell$.  

\begin{figure*}[t!]
\centering
\includegraphics[width=0.24\textwidth]{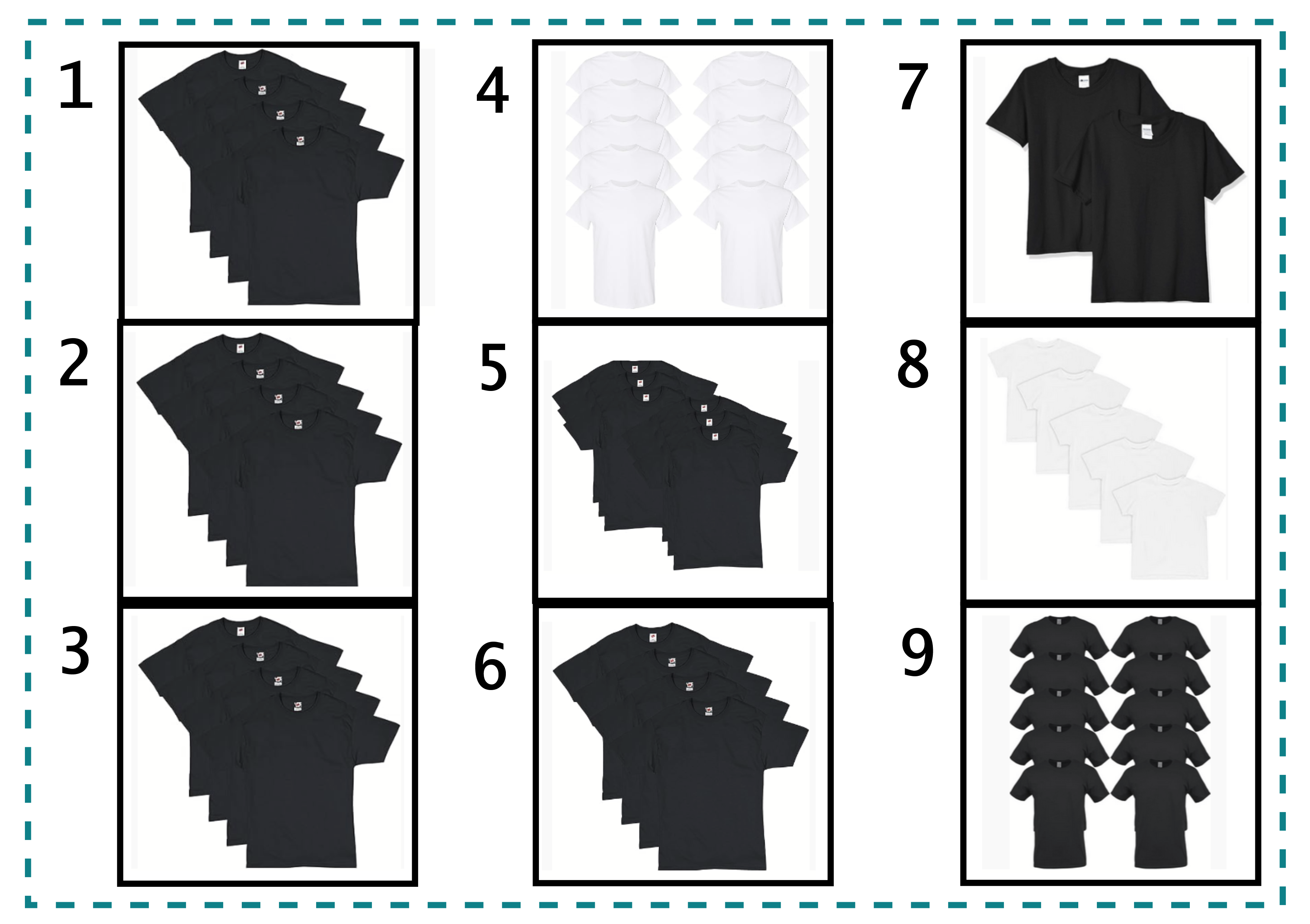} \hfill
\includegraphics[width=0.24\textwidth]{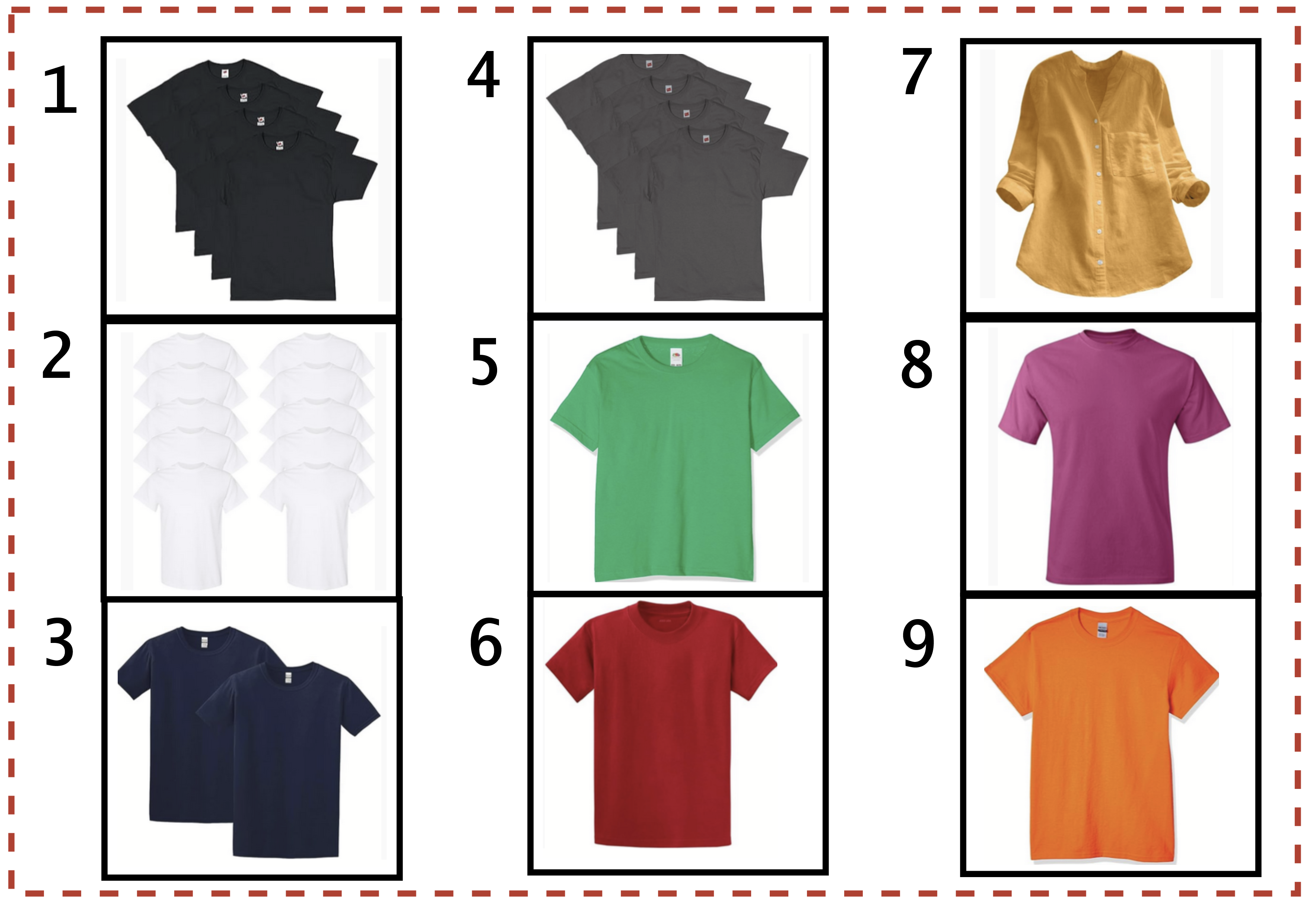} \hfill
\includegraphics[width=0.23\textwidth]{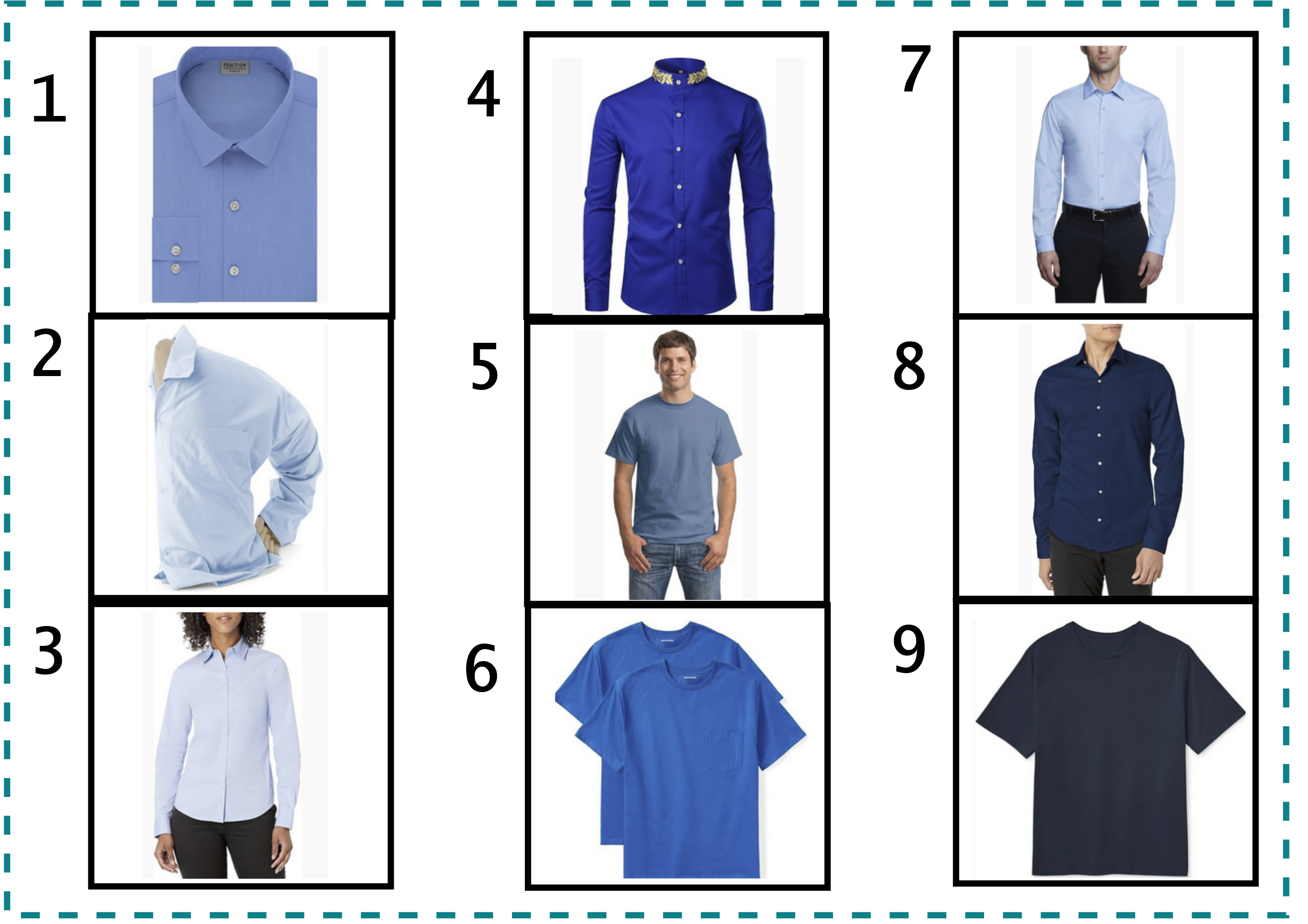} \hfill
\includegraphics[width=0.23\textwidth]{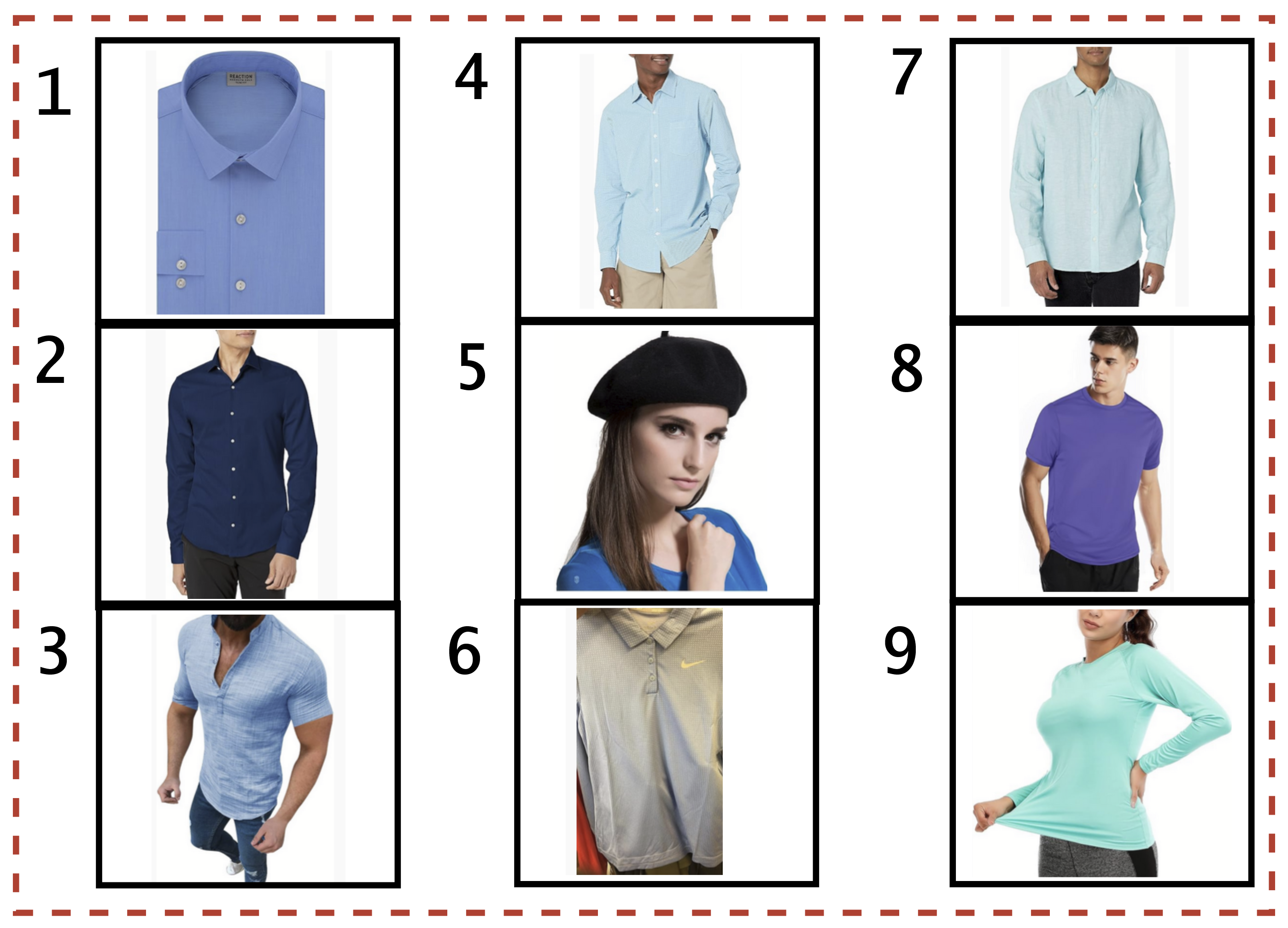}\\
\caption{Neighbor search results ($k=9$) on the Amazon dataset. From left: \textbf{First} and \textbf{Second} images - ANN and Nash-based results for query ``shirts'', respectively.
\textbf{Third} and \textbf{Fourth} images - ANN and Nash-based results for query ``blue shirt'', respectively.
Note that the Nash-based method selects diverse colors for the query ``shirts'' but conforms to the blue color for the query ``blue shirt''.}
\label{fig:QueryResponses}
\end{figure*}

Prior work \citep{anand2025graphbased} imposed constraints for achieving diversity in NNS. These constraints enforced that, for each $\ell \in [c]$ and among the $k$ returned vectors, at most $k'$ many can have attribute $\ell$. Such hard constraints rely on a fixed ad hoc quota parameter $k'$ and may fail to adapt to the intent expressed in the query. In contrast, our NSW-based approach balances relevance and diversity in a query-dependent manner. For example, in the apparel-search setup, if the search query is ``blue shirt,'' then a constraint on the color attribute `blue' (i.e., when $\ell$ stands for `blue') would limit the relevance by excluding valid vectors. NSW, however, for the ``blue shirt'' query, is free to select all the $k$ vectors with attribute `blue' upholding relevance; see Figure \ref{fig:QueryResponses} for supporting empirical results. On the other hand, if the apparel-search query is just ``shirts,'' then NSW criterion is inclined to select vectors with different color attributes. These features of NSW are substantiated by the stylized instances given in Examples \ref{example-one} and \ref{example-two} (Section \ref{section:our-results}).

We reiterate that our formulation does not require a quota parameter $k'$ to force diversity. For NSW, diversity (fairness) across attributes is obtained via normative properties of Nash social welfare. Hence, with axiomatic support, NSW stands as a meaningful criterion in  neighbor search, as it is in the context of economic and allocation policies.

Our welfarist formulation extends further to control the trade-off between relevance and diversity. Specifically, we also consider $p$-mean welfare. Formally, for exponent parameter $p\! \in\!\! (-\infty, 1]$, the $p$th  mean $M_p(\cdot)$, of $c$ utilities $u_1,u_2,\ldots, u_c \in \mathbb{R}_+$, is defined as  $M_p(u_1, \ldots, u_c) \coloneqq \left( \frac{1}{c} \sum_{\ell=1}^c u_\ell^p \right)^{1/p}$. The $p$-mean welfare, $M_p(\cdot)$, captures a range of objectives with different values of $p$: it corresponds to the utilitarian welfare (arithmetic mean) when $p\! =\! 1$, the NSW (geometric mean) with $p\! \to\! 0$, and the egalitarian welfare when $p\! \to\! - \infty$. Notably, setting $p\!=\!1$, we get back the standard nearest neighbor objective, i.e., maximizing $M_1(\cdot)$ corresponds to finding the $k$ nearest neighbors and this objective is not concerned with diversity across attributes. At the other extreme, $p\! \to\! -\infty$ aims to find as attribute-diverse a set of $k$ vectors as possible (while paying scarce attention to relevance). 

We study, both theoretically and experimentally, two diversity settings: (i) single-attribute setting and (ii) multi-attribute setting. In the single-attribute setting, each input vector $v \in P$ is associated with exactly one attribute $\ell \in [c]$ -- this captures, for instance, the display-advertisement setup, wherein each advertisement $v$ belongs to exactly one seller $\ell$. In the more general multi-attribute setting, each input vector $v \in P$ can have more than one attribute; in apparel-search, for instance, the products can be associated with multiple attributes, such as color, brand, and price. 

We note that the constraint-based formulation for diversity considered in \cite{anand2025graphbased} primarily addresses single-attribute setting. In fact, generalizing such constraints to the multi-attribute context leads to a formulation wherein it is {\rm NP}-hard even to determine whether there exist $k$ vectors that satisfy the constraints, i.e., it would be computationally hard to find any size-$k$ constraint-feasible subset $S$, let alone an optimal one.\footnote{This hardness result follows via a reduction from the Maximum Independent Set problem.}

By contrast, our NSW formulation does not run into such a feasibility barrier. Here, for any candidate subset $S$ of $k$ vectors, each included vector $v \in S$ contributes to the utility $u_\ell$ of every attribute $\ell$ associated with $v$. As before, the NSW induced by $S$ is the geometric mean of the induced utilities, $u_1, u_2, \ldots, u_c$, and the objective is to find a subset of $k$ vectors with as large NSW as possible. 

\subsection*{Our Contributions}
\begin{itemize}
\item We view the NSW formulation for diversity, in both single-attribute and multi-attribute settings, as a key contribution of the current paper.  Another relevant contribution of this work is the generalization to $p$-mean welfare, which provides a systematic way to trade off relevance and diversity.

\item We also develop efficient algorithms, with provable guarantees, for the NSW and $p$-mean welfare formulations. For the single-attribute setting, we develop an efficient greedy algorithm for finding $k$ vectors that optimize the Nash social welfare among the $c$ attributes (Theorem \ref{theorem:single-attribute}). In addition, this algorithm can be provably combined with any sublinear ANN method (as a subroutine) to find near-optimal solutions for the Nash objective in sublinear time (Corollary~\ref{theorem:single-attribute-approx-oracle}).

\item For the multi-attribute setting, we first show that finding the set of $k$ vectors that maximize the Nash social welfare is {\rm NP}-hard (Theorem \ref{theorem:multi-attribute-hardness}). We complement this hardness result, by developing a polynomial-time approximation algorithm that achieves an approximation ratio of $\left(1 - 1/e\right) \approx 0.63$ for maximizing the logarithm of the Nash social welfare (Theorem \ref{theorem:multi-attribute-greedy-guarantee}).

\item We complement our theoretical results with experiments on both real-world and semi-synthetic datasets. These experiments demonstrate that the NSW objective effectively captures the trade-off between diversity and relevance in a query-dependent manner. We further analyze the behavior of the $p$-mean welfare objective across different values of $p \in (-\infty, 1]$, observing that it interpolates smoothly between prioritizing for diversity, when $p$ is small, and focusing on relevance, when $p$ is large. Finally, we benchmark the solution quality and running times of various algorithms for solving the NSW and $p$-mean formulations proposed in this work.

\end{itemize}

\section{Problem Formulation and Main Results}
\label{section:model_and_results}
 
We are interested in neighbor search algorithms that not only achieve a high relevance, but also find a diverse set of vectors for each query. To quantify diversity we work with a model wherein each input vector $v \in P$ is assigned one or more attributes from the set $[c] = \{1,2,\ldots, c\}$. In particular, write $\col(v) \subseteq [c]$ to denote the attributes assigned to vector $v \in P$. Also, let $D_\ell \subseteq P$ denote the subset of vectors that are assigned attribute $\ell \in [c]$, i.e., $D_\ell \coloneqq \{ v \in P \mid \ell \in \col(v) \}$.  

\subsection{Our Results}
\label{section:our-results}
An insight of this work is to equate these $c$ attributes with $c$ distinct agents. Here, the output of a neighbor search algorithm---i.e., the selected subset $S \subseteq P$---induces utility among these agents. With this perspective, we define the Nash Nearest Neighbor Search problem (NaNNS) below. This novel formulation for diversity is a  key contribution of this work. For any query $q \in \mathbb{R}^d$ and subset $S \subseteq P$, we define utility $u_\ell(S) \coloneqq \sum_{v \in S \cap D_\ell} \sigma(q, v)$, for each $\ell \in [c]$. That is, $u_\ell(S)$ is equal to the cumulative similarity between $q$ and the vectors in $S$ that belong to $D_\ell$. Equivalently, $u_\ell(S)$ is the cumulative similarity of the vectors in $S$ that have attribute $\ell$.\footnote{Note that in the above-mentioned display-advertising example, $u_\ell(\cdot)$ is the cumulative similarity between the (search) query and the selected advertisements that are from seller $\ell$.}

We employ Nash social welfare to identify size-$k$ subsets $S$ that are both relevant (with respect to similarity) and support diversity among the $c$ attribute classes. The Nash social welfare among $c$ agents is defined as the geometric mean of the agents' utilities. Specifically, in the above-mentioned utility model and with a smoothening parameter $\eta >0$, the Nash social welfare (NSW) induced by any subset $S \subseteq P$ among the $c$ attributes is defined as 
\begin{align}
   \NSW(S) \coloneqq \ ( \prod_{\ell=1}^c \ ( u_\ell(S) + \eta ) )^{1/c}. 
\end{align}
Throughout, $\eta>0$ will be a fixed smoothing constant that ensures that NSW remains nonzero. 

\begin{definition}[NaNNS]
\label{definition:NaNNS}
Nash nearest neighbor search (NaNNS) corresponds to the following the optimization problem  
$\argmax_{S \subseteq P: |S| = k} \ \NSW(S)$, or, equivalently, 
\begin{align}
\label{optimization-problem:log-nsw}
    \argmax_{S \subseteq P: \lvert S \rvert = k} \lognsw (S) 
\end{align}
\end{definition}
Here, we have $\lognsw (S)  = \frac{1}{c}\sum_{\ell \in [c]} \log(u_\ell(S) + \eta)$.

To further appreciate the welfarist approach, note that one recovers the standard nearest neighbor problem, NNS, in the single-attribute setting, if---instead of the geometric mean---we maximize the arithmetic mean. That is, maximizing the utilitarian social welfare gives us  
$\max_{S \subseteq P: |S| = k} \ \sum_{\ell =1}^c u_\ell(S) = \max_{S \subseteq P: |S| = k} \ \sum_{v \in S} \sigma(q, v)$. 

As stated in the introduction, depending on the query and the problem instance, solutions obtained via NaNNS can adjust between the ones obtained through standard NNS and those obtained via hard constraints. This feature is illustrated in the following stylized examples; see Appendix~\ref{section:example-proofs} for the associated proofs. 

The first example shows that if all vectors have same similarity, then an optimal solution, $S^*$, for NaNNS is completely diverse, i.e., all the vectors in $S^*$ have different attributes. 
\begin{restatable}[Complete Diversity via NaNNS]{example}{ExampleOne}
\label{example-one}
    Consider an instance in which, for a given query $q \in \mathbb{R}^d$, all vectors in $P$ are equally similar with the query: $\sigma(q, v) = 1$ for all $v \in P$. Also, let $\lvert \col(v) \rvert = 1$ for all $v \in P$, and write $S^* \in \argmax_{S \subseteq P: \ \lvert S \rvert = k} \NSW(S)$. If $c \geq k$, then here it holds that $|S^* \cap D_\ell| \leq 1$ for all $\ell \in [c]$.  
\end{restatable}
The second example shows that if the vectors of only one attribute have high similarity with the given query, then a Nash optimal solution $S^*$ contains only vectors with that attribute.

\begin{restatable}[Complete Relevance via NaNNS]{example}{ExampleTwo}
\label{example-two}
    Consider an instance in which, for a given query $q \in \mathbb{R}^d$ and a particular $\ell^* \in [c]$, only vectors $v \in D_{\ell^*}$ have similarity $\sigma(q, v) = 1$ and all other vectors $v' \in P \setminus D_{\ell^*}$ have similarity $\sigma(q, v') = 0$. Also, suppose that $\lvert \col(v) \rvert = 1$ for each $v \in P$, along with $\lvert D_{\ell^*} \rvert \geq k$. Then, for a Nash optimal solution $S^* \in \argmax_{S \subseteq P, \lvert S \rvert = k} \NSW(S)$, it holds that $\lvert S^* \cap D_{\ell^*} \rvert = k$. That is, for all other $\ell \in [c] \setminus \{\ell^*\}$ we have $\lvert S^* \cap D_\ell \rvert = 0$. 
\end{restatable}

With the above-mentioned utility model for the $c$ attributes, we also identify an extended formulation based on generalized $p$-means. Specifically, for exponent parameter $p \in (-\infty, 1]$, the $p$th  mean $M_p(\cdot)$, of $c$ nonnegative numbers $w_1,w_2,\ldots, w_c \in \mathbb{R}_+$, is defined as 
\begin{align}
    M_p(w_1, \ldots, w_c) \coloneqq \left( \frac{1}{c} \sum_{\ell=1}^c w_\ell^p \right)^{1/p}
\end{align}
Note that $M_1(w_1, \ldots, w_c)$ is the arithmetic mean $\frac{1}{c} \sum_{\ell=1}^c w_\ell$. Here, when $p \to 0$, we obtain the geometric mean (Nash social welfare): $M_0(w_1,\ldots,w_c) = \left( \prod_{\ell=1}^c w_\ell \right)^{1/c}$. Further, $p \to -\infty$ gives us egalitarian welfare, $M_{-\infty}(w_1, \ldots, w_\ell) = \min_{1 \leq \ell \leq c} \ w_\ell$.

Analogous to NaNNS, we consider a fixed smoothing constant $\eta>0$, and, for each subset $S \subseteq P$, define $M_p(S) \coloneqq M_p\left(u_1(S) + \eta, \ldots, u_c(S) + \eta \right)$. 

With these constructs in hand and generalizing both NNS and NaNNS, we have the $p$-mean nearest neighbor search ($p$-NNS) problem defined as follows.

\begin{definition}[$p$-NNS]
\label{definition:p-NNS}
For any exponent parameters $p \in (-\infty, 1]$, the $p$-mean nearest neighbor search ($p$-NNS) corresponds to the following the optimization problem  
\begin{align}
    \max_{S \subseteq P: \ \lvert S \rvert = k} \  M_p \left( S \right) = \max_{S \subseteq P: \ \lvert S \rvert = k} 
    \left( \frac{1}{c} \sum_{\ell=1}^c \left( u_\ell(S) +\eta \right)^p \right)^{1/p}
\end{align}
\end{definition}

\paragraph{Diversity in Single- and Multi-Attribute Settings.} The current work addresses two diversity settings: the single-attribute setup and, the more general, the multi-attribute one. The single-attribute setting refers to the case wherein $\lvert \col(v) \rvert = 1$ for each input vector $v \in P$ and, hence, the attribute classes $D_\ell$s are pairwise disjoint. In the more general multi-attribute setting, we have $\lvert \col(v) \rvert \geq 1$; here, the sets $D_\ell$-s intersect.\footnote{
For a motivating instantiation for multi-attributes, note that, in the apparel-search context, it is possible for a product (input vector) $v$ to have multiple attributes based on $v$'s seller and its color(s).} Notably, the NaNNS seamlessly applies to both these settings.

\paragraph{Algorithmic Results for Single-Attribute NaNNS and $p$-NNS.} 
In addition to introducing the NaNNS and $p$-NNS formulations for capturing diversity, we develop algorithmic results for these problems, thereby demonstrating the practicality of our approach in neighbor search. In particular, in the single-attribute setting, we show that both NaNNS and $p$-NNS admit efficient algorithms.

\begin{restatable}{theorem}{SingleAttributeResult}
\label{theorem:single-attribute} 
In the single-attribute setting, given any query $q \in \mathbb{R}^d$ and an (exact) oracle \texttt{ENN} for $k$ most similar vectors from any set, Algorithm~\ref{algo:greedy-nash-ann} (\texttt{Nash-ANN}) returns an optimal solution for NaNNS, i.e., it returns a size-$k$ subset $\Alg \subseteq P$ that satisfies 
       $\Alg \in \argmax_{S \subseteq P: \lvert S \rvert = k} \NSW (S)$.
    Furthermore, the algorithm runs in time $O(kc) + \sum_{\ell=1}^c \texttt{ENN}(D_\ell, q)$, where \texttt{ENN}$(D_\ell, q)$ is the time required by the exact oracle to find $k$ most similar vectors to $q$ in $D_\ell$.
\end{restatable}

Further, to establish the practicality of our formulations, we present an approximate algorithm for NaNNS that leverages any standard ANN algorithm as an oracle (subroutine), i.e., works with any $\alpha$-approximate ANN oracle ($\alpha \in (0,1)$) which returns a subset $S$ containing $k$ vectors satisfying $\sigma(q, v_{(i)}) \geq \alpha \  \sigma(q, v^*_{(i)})$, for all $i \in [k]$, where $v_{(i)}$ and $v^*_{(i)}$ are the $i$-th most similar vectors to $q$ in $S$ and $P$, respectively. 
Formally, 

\begin{restatable}{corollary}{SingleAttributeCorollaryApproximateOracle}
\label{theorem:single-attribute-approx-oracle}
    In the single-attribute setting, given any query $q \in \mathbb{R}^d$ and an $\alpha$-approximate oracle \texttt{ANN} for $k$ most similar vectors from any set, Algorithm~\ref{algo:greedy-nash-ann} (\texttt{Nash-ANN}) returns an $\alpha$-approximate solution for NaNNS, i.e., it returns a size-$k$ subset $\Alg \subseteq P$ with  %
        $\NSW(\Alg) \geq \alpha \max_{S \subseteq P: \  \lvert S \rvert = k} \NSW (S)$.
    The algorithm runs in time $O(kc) + \sum_{\ell=1}^c \texttt{ANN}(D_\ell, q)$, where \texttt{ANN}$(D_\ell, q)$ is the time required by the oracle to find $k$ similar vectors to $q$ in $D_\ell$.
\end{restatable}
Furthermore, both~\Cref{theorem:single-attribute} and~\Cref{theorem:single-attribute-approx-oracle} generalize to $p$-NNS problem with a slight modification in \Cref{algo:greedy-nash-ann}. Specifically, there exists exact, efficient algorithm (\Cref{algo:p-mean-ann}) for the $p$-NNS problem (\Cref{theorem:single-attribute-p-mean} and~\Cref{theorem:single-attribute-p-mean-approx-oracle}). Appendix~\ref{appendix:p-mean} details this algorithm and the results for $p$-NNS in the single-attribute setting.

\paragraph{Algorithmic Results for Multi-Attribute NaNNS.} Next, we address the multi-attribute setting. While the optimization problem~\eqref{optimization-problem:log-nsw} in the single attribute setting can be solved efficiently, the problem is {\rm NP}-hard in the multi-attribute setup (see~\Cref{subsec:proof-of-np-hardness-of-multi-attribute} for the proof).

\begin{restatable}{theorem}{TheoremNPHardness}
\label{theorem:multi-attribute-hardness}
In the multi-attribute setting, with parameter $\eta = 1$, NaNNS is {\rm NP}-hard. 
\end{restatable}

Complementing this hardness result, we show that, considering the logarithm of the objective, NaNNS in the multi-attribute setting admits a polynomial-time $\left(1 - \frac{1}{e}\right)$-approximation algorithm. This result is established in \Cref{subsec:submodular-greedy-approx-guarantee}. 

\begin{restatable}{theorem}{TheoremMultiAttributeApprox}
\label{theorem:multi-attribute-greedy-guarantee}
    In the multi-attribute setting with parameter $\eta = 1$, there exists a polynomial-time algorithm  (\Cref{algo:greedy-submodular}) that, given any query $q \in \mathbb{R}^d$, finds a size-$k$ subset $\Alg \subseteq P$ with $\lognsw(\Alg) \geq \left(1 - \frac{1}{e}\right) \lognsw(\Opt)$; here, $\Opt$ denotes an optimal solution for the optimization problem (\ref{optimization-problem:log-nsw}).
\end{restatable}
\paragraph{Experimental Validation of our Formulation and Algorithms.} We complement our theoretical results with several experiments on real-world datasets. Our findings highlight that the Nash-based formulation strikes a balance between diversity and relevance. 
Specifically, we find that, across datasets and in both single- and multi-attribute settings, the Nash formulation maintains relevance and consistently achieves high diversity. By contrast, the hard-constrained formulation from \citep{anand2025graphbased} is highly sensitive to the choice of the constraint parameter $k'$, and in some cases incurs a substantial drop in relevance. Section \ref{section:experiments} details our experimental evaluations.

\section{NaNNS in the Single-Attribute Setting}
\label{section:algorithm}

\begin{algorithm}[h]
\caption{\texttt{Nash-ANN}: Algorithm for NaNNS in the single-attribute setting}
\label{algo:greedy-nash-ann}

\KwIn{Query $q \in \mathbb{R}^d$ and, for each attribute $\ell \in [c]$,
the set of input vectors $D_\ell \subset \mathbb{R}^d$}

For each $\ell \in [c]$, fetch $\widehat{D}_\ell$, the $k$ (exact or approximate)
nearest neighbors of $q \in \mathbb{R}^d$ from $D_\ell$.
\label{line:hat-D}

For every $\ell \in [c]$ and each index $i \in [k]$, let
$v^\ell_{(i)}$ denote the $i$th most similar vector to $q$
in $\widehat{D}_\ell$.
\label{line:definition-of-ith-most-similar-point}

Initialize subset $\Alg \gets \emptyset$, along with count
$k_\ell \gets 0$ and utility $w_\ell \gets 0$, for each $\ell \in [c]$.

\While{$|\Alg| < k$}{
    \label{line:while-loop-start}

    $a \gets \argmax\limits_{\ell \in [c]}
    \Big(\log\big(w_\ell + \eta + \sigma(q, v^\ell_{(k_\ell + 1)})\big)
    - \log(w_\ell + \eta)\Big)$.
    \Comment{Ties broken arbitrarily}
    \label{line:marginal}

    $\Alg \gets \Alg \cup \{ v^a_{(k_a + 1)} \}$.
    $w_a \gets w_a + \sigma(q, v^a_{(k_a + 1)})$.
    $k_a \gets k_a + 1$.
    \label{line:update}
}
\label{line:while-loop-ends}

\Return{$\Alg$}

\end{algorithm}

In this section, we first provide our exact, efficient algorithm (\Cref{algo:greedy-nash-ann}) for NaNNS in the single-attribute setting and then present the proof of optimality of the algorithm (i.e., proof of~\Cref{theorem:single-attribute} and~\Cref{theorem:single-attribute-approx-oracle}). %

The algorithm has two parts: a prefetching step and a greedy, iterative selection. In the prefetching step, for each attribute $\ell \in [c]$, we populate $k$ vectors from within $D_\ell$\footnote{Recall that in the single-attribute setting, the input vectors $P$ are partitioned into subsets $D_1,\ldots, D_c$, where $D_\ell$ denotes the subset of input vectors with attribute $\ell \in [c]$.} that are most similar to the given query $q \in \mathbb{R}^d$. Such a size-$k$ subset, for each $\ell \in [c]$, can be obtained by executing any nearest neighbor search algorithm within $D_\ell$ and with respect to query $q$. Alternatively, we can execute any standard ANN algorithm as a subroutine and find sufficiently good approximations for the $k$ nearest neighbors (of $q$) within each $D_\ell$.    

Write $\widehat{D}_\ell \subseteq D_\ell$ to denote the $k$---exact or approximate---nearest neighbors of $q \in \mathbb{R}^d$ in $D_\ell$. We note that our algorithm is robust to the choice of the search algorithm (subroutine) used for finding $\widehat{D}_\ell$s: If $\widehat{D}_\ell$s are exact nearest neighbors, then \Cref{algo:greedy-nash-ann} optimally solves NaNNS in the single-attribute setting (\Cref{theorem:single-attribute}). Otherwise, if $\widehat{D}_\ell$s are obtained via an ANN algorithm with approximation guarantee $\alpha \in (0,1)$, then \Cref{algo:greedy-nash-ann} achieves an approximation ratio of $\alpha$ (\Cref{theorem:single-attribute-approx-oracle}). 

The algorithm then considers the vectors with each $\widehat{D}_\ell$ in decreasing order of their similarity with $q$. Confining to this order, the algorithm populates the $k$ desired vectors iteratively. In each iteration, the algorithm greedily selects a new vector based on maximizing the marginal increase in $\lognsw(\cdot)$; see Lines \ref{line:marginal} and \ref{line:update} in \Cref{algo:greedy-nash-ann}. 
Theorem \ref{theorem:single-attribute} and Corollary \ref{theorem:single-attribute-approx-oracle} (stated previously) provide our main results for \Cref{algo:greedy-nash-ann}. Proofs of~\Cref{theorem:single-attribute} and~\Cref{theorem:single-attribute-approx-oracle} are presented below.

\subsection{Proofs of Theorem \ref{theorem:single-attribute} and Corollary \ref{theorem:single-attribute-approx-oracle}}
Now we state the proof of optimality of~\Cref{algo:greedy-nash-ann}. To obtain the proof, we make use of two lemmas stated below. We defer their proofs till after the proof of~\Cref{theorem:single-attribute} and~\Cref{theorem:single-attribute-approx-oracle}.

As in \Cref{algo:greedy-nash-ann}, write $\widehat{D}_\ell$ to denote the $k$ nearest neighbors of the given query $q$ in the set $D_\ell$. Recall that in the single-attribute setting the  sets $D_\ell$s are disjoint across $\ell \in [c]$. Also, $v^\ell_{(j)} \in \widehat{D}_\ell$ denotes the $j^{\text{th}}$ most similar vector to $q$ in $\widehat{D}_\ell$, for each index $j \in [k]$. For a given attribute $\ell \in [c]$, we define the logarithm of cumulative similarity upto the $i^\text{th}$ most similar vector as
\begin{align}
\label{eq:defn-of-f_ell}
F_\ell(i) \coloneqq \log \left( \sum_{j=1}^i \sigma(q, v^\ell_{(j)}) + \eta \right)
\end{align}
Note that $F_\ell(i)$ is also equal to the logarithm of the cumulative similarity of the $i$ most similar (to $q$) vectors in $D_\ell$ when the neighbor search oracle is exact. 
The lemma below shows that $F_\ell(\cdot)$ satisfies a useful decreasing marginals property. 

\begin{lemma}[Decreasing Marginals]
\label{lemma:decreasing-marginal-single-attribute}
    For all attributes $\ell \in [c]$ and indices $i', i \in [k]$, with $i'< i$, it holds that
    \begin{align*}
    F_\ell(i') - F_\ell(i'-1) \geq F_\ell(i) - F_\ell(i-1)~.
    \end{align*}
\end{lemma}

The following lemma asserts the Nash optimality of the subset returned by \Cref{algo:greedy-nash-ann}, $\Alg$, within a relevant class of solutions. 

\begin{lemma}
    \label{lemma:size-match}
In the single-attribute setting, let $\Alg$ be the subset of vectors returned by \Cref{algo:greedy-nash-ann} and $S$ be any subset of input vectors with the property that $|S \cap D_\ell| = |\Alg \cap D_\ell|$, for each $\ell \in [c]$. Then, $\NSW(\Alg) \geq \NSW(S)$. 
\end{lemma}

Now we are ready to state the proof of~\Cref{theorem:single-attribute}.
\SingleAttributeResult*
\begin{proof}
    The runtime of \Cref{algo:greedy-nash-ann} can be established by noting that Line \ref{line:hat-D} requires $\sum_{\ell=1}^c \texttt{ENN}(D_\ell, q)$ time to populate the subsets $\widehat{D}_\ell$s, and the while-loop (Lines~\ref{line:while-loop-start}-\ref{line:while-loop-ends}) iterates $k$ times and each iteration (specifically, Line \ref{line:marginal}) runs in $O(c)$ time. Hence, as stated, the time complexity of the algorithm is $O(kc) + \sum_{\ell=1}^c \texttt{ENN}(D_\ell, q)$.   

    Next, we prove the optimality of the returned set $\Alg$. Let $\Opt \in \argmax_{S \subseteq P: \lvert S \rvert = k} \NSW (S)$ be an optimal solution with attribute counts $|\Opt \cap D_\ell|$ as close to $|\Alg \cap D_\ell|$ as possible. That is, among the optimal solutions, it is one that minimizes $\sum_{\ell=1}^c | k^*_\ell - k_\ell|$, where $k^*_\ell = |\Opt \cap D_\ell|$ and $k_\ell = |\Alg \cap D_\ell|$, for each $\ell \in [c]$. 
    We will prove that $\Opt$ satisfies $k^*_\ell = k_\ell$ for each $\ell \in [c]$. This guarantee, along with Lemma \ref{lemma:size-match}, implies that, as desired, $\Alg$ is a Nash optimal solution. 

    Assume, towards a contradiction, that $k^*_\ell \neq k_\ell$ for some $\ell \in [c]$. Since $|\Opt| = |\Alg| = k$, there exist attributes $x, y \in [c]$ with the property that  
    $    k^*_x  < k_x$ and $k^*_y  > k_y 
    $.
    For a given attribute $\ell \in [c]$, write the logarithm of cumulative similarity upto the $i^\text{th}$ most similar vector as
$F_\ell(i) \coloneqq \log \big( \sum_{j=1}^i \sigma(q, v^\ell_{(j)}) + \eta \big)$, where $v^\ell_{(j)}$ is defined in Line~\ref{line:definition-of-ith-most-similar-point} of~\Cref{algo:greedy-nash-ann}. 

    Next, note that for any attribute $\ell \in [c]$, if \Cref{algo:greedy-nash-ann}, at any point during its execution, has included $k'_\ell$ vectors of attribute $\ell$ in $\Alg$, then at that point the maintained utility $w_\ell = \sum_{j=1}^{k'_\ell} \sigma(q, v^\ell_{(j)})$. Hence, at the beginning of any iteration of the algorithm, if the $k'_\ell$ denotes the number of selected vectors of each attribute $\ell \in [c]$, then the marginals considered in Line \ref{line:marginal} are $F_\ell \left(k'_\ell+1 \right) - F_\ell \left(k'_\ell \right)$. These observations and the selection criterion in Line \ref{line:marginal} of the algorithm give us the following inequality for the counts $k_x=|\Alg \cap D_x|$ and $k_y=|\Alg \cap D_y|$ of the returned solution $\Alg$:
    \begin{align}
        F_x(k_x) - F_x(k_x - 1) \geq F_y(k_y+1) - F_y(k_y)  \label{ineq:greey-choice-main}
    \end{align}
        Specifically, equation (\ref{ineq:greey-choice-main}) follows by considering the iteration in which $k_x^{\text{th}}$ (last) vector of attribute $x$ was selected by the algorithm. Before that iteration the algorithm had selected $(k_x-1)$ vectors of attribute $x$, and let $k'_y$ denote the number of vectors with attribute $y$ that have been selected till that point. Note that $k'_y \leq k_y$. The fact that the $k_x^{\text{th}}$ vector was (greedily) selected in Line \ref{line:marginal}, instead of including an additional vector of attribute $y$, gives $F_x(k_x) - F_x(k_x - 1) \geq F_y(k'_y+1) - F_y(k'_y) \geq F_y(k_y+1) - F_y(k_y)$; here, the last inequality follows from~\Cref{lemma:decreasing-marginal-single-attribute}. 
        Therefore we have, 
        {\small
        \begin{align}
            F_x(k^*_x + 1) - F_x(k^*_x) \overset{\text{(i)}}{\geq} F_x(k_x) - F_x(k_x - 1) \overset{\text{(ii)}}{\geq} F_y(k_y+1) - F_y(k_y) \overset{\text{(iii)}}{\geq} F_y(k^*_y) - F_y(k^*_y-1)  \label{ineq:swap-main}
        \end{align}
        }
Here, inequality (i) follows from $k^*_x < k_x$ and~\Cref{lemma:decreasing-marginal-single-attribute}, inequality (ii) is due to equation (\ref{ineq:greey-choice-main}), and inequality (iii) is via $k^*_y > k_y$ and~\Cref{lemma:decreasing-marginal-single-attribute}. 
         
    Next, observe that the definition of $\widehat{D}_\ell$ ensures that $v^\ell_{(i)}$ is, in fact, the $i^{\text{th}}$ most similar (to $q$) vector among the ones that have attribute $\ell$, i.e., $i^{\text{th}}$ most similar in all of $D_\ell$. Since $\Opt$ is an optimal solution, the $k^*_\ell = |\Opt \cap D_\ell|$ vectors of attribute $\ell$ in $\Opt$ are the most similar $k^*_\ell$ vectors from $D_\ell$. That is, $\Opt \cap D_\ell = \{v^\ell_{(1)}, \ldots, v^\ell_{(k^*_\ell)} \}$, for each $\ell \in [c]$. This observation and the definition of $F_\ell(\cdot)$ imply that the logarithm of $\Opt$'s NSW satisfies 
$    \lognsw(\Opt) = \frac{1}{c} \sum_{\ell=1}^c F_\ell(k^*_\ell)
$. 
Now, consider a subset of vectors $S$ obtained from $\Opt$ by including vector $v^x_{(k^*_x + 1)}$ and removing $v^y_{(k^*_y)}$, i.e., $S = \left( \Opt \cup \left\{ v^x_{(k^*_x + 1)} \right\} \right) \setminus \left\{v^y_{(k^*_y)} \right\}$. Note that
{
\begin{align*}
    \lognsw(S) - \lognsw(\Opt) & = \frac{1}{c}\Big( F_x(k^*_x + 1) - F_x(k^*_x) \Big)  + \frac{1}{c}\Big(  F_y(k^*_y-1) - F_y(k^*_y) \Big) \geq 0, 
\end{align*}
}
where the last inequality is via equation~(\ref{ineq:swap-main}). Hence, $\NSW(S) \geq \NSW(\Opt)$. Given that $\Opt$ is a Nash optimal solution, the last inequality must hold with an equality, $\NSW(S) = \NSW(\Opt)$, i.e., $S$ is an optimal solution as well. This, however, contradicts the choice of $\Opt$ as an optimal solution that minimizes $\sum_{\ell=1}^c | k^*_\ell - k_\ell|$; note that $\sum_{\ell=1}^c \left|\widehat{k}_\ell - k_\ell \right| < \sum_{\ell=1}^c \left| k^*_\ell - k_\ell \right|$, where $\widehat{k}_\ell \coloneqq |S \cap D_\ell|$. 

Therefore, by way of contradiction, we obtain that $|\Opt \cap D_\ell| = |\Alg \cap D_\ell|$ for each $\ell \in [c]$. As mentioned previously, this guarantee along with Lemma \ref{lemma:size-match} imply that $\Alg$ is a Nash optimal solution. This completes the proof of the theorem.     
\end{proof}

\SingleAttributeCorollaryApproximateOracle*

\begin{proof}
    The running time of the algorithm follows via an argument similar to the one used in the proof of~\Cref{theorem:single-attribute}. Therefore, we only argue correctness here. 
    
    For every $\ell \in [c]$, let the $\alpha$-approximate oracle return $\widehat{D}_\ell$. Recall that $v^\ell_{(i)}$, $i \in [k]$, denotes the $i^{\text{th}}$ most similar point to $q$ in the set $\widehat{D}_\ell$. Further, for every $\ell \in [c]$, let $D^*_{\ell}$ be the set of $k$ most similar points to $q$ within $D_\ell$ and define $v^{*\ell}_{(i)}$, $i \in [k]$, to be the $i^{\text{th}}$ most similar point to $q$ in $D^*_{\ell}$. Recall that by the guarantee of the $\alpha$-approximate NNS oracle, we have $\sigma(q, v^\ell_{(i)}) \geq \alpha \cdot \sigma(q, v^{*\ell}_{(i)})$ for all $i \in [k]$. Let $\Opt$ be an optimal solution to the NaNNS problem containing $k^*_\ell$ most similar points of attribute $\ell$ for every $\ell \in [c]$.
    
    Finally, let $\widehat{\Opt}$ be the optimal solution to the NaNNS problem when the set of vectors to search over is $P = \cup_{\ell \in [c]} \widehat{D}_{\ell}$. 

    By an argument similar to the proof of~\Cref{theorem:single-attribute}, we have $\NSW(\Alg) = \NSW(\widehat{\Opt})$. Therefore
    {\allowdisplaybreaks
    \begin{align*}
        \NSW(\Alg) &= \NSW(\widehat{\Opt}) \\
        &\geq \left(\prod_{\ell \in [c]} \left(\sum_{i=1}^{k^*_\ell} \sigma(q, v^\ell_{(i)} ) + \eta \right)\right)^{\frac{1}{c}} \tag{$\bigcup_{\ell \in [c]: k^*_\ell \geq 1} \{v^\ell_{(1)}, \ldots, v^\ell_{(k^*_\ell)}\}$ is a feasible solution} \\
        &\geq \left(\prod_{\ell \in [c]} \left(\sum_{i=1}^{k^*_\ell} \alpha \sigma(q, v^{*\ell}_{(i)} ) + \eta \right)\right)^{\frac{1}{c}} \tag{by $\alpha$-approximate guarantee of the oracle; $k^*_{\ell} \leq k$} \\
        &\geq \left(\prod_{\ell \in [c]} \alpha \left(\sum_{i=1}^{k^*_\ell}  \sigma(q, v^{*\ell}_{(i)} ) + \eta \right)\right)^{\frac{1}{c}} \tag{$\alpha \in (0, 1)$} \\
        &= \alpha\ \NSW(\Opt) \tag{definition of $\Opt$}
    \end{align*}
    }
    Hence, the corollary stands proved.
\end{proof}

We complete this section by stating the proofs of Lemmas~\ref{lemma:decreasing-marginal-single-attribute} and~\ref{lemma:size-match}.
\begin{proof}[Proof of~\Cref{lemma:decreasing-marginal-single-attribute}]
    Note that $\exp(F_\ell(i)) = \sum_{j=1}^i \sigma(q, v^\ell_{(j)}) + \eta = \exp(F_\ell(i-1)) + \sigma(q, v^\ell_{(i)})$. Therefore, we have
    \begin{align}
    \exp(F_\ell(i) - F_\ell(i-1)) = \frac{\exp(F_\ell(i))}{\exp(F_\ell(i-1))} = \frac{\exp(F_\ell(i-1)) + \sigma(q, v^\ell_{(i)})}{\exp(F_\ell(i-1))} = 1 + \frac{\sigma(q, v^\ell_{(i)})}{\exp(F_\ell(i-1))} \label{ineq:exp-exp}
    \end{align}
    Similarly, we have $\exp(F_\ell(i') - F_\ell(i'-1)) = 1 + \frac{\sigma(q, v^\ell_{(i')})}{\exp(F_\ell(i'-1))}$. 
    
    In addition, the indexing of the vectors $v^\ell_{(j)}$ ensures that $\sigma(q, v^\ell_{(i')}) \geq \sigma(q, v^\ell_{(i)})$ for $i'<i$. Moreover, $\exp(F_\ell(i))$ is non-decreasing since it is the cumulative sum of non-negative similarities upto $i^{\text{th}}$ vector $v^\ell_{(i)}$. Hence, $\exp(F_\ell(i)) \geq \exp(F_\ell(i'))$ for $i' < i$. Combining these inequalities, we obtain 
    \begin{align*}
        \frac{\sigma(q, v^\ell_{(i')})}{\exp(F_\ell(i'-1))} \geq \frac{\sigma(q, v       ^\ell_{(i)})}{\exp(F_\ell(i-1))}.
    \end{align*}
    That is, $1 + \frac{\sigma(q, v^\ell_{(i')})}{\exp(F_\ell(i'-1))} \geq 1 + \frac{\sigma(q, v^\ell_{(i)})}{\exp(F_\ell(i-1))}$. Hence, equation (\ref{ineq:exp-exp}) gives us 
    \begin{align}
         \exp(F_\ell(i') - F_\ell(i'-1)) \geq \exp(F_\ell(i) - F_\ell(i - 1)) \label{ineq:exp-ineq}
    \end{align}
Since $\exp(\cdot)$ is an increasing function, inequality (\ref{ineq:exp-ineq}) implies
    \begin{align*}
    F_\ell(i') - F_\ell(i'-1) \geq F_\ell(i) - F_\ell(i-1)~.
    \end{align*}
    The lemma stands proved. 
\end{proof}

\begin{proof}[Proof of~\Cref{lemma:size-match}]
Assume, towards a contradiction, that there exists a subset of input vectors $S$ that satisfies $|S \cap D_\ell| = |\Alg \cap D_\ell|$, for each $\ell \in [c]$, and still induces NSW strictly greater than that of $\Alg$. This strict inequality implies that there exists an attribute $a \in [c]$ with the property that the utility $u_a(S) > u_a(\Alg)$.%
\begin{align}
\sum_{t \in S \cap D_a} \sigma(q, t) > \sum_{v \in \Alg \cap D_a} \sigma(q, v) \label{ineq:strict-util}    
\end{align}
On the other hand, note that the construction of \Cref{algo:greedy-nash-ann} and the definition of $\widehat{D}_a$ ensure that the vectors in $\Alg \cap D_a$ are, in fact, the most similar to $q$ among all the vectors in $D_a$. This observation and the fact that $|S\cap D_a| = |\Alg \cap D_a|$ gives us $\sum_{v \in \Alg \cap D_a} \sigma(q, v) \geq \sum_{t \in S \cap D_a} \sigma(q, t)$. This equation, however, contradicts the strict inequality (\ref{ineq:strict-util}). 

Therefore, by way of contradiction, we obtain that there does not exist a subset $S$ such that $|S \cap D_\ell| = |\Alg \cap D_\ell|$, for each $\ell \in [c]$, and $\NSW(\Alg) < \NSW(S)$. The lemma stands proved. 
\end{proof}

\section{NaNNS in the Multi-Attribute Setting}
\label{sec:multi-attribute-details}

In this section, we first establish the NP-Hardness of the NaNNS problem in the multi-attribute setting by stating the proof of~\Cref{theorem:multi-attribute-hardness}. Thereafter, we describe an efficient algorithm that obtains a constant approximation in terms of the logarithm of the Nash Social Welfare objective, and prove the approximation ratio (\Cref{theorem:multi-attribute-greedy-guarantee}).

\subsection{Proof of \Cref{theorem:multi-attribute-hardness}}
\label{subsec:proof-of-np-hardness-of-multi-attribute}

Recall that in the multi-attribute setting, input vectors $v \in P$ are associated with one or more attributes, $\lvert \col(v) \rvert \geq 1$. 

\TheoremNPHardness*
\begin{proof}
    Consider the decision version of the optimization problem: given a threshold $W \in \mathbb{Q}$, decide whether there exists a size-$k$ subset $S \subseteq P$ such that $\lognsw(S) \geq W$. We will refer to this problem as \texttt{NaNNS}. Note that the input in an \texttt{NaNNS} instance consists of: a set of $n$ vectors $P \subset \mathbb{R}^d$, a similarity function $\sigma: \mathbb{R}^d \times \mathbb{R}^d \to \mathbb{R}_+$, an integer $k \in \mathbb{N}$, the sets $D_\ell = \{p \in P: \ell \in \col(p)\}$ for every attribute $\ell \in [c]$, a query point $q \in \mathbb{R}^d$, and threshold $W \in \mathbb{Q}$. We will  show that \texttt{NaNNS} is NP-complete by reducing \textsc{Exact Regular Set Packing} (\ersp) to it. \ersp\ is known to be NP-complete \cite{gareyandjohnson} and is also W[1]-hard with respect to solution size~\cite{AUSIELLO1980136}.
    
    In \ersp, we are given a universe of $n$ elements, $\mathcal{U} = \{1, 2, \ldots, n\}$, an integer $k \in \mathbb{N}$, and a collection of $m$ subsets $\mathcal{S} = \{S_1, \ldots, S_m \}$, with each subset $S_i \subseteq \mathcal{U}$ of cardinality $\tau$ (i.e., $\lvert S_i \rvert = \tau$ for each $i \in [m]$). The objective here is to decide whether there exists a size-$k$ sub-collection $I \subseteq \mathcal{S}$ such that for all distinct $S, S' \in I$ we have $S \cap S' = \emptyset$.

    For the reduction, we start with the given instance of \ersp~and construct an instance of \texttt{NaNNS}: Consider $\mathcal{U}$ as the set of attributes, i.e., set $c= n$. In addition, we set the input vectors $P = \left\{\frac{1}{\tau} \ \mathbf{1}_S\ \vert\ S \in \mathcal{S} \right\}$; here, $\mathbf{1}_S \in \mathbb{R}^n$ is the characteristic vector (in $\mathbb{R}^n$) of the subset $S$, i.e., for each $i \in [n]$, the $i$-th coordinate of $\mathbf{1}_S$ is $\mathds{1}\{i \in S\}$. Note that the set of vectors $P$ is of cardinality $m$.
    
    Furthermore, we set the query vector $q = \mathbf{1}$ as the all-ones vector in $\mathbb{R}^n$. In this \texttt{NaNNS} instance, each input vector $\frac{1}{\tau} \ \mathbf{1}_S$ is assigned attribute $\ell \in [n]$ iff element $\ell \in S$. That is, $D_\ell = \{\frac{1}{\tau} \ \mathbf{1}_S \ \vert\ S \in \mathcal{S} \text{ and } \ell \in S\}$. The number of neighbors to be found in the constructed \texttt{NaNNS} is equal to $k$, which is the count in the given the \ersp\ instance. Also, the similarity function $\sigma: \mathbb{R}^n \times \mathbb{R}^n \to \mathbb{R}$ is taken to be the standard dot-product. Finally, we set the threshold $W = \frac{\tau k \log 2}{c}$. 
    
    Note that the reduction takes time polynomial in $n$ and $m$. In addition, for each input vector $v \in P$ it holds that $v = \frac{1}{\tau} \cdot \mathbf{1}_S$ for some $S \in \mathcal{S}$ and, hence, $\sigma(q, v) = \langle \frac{1}{\tau} \ \mathbf{1}_S, \mathbf{1} \rangle = 1$.

    Now we establish the correctness of the reduction.

    \noindent
    Forward direction ``$\Rightarrow$": Suppose the given \ersp\ instance admits a (size-$k$) solution $I^* \subset \mathcal{S}$. Consider the subset of vectors $N^* \coloneqq \{\frac{1}{\tau} \ \mathbf{1}_S \vert S \in I^* \}$. Indeed, $N^* \subseteq P$ and $\lvert N^* \rvert = k$, hence $N^*$ is a feasible set of the \texttt{NaNNS} problem. Now, since $I^*$ is a solution to the \ersp\ instance, for distinct $S$, $S' \in I^*$ we have $S \cap S' = \emptyset$. In particular, if for an element $\ell \in [c]$, we have $\ell \in S$ for some $S \in I^*$, then $\ell \notin S'$ for all $S' \in I^* \setminus \{S\}$. Therefore, $\lvert N^* \cap D_\ell \rvert \leq 1$ for all $\ell \in [c]$, which in turn implies that $u_\ell(N^*)$ is either $1$ or $0$ for every $\ell \in [c]$. Finally, note that each vector $v \in P$ belongs to exactly $\tau$ attributes, i.e., $\lvert \col(v) \rvert = \tau$. Hence,
    {\allowdisplaybreaks
    \begin{align*}
        \lognsw (N^*) = \frac{1}{c} \sum_{\ell=1}^c \log(1 + u_\ell(N^*)) = \frac{1}{c} \sum_{v \in N^*} \sum_{\ell \in \col(v)} \log(1 + 1) = \frac{\tau k \log 2}{c}~.
    \end{align*}
    }
    Therefore, if the given \ersp\ instance admits a solution (exact packing), then the constructed \texttt{NaNNS} instance has $k$ neighbors with sufficiently high $\lognsw$.

\noindent
    Reverse direction ``$\Leftarrow$": Suppose $N^* \subseteq P$ is a solution in the constructed \texttt{NaNNS} instance with $\lvert N^* \rvert = k$ and $\lognsw(N^*) \geq W$. Define $I^* \coloneqq \{S \ \vert\ \frac{1}{\tau} \cdot \mathbf{1}_S \in N^*\}$ and note that $\lvert I^* \rvert = k$. We will show that $I^*$ is a solution for the given \ersp\ instance, i.e., it consists of disjoint subsets. 
    
    Towards this, first note that $N^*$ induces social welfare: 
    {\allowdisplaybreaks
    \begin{align}
    \sum_{\ell \in [c]} u_\ell(N^*) = \sum_{\ell \in [c]} \sum_{v \in N^* \cap D_\ell} \sigma(q, v) = \sum_{v \in N^*} \sum_{\ell \in \col(v)} \sigma(q,v) = \tau k \label{eqn:sw}
    \end{align}
    }

    Furthermore, any attribute $\ell \in [c]$ has a non-zero utility under $N^*$ iff $\ell \in S$ for some subset $S \in I^*$. Hence, $\mathcal{A} \coloneqq \cup_{S \in I^*} S$ corresponds to the set of attributes with non-zero utility under $N^*$. We have $1 \leq \lvert \mathcal{A} \rvert \leq \tau k$. Next, using the fact that $\lognsw(N^*) \geq W =\frac{\tau k \log 2}{c}$ we obtain 
    {\allowdisplaybreaks
    \begin{align*}
        \frac{\tau k \log 2}{c} \leq \lognsw(N^*) &= \frac{1}{c} \sum_{\ell \in [c]} \log(1 + u_\ell(N^*)) \\
        &= \frac{1}{c} \sum_{\ell \in \mathcal{A}} \log(1 + u_\ell(N^*)) \\
        &= \frac{\lvert \mathcal{A} \rvert}{c} \  \frac{1}{\lvert \mathcal{A} \rvert} \sum_{\ell \in \mathcal{A}} \log(1 + u_\ell(N^*)) \\
        &\leq \frac{\lvert \mathcal{A} \rvert}{c} \  \log\left(\frac{1}{\lvert \mathcal{A} \rvert}\sum_{\ell \in \mathcal{A}} \left(1 + u_\ell(N^*) \right) \right) \tag{concavity of $\log$} \\
        &= \frac{\lvert \mathcal{A} \rvert}{c} \  \log\left(1 + \frac{\sum_{\ell \in \mathcal{A}} u_\ell(N^*)}{\lvert \mathcal{A} \rvert} \right) \\
        &= \frac{\lvert \mathcal{A} \rvert}{c} \  \log\left(1 + \frac{\tau k}{\lvert \mathcal{A} \rvert} \right) \tag{via (\ref{eqn:sw})} \\
        &\leq \frac{\tau k \log 2}{c}. %
    \end{align*}}
    Here, the last inequality follows from \Cref{lemma:log-inequality-technical} (stated and proved below). 
    Hence, all the inequalities in the derivation above must hold with equality. In particular, we must have $\lvert \mathcal{A} \rvert = \tau k$ by the quality condition of~\Cref{lemma:log-inequality-technical}. Therefore, for distinct sets $S, S' \in I^*$ it holds that $S \cap S' = \emptyset$. Hence, as desired, $I^*$ is a solution of the \ersp\ instance.    

    This completes the correctness of the reduction, and the theorem stands proved. 
\end{proof}

\begin{restatable}{lemma}{LogInequality}
\label{lemma:log-inequality-technical}
    For any $a > 0$ and for all $x \in (0, a]$ it holds that $x \log(1 + \frac{a}{x}) \leq a \log 2$. Furthermore, the equality holds iff $x = a$.
\end{restatable}
\begin{proof}
    Write $f(x) \coloneqq x \log(1 + \frac{a}{x})$. At the end points of the domain $(0, a]$, the function $f(\cdot)$ satisfies: 
    $f(a) = a \log(2)$ and $$\lim_{x \to 0^+} f(x) = \lim_{x \to 0^+} x \log(a+x) - x \log x = \lim_{x \to 0^+} x \log(a + x) - \lim_{x \to 0^+} x \log(x) = 0 - 0 = 0.$$
    Note that $f'(x) = \log(1 + \frac{a}{x}) - \frac{a}{a+x}$. We will show that $f'(x) > 0$ for all $x \in (0, a]$ which will conclude the proof.

\noindent
    \textbf{Case 1}: $x \in (0, \frac{a}{2}]$. We have $\log(1 + \frac{a}{x}) \geq \log(1 + \frac{a}{a/2}) = \log(3) > 1$. On the other hand, $\frac{a}{a + x} \leq 1$.

\noindent
    \textbf{Case 2}: $x \in (\frac{a}{2}, a]$. In this case, $\log(1 + \frac{a}{x}) \geq \log(1 + \frac{a}{a}) = \log(2) > 0.693$. However, $\frac{a}{a+x} < \frac{a}{a + \frac{a}{2}} = \frac{2}{3} \leq 0.667$.

    Therefore, $f'(x) = \log(1 + \frac{a}{x}) - \frac{a}{a + x} > 0$ for all $x \in (0,a]$, which completes the proof.
\end{proof}

\subsection{Algorithm for the Multi-Attribute Setting}
This section details \Cref{algo:greedy-submodular}, based on which we obtain \Cref{theorem:multi-attribute-greedy-guarantee}. The algorithm greedily selects the $k$ neighbors for the given query $q$. Specifically, the algorithm iterates $k$ times, and in each iteration, it selects a new vector (from the given set $P$) whose inclusion in the current solution $\Alg$ yields that maximum increase in $\lognsw$ (Line~\ref{line:find-maximal-log-nsw-vector}). After $k$ iterations, the algorithm returns the $k$ selected vectors $\Alg$.

\begin{algorithm}[h]
\caption{\texttt{MultiNashANN}: Approximation algorithm in the multi-attribute setting}
\label{algo:greedy-submodular}

\KwIn{Query $q \in \mathbb{R}^d$, and the set of input vectors $P \subset\mathbb{R}^d$}

Initialize $\Alg = \emptyset$.

\For{$i = 1$ to $k$}{
    Set $\widehat{v} = 
    \argmax_{v \in P \setminus \Alg} \big( 
    \lognsw(\Alg \cup \{v\}) \ -  \ \lognsw (\Alg) \big)$.
    \label{line:find-maximal-log-nsw-vector}

    Update $\Alg \gets \Alg \cup \{\widehat{v}\}$.
    \label{line:add-maximal-log-nsw-vector}
}

\Return{$\Alg$}

\end{algorithm}

\subsection{Proof of \Cref{theorem:multi-attribute-greedy-guarantee}}
\label{subsec:submodular-greedy-approx-guarantee}
We now establish~\Cref{theorem:multi-attribute-greedy-guarantee}, which provides the approximation ratio achieved by ~\Cref{algo:greedy-submodular}.

\TheoremMultiAttributeApprox*
\begin{proof}
For each subset $S \subseteq P$, write function $f(S) \coloneqq \lognsw(S)$. Since parameter $\eta =1$, we have $f(\emptyset) = 0$ and the function is nonnegative. Moreover, we will show that this set function $f: 2^P \to \mathbb{R}_+$ is monotone and submodular. Given that \Cref{algo:greedy-submodular} follows the marginal-gain greedy criterion in each iteration, it achieves a $(1 - \frac{1}{e})$-approximation for the submodular maximization problem (\ref{optimization-problem:log-nsw}).

To establish the monotonicity of $f$, consider any pair of subsets $S \subseteq T \subseteq P$. Here, for each $\ell \in [c]$, we have $D_\ell \cap S \subseteq D_\ell \cap T$. Hence, $u_\ell(S) \leq u_\ell(T)$. Further, since $\log$ is an increasing function, it holds that $\log(u_\ell(S) + 1) \leq \log(u_\ell(T) + 1)$, for each $\ell \in [c]$. Hence, $f(S) \leq f(T)$, and we obtain that $f$ is monotone.

For submodularity, let $S \subseteq T \subseteq P$ be any two subsets and let $w \in P \setminus T$. Write $S + w$ and $T + w$ to denote the sets $S \cup \{w\}$ and $T \cup \{w\}$, respectively. Here, we have
    {\allowdisplaybreaks
    \begin{align*}
        &f(S + w) - f(S) - f(T + w) + f(T) \\
        &= \frac{1}{c} \sum_{\ell \in [c]} \log\left(\frac{1 + \sum_{v \in D_\ell \cap (S + w)} \sigma(q, v)}{1 + \sum_{v \in D_\ell \cap S} \sigma(q, v)} \cdot \frac{1 + \sum_{v \in D_\ell \cap T} \sigma(q, v)}{1 + \sum_{v \in D_\ell \cap (T + w)} \sigma(q, v)}\right) \\
        &= \frac{1}{c} \sum_{\ell \in \col(w)} \log\left(\left(1 + \frac{\sigma(q, w)}{1 + \sum_{v \in D_\ell \cap S} \sigma(q, v)} \right) \cdot \left(1 + \frac{\sigma(q, w)}{1 + \sum_{v \in D_\ell \cap T} \sigma(q, v)}\right)^{-1}\right) \\
        &= \frac{1}{c} \sum_{\ell \in \col(w)} \log\left(\left(1 + \frac{\sigma(q, w)}{1 + u_\ell(S)} \right) \cdot \left(1 + \frac{\sigma(q, w)}{1 + u_\ell(T)}\right)^{-1}\right) \\
        &\geq 0 \tag{$u_\ell(S) \leq u_\ell(T)$ for $S \subseteq T$}
    \end{align*}
    }
    Therefore, upon rearranging, we obtain $f(S + w) - f(S) \geq f(T + w) - f(T)$, i.e., $f$ is submodular. Hence,~\Cref{algo:greedy-submodular}, which follows marginals-gain greedy selection, achieves a $(1 - \frac{1}{e})$-approximation \citep{analysis1978nemhauser} for the optimization problem (\ref{optimization-problem:log-nsw}).
\end{proof}

\section{Experimental Evaluations}
\label{section:experiments}
In this section, we validate the welfare-based formulations and the performance of our proposed algorithms against existing methods on a variety of real and semi-synthetic datasets. We perform three sets of experiments:

\begin{itemize}
    \item In the first set of experiments (\Cref{fig:mainPaperResults-NaNNS}), we compare \ouralgo ~(\Cref{algo:greedy-nash-ann}) with prior work on hard-constraint based diversity \citep{anand2025graphbased}. Here, we show that \ouralgo\ strikes a balance between relevance and diversity both in the single- and multi-attribute settings.
    \item In the second set of experiments (\Cref{fig:mainPaperResults-p-NNS}), we study the effect of varying the exponent parameter $p$ in the $p$-NNS objective on relevance and diversity, in both single- and multi-attribute settings.
    \item In the final set of experiments (Tables~\ref{tab:sift20-k50} and~\ref{tab:QPSandLatency}), we compare our algorithm, \ouralgo~(with provable guarantees), and a heuristic we propose to improve the runtime of \ouralgo. The heuristic directly utilizes a standard \diskann ~algorithm to first fetch a sufficiently large candidate set of vectors (irrespective of their attributes). Then, it applies the greedy procedure for Nash social (or $p$-mean) welfare maximization (similar to Lines \ref{line:while-loop-start}-\ref{line:while-loop-ends} in~\Cref{algo:greedy-nash-ann}) only within this set.
\end{itemize}
 In what follows, we provide the details of the experimental set-ups, the baseline algorithms and the results of the experiments.
 
Additional plots for the experiments appear in Appendices \ref{appendix:exp-single-attribute} and \ref{appendix:exp-multi-attribute}. %

\subsection{Metrics for Measuring Relevance and Diversity}
\label{subsec:description-various-metrics--of-diversity-and-relevance}

 \noindent \textbf{Relevance Metrics}: In our experiments, we capture the relevance of a solution to the query through two metrics detailed below.

\begin{enumerate}
    \item \textbf{Approximation Ratio}: For a given query $q$, let $S$ be the set of $k$ vectors returned by an NNS algorithm that we wish to study, and let $O$ be the $k$ most similar vectors to $q$ in $P$. Then the approximation ratio of the algorithm is defined as the ratio
    $
        \frac{\sum_{v \in S} \sigma(q,v)}{\sum_{v \in O} \sigma(q,v)}
    $. Therefore, a higher approximation ratio indicates a more relevant solution. Note that the highest possible value of this metric is $1$.
    
    \item \textbf{Recall}: For a given query $q$, let $S$ be the set of $k$ vectors returned by an NNS algorithm that we wish to study and let $O$ be the $k$ most similar vectors to $q$ in $P$. The recall of the algorithm is defined as the quantity $\frac{\lvert S \cap O \rvert}{\lvert O \rvert}$. Therefore, higher the recall, more relevant the solution, and the maximum possible value of recall is $1$.
\end{enumerate}

\begin{remark}
\label{remark:approx-ratio-better-metric-than-recall}
Although recall is a popular metric in the context of the standard NNS problem, it is important to note that it is a fragile metric when the objective is to retrieve a relevant and diverse set of vectors for a given query. This can be illustrated with the following stylized example in the single-attribute setting. Suppose for a given query $q$, all the vectors in the similarity-wise optimal set $O$ have similarity $1$ and share the same attribute $\ell^* \in [c]$, i.e., for each $u \in O$ we have $\sigma(q, u) = 1$ and $\col(u) = \ell^*$. That is, the set $O$ of the $k$ most similar vectors to $q$ are not at all diverse. However, it is possible to have another set $S$ of $k$ vectors each with a distinct attribute and $\sigma(q, v) = 0.99$ for each $v \in S$. Such a set provides a highly relevant set of vectors that are also completely diverse. However, for the set $S$, the recall is actually $0$ (since $S \cap O = \emptyset$), but the approximation ratio is $0.99$. Hence, in the context of diverse nearest neighbor search problem, approximation ratio may be a more meaningful relevance metric than recall. \\
\end{remark}

\noindent \textbf{Diversity Metrics}: To measure the diversity of the solutions obtained by various algorithms, we consider the following metrics.

\begin{itemize}
    \item \textbf{Entropy}: Let $S \subseteq P$ be a size-$k$ subset  computed by an algorithm. Then the entropy of the set $S$ in the single-attribute setting is given by the quantity $ \sum_{\ell \in [c]: p_\ell > 0} -p_\ell \log(p_\ell)$ where $p_\ell = \frac{\lvert S \cap D_\ell \rvert}{\lvert S \rvert}$. Note that a higher entropy value indicates greater diversity. Moreover, it is not hard to see that the highest possible value of entropy is $\log(k)$ (achieved when $S$ contains at most $1$ vector from each attribute).

    \item \textbf{Inverse Simpson Index}: For a given set $S \subseteq P$  in the single-attribute setting, the inverse Simpson index is defined as $\frac{1}{\sum_{\ell = 1}^c p_\ell^2}$ where $p_\ell$ is the same as in the definition of entropy above. A higher value of this metric indicates greater diversity.

    \item \textbf{Distinct Attribute Count}: In the single-attribute setting, the distinct attribute count of a set $S \subseteq P$ is the number of different attributes that have at least one vector in $S$, i.e., the count is equal to $\lvert \{ \ell \in [c] \mid \lvert S \cap D_\ell \rvert > 0 \}\rvert$.
\end{itemize}

Note that the diversity metrics defined above are for the single-attribute setting. In the multi-attribute setting, in our experiments, we focus on settings where the attribute set $[c]$ is partitioned into $m$ sets $\{C_i \}_{i=1}^m$ (i.e., $[c] = \sqcup_{i=1}^m C_i$) and every input vector $v \in P$ is associated with exactly one attribute from each $C_i$. In particular, $\lvert \col(v) \rvert = m$ and $\lvert \col(v) \cap C_i \rvert = 1$ for each $1 \leq i \leq m$. We call each $C_i$ an attribute class. To measure diversity in the multi-attribute setting, we consider the aforementioned diversity metrics like entropy and inverse Simpson index restricted to an attribute class $C_i$. For instance, the entropy a set $S \subseteq P$ restricted to a particular $C_i$ is given by $ \sum_{\ell \in C_i: p_\ell > 0} -p_\ell \log(p_\ell)$, where $p_\ell = \frac{\lvert S \cap D_\ell \rvert}{\lvert S \rvert}$. Similarly, the inverse Simpson index of a set $S \subseteq P$ restricted to $C_i$ is given by $ \frac{1}{ \sum_{\ell \in C_i} p_\ell^2}$.

\subsection{Experimental Setup and Datasets}
\label{subsec:experiment-set-up-and-dataset}
\noindent \textbf{Hardware Details.} All the experiments were performed in memory on an Intel(R) Xeon(R) Silver $4314$ CPU ($64$ cores,  $2.40$GHz) with $128$ GB RAM. We set the number of threads to $32$. 

\noindent \textbf{Datasets.} We report results on both semi-synthetic and real-world datasets consistent with prior works \citep{anand2025graphbased}. These are summarized in \Cref{tab:datasets} and detailed below.

\renewcommand{\arraystretch}{1.1}
\begin{table*}[t!]
\centering
\caption{Summary of considered datasets. For synthetic attributes, we use two strategies: clustering-based (suffixed by \texttt{Clus}) and distribution-based (suffixed by \texttt{Prob}), see \Cref{subsec:experiment-set-up-and-dataset} for details.}
\label{tab:datasets}
\scalebox{0.9}{
\begin{tabular}{|l|c|c|c|c|c|}
\hline
\textbf{Dataset} & \textbf{\# Input Vectors} & \textbf{\# Query Vectors} & \textbf{Dimension} & \textbf{Attributes} \\
\hline
\amazon & $92,092$ & $8,956$ & $768$ & product color  \\
\hline
\arxiv  & $200,000$  & $50,000$ & $1536$ & year, paper category   \\
\hline
\texttt{Sift1m} & $1,000,000$ & $10,000$ & $128$ & synthetic   \\
\hline
\texttt{Deep1b} & $9,990,000$ & $10,000$ & $96$ & synthetic   \\
\hline
\end{tabular}
}
\end{table*}

\begin{enumerate}

\item  \textbf{Amazon Products Dataset} (\amazon):  The dataset, also known as the Shopping Queries Image Dataset (SQID)~\citep{amazonEmbedding}, 
includes vector embeddings of about $190,000$ product images and about $9,000$ text queries by users. The  embeddings of both product images and query texts are obtained via OpenAI's CLIP model~\citep{radford2021learning}, which maps both images and texts to a shared vector space. Given this dataset, our task is to retrieve relevant and diverse product images for a given text query. SQID also contains metadata for every product image, such as product image url, product id, product description, product title, and product color. The dataset is publicly available on Hugging Face platform.\footnote{\href{https://huggingface.co/datasets/crossingminds/shopping-queries-image-dataset}{https://huggingface.co/datasets/crossingminds/shopping-queries-image-dataset}}

For this dataset, we choose the set of all possible product colors in the dataset as our set of attributes $[c]$. We noted that for a lot of products, the color of the product in its image did not match the product color in the associated metadata. Hence, to associate a clean attribute (color) to each vector (product) in the dataset, we use the associated metadata as follows: we assign to the vector the majority color among the colors listed in the product color, product description, and title of the product. In case of a tie, we assign a separate attribute (color) called `color\_mix’. Further, we remove from consideration product images whose metadata does not contain any valid color names. The processed dataset contains $92,092$ vector embeddings of product images and constitutes our set $P$. Note that the dataset exhibits a skewed color distribution, shown in Figure~\ref{fig:amazonColorDistri}, with dominant colors such as black and white. No processing is applied to the query set which contains $8,956$ vectors. The vector embeddings of both images and queries are $d = 768$ dimensional. We use $\sigma(u,v) = 1 + \frac{u^\top v}{\lVert u \rVert \cdot \lVert v \rVert}$ as the similarity function between two vectors $u$ and $v$. Since the CLIP model was trained using the cosine similarity metric in the loss function (see~\citep{amazonEmbedding}, Section 4.2), this similarity function is a natural choice for the \amazon~dataset. 

\begin{figure}
    \centering
    \includegraphics[width=0.27\linewidth]{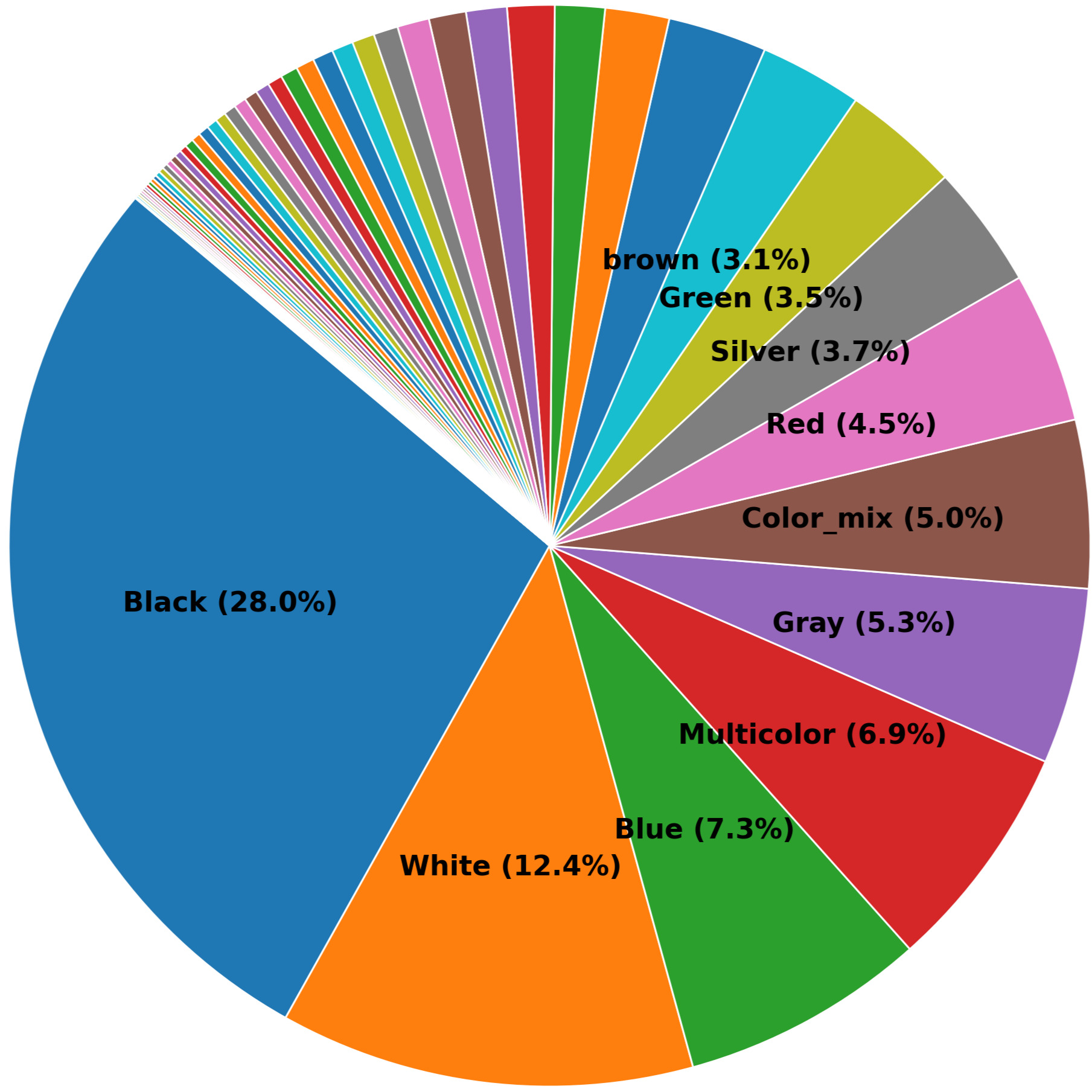}
    \caption{Distribution of product colors in the processed (cleaned) \amazon\ dataset. }
    \label{fig:amazonColorDistri}
\end{figure}

\item  \textbf{ArXiv OpenAI Embedding} (\arxiv): This dataset published by Cornell University consists of vector embeddings and metadata for approximately $250,000$ machine learning papers from arXiv.org~\citep{arxivEMbd}. The vector embeddings of the papers were generated via OpenAI’s \texttt{text-embedding-ada-002} model, executed on each paper’s title, authors, year, and abstract. The dataset is publicly available on Kaggle~\citep{arxivEMbd}.\footnote{\href{https://www.kaggle.com/datasets/awester/arxiv-embeddings}{https://www.kaggle.com/datasets/awester/arxiv-embeddings}} Note that there are no user queries prespecified in this dataset.

We use this dataset to study both the single- and multi-attribute setting. For the single-attribute setting, we only consider the year in which a paper was last updated as its attribute. For the multi-attribute setting, the paper's update year and its arXiv category are the two associated attributes. Therefore, we have $m=2$ attribute classes: update years and arXiv categories. Specifically, for the experiments we only consider papers with update-year between $2012$ and $2025$ and belonging to one or more of the following arXiv categories: \texttt{cs.ai}, \texttt{math.oc}, \texttt{cs.lg}, \texttt{cs.cv}, \texttt{stat.ml}, \texttt{cs.ro}, \texttt{cs.cl}, \texttt{cs.ne}, \texttt{cs.ir}, \texttt{cs.sy}, \texttt{cs.hc}, \texttt{cs.cr}, \texttt{cs.cy}, \texttt{cs.sd}, \texttt{eess.as}, and \texttt{eess.iv}.

\begin{figure}[h!]
    \centering
    \begin{subfigure}[b]{0.48\textwidth}
        \centering
        \includegraphics[width=0.6\linewidth]{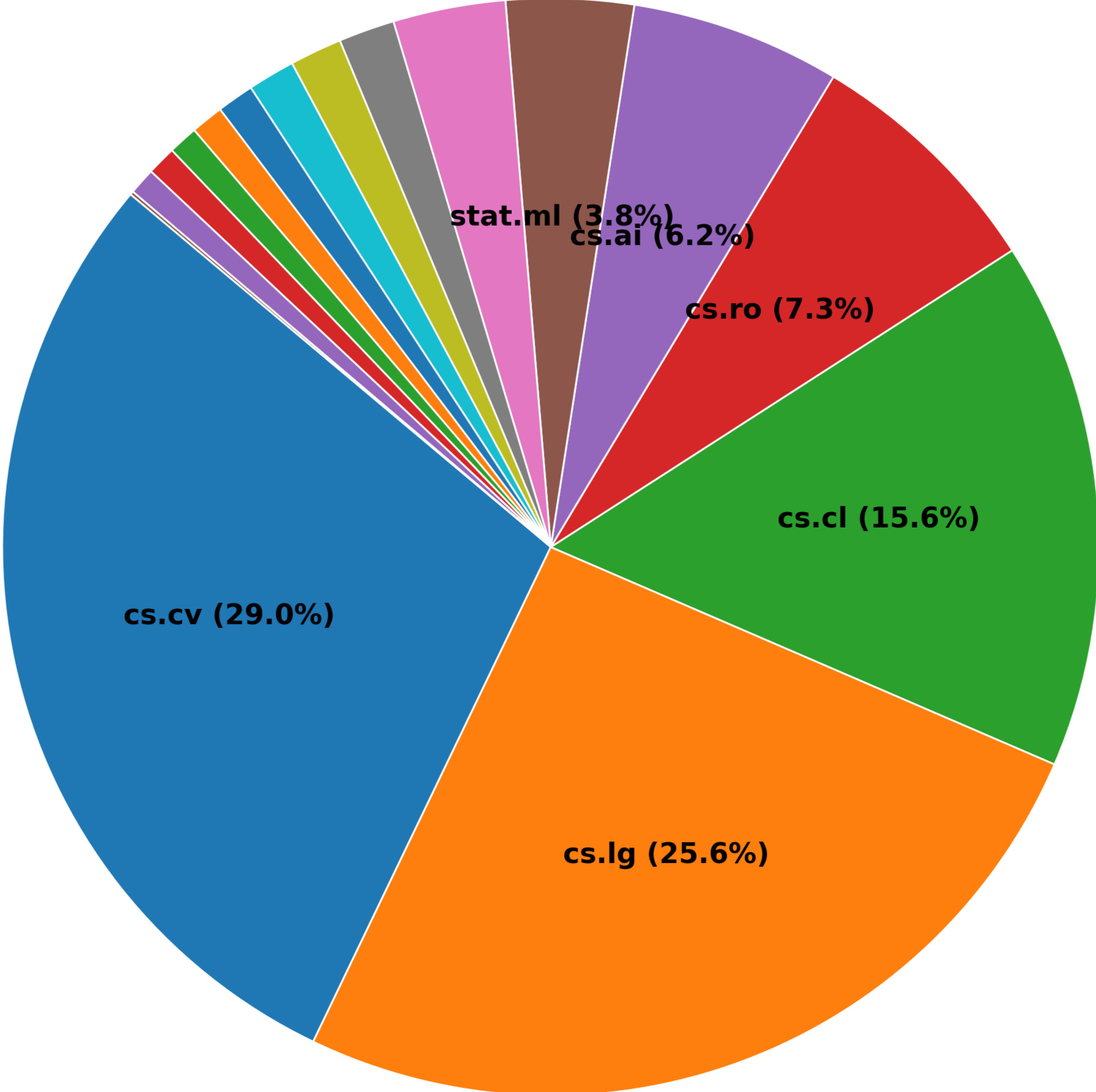}
        \caption{}
        \label{fig:sub1}
    \end{subfigure}%
    \hfill
    \begin{subfigure}[b]{0.48\textwidth}
        \centering
        \includegraphics[width=0.6\linewidth]{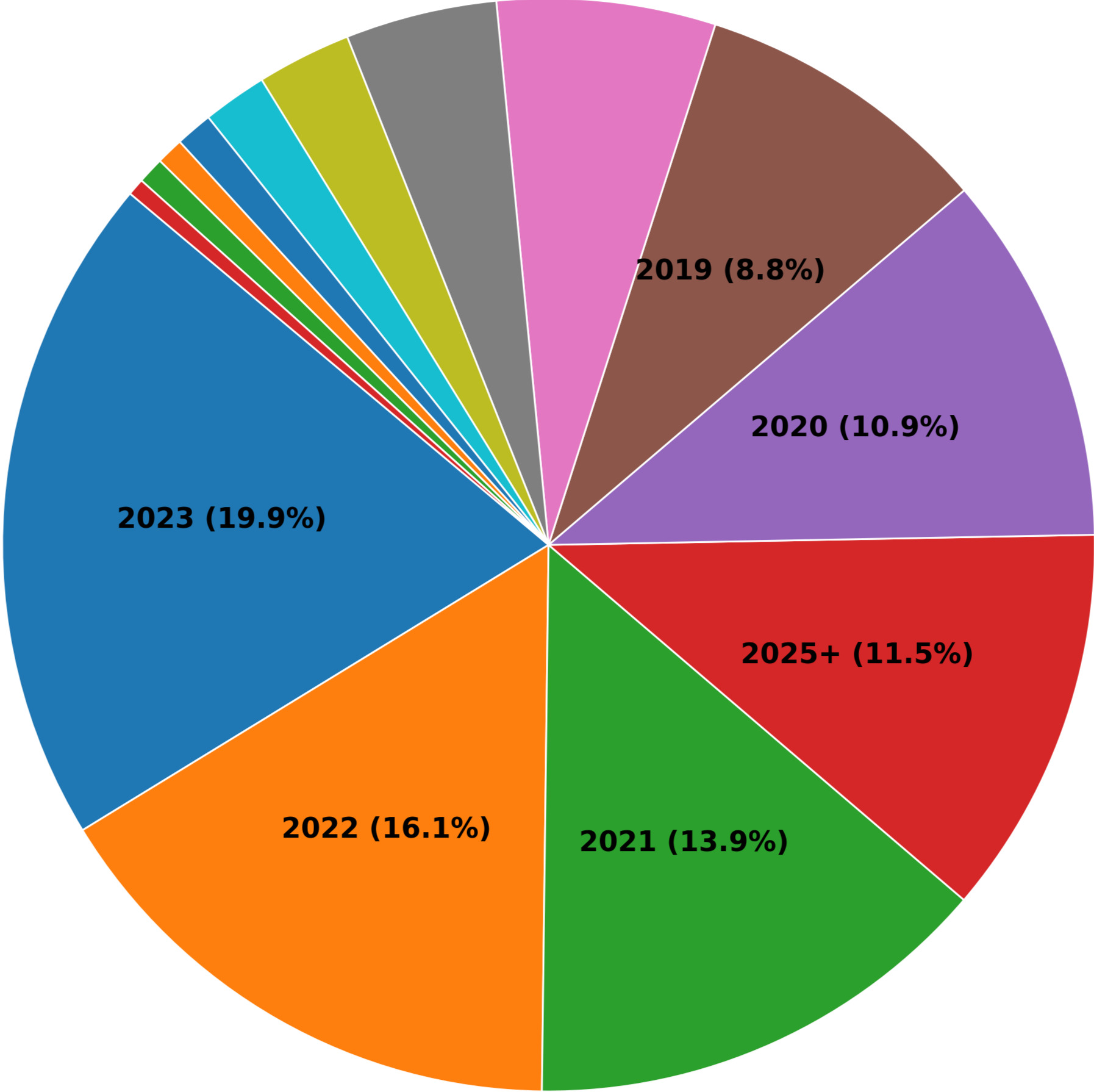}
        \caption{}
        \label{fig:sub2}
    \end{subfigure}
    \caption{Distribution of (a) arXiv categories (b) last update year in the \arxiv\ dataset.}
    \label{fig:arxivDistri}
\end{figure}

Since this dataset does not contain  predefined queries, we randomly split the dataset in $4:1$ ratio, and use the larger part as the set $P$ and the smaller part as the queries. Here, by using the vector embeddings of papers themselves as queries, we aim to simulate the task of finding papers similar to a given query paper. The similarity function used for this dataset is the reciprocal of the Euclidean distance, i.e., $\sigma(u,v) = \frac{1}{\lVert u -v \rVert + \delta}$, for any two vectors $u$ and $v$; here, $\delta>0$ is a small constant to avoid division by zero in case $\lVert u - v \rVert = 0$. Typically, we set $\delta = \eta$, where $\eta$ is the smoothening parameter used in the definition of $\NSW(\cdot)$. The distribution of the attributes across the vectors is shown in Figure~\ref{fig:arxivDistri}.

\item  \textbf{SIFT Embeddings}: This is a standard benchmarking dataset for approximate nearest neighbor search in the Euclidean distance metric~\citep{tensorflowdatasets}. The dataset consists of SIFT vector embeddings in $d=128$ dimensional space. In particular, here input set $P$ contains $1,000,000$ vectors and we have $10,000$ vectors as queries. The embeddings are available at \citep{tensorflowdatasets}.\footnote{\href{https://www.tensorflow.org/datasets/catalog/sift1m}{https://www.tensorflow.org/datasets/catalog/sift1m}} Note that this dataset does not contain any metadata that can be naturally used to assign attributes to the given vectors. Therefore, we utilize the following two methods for synthetically assigning attributes to the input vectors. 

\begin{itemize}
    \item \textbf{Clustering-based} (\siftC\hspace{-2pt}): Since attributes such as colors often occupy distinct regions in the embedding space, we apply k-means clustering~\cite{lloyd1982least,mcqueen1967some} to identify $20$ clusters. Each cluster is then assigned a unique index, and vectors in the same cluster are associated with the index of the cluster, which serves as our synthetic attribute. This simulates a single-attribute setting with $c = 20$.

    To simulate the multi-attribute setting, we extend the above method of attribute generation. Given an input vector $v$ of $128$ dimensions, we split it into $4$ vectors $\{v^i\}_{i=1}^4$, each of $32$ dimensions. In particular, the vector $v^1$ consists of the first $32$ components of $v$, the vector $v^2$ is obtained from the next $32$ components of $v$, so on.     
    Next, for each $i \in [4]$, we apply k-means clustering on $P^i \coloneqq \{v^i: v \in P\}$ to identify $20$ clusters. We write $C_i = \{c^i_1, \ldots, c^i_{20}\}$ to denote these $20$ clusters for $P^i$ and if vector $v^i$ belongs to a cluster $c^i_j$, then $v^i$ is assigned the attribute $c^i_j$, i.e., we set $\col(v^i) = c^i_j$. Finally, for a vector $v \in P$, we assign $\col(v) = \{\col(v^1), \ldots, \col(v^4)\}$. Note that this method of simulating the multi-attribute setting yields $m = 4$ attribute classes.

    \item \textbf{Probability distribution-based} (\siftP\hspace{-2pt}): As in prior work \citep{anand2025graphbased}, we also consider a setting wherein the input vectors have randomly assigned attributes. Specifically, we consider the single-attribute setting with $c=20$. Here, we assign each vector $v \in P$ an attribute as follows: with probability $0.9$, select an attribute from $\{1,2,3\}$ uniformly at random, otherwise, with the remaining $0.1$ probability, select an attribute from $\{4, \ldots, 20\}$ uniformly at random. This results in a skewed distribution of vectors among attributes that mimics real-world settings (e.g., market dominance by a few sellers). %
\end{itemize}

\item \textbf{Deep Descriptor Embeddings}:  This is another standard dataset for nearest neighbor search \citep{tensorflowdatasets}.\footnote{\href{https://www.tensorflow.org/datasets/catalog/deep1b}{https://www.tensorflow.org/datasets/catalog/deep1b}} The version of the dataset used in the current work contains  $9,990,000$ input vectors and $10,000$ separate query vectors. Here, the vectors are $96$ dimensional and the distances between them are evaluated using the cosine distance. 

As in the case of SIFT dataset, input vectors in \texttt{Deep1b} do not have predefined attributes. Hence, we use the above-mentioned methods (based on clustering and randomization) to synthetically assign attributes to the vectors. This gives us the clustering-based \deepC\ and probability distribution-based \deepP\ versions of the dataset.

\end{enumerate}

\paragraph{\textbf{Choice of Parameter $\eta$}:} For our methods, we tune and set the smoothing parameter, $\eta$, to $0.01$ for the \arxiv, \siftC\ and \siftP\ datasets, and set it to $0.0001$ to analyze $p$-NNS. For other datasets, namely \amazon, \deepC\ and \deepP, we set $\eta$ to $50$.

\subsection{Algorithms}
\label{subsec:algorithm-details-for-experiments}
Next, we describe the algorithms executed in the experiments.

\begin{enumerate}
\item \textbf{\diskann}: This is the standard ANN algorithm that aims to maximize the similarity of the retrieved vectors to the given query without any diversity considerations. In our experiments, we use the graph based DiskANN method of \cite{DiskANN19} as the standard ANN algorithm. We instantiate DiskANN with candidate list size $L=2000$ and the maximum graph degree as $128$.\footnote{Both these choices are sufficiently larger than the standard values, $L=200$ and maximum graph degree $64$.} Here, we also set the pruning factor at $1.3$, which is consistent with the existing recommendation in \cite{anand2025graphbased}.

\item \textbf{\divann}: This refers to the algorithm of \cite{anand2025graphbased} that solves the hard- constraint-based formulation for diversity in the single-attribute setting. Recall that \cite{anand2025graphbased} aims to maximize the similarity of the retrieved vectors to the given query subject to the constraint that no more than $k'$ vectors in the retrieved set should have the same attribute. Note that the smaller the value of $k'$, the more diverse the retrieved set of vectors. 
Moreover, $k'$ has to be provided as an input to this algorithm. In our experiments, we set different values $k'$, such as $k' \in \{1, 2, 5\}$ when $k=10$, and $k'\in \{1,2,5,10\}$ when $k=50$.

\item \textbf{\texttt{Nash-ANN} and \texttt{$p$-mean-ANN}}: \texttt{Nash-ANN} refers to \Cref{algo:greedy-nash-ann} and \texttt{$p$-mean-ANN} refers to \Cref{algo:p-mean-ann} (stated in Appendix~\ref{appendix:p-mean}). Recall that \Cref{algo:greedy-nash-ann} and \Cref{algo:p-mean-ann} optimally solve the NaNNS and the $p$-NNS problems, respectively, in the single-attribute setting given access to an exact nearest neighbor search oracle (Theorems \ref{theorem:single-attribute} and~\ref{theorem:single-attribute-p-mean}). Further note that the $p$-mean welfare function, $M_p(\cdot)$, reduces to the Nash social welfare (geometric mean) when the exponent parameter $p \to 0^{+}$.
For readability and at required places, we will write $p=0$ to denote \ouralgo. We conduct experiments with varying values of $p \in \{-10, -1, -0.5, 0, 0.5, 1\}$.

\item \textbf{\texttt{Multi Nash-ANN} and \texttt{Multi Div-ANN}}:  
In the multi-attribute setting, there are no prior methods to address diversity. Hence, for comparisons, we first fetch $L = 10000$  candidate vectors from $P$ for each query $q$, using the standard \diskann~method, and then apply the following algorithms on the candidate vectors: (i) our algorithm for the NaNNS problem in the multi-attribute setting (\Cref{algo:greedy-submodular}), which we term as \multinash, and (ii) an adaptation of the algorithm of \cite{anand2025graphbased}, referred hereon as \multidivann, which greedily selects the most similar vectors to the query subject to the constraint that there are no more than $k'$ vectors from each attribute.\footnote{Note that one vector can have multiple attributes, hence contributing to the constraint of multiple attributes. Therefore, the issue of identifying an appropriate $k'$ is exacerbated on moving to the multi-attribute setting.}
We compare \multinash~($p=0$) against \multidivann~ under different  choices of $k'$.

\item \multipmean: In the multi-attribute setting, we also implement an analogue of \multinash\ that (in lieu of $\NSW$) focuses on the $p$-mean welfare, $M_p$. The objective in these experiments is to understand the impact of varying the parameter $p$ and the resulting tradeoff between relevance and diversity.
Here, for each given query, we first fetch a set of $L=10000$ candidate vectors using \diskann. Then, we populate a set of $k$ vectors by executing a marginal-gains greedy method over the $L$ candidate vectors. In particular, we iterate $k$ times, and in each iteration, select a new candidate vector that: (i) for $p \in (0,1]$, yields the maximum increase in $M_p(\cdot)^p$, or (ii) for $p < 0$, leads to the maximum decrease in $M_p(\cdot)^p$. 
\end{enumerate}

\subsection{Results: Balancing Relevance and Diversity}

\paragraph{Single-attribute setting.} 
We first compare, in the single-attribute setting, the performance of our algorithm, \ouralgo, with \diskann~and \divann\ (with different values of $k'$). The results for the \amazon\ and \deepC\ datasets with $k=50$ are shown in Figure~\ref{fig:mainPaperResults-NaNNS} (columns one and two). Here, \diskann\ finds the most relevant set of neighbors (approximation ratio close to $1$), albeit with the lowest entropy (diversity). Moreover, as can be seen in the plots, the most diverse (highest entropy) solution is obtained when we set, in \divann, $k' = 1$; this restricts each $\ell \in [c]$ to contribute at most one vector in the output of \divann. Also, note that one can increase the approximation ratio (i.e., increase relevance) of \divann~while incurring a loss in entropy (diversity), by increasing the value of the constraint parameter $k'$. However, selecting a `right' value for $k'$ is non-obvious, since this choice needs to be tailored to the dataset and, even within it, to queries (recall the ``blue shirt'' query in Figure \ref{fig:QueryResponses}). 

By contrast, \ouralgo\ does not require such ad hoc adjustments and, by design, finds a balance between relevance and diversity. Indeed, as can be seen in Figure~\ref{fig:mainPaperResults-NaNNS} (columns 1 and 2), \ouralgo\ maintains an approximation ratio close to $1$ while achieving diversity similar to \divann\ with $k'=1$. Moreover, \ouralgo~  Pareto dominates \divann\  with $k' = 2$ for \amazon ~dataset and $k' = 5$ for \deepC\ dataset on the fronts of approximation ratio and entropy.  
The results for other datasets and metrics follow similar trends and are given in Appendix \ref{appendix:exp-single-attribute}.

\begin{figure}[t!]
\centering
\includegraphics[width=0.24\textwidth]{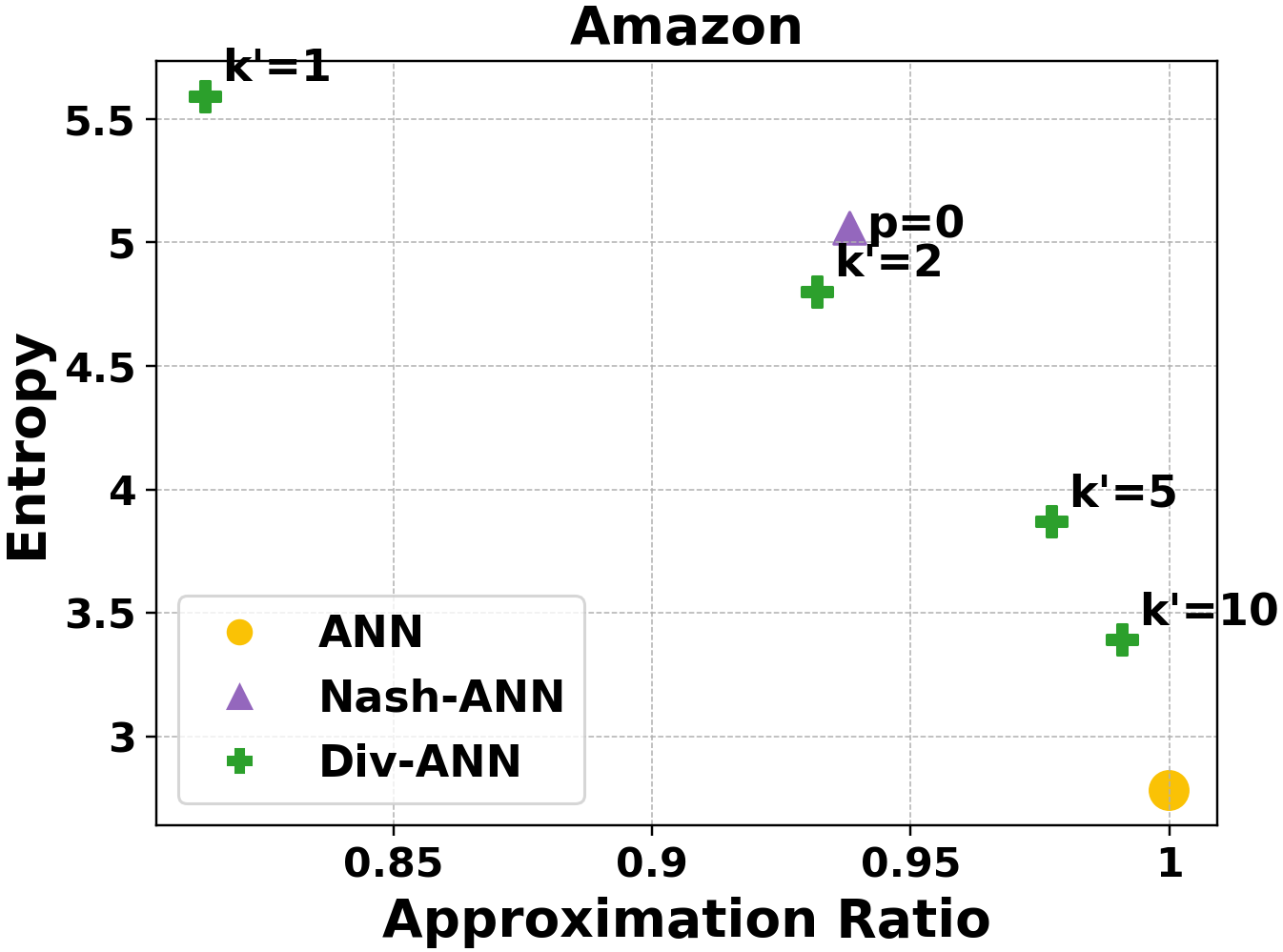} \hfill
\includegraphics[width=0.24\textwidth]{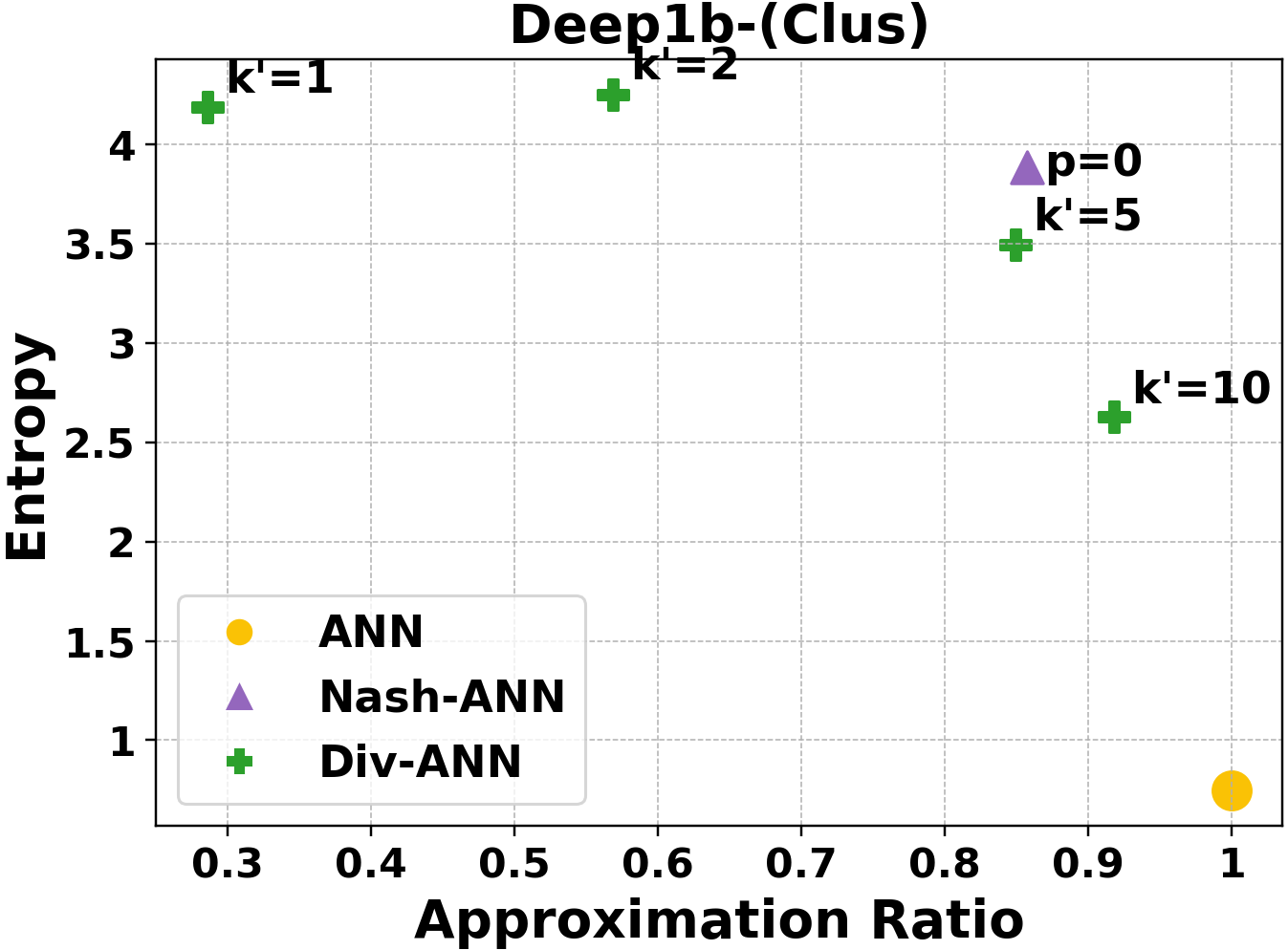} \hfill
\includegraphics[width=0.24\textwidth]{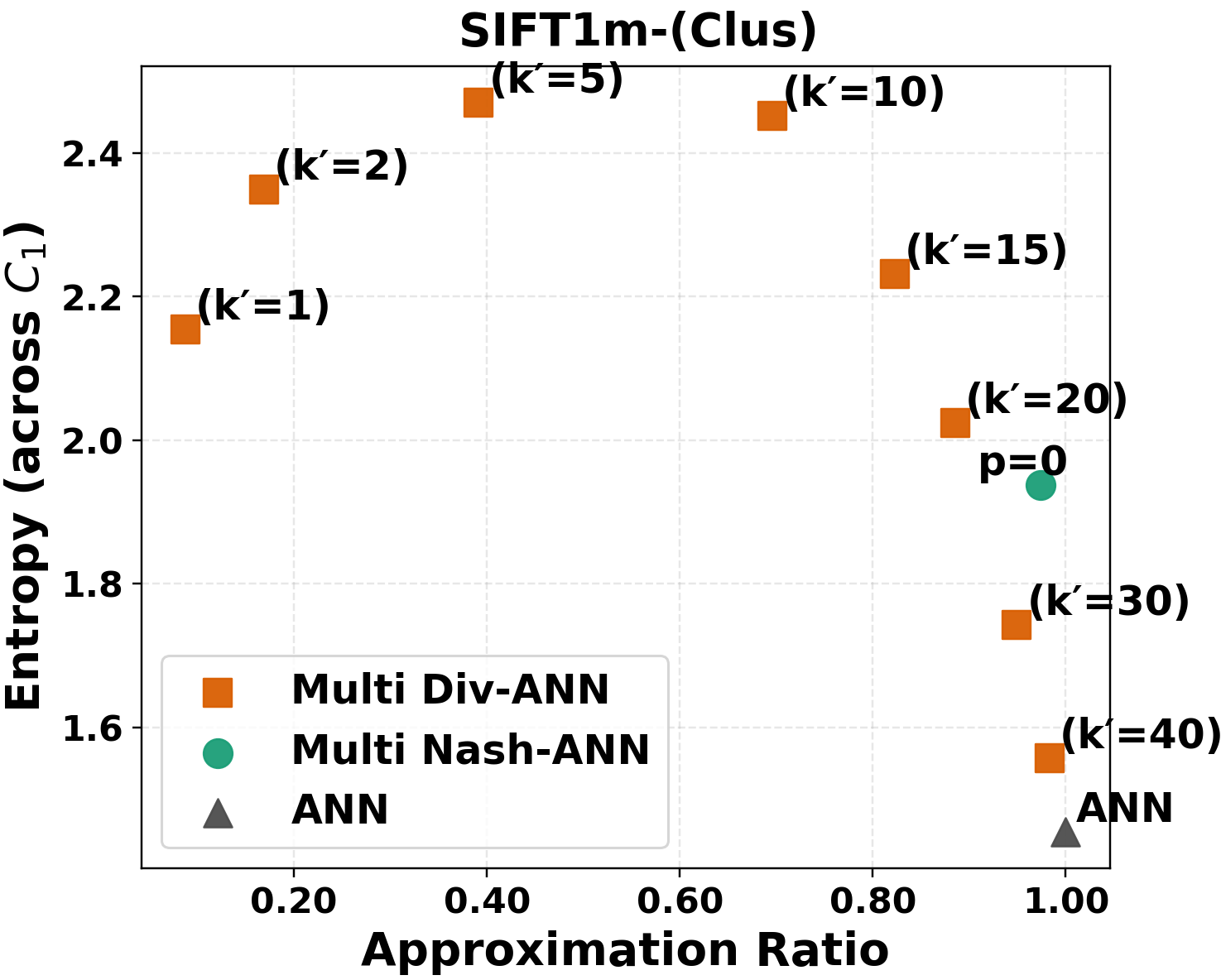} \hfill
\includegraphics[width=0.24\textwidth]{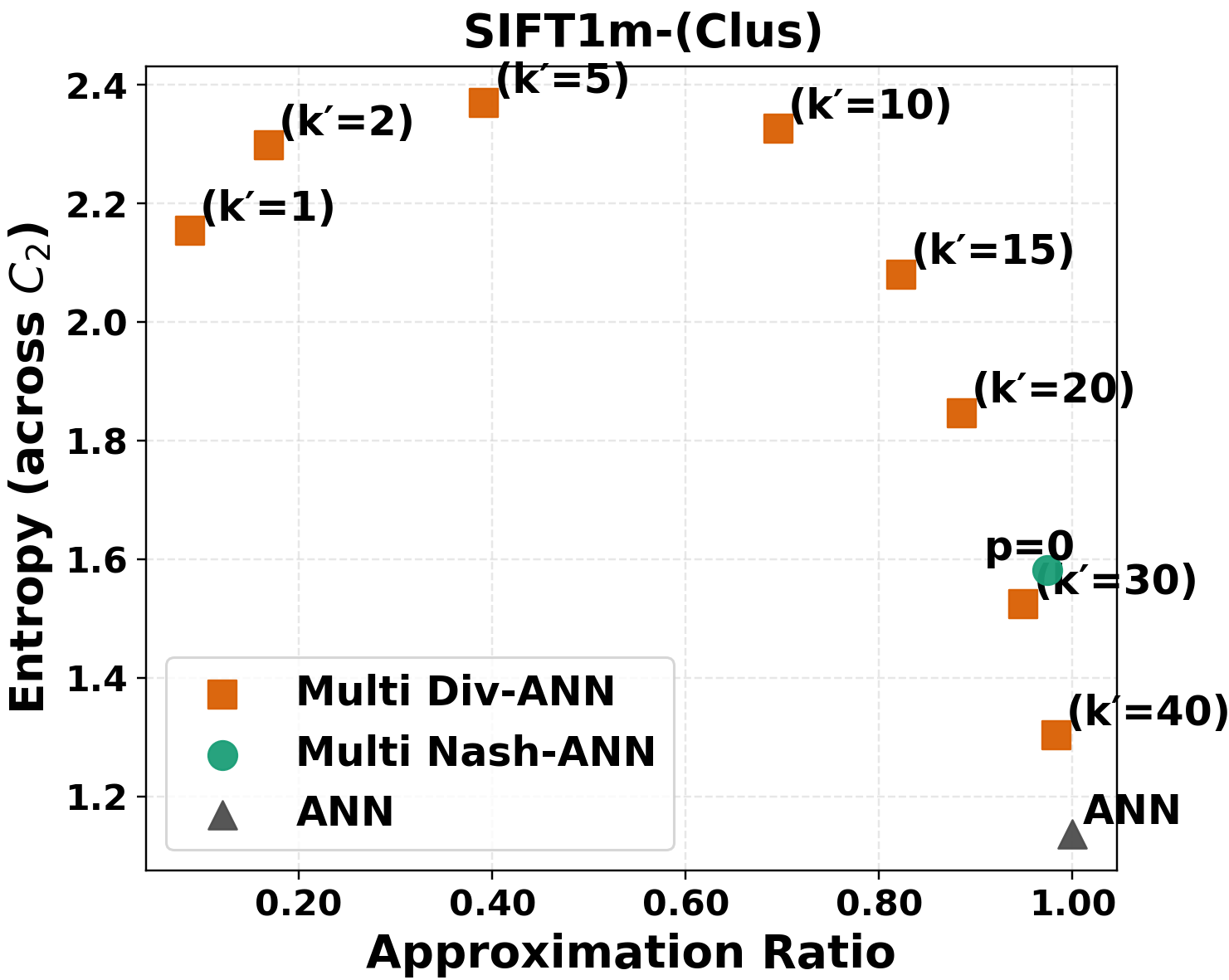} \hfill

\caption{{\small {\bf Columns 1 and 2} - Comparison of approximation ratio versus entropy trade-offs between \ouralgo, and \divann\ with varying $k'$, for $k$  = $50$ on \amazon\ and \deepC\ datasets in the single-attribute setting. {\bf Columns 3 and 4} - Comparison of approximation ratio versus entropy trade-offs (across attribute classes $C_1$ and $C_2$) between \multinash, and \multidivann\ with varying $k'$ on \siftC\ dataset with $k=50$ in the multi-attribute setting. }}
\label{fig:mainPaperResults-NaNNS}
\end{figure}

\paragraph{Multi-attribute setting.}
In the multi-attribute setting, we report results for \multinash\ and \multidivann\ on the \siftC\ dataset (Figure \ref{fig:mainPaperResults-NaNNS}, columns 3 and 4) for $k=50$ and $c=80$. These eighty attributes are partitioned into four sets, $\{C_i\}_{i=1}^4$, with each set of size $|C_i| = 20$, i.e., $[c] = \cup_{i=1}^4 C_i$. Further, each input vector $v$ is associated with four attributes ($|\col(v)|=4$), one from each $C_i$; see \Cref{subsec:experiment-set-up-and-dataset} for further details. %
Here, to quantify diversity, we separately consider for each $i \in [4]$, the entropy across attributes within a $C_i$. In \Cref{fig:mainPaperResults-NaNNS} (columns 3 and 4), we compare the approximation ratio versus entropy trade-offs of \multinash\ against \multidivann\ with varying $k'$. Here we show the results for attribute classes $C_1$ (column 1) and $C_2$ (column 2) whereas the results for $C_3$ and $C_4$ are given in \Cref{fig:mainPaperResultsMultiAttriB}. We observe that \multinash\ maintains a high approximation ratio (relevance) while simultaneously achieving significantly higher entropy (higher diversity) than \diskann. By contrast, in the constraint-based method \multidivann, low values of $k'$ lead to a notable drop in the approximation ratio, whereas increasing $k'$ reduces entropy. For example, for $k'$ below $15$, one obtains approximation ratio less than $0.8$, and to reach an approximation ratio comparable to \multinash, one needs $k'$ as high as $30$. Additional results for the \arxiv\ dataset in the multi-attribute setting are provided in Appendix~\ref{appendix:exp-multi-attribute}, and they exhibit trends similar to the ones in \Cref{fig:mainPaperResults-NaNNS}. These findings demonstrate that \multinash~achieves a balance between relevance and diversity. In summary
{
\begin{quote}

    \emph{Across datasets, and in both single- and multi-attribute settings, the Nash formulation consistently improves entropy (diversity) over \kANN, while maintaining an approximation ratio (relevance) of roughly above $0.9$. By contrast, the hard-constrained formulation is highly sensitive to the choice of the constraint parameter $k'$, and in some cases, incurs a substantial drop in approximation ratio (even lower than $0.2$).}

\end{quote}
}

\begin{figure}[t!]
\centering
\includegraphics[width=0.24\textwidth]{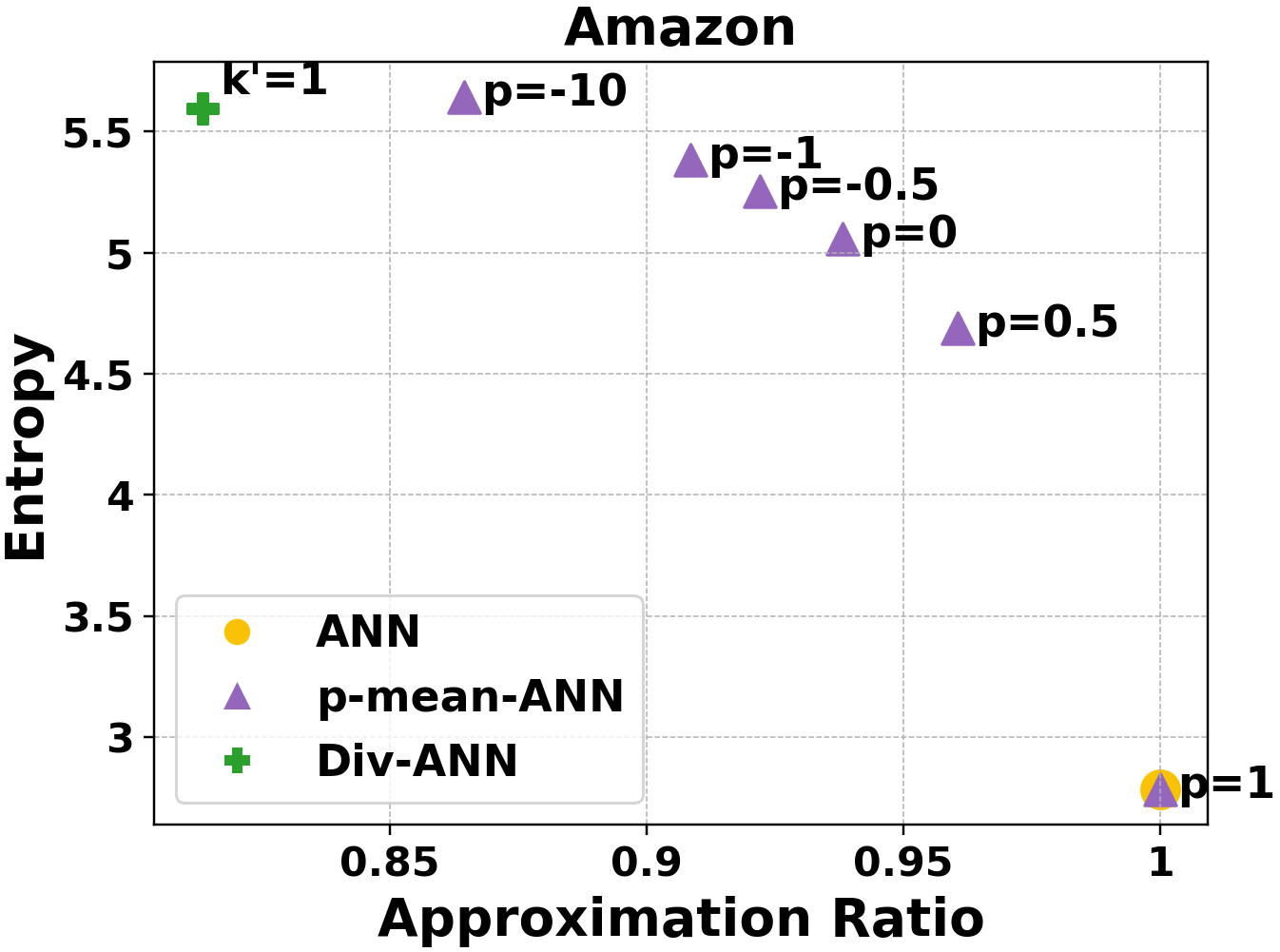} \hfill
\includegraphics[width=0.24\textwidth]{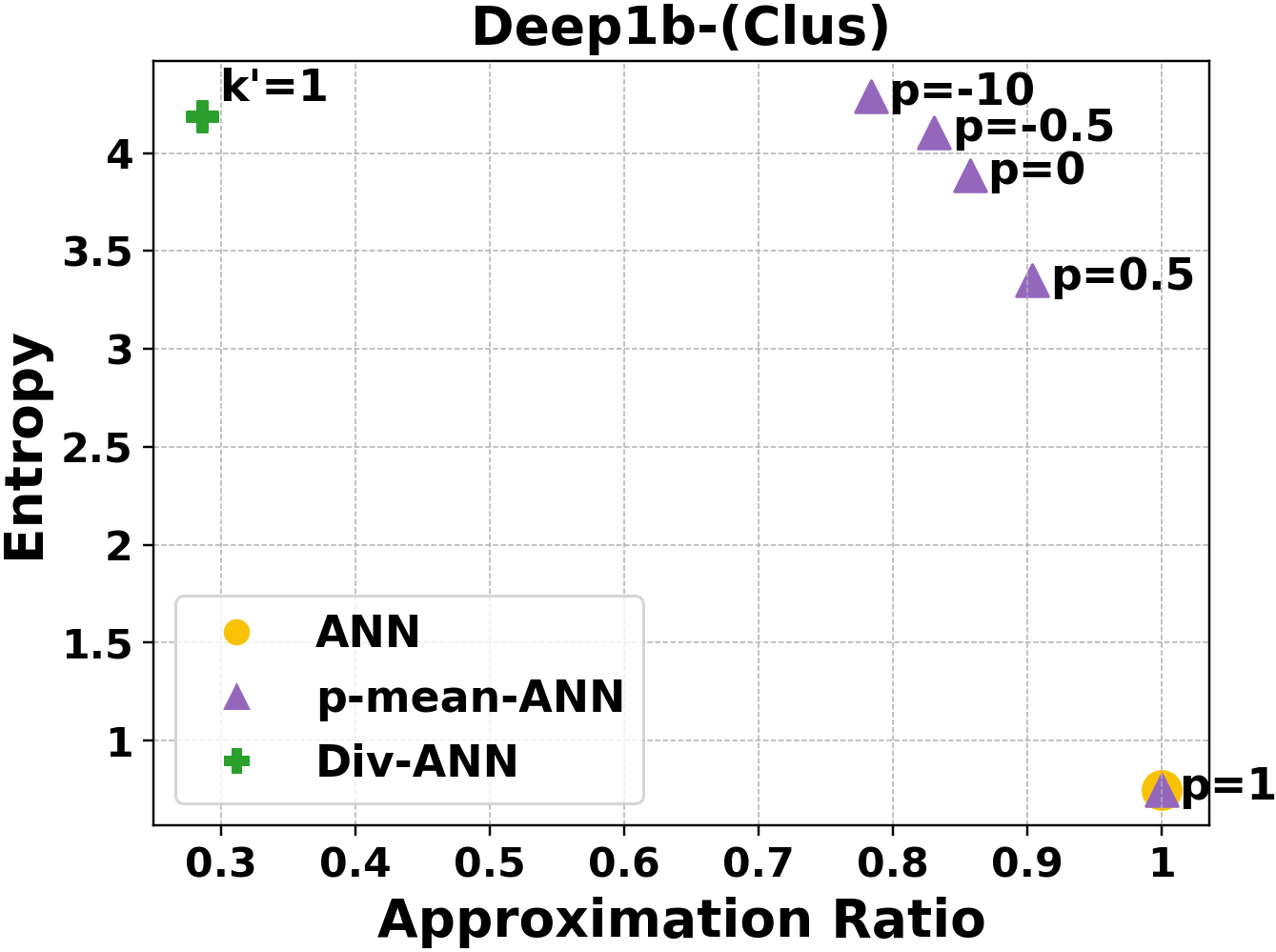}
\includegraphics[width=0.24\textwidth]{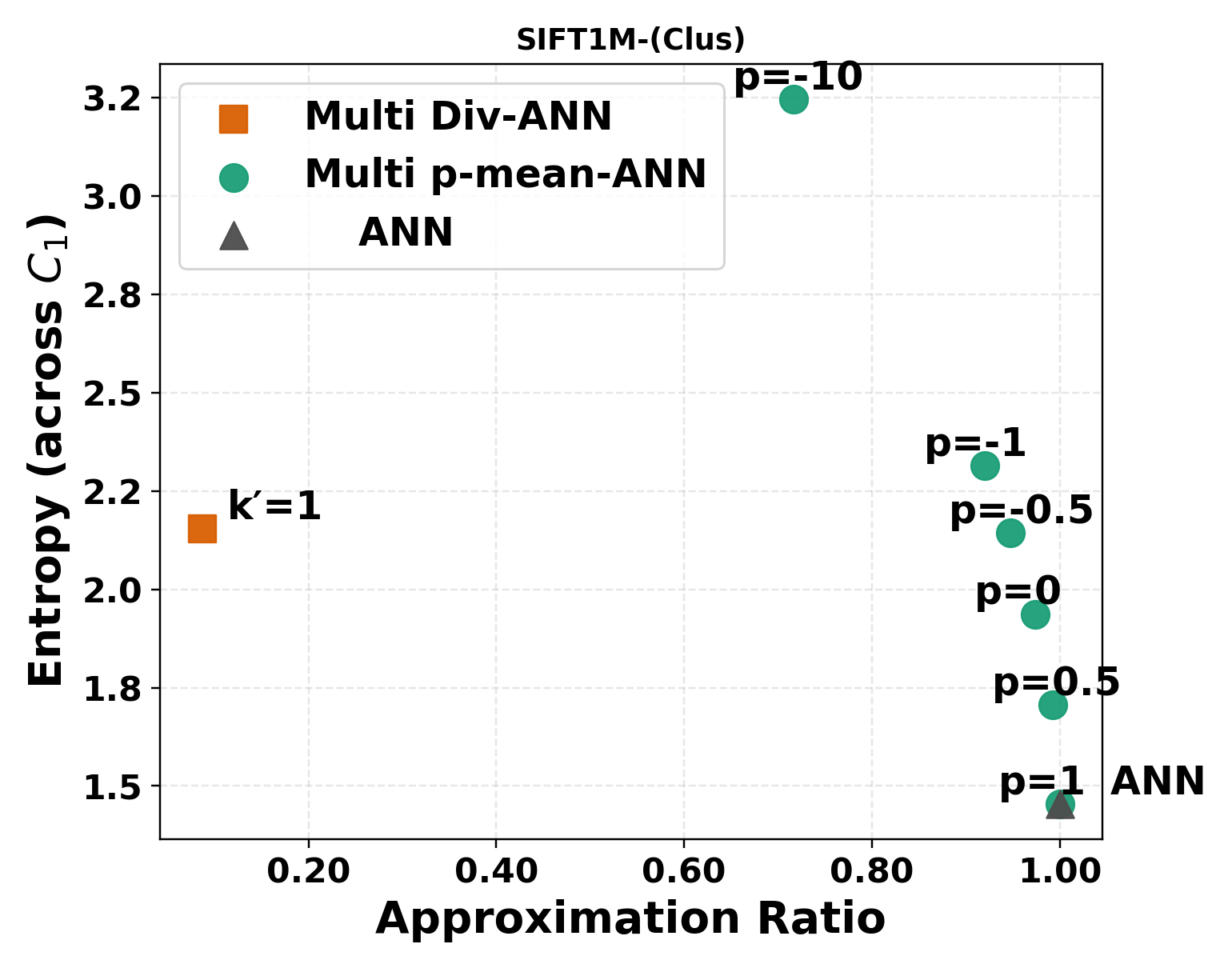} \hfill
\includegraphics[width=0.24\textwidth]{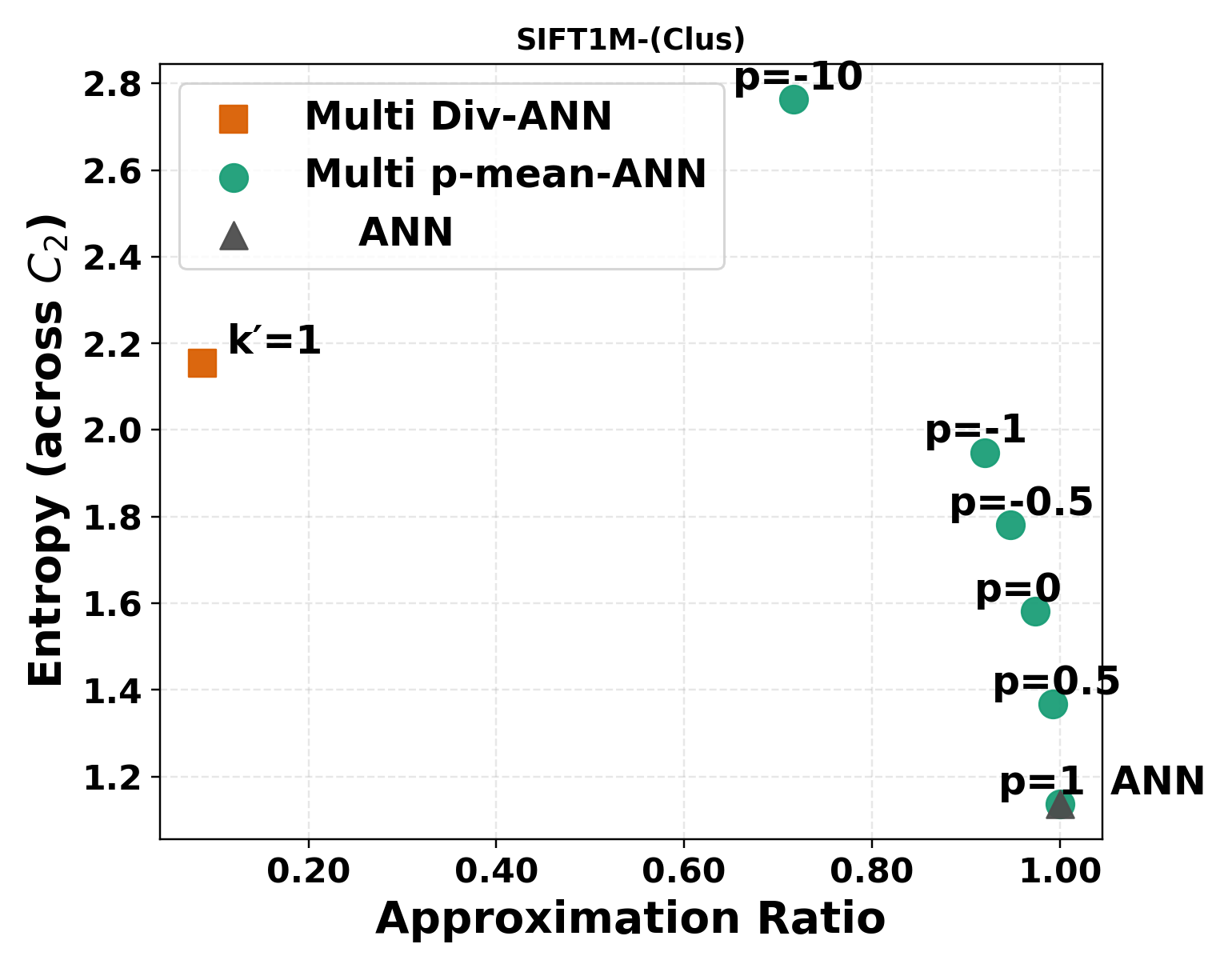}

\caption{{\small {\bf Columns 1 and 2} - Approximation ratio versus entropy trade-offs for {\tt $p$-mean-ANN} at various $p$ values, for $k$ = $50$ on \amazon\ and \deepC\ datasets in the single-attribute setting. {\bf Columns 3 and 4} - Approximation ratio versus entropy trade-offs (across attribute classes $C_1$ and $C_2$) for \multipmean\ with varying $p$ on \siftC\ dataset with $k$=$50$ in the multi-attribute setting.}}
\label{fig:mainPaperResults-p-NNS}
\end{figure}

\paragraph{Results for $p$-NNS.} Recall that, selecting the exponent parameter $p \in (-\infty, 1]$ enables us to interpolate $p$-NNS between the standard NNS problem ($p=1$), NaNNS ($p=0$), and optimizing solely for diversity ($p \to -\infty$). We execute \pMeanANN~ for $p \in \{-10, -1, -0.5, 0, 0.5, 1\}$ in both single- and multi-attribute settings and show that a trade-off between relevance (approximation ratio) and diversity (entropy) can be achieved by varying $p$.

For the single-attribute setting, Figure \ref{fig:mainPaperResults-p-NNS} (columns 1 and 2), and for the multi-attribute setting, Figure \ref{fig:mainPaperResults-p-NNS} (columns 3 and 4) capture this feature on \siftC\ dataset with $k=50$: For lower values of $p$, we have higher entropy but lower approximation ratio, while $p = 1$ matches \diskann. For the multi-attribute setting, we show results for attribute classes $C_1$ (column 3) and $C_2$ (column 4) in \Cref{fig:mainPaperResults-p-NNS}, whereas the results for $C_3$ and $C_4$ are shown in \Cref{fig:mainPaperResultsMultiAttriB}. Note that in the multi-attribute setting, \texttt{Multi $p$-mean-ANN} with $p=-10$ Pareto dominates \multidivann\ with $k'=1$ in terms of approximation ratio and entropy. Moreover, analogous results are obtained for other datasets and metrics; see Appendix~\ref{appendix:exp-single-attribute} and \ref{appendix:exp-multi-attribute}. 

\subsection{A Faster Heuristic for the Single Attribute Setting: \texttt{$p$-FetchUnion-ANN}}
\label{subsec:FetchNashUnionNash}

Further, we empirically study a faster heuristic algorithm for our NSW and $p$-mean welfare formulations. Specifically, the heuristic---called \texttt{$p$-FetchUnion-ANN}---first fetches a sufficiently large candidate set of vectors (irrespective of their attributes) using the \diskann~algorithm. Then, it applies the Nash (or $p$-mean) selection (similar to Line~\ref{line:marginal} in~\Cref{algo:greedy-nash-ann} or Lines~\ref{line:marginal-positive-p}-\ref{line:marginal-negative-p} in~\Cref{algo:p-mean-ann}) within this set. That is, instead of starting out with $k$ neighbors for each $\ell \in [c]$ (as in Line \ref{line:hat-D} of \Cref{algo:greedy-nash-ann}), the alternative here is to work with sufficiently many neighbors from the set $\cup_{\ell=1}^c D_\ell$. 

Table~\ref{tab:sift20-k50} shows that this heuristic consistently achieves performance comparable to \texttt{$p$-mean-ANN} in terms of approximation ratio and diversity in the \siftC\ dataset. Since \texttt{$p$-FetchUnion-ANN} retrieves a larger pool of vectors with high similarity, it achieves improved  approximation ratio over \texttt{$p$-Mean-ANN}. However, it comes at the cost of reduced entropy, which can be explained by the fact that in restricting its search to an initially fetched large pool of vectors, \texttt{$p$-FetchUnion-ANN} may miss a more diverse solution that exists over the entire dataset. Another important aspect of \texttt{$p$-FetchUnion-ANN} is that, because it retrieves all neighbors from the union at once, the heuristic delivers substantially higher throughput (measured as queries answered per second, QPS) and therefore lower latency, which can be seen in Table~\ref{tab:QPSandLatency} for the \siftC\ dataset. In particular, \texttt{$p$-FetchUnion-ANN} serves almost \textbf{10$\times$} more queries on \siftC\ dataset than \texttt{$p$-mean-ANN}. The latency values exhibit a similar trend with reductions of similar magnitude. In summary, these observations position the heuristic as a notably fast method for NaNNS and $p$-NNS, particularly when $c$ is large. We provide comparison of \texttt{$p$-FetchUnion-ANN} with \texttt{$p$-mean-ANN} on other datasets in Tables~\ref{tab:amazon-k50} to~\ref{tab:QPSandLatencyAmazon} in Appendix~\ref{appendix:FetchNashUnionNash}.

\begin{table}[tb]
\centering
\caption{Comparison of performance across $p$ values for \siftC\ at $k=50$. }
\label{tab:sift20-k50}
\scalebox{0.75}{
\setlength{\tabcolsep}{6pt}
\renewcommand{\arraystretch}{1.2}
\begin{tabular}{llcccccc}
\toprule
\textbf{Metric} & \textbf{Algorithm} & $p=-10$ & $p=-1$ & $p=-0.5$ & $p=0$ & $p=0.5$ & $p=1$ \\
\midrule
\multirow{4}{*}{Approx.~Ratio}
  & \texttt{$p$-Mean-ANN} & 0.749$\pm$0.051 & 0.810$\pm$0.045 & 0.812$\pm$0.043 & 0.846$\pm$0.036 & 0.932$\pm$0.028 & 1.000$\pm$0.000 \\
  & \texttt{$p$-FetchUnion-ANN} & 0.979$\pm$0.014 & 0.980$\pm$0.013 & 0.980$\pm$0.013 & 0.981$\pm$0.012 & 0.983$\pm$0.011 & 1.000$\pm$0.000 \\
  & \diskann\  & \multicolumn{6}{c}{1.000$\pm$0.000} \\
  & \divann\  ($k'\!=\!1$) & \multicolumn{6}{c}{0.315$\pm$0.021} \\
\midrule
\multirow{4}{*}{Entropy}
  & \texttt{$p$-Mean-ANN} & 4.285$\pm$0.012 & 4.293$\pm$0.002 & 4.293$\pm$0.001 & 4.197$\pm$0.045 & 3.506$\pm$0.275 & 0.892$\pm$0.663 \\
  & \texttt{$p$-FetchUnion-ANN} & 2.235$\pm$0.802 & 2.238$\pm$0.802 & 2.239$\pm$0.802 & 2.239$\pm$0.802 & 2.231$\pm$0.800 & 0.892$\pm$0.663 \\
  & \diskann\  & \multicolumn{6}{c}{0.892$\pm$0.663} \\
  & \divann\  ($k'\!=\!1$) & \multicolumn{6}{c}{4.289$\pm$0.053} \\
\bottomrule
\end{tabular}
}%

\end{table}

\begin{table}[tb]
\centering
\caption{Comparison of performance on QPS and Latency across $p$ on \siftC\ dataset for $k=50$. }
\label{tab:QPSandLatency}
\scalebox{0.75}{
\setlength{\tabcolsep}{6pt}
\renewcommand{\arraystretch}{1.2}
\begin{tabular}{llcccccc}
\toprule
\textbf{Metric} & \textbf{Algorithm} & $p=-10$ & $p=-1$ & $p=-0.5$ & $p=0$ & $p=0.5$ & $p=1$ \\
\midrule
\multirow{2}{*}{Query per Second}
  & \texttt{$p$-Mean-ANN} & 120.86  & 115.78  & 107.01 & 135.98 & 122.59 & 122.59 \\
  & \texttt{$p$-FetchUnion-ANN} & 1324.53 & 1324.62 & 1337.28  & 1442.03 & 1443.38 & 1327.03  \\
\midrule
\multirow{2}{*}{Latency ($\mu s$)}
  & \texttt{$p$-Mean-ANN} & 264566.00 & 276129.00 & 298804.00  & 230318.00  & 235144.00  & 260800.00 \\
  & \texttt{$p$-FetchUnion-ANN} &  24133.80 &  24134.00 & 23907.00 & 22170.20 & 22149.30 & 28990.40 \\
  \midrule
\multirow{2}{*}{$99.9$th percentile of Latency}
  & \texttt{$p$-Mean-ANN} &   484601.00  &  513036.00 &  478821.00 & 477925.00 & 482777.00 & 479132.00 \\
  & \texttt{$p$-FetchUnion-ANN} &   52943.40  & 53474.70 & 54283.40 & 56128.70 & 53082.20  & 24088.70  \\
\bottomrule
\end{tabular}
}%

\end{table}

\section{Conclusion}
\label{section:conclusion}
In this work, we formulated diversity in neighbor search with a welfarist perspective, using Nash social welfare (NSW) and $p$-mean welfare as the underlying objectives. Our NSW formulation balances diversity and relevance in a query-dependent manner, satisfies several desirable axiomatic properties, and is naturally applicable in both single-attribute and multi-attribute settings. With these properties, our formulation overcomes key limitations of the prior hard-constrained approach~\citep{anand2025graphbased}. Furthermore, the more general $p$-mean welfare interpolates  between complete relevance ($p=1$) and complete diversity ($p=-\infty$), offering practitioners a tunable parameter for real-world needs. Our formulations also admit provable and practical algorithms suited for low-latency scenarios. Experiments on real-world and semi-synthetic datasets validate their effectiveness in balancing diversity and relevance against existing baselines. 

An important direction for future work is the design of sublinear-time approximation algorithms, in both single- and multi-attribute settings, that directly optimize our welfare objectives as part of ANN algorithms, thereby further improving efficiency. Another promising avenue is to extend welfare-based diversity objectives to settings without explicit attributes.

\section*{Acknowledgement}
\addcontentsline{toc}{section}{Acknowledgement}
Siddharth Barman, Nirjhar Das, and Shivam Gupta acknowledge the support of the Walmart Center for Tech Excellence (CSR WMGT-23-0001) and an Ittiam CSR Grant (OD/OTHR-24-0032).

\bibliographystyle{alpha}
\bibliography{references}

\newpage
\appendix

\section{Proofs of Examples \ref{example-one} and \ref{example-two}}
\label{section:example-proofs}
This section provides the proofs for Examples \ref{example-one} and \ref{example-two}. Recall that these stylized examples highlight how NaNNS dynamically recovers complete diversity or complete relevance based on the data.

\ExampleOne*

\begin{proof}
    Towards a contradiction, suppose there exists $T \in \argmax_{S \subseteq P: \lvert S \rvert = k} \NSW(S)$ such that $\lvert T \cap D_{\ell^*} \rvert > 1$ for some $\ell^* \in [c]$. Note that according to the setting specified in the example, $u_\ell(T) = \lvert T \cap D_\ell \rvert + \eta$ for all $\ell \in [c]$.

    Since $c \geq k$ and $\lvert T \cap D_{\ell^*} \rvert > 1$, there exists $\ell' \in [c]$ such that $\lvert T \cap D_{\ell'} \rvert = 0$. Let $v^* \in T \cap D_{\ell^*}$ and $v' \in D_{\ell'}$ be two vectors. Consider the set $T' = (T \setminus \{v^*\} ) \cup \{v'\}$. We have
    \begin{align*}
        \frac{\NSW(T')}{\NSW(T)} &= \left(\frac{(u_{\ell'}(T') + \eta)}{(u_{\ell'}(T) + \eta)} \cdot \frac{(u_{\ell^*}(T') + \eta)}{(u_{\ell^*}(T) + \eta)} \prod_{\ell \in [c] \setminus \{\ell^*, \ell'\}} \frac{(u_\ell(T') + \eta)}{(u_\ell(T) + \eta)}\right)^\frac{1}{c} \\
        &= \left(\frac{(1 + \eta)}{\eta} \cdot \frac{(u_{\ell^*}(T) - 1 + \eta)}{(u_{\ell^*}(T) + \eta)} \prod_{\ell \in [c] \setminus \{\ell^*, \ell'\}} \frac{(u_\ell(T) + \eta)}{(u_\ell(T) + \eta)}\right)^\frac{1}{c} \\
        &= \left( \frac{(1 + \eta)}{\eta} \cdot \frac{(u_{\ell^*}(T) - 1 + \eta)}{(u_{\ell^*}(T) + \eta)} \right)^\frac{1}{c} \\
        &= \left(\frac{u_{\ell^*}(T) - 1 + \eta u_{\ell^*}(T) + \eta^2}{\eta u_{\ell^*}(T) + \eta^2}\right)^\frac{1}{c} > 1 \tag{$u_{\ell^*}(T) \geq 2$}
    \end{align*}
    Therefore, we have $\NSW(T') > \NSW(T)$, which contradicts the optimality of $T$. Hence, we must have $\lvert T \cap D_\ell \rvert \leq 1$ for all $\ell \in [c]$, which proves the claim.
\end{proof}

\ExampleTwo*

\begin{proof}
    Towards a contradiction, suppose there exists $T \in \argmax_{S \subseteq P: \lvert S \rvert = k} \NSW(S)$ such that $\lvert T \cap D_{\ell^*} \rvert < k$. Therefore, there exists $\ell' \in [c] \setminus \{\ell^*\}$ such that $\lvert T \cap D_{\ell'} \rvert \geq 1$. Let $v^* \in D_{\ell^*} \setminus T$ and let $v' \in T \cap D_{\ell'}$. Note that $u_{\ell'}(T) = 0$ since $\sigma(q, v) = 0$ for all $v \in D_\ell$ for any $\ell \in [c] \setminus \{\ell^*\}$. Moreover, we also have $u_{\ell^*}(T) = \lvert T \cap D_{\ell^*} \rvert$.

    Consider the set $T' = (T \setminus \{v'\}) \cup \{v^*\}$. We have,
    {\allowdisplaybreaks
    \begin{align*}
        \frac{\NSW(T')}{\NSW(T)} &= \left(\frac{(u_{\ell'}(T') + \eta)}{(u_{\ell'}(T) + \eta)} \cdot \frac{(u_{\ell^*}(T') + \eta)}{(u_{\ell^*}(T) + \eta)} \prod_{\ell \in [c] \setminus \{\ell^*, \ell'\}} \frac{(u_\ell(T') + \eta)}{(u_\ell(T) + \eta)}\right)^\frac{1}{c} \\
        &= \left(\frac{(u_{\ell'}(T) - \sigma(q, v') + \eta)}{(u_{\ell'}(T) + \eta)} \cdot \frac{(u_{\ell^*}(T) + \sigma(q, v^*) + \eta)}{(u_{\ell^*}(T) + \eta)} \prod_{\ell \in [c] \setminus \{\ell^*, \ell'\}} \frac{(u_\ell(T) + \eta)}{(u_\ell(T) + \eta)}\right)^\frac{1}{c} \\
        &= \left( \frac{(0 - 0 + \eta)}{0 + \eta} \cdot \frac{(\lvert T \cap D_{\ell^*} \rvert + 1 + \eta)}{(\lvert T \cap D_{\ell^*} \rvert + \eta)} \right)^\frac{1}{c} \\
        &= \left(1 + \frac{1}{ \lvert T \cap D_{\ell^*} \rvert + \eta}\right)^\frac{1}{c} > 1~.
    \end{align*}}
    Therefore, we have obtained $\NSW(T') > \NSW(T)$, which contradicts the optimality of $T$. Hence, it must be the case that $\lvert T \cap D_{\ell^*} \rvert = k$, which proves the claim.
\end{proof}

\section{Extensions for $p$-NNS}
\label{appendix:p-mean}

This section extends our NaNNS results to $p$-mean nearest neighbor search ($p$-NNS) in the single-attribute setting. We state our algorithm (\Cref{algo:p-mean-ann}) and present corresponding guarantees (\Cref{theorem:single-attribute-p-mean} and~\Cref{theorem:single-attribute-p-mean-approx-oracle}) for finding an  optimal solution for the $p$-NNS problem. 

Recall that, for any $p \in (-\infty, 1]$, the $p$-mean welfare of $c$ agents with utilities $w_1, \ldots, w_c$ is given by 
   $M_p(w_1, \ldots, w_c) \coloneqq \left(\frac{1}{c}\sum_{\ell =1}^c w_\ell^p\right)^\frac{1}{p}$.

Here, as in Section \ref{section:our-results}, the utility is defined as $u_\ell(S) = \sum_{v \in S \cap D_\ell} \ \sigma(q, v)$, for any subset of vectors $S$ and attribute $\ell \in [c]$. Also, with a fixed smoothing constant $\eta>0$, we will write $M_p(S) \coloneqq M_p\left(u_1(S) + \eta, \ldots, u_c(S) + \eta \right)$, and the $p$-NNS problem is defined as 
\begin{align}
\max_{S \subseteq P: \lvert S \rvert = k} M_p \left(S \right)
\end{align}

Throughout this section, we will write let $F_\ell(i) \coloneqq \left(\sum_{j=1}^i \sigma(q, v^\ell_{(j)}) + \eta \right)^p$, for each attribute $\ell \in [c]$; here, as before, $v^\ell_{(j)}$ denotes the $j^{\text{th}}$ most similar vector to $q$ in $\widehat{D}_\ell$, the set of $k$ most similar vectors to $q$ from $D_\ell$.

\begin{lemma}[Decreasing Marginals for $p > 0$]
\label{lemma:decreasing-marginal-of-p-mean-for-positive-p}
    Fix any $p \in (0, 1]$ and attribute $\ell \in [c]$. Then, for indices $1 \leq i' < i \leq k$, it holds that
    \begin{align*}
        F_\ell(i') - F_\ell(i'-1) \geq F_\ell(i) - F_\ell(i-1)~.
    \end{align*}
\end{lemma}

\begin{proof}
    Write $G(j) \coloneqq F_\ell(j) - F_\ell(j-1)$ for all $j \in [k]$, and $f_\ell(i) = \sum_{j=1}^i \sigma(q, v^\ell_{(j)})$ for $i \in [k]$. We will establish the lemma by showing that $G(j)$ is decreasing in $j$. Towards this, note that, for indices $j \geq 2$, we have  
    \begin{align*}
        &G(j - 1) - G(j) \\
        &= F_\ell(j-1) - F_\ell(j-2) - F_\ell(j) + F_\ell(j-1)\\
        &= 2 F_\ell(j-1) - (F_\ell(j) + F_\ell(j-2)) \\
        &= 2 (f_\ell(j-1) + \eta)^p - \left((f_\ell(j) + \eta)^p + (f_\ell(j-2) + \eta)^p\right) \\
        &= 2(f_\ell(j-2) + \eta)^p \left(\left(1 + \frac{\sigma(q, v^\ell_{(j-1)})}{f_\ell(j-2) + \eta}\right)^p -  \frac{1}{2}\left(1 + \left(1 + \frac{\sigma(q, v^\ell_{(j-1)}) + \sigma(q, v^\ell_{(j)})}{f_\ell(j-2) + \eta}\right)^p\right)\right) \\
        &\geq 2(f_\ell(j-2) + \eta)^p \left(\left(1 + \frac{\sigma(q, v^\ell_{(j-1)})}{f_\ell(j-2) + \eta}\right)^p -  \frac{1}{2}\left(1 + \left(1 + \frac{2\sigma(q, v^\ell_{(j-1)})}{f_\ell(j-2) + \eta}\right)^p\right)\right) \tag{$\sigma(q, v^\ell_{(j-1)}) \geq \sigma(q, v^\ell_{(j)})$; $x \mapsto x^p$ is increasing for $p \in (0, 1]$ and $x \geq 0$} \\
        &\geq 2(f_\ell(j-2) + \eta)^p \left(\left(1 + \frac{\sigma(q, v^\ell_{(j-1)})}{f_\ell(j-2) + \eta}\right)^p -  \left(\frac{1}{2}\cdot 1 + \frac{1}{2} \cdot \left(1 + \frac{2\sigma(q, v^\ell_{(j-1)})}{f_\ell(j-2) + \eta}\right)\right)^p\right) \tag{$x \mapsto x^p$ is concave for $p \in (0, 1]$ and $x \geq 0$} \\
        &= 0.
    \end{align*}
    Therefore, $G(j) \leq G(j-1)$ for each $2 \leq j \leq k$. Equivalently, for indices $1 \leq i' < i \leq k$, it holds that $G(i') \geq G(i)$. This completes the proof of the lemma. 
\end{proof}

\begin{lemma}[Increasing Marginals for $p < 0$]
\label{lemma:increasing-marginal-of-p-mean-for-negative-p}
    Fix any parameter $p \in (-\infty, 0)$ and attribute $\ell \in [c]$. Then, for indices $1 \leq i' < i \leq k$, it holds that
    \begin{align*}
        F_\ell(i') - F_\ell(i'-1) \leq F_\ell(i) - F_\ell(i-1)~.
    \end{align*}
\end{lemma}

\begin{proof}
    The proof is similar to that of~\Cref{lemma:decreasing-marginal-of-p-mean-for-positive-p}, except that we now seek the reverse inequality.  In particular, write $G(j) \coloneqq F_\ell(j) - F_\ell(j-1)$ for all $j \in [k]$, and $f_\ell(i) = \sum_{j=1}^i \sigma(q, v^\ell_{(j)})$ for $i \in [k]$. 
    
    To show that $G(j) \geq G(j-1)$ for all indices $2 \leq j \leq k$, note that 
         
    \begin{align*}
        &G(j - 1) - G(j) \\
        &= 2(f_\ell(j-2) + \eta)^p \left(\left(1 + \frac{\sigma(q, v^\ell_{(j-1)})}{f_\ell(j-2) + \eta}\right)^p -  \frac{1}{2}\left(1 + \left(1 + \frac{\sigma(q, v^\ell_{(j-1)}) + \sigma(q, v^\ell_{(j)})}{f_\ell(j-2) + \eta}\right)^p\right)\right) \\
        &\leq 2(f_\ell(j-2) + \eta)^p \left(\left(1 + \frac{\sigma(q, v^\ell_{(j-1)})}{f_\ell(j-2) + \eta}\right)^p -  \frac{1}{2}\left(1 + \left(1 + \frac{2\sigma(q, v^\ell_{(j-1)})}{f_\ell(j-2) + \eta}\right)^p\right)\right) \tag{$\sigma(q, v^\ell_{(j-1)}) \geq \sigma(q, v^\ell_{(j)})$; $x \mapsto x^p$ is decreasing for $p \in (-\infty, 0)$ and $x \geq 0$} \\
        &\leq 2(f_\ell(j-2) + \eta)^p \left(\left(1 + \frac{\sigma(q, v^\ell_{(j-1)})}{f_\ell(j-2) + \eta}\right)^p -  \left(\frac{1}{2}\cdot 1 + \frac{1}{2} \cdot \left(1 + \frac{2\sigma(q, v^\ell_{(j-1)})}{f_\ell(j-2) + \eta}\right)\right)^p\right) \tag{$x \mapsto x^p$ is convex for $p \in (-\infty, 0)$ and $x \geq 0$} \\
        &= 0~.
    \end{align*}
    Therefore, we have $G(j) \geq G(j-1)$ for all $2 \leq j \leq k$. That is, $G(i') \leq G(i)$   
    for indices $1 \leq i' < i \leq k$. Hence, the lemma stands proved. 
\end{proof}

\begin{algorithm}[t]
\caption{\texttt{p-Mean-ANN}: Algorithm for $p$-NNS in the single-attribute setting}
\label{algo:p-mean-ann}

\KwIn{Query $q \in \mathbb{R}^d$ and, for each attribute $\ell \in [c]$, 
the set of input vectors $D_\ell \subset \mathbb{R}^d$ and parameter $p \in (-\infty, 1] \setminus \{0\}$}
\label{line:p-ANN-algo-start}

For each $\ell \in [c]$, fetch the $k$ (exact or approximate) nearest neighbors of
$q \in \mathbb{R}^d$ from $D_\ell$.
Write $\widehat{D}_\ell \subseteq D_\ell$ to denote these sets.
\label{line:hat-D-p-mean}

For every $\ell \in [c]$ and each index $i \in [k]$, let
$v^\ell_{(i)}$ denote the $i$th most similar vector to $q$
in $\widehat{D}_\ell$.
\label{line:p-ANN-define-v^l_(i)}

Initialize subset $\Alg = \emptyset$, along with count
$k_\ell = 0$ and utility $w_\ell = 0$, for each $\ell \in [c]$.
\label{line:p-ANN-initialize-ALG}

\While{$|\Alg| < k$}{
    \If{$p \in (0,1]$}{
        $a \gets \argmax\limits_{\ell \in [c]}
        \Big((w_\ell + \eta + \sigma(q, v^\ell_{(k_\ell + 1)}))^p
        - (w_\ell + \eta)^p\Big)$.
        \Comment{Ties broken arbitrarily}
        \label{line:marginal-positive-p}
    }
    \ElseIf{$p < 0$}{
        $a \gets \argmin\limits_{\ell \in [c]}
        \Big((w_\ell + \eta + \sigma(q, v^\ell_{(k_\ell + 1)}))^p
        - (w_\ell + \eta)^p\Big)$.
        \Comment{Ties broken arbitrarily}
        \label{line:marginal-negative-p}
    }

    Update $\Alg \gets \Alg \cup \{ v^a_{(k_a + 1)} \}$ along with $w_a \gets w_a + \sigma(q, v^a_{(k_a + 1)})$ and     $k_a \gets k_a + 1$.
    \label{line:update-p-mean}
}

\Return{$\Alg$}

\end{algorithm}

\begin{lemma}
\label{lemma:size-match-p-mean-NNS}
In the single-attribute setting, let $\Alg$ be the subset of vectors returned by \Cref{algo:p-mean-ann} and $S$ be any subset of input vectors with the property that $|S \cap D_\ell| = |\Alg \cap D_\ell|$, for each $\ell \in [c]$. Then, $M_p(\Alg) \geq M_p(S)$.
\end{lemma}

\begin{proof}
Assume, towards a contradiction, that there exists a subset of input vectors $S$ that satisfies $|S \cap D_\ell| = |\Alg \cap D_\ell|$, for each $\ell \in [c]$, and still induces p-mean welfare strictly greater than that of $\Alg$. This strict inequality combined with the fact that $M_p(w_1, \ldots, w_c)$ is an increasing function of $w_i$s implies that there exists an attribute $a \in [c]$ with the property that the utility $u_a(S) > u_a(\Alg)$. That is, 
\begin{align}
\sum_{t \in S \cap D_a} \sigma(q, t) > \sum_{v \in \Alg \cap D_a} \sigma(q, v) \label{ineq:strict-util-p-mean}   
\end{align}
On the other hand, note that the construction of \Cref{algo:p-mean-ann} and the definition of $\widehat{D}_a$ ensure that the vectors in $\Alg \cap D_a$ are, in fact, the most similar to $q$ among all the vectors in $D_a$. This observation and the fact that $|S\cap D_a| = |\Alg \cap D_a|$ gives us $\sum_{v \in \Alg \cap D_a} \sigma(q, v) \geq \sum_{t \in S \cap D_a} \sigma(q, t)$. This equation, however, contradicts the strict inequality (\ref{ineq:strict-util-p-mean}). Therefore, by way of contradiction, we obtain that there does not exist a subset $S$ such that $|S \cap D_\ell| = |\Alg \cap D_\ell|$, for each $\ell \in [c]$, and $M_p(\Alg) < M_p(S)$. The lemma stands proved. 
\end{proof}

\begin{restatable}{theorem}{SingleAttributePMean}
\label{theorem:single-attribute-p-mean} 
In the single-attribute setting, given any query $q \in \mathbb{R}^d$ and an (exact) oracle \texttt{ENN} for $k$ most similar vectors from any set, Algorithm~\ref{algo:p-mean-ann} (\texttt{p-mean-ANN}) returns an optimal solution for $p$-NNS, i.e., it returns a size-$k$ subset $\Alg \subseteq P$ that satisfies 
       $\Alg \in \argmax_{S \subseteq P: \lvert S \rvert = k} M_p (S)$.
    Furthermore, the algorithm runs in time $O(kc) + \sum_{\ell=1}^c \texttt{ENN}(D_\ell, q)$, where \texttt{ENN}$(D_\ell, q)$ is the time required by the exact oracle to find $k$ most similar vectors to $q$ in $D_\ell$.
\end{restatable}

\begin{proof}
    The running time of the algorithm follows via arguments similar to the ones used in the proof of~\Cref{theorem:single-attribute}.

    For the correctness analysis, we divide the proof into two: $p \in (0, 1]$ and $p < 0$.

\noindent
    \textbf{Case 1}: $p \in (0, 1]$. Since $x \mapsto x^{p}$ is an increasing function for $x \geq 0$, the problem $\max_{S \subseteq P, \lvert S \rvert = k} M_p(S)$ is equivalent to $\max_{S \subseteq P, \lvert S \rvert = k} M_p(S)^p = \max_{S \subseteq P, \lvert S \rvert = k} \frac{1}{c} \sum_{\ell = 1}^c \left( u_\ell(S) + \eta\right)^p$. 
        
    The proof here is similar to that of ~\Cref{theorem:single-attribute}. Let $k_\ell = \lvert \Alg \cap D_\ell \rvert$ for all $\ell \in [c]$. Further, let $\Opt \in \argmax_{S \subseteq P, \lvert S \rvert = k} \frac{1}{c} \sum_{\ell = 1}^c \left( u_\ell(S) + \eta \right)^p$ and $k^*_\ell = \lvert \Opt \cap D_\ell \rvert$ for all $\ell \in [c]$, where $\Opt$ is chosen such that $\sum_{\ell = 1}^c \lvert k^*_\ell - k_\ell \rvert$ is minimized.

    We will prove that $\Opt$ satisfies $k^*_\ell = k_\ell$ for each $\ell \in [c]$. This guarantee, along with Lemma \ref{lemma:size-match-p-mean-NNS}, implies that, as desired, $\Alg$ is a p-mean welfare optimal solution. 

    Assume, towards a contradiction, that $k^*_\ell \neq k_\ell$ for some $\ell \in [c]$. Since $|\Opt| = |\Alg| = k$, there exist attributes $x, y \in [c]$ with the property that  
    \begin{align}
        k^*_x  < k_x \qquad & \text{ and } \qquad k^*_y  > k_y \label{ineq:opt-alg-sizes-p-mean-positive-p}
    \end{align}
    Recall that via~\Cref{lemma:decreasing-marginal-of-p-mean-for-positive-p}, we have for any pair of indices $1 \leq i' < i \leq k$, 
    \begin{align}
     F_\ell(i') - F_\ell(i'-1) \geq F_\ell(i) - F_\ell(i-1) \label{eqn:lemma-marginal-p-mean-positive-p}
    \end{align}  

    Next, note that for any attribute $\ell \in [c]$, if \Cref{algo:p-mean-ann}, at any point during its execution, has included $k'_\ell$ vectors of attribute $\ell$ in $\Alg$, then at that point the maintained utility $w_\ell = \sum_{j=1}^{k'_\ell}\sigma(q, v^\ell_{(j)})$. Hence, at the beginning of any iteration of the algorithm, if the $k'_\ell$ denotes the number of selected vectors of each attribute $\ell \in [c]$, then the marginals considered in Line \ref{line:marginal-positive-p} are $F_\ell \left(k'_\ell+1 \right) - F_\ell \left(k'_\ell \right)$. These observations and the selection criterion in Line \ref{line:marginal-positive-p} of the algorithm give us the following inequality for the counts $k_x=|\Alg \cap D_x|$ and $k_y=|\Alg \cap D_y|$ of the returned solution $\Alg$:
    
    \begin{align}
        F_x(k_x) - F_x(k_x - 1) \geq F_y(k_y+1) - F_y(k_y)  \label{ineq:greey-choice-p-mean-positive-p}
    \end{align}
        Specifically, equation (\ref{ineq:greey-choice-p-mean-positive-p}) follows by considering the iteration in which $k_x^{\text{th}}$ (last) vector of attribute $x$ was selected by the algorithm. Before that iteration the algorithm had selected $(k_x-1)$ vectors of attribute $x$, and let $k'_y$ denote the number of vectors with attribute $y$ that have been selected till that point. Note that $k'_y \leq k_y$. The fact that the $k_x^{\text{th}}$ vector was (greedily) selected in Line \ref{line:marginal-positive-p}, instead of including an additional vector of attribute $y$, gives $F_x(k_x) - F_x(k_x - 1) \geq F_y(k'_y+1) - F_y(k'_y) \geq F_y(k_y+1) - F_y(k_y)$; here, the last inequality follows from equation (\ref{eqn:lemma-marginal-p-mean-positive-p}). Hence, equation (\ref{ineq:greey-choice-p-mean-positive-p}) holds. 

        Moreover, 
        \begin{align}
            F_x(k^*_x + 1) - F_x(k^*_x) & \geq F_x(k_x) - F_x(k_x - 1) \tag{via eqns.~(\ref{ineq:opt-alg-sizes-p-mean-positive-p}) and (\ref{eqn:lemma-marginal-p-mean-positive-p})} \\
            & \geq F_y(k_y+1) - F_y(k_y) \tag{via eqn.~(\ref{ineq:greey-choice-p-mean-positive-p})} \\
            & \geq F_y(k^*_y) - F_y(k^*_y-1)  \label{ineq:swap-p-mean-positive-p}
        \end{align}
The last inequality follows from equations (\ref{ineq:opt-alg-sizes-p-mean-positive-p}) and (\ref{eqn:lemma-marginal-p-mean-positive-p}). 
         
    Recall that $v^\ell_{(i)}$ denotes the $i^{\text{th}}$ most similar (to $q$) vector in the set $\widehat{D}_\ell$. The definition of $\widehat{D}_\ell$ ensures that $v^\ell_{(i)}$ is, in fact, the $i^{\text{th}}$ most similar (to $q$) vector among the ones that have attribute $\ell$, i.e., $i^{\text{th}}$ most similar in all of $D_\ell$. Since $\Opt$ is an optimal solution, the $k^*_\ell = |\Opt \cap D_\ell|$ vectors of attribute $\ell$ in $\Opt$ are the most similar $k^*_\ell$ vectors from $D_\ell$. That is, $\Opt \cap D_\ell = \left\{v^\ell_{(1)}, \ldots, v^\ell_{(k^*_\ell)} \right\}$, for each $\ell \in [c]$. This observation and the definition of $F_\ell(\cdot)$ imply that the $p$-th power of $\Opt$'s $p$-mean welfare satisfies 
\begin{align}
    M_p(\Opt)^p = \frac{1}{c} \sum_{\ell=1}^c F_\ell(k^*_\ell). 
\end{align}
Now, consider a subset of vectors $S$ obtained from $\Opt$ by including vector $v^x_{(k^*_x + 1)}$ and removing $v^y_{(k^*_y)}$, i.e., $S = \left( \Opt \cup \left\{ v^x_{(k^*_x + 1)} \right\} \right) \setminus \left\{v^y_{(k^*_y)} \right\}$. Note that 
\begin{align*}
    M_p(S)^p - M_p(\Opt)^p & = \frac{1}{c}\Big( F_x(k^*_x + 1) - F_x(k^*_x) \Big)  + \frac{1}{c}\Big(  F_y(k^*_y-1) - F_y(k^*_y) \Big) \\ 
    & \geq 0 \tag{via eqn.~(\ref{ineq:swap-p-mean-positive-p})} 
\end{align*}
Hence, $M_p(S) \geq M_p(\Opt)$. Given that $\Opt$ is a $p$-mean welfare optimal solution, the last inequality must hold with equality, $M_p(S) = M_p(\Opt)$, i.e., $S$ is an optimal solution as well. This, however, contradicts the choice of $\Opt$ as an optimal solution that minimizes $\sum_{\ell=1}^c | k^*_\ell - k_\ell|$ -- note that $\sum_{\ell=1}^c \left|\widehat{k}_\ell - k_\ell \right| < \sum_{\ell=1}^c \left| k^*_\ell - k_\ell \right|$, where $\widehat{k}_\ell \coloneqq |S \cap D_\ell|$. 

Therefore, by way of contradiction, we obtain that $|\Opt \cap D_\ell| = |\Alg \cap D_\ell|$ for each $\ell \in [c]$. As mentioned previously, this guarantee, along with Lemma \ref{lemma:size-match-p-mean-NNS}, implies that $\Alg$ is a p-mean welfare optimal solution. This completes the proof of the theorem for the case $p \in (0, 1]$. \\

\noindent
\textbf{Case 2}: $p < 0$. The arguments here are similar to the ones used in the previous case. However, since the exponent parameter $p$ is negative, the key inequalities are reversed. For completeness, we present the proof below.

Note that, with the $p<0$, the map $x \mapsto x^{p}$ is a decreasing function for $x \geq 0$. Hence, the optimization problem $\max_{S \subseteq P, \lvert S \rvert = k} M_p(S)$ is equivalent to $\min_{S \subseteq P, \lvert S \rvert = k} M_p(S)^p = \min_{S \subseteq P, \lvert S \rvert = k} \frac{1}{c} \sum_{\ell = 1}^c \left( u_\ell(S) + \eta \right)^p$.  

Let $k_\ell = \lvert \Alg \cap D_\ell \rvert$ for all $\ell \in [c]$. Further, let $\Opt \in \argmin_{S \subseteq P, \lvert S \rvert = k} \frac{1}{c} \sum_{\ell = 1}^c \left( u_\ell(S) + \eta \right)^p$ and $k^*_\ell = \lvert \Opt \cap D_\ell \rvert$ for all $\ell \in [c]$, where $\Opt$ is chosen such that $\sum_{\ell = 1}^c \lvert k^*_\ell - k_\ell \rvert$ is minimized.

    We will prove that $\Opt$ satisfies $k^*_\ell = k_\ell$ for each $\ell \in [c]$. This guarantee, along with Lemma \ref{lemma:size-match-p-mean-NNS}, implies that, as desired, $\Alg$ is a p-mean welfare optimal solution. 

    Assume, towards a contradiction, that $k^*_\ell \neq k_\ell$ for some $\ell \in [c]$. Since $|\Opt| = |\Alg| = k$, there exist attributes $x, y \in [c]$ with the property that  
    \begin{align}
        k^*_x  < k_x \qquad & \text{ and } \qquad k^*_y  > k_y \label{ineq:opt-alg-sizes-p-mean-negative-p}
    \end{align}
    Recall that via~\Cref{lemma:increasing-marginal-of-p-mean-for-negative-p}, we have for any pair of indices $1 \leq i' < i \leq k$, 
    \begin{align}
     F_\ell(i') - F_\ell(i'-1) \leq F_\ell(i) - F_\ell(i-1) \label{eqn:lemma-marginal-p-mean-negative-p}
    \end{align}  

    Next, note that for any attribute $\ell \in [c]$, if \Cref{algo:p-mean-ann}, at any point during its execution, has included $k'_\ell$ vectors of attribute $\ell$ in $\Alg$, then at that point the maintained utility $w_\ell = \sum_{j=1}^{k'_\ell}\sigma(q, v^\ell_{(j)})$. Hence, at the beginning of any iteration of the algorithm, if the $k'_\ell$ denotes the number of selected vectors of each attribute $\ell \in [c]$, then the marginals considered in Line \ref{line:marginal-negative-p} are $F_\ell \left(k'_\ell+1 \right) - F_\ell \left(k'_\ell \right)$. These observations and the selection criterion in Line \ref{line:marginal-negative-p} of the algorithm give us the following inequality for the counts $k_x=|\Alg \cap D_x|$ and $k_y=|\Alg \cap D_y|$ of the returned solution $\Alg$:
    \begin{align}
        F_x(k_x) - F_x(k_x - 1) \leq F_y(k_y+1) - F_y(k_y)  \label{ineq:greey-choice-p-mean-negative-p}
    \end{align}
        Specifically, equation (\ref{ineq:greey-choice-p-mean-negative-p}) follows by considering the iteration in which $k_x^{\text{th}}$ (last) vector of attribute $x$ was selected by the algorithm. Before that iteration the algorithm had selected $(k_x-1)$ vectors of attribute $x$, and let $k'_y$ denote the number of vectors with attribute $y$ that have been selected till that point. Note that $k'_y \leq k_y$. The fact that the $k_x^{\text{th}}$ vector was (greedily) selected in Line \ref{line:marginal-negative-p}, instead of including an additional vector of attribute $y$, gives $F_x(k_x) - F_x(k_x - 1) \leq F_y(k'_y+1) - F_y(k'_y) \leq F_y(k_y+1) - F_y(k_y)$; here, the last inequality follows from equation (\ref{eqn:lemma-marginal-p-mean-negative-p}). Hence, equation (\ref{ineq:greey-choice-p-mean-negative-p}) holds. 

        Moreover, 
        \begin{align}
            F_x(k^*_x + 1) - F_x(k^*_x) & \leq F_x(k_x) - F_x(k_x - 1) \tag{via eqns.~(\ref{ineq:opt-alg-sizes-p-mean-negative-p}) and (\ref{eqn:lemma-marginal-p-mean-negative-p})} \\
            & \leq F_y(k_y+1) - F_y(k_y) \tag{via eqn.~(\ref{ineq:greey-choice-p-mean-negative-p})} \\
            & \leq F_y(k^*_y) - F_y(k^*_y-1)  \label{ineq:swap-p-mean-negative-p}
        \end{align}
The last inequality follows from equations (\ref{ineq:opt-alg-sizes-p-mean-negative-p}) and (\ref{eqn:lemma-marginal-p-mean-negative-p}). 
         
    Recall that $v^\ell_{(i)}$ denotes the $i^{\text{th}}$ most similar (to $q$) vector in the set $\widehat{D}_\ell$. The definition of $\widehat{D}_\ell$ ensures that $v^\ell_{(i)}$ is, in fact, the $i^{\text{th}}$ most similar (to $q$) vector among the ones that have attribute $\ell$, i.e., $i^{\text{th}}$ most similar in all of $D_\ell$. Since $\Opt$ is an optimal solution, the $k^*_\ell = |\Opt \cap D_\ell|$ vectors of attribute $\ell$ in $\Opt$ are the most similar $k^*_\ell$ vectors from $D_\ell$. That is, $\Opt \cap D_\ell = \left\{v^\ell_{(1)}, \ldots, v^\ell_{(k^*_\ell)} \right\}$, for each $\ell \in [c]$. This observation and the definition of $F_\ell(\cdot)$ imply that the $p$-th power of $\Opt$'s $p$-mean welfare satisfies 
\begin{align}
    M_p(\Opt)^p = \frac{1}{c} \sum_{\ell=1}^c F_\ell(k^*_\ell). 
\end{align}
Now, consider a subset of vectors $S$ obtained from $\Opt$ by including vector $v^x_{(k^*_x + 1)}$ and removing $v^y_{(k^*_y)}$, i.e., $S = \left( \Opt \cup \left\{ v^x_{(k^*_x + 1)} \right\} \right) \setminus \left\{v^y_{(k^*_y)} \right\}$. Note that 
\begin{align*}
    M_p(S)^p - M_p(\Opt)^p & = \frac{1}{c}\Big( F_x(k^*_x + 1) - F_x(k^*_x) \Big)  + \frac{1}{c}\Big(  F_y(k^*_y-1) - F_y(k^*_y) \Big) \\ 
    & \leq 0 \tag{via eqn.~(\ref{ineq:swap-p-mean-negative-p})} 
\end{align*}
Hence, $M_p(S) \geq M_p(\Opt)$. Given that $\Opt$ is a $p$-mean welfare optimal solution, the last inequality must hold with equality, $M_p(S) = M_p(\Opt)$, i.e., $S$ is an optimal solution as well. This, however, contradicts the choice of $\Opt$ as an optimal solution that minimizes $\sum_{\ell=1}^c | k^*_\ell - k_\ell|$ -- note that $\sum_{\ell=1}^c \left|\widehat{k}_\ell - k_\ell \right| < \sum_{\ell=1}^c \left| k^*_\ell - k_\ell \right|$, where $\widehat{k}_\ell \coloneqq |S \cap D_\ell|$. 

Therefore, by way of contradiction, we obtain that $|\Opt \cap D_\ell| = |\Alg \cap D_\ell|$ for each $\ell \in [c]$. As mentioned previously, this guarantee, along with Lemma \ref{lemma:size-match-p-mean-NNS}, implies that $\Alg$ is a p-mean welfare optimal solution. This completes the proof of the theorem for the case $p < 0$.

Combining the two cases, we have that the theorem holds for all $p \in (-\infty, 1] \setminus \{0\}$. 
\end{proof}

\begin{restatable}{corollary}{SingleAttributePMeanApproximateOracle}
\label{theorem:single-attribute-p-mean-approx-oracle}
    In the single-attribute setting, given any query $q \in \mathbb{R}^d$ and an $\alpha$-approximate oracle \texttt{ANN} for $k$ most similar vectors from any set, Algorithm~\ref{algo:p-mean-ann} (\texttt{p-mean-ANN}) returns an $\alpha$-approximate solution for $p$-NNS, i.e., it returns a size-$k$ subset $\Alg \subseteq P$ with  %
        $M_p(\Alg) \geq \alpha \max_{S \subseteq P: \  \lvert S \rvert = k} M_p (S)$.
    The algorithm runs in time $O(kc) + \sum\limits_{\ell=1}^c \texttt{ANN}(D_\ell, q)$, with \texttt{ANN}$(D_\ell, q)$ being the time required by the approximate oracle to find $k$ similar vectors to $q$ in $D_\ell$.
\end{restatable}

\begin{proof}
    The running time of the algorithm follows via an argument similar to one used in the proof of~\Cref{theorem:single-attribute-p-mean}. Hence, we only argue correctness here. 
    
    For every $\ell \in [c]$, let the $\alpha$-approximate oracle return $\widehat{D}_\ell$. Recall that $v^\ell_{(i)}$, $i \in [k]$, denotes the $i^{\text{th}}$ most similar point to $q$ in the set $\widehat{D}_\ell$. Further, for every $\ell \in [c]$, let $D^*_{\ell}$ be the set of $k$ most similar points to $q$ within $D_\ell$ and, for each $i \in [k]$, define $v^{*\ell}_{(i)}$ to be the $i^{\text{th}}$ most similar point to $q$ in $D^*_{\ell}$. Recall that by the guarantee of the $\alpha$-approximate NNS oracle, we have $\sigma(q, v^\ell_{(i)}) \geq \alpha \  \sigma(q, v^{*\ell}_{(i)})$ for each $i \in [k]$. Let $\Opt$ be an optimal solution to the $p$-NNS problem. Note that, for each attribute $\ell in [c]$, the optimal solution $\Opt$ contains in it the $k^*_\ell$ most similar vectors with attribute $\ell$. 
    
    Finally, let $\widehat{\Opt}$ be the optimal solution to the $p$-NNS problem when the set of vectors to search over is $P = \cup_{\ell \in [c]} \widehat{D}_{\ell}$. 

    By arguments similar to the ones used in the proof of~\Cref{theorem:single-attribute-p-mean}, we have $M_p(\Alg) = M_p(\widehat{\Opt})$. Therefore
    {\allowdisplaybreaks
    \begin{align*}
        M_p(\Alg) &= M_p(\widehat{\Opt}) \\
        &\geq \left(\frac{1}{c}\sum_{\ell \in [c]} \left(\sum_{i=1}^{k^*_\ell} \sigma(q, v^\ell_{(i)} ) + \eta \right)^p\right)^{\frac{1}{p}} \tag{$\cup_{\ell \in [c]: k^*_\ell \geq 1} \{v^\ell_{(1)}, \ldots, v^\ell_{(k^*_\ell)}\}$ is a feasible solution} \\
        &\geq \left(\frac{1}{c}\sum_{\ell \in [c]} \left(\sum_{i=1}^{k^*_\ell} \alpha \sigma(q, v^{*\ell}_{(i)} ) + \eta \right)^p\right)^{\frac{1}{p}} \tag{by $\alpha$-approximate guarantee of the oracle, and $M_p$ is increasing in its argument; $k^*_\ell \leq k$} \\
        &\geq \left(\frac{1}{c}\sum_{\ell \in [c]} \alpha^p \left(\sum_{i=1}^{k^*_\ell}  \sigma(q, v^{*\ell}_{(i)} ) + \eta \right)^p\right)^{\frac{1}{p}} \tag{$\alpha \in (0, 1)$} \\
        &= \alpha\ M_p(\Opt) \tag{definition of $\Opt$}
    \end{align*}
    }
The corollary stands proved.
\end{proof}

\section{Experimental Evaluation and Analysis}
\label{appendix:experiments}
In this section, we present additional experimental results to further validate the performance of \ouralgo\ and compare it with existing methods. The details of the evaluation metrics and the description of the datasets used in our study are already presented in~\Cref{section:experiments}. Here, we report results for the single-attribute setting (\Cref{appendix:exp-single-attribute}), where we compare the approximation ratio alongside all diversity metrics for $k=10$ and $k=50$. We also include recall values for both $k=10$ and $k=50$ (\Cref{appendix:recallEntropy}). The key observation in all these plots is that the NSW objective effectively strikes a balance between relevance and diversity without having to specify any ad hoc constraints. Furthermore, we report experimental results for the multi-attribute setting on both a synthetic dataset (\texttt{Sift1m}) and a real-world dataset (\arxiv). Finally, we experimentally validate the performance-efficiency trade-offs of a faster heuristic variant of \pMeanANN\ (detailed in Appendix~\ref{subsec:FetchNashUnionNash}) that can be used in addition to any existing (standard) ANN algorithm.

\subsection{Balancing Relevance and Diversity: Single-attribute Setting}
\label{appendix:exp-single-attribute}
In this experiment, we evaluate how well \pMeanANN\ (and the special case of $p$ = $0$, \ouralgo) balances relevance and diversity in the single-attribute setting. We begin by examining the tradeoff between approximation ratio and entropy achieved by our algorithms on additional datasets beyond those used in the main paper. Moreover, we also report results for other diversity metrics such as the inverse Simpson index (\Cref{appendix:paretoSingleAttriInvSimpsonAprxRatio}) and the number of distinct attributes appearing in the $k$ neighbors (\Cref{appendix:paretoSingleAttriDistinctCountAprxRatio}) retrieved by our algorithms. These experiments corroborate the findings in the main paper, namely, \ouralgo\ and \pMeanANN\ are able to strike a balance between relevance and diversity whereas \diskann\ only optimizes for relevance (hence low diversity) and \divann\ only optimizes for diversity (hence low relevance).

\subsubsection{Performance of \ouralgo}
We report the results for different datasets in Figures~\ref{fig:AppendixResultsSingleAttriSIFTCLus}, \ref{fig:AppendixResultsSingleAttriSIFTProb}, \ref{fig:AppendixResultsSingleAttriDeepProb}, and \ref{fig:AppendixResultsSingleAttriArxiv}. On the \siftC\ dataset (\Cref{fig:AppendixResultsSingleAttriSIFTCLus}), \ouralgo\ achieves entropy close to that of the most diverse solution (\divann\ with $k' = 1$) in both $k=10$ and $k=50$ cases. Moreover, \ouralgo\ achieves significantly higher approximation ratio than \divann\ in both $k=10$ and $k=50$ cases when $k'=1$. For $k=10$ case, \ouralgo\ Pareto dominates \divann\ even with the relaxed constraint of $k' = 5$ for $k = 10$. When the number of required neighbors is increased to $k = 50$, no other method Pareto dominates \ouralgo. Similar observations hold for the \siftP\ (\Cref{fig:AppendixResultsSingleAttriSIFTProb}) and \deepP\ (\Cref{fig:AppendixResultsSingleAttriDeepProb}) datasets. In the results on the \arxiv\ dataset (\Cref{fig:AppendixResultsSingleAttriArxiv}) with $k = 10$, we observe that \divann\ already achieves a high approximation ratio. However, \ouralgo\ matches the entropy of \divann\ with $k' = 1$ while improving on the approximation ratio. For $k = 50$, \ouralgo\ nearly matches the entropy of \divann\ with $k'=1,2$ whereas it significantly improves on the approximation ratio. In summary, the experimental results clearly demonstrate the ability of \ouralgo\ to adapt to the varying nature of queries and consistently strike a balance between relevance and diversity.

\begin{figure}[t!]
\centering
\begin{tabular}{@{}c@{}c@{}c@{}c@{}}
\includegraphics[width=0.48\textwidth]{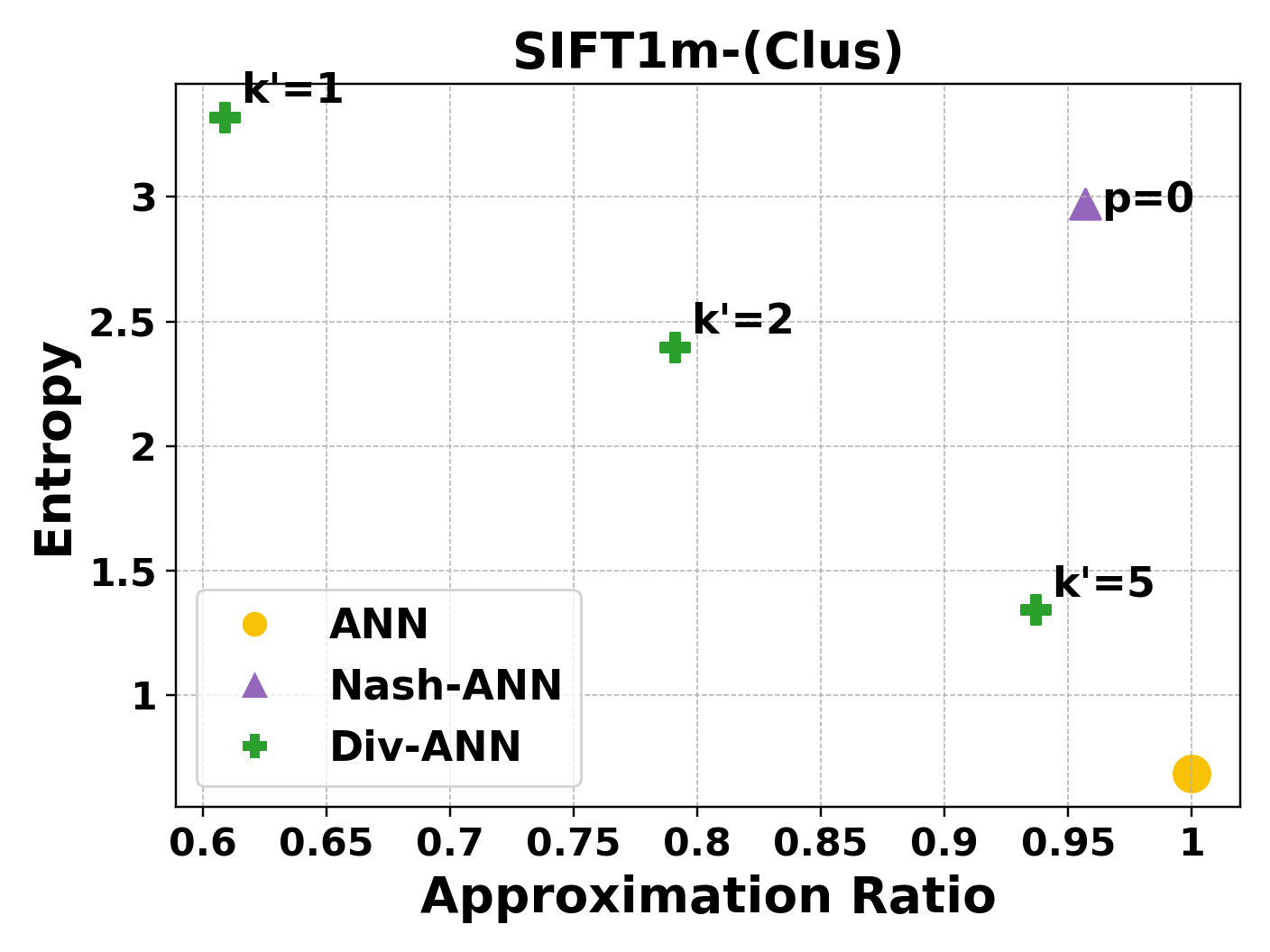} &
\includegraphics[width=0.48\textwidth]{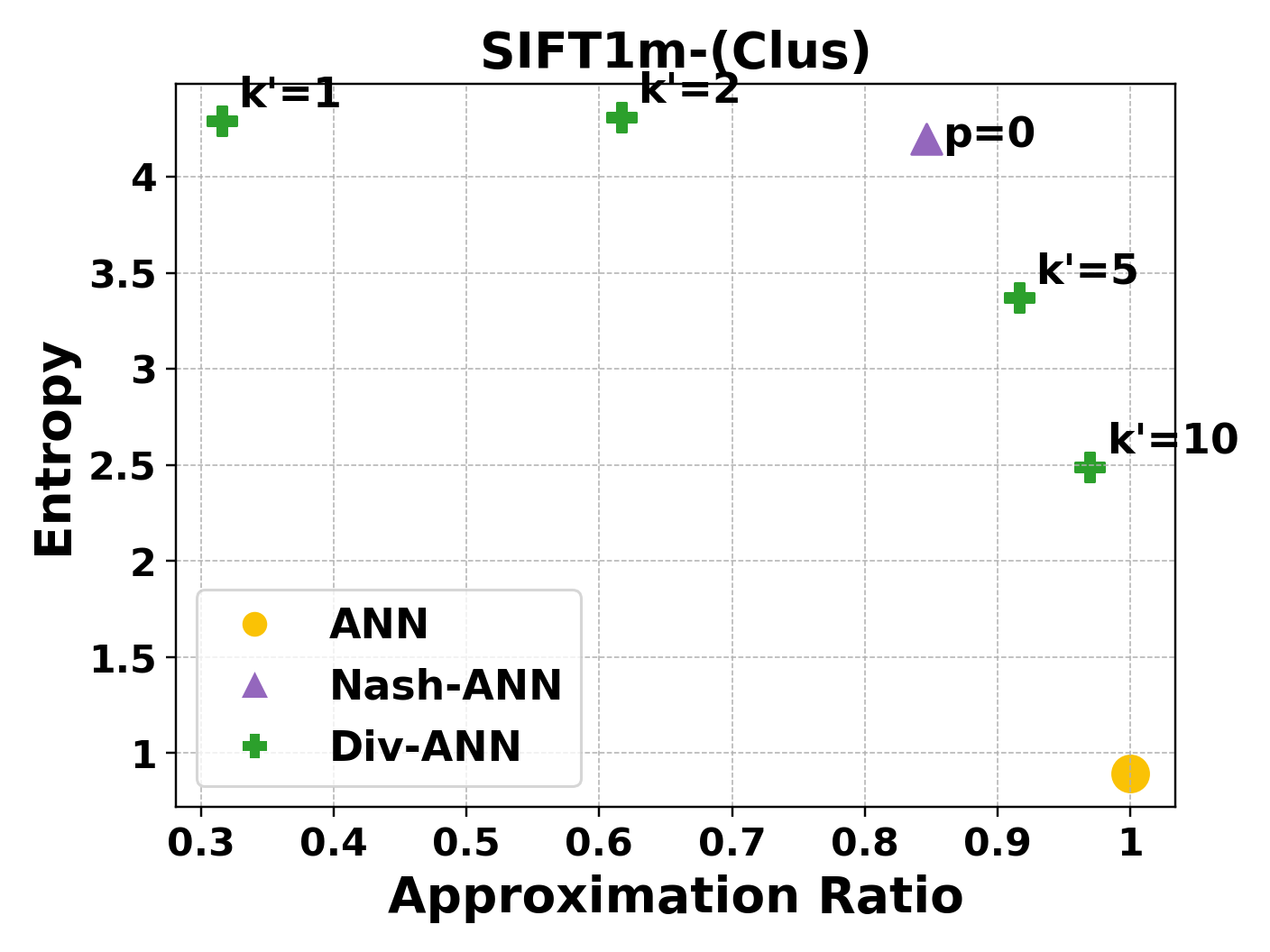} \\
\end{tabular}
\caption{The plots show approximation ratio versus entropy trade-offs for various algorithms for $k$  = $10$ \textbf{(Left)} and  $k$  = $50$ \textbf{(Right)} in  single-attribute setting  on \siftC\ dataset.}
\label{fig:AppendixResultsSingleAttriSIFTCLus}
\end{figure}

\begin{figure}[t!]
\centering
\begin{tabular}{@{}c@{}c@{}c@{}c@{}}
\includegraphics[width=0.48\textwidth]{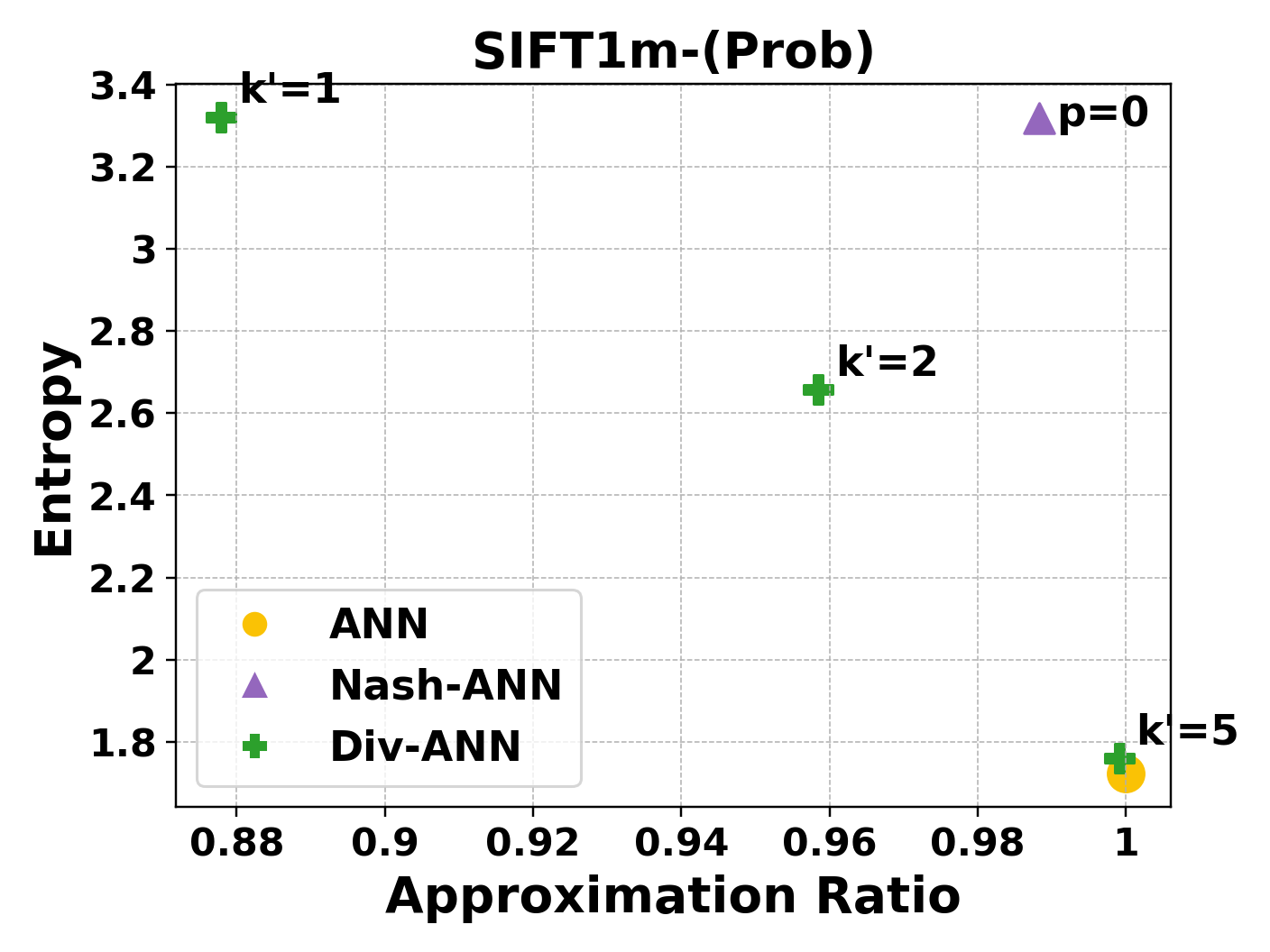} &
\includegraphics[width=0.48\textwidth]{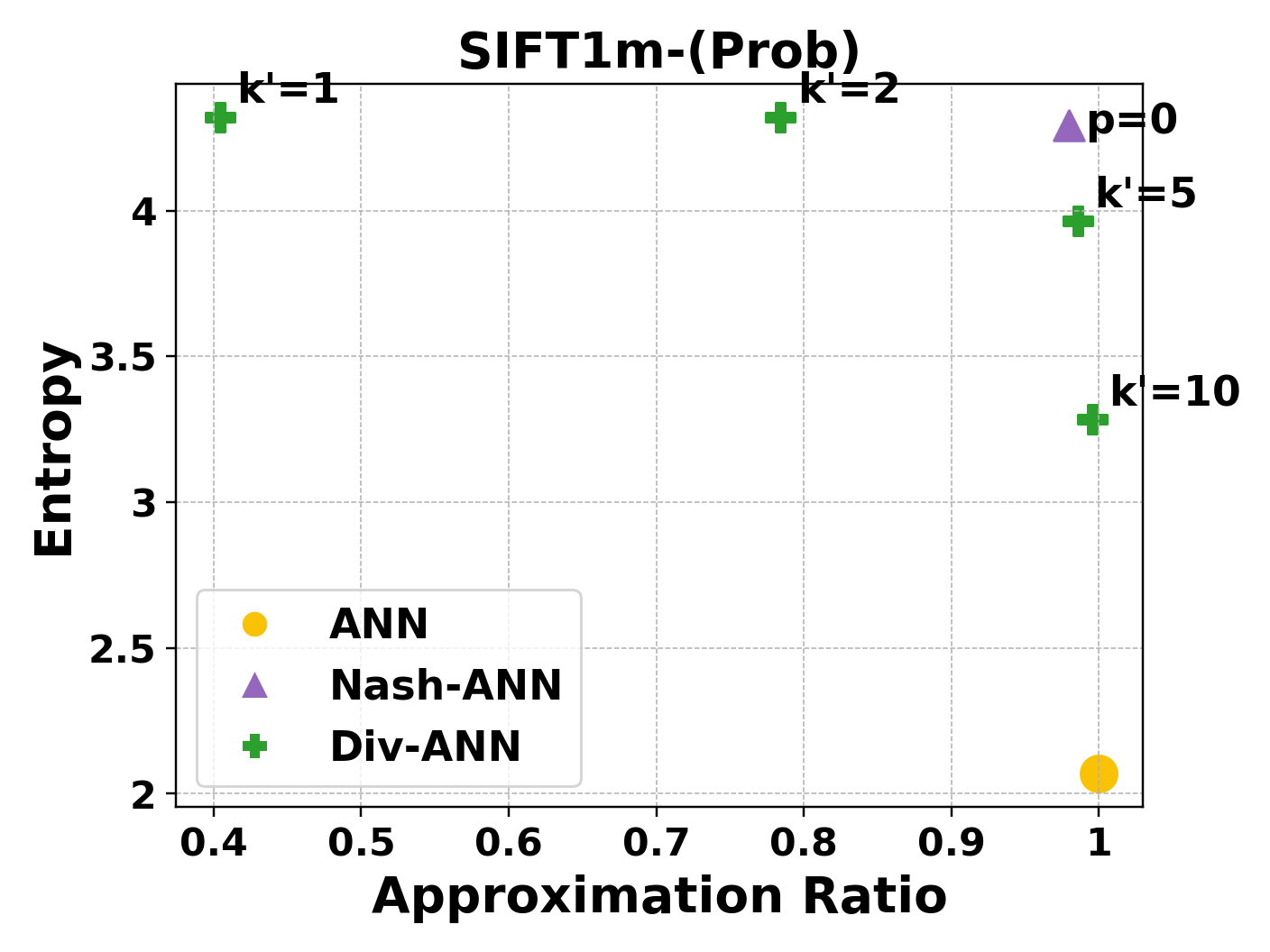} \\
\end{tabular}
\caption{The plots show approximation ratio versus entropy trade-offs for various algorithms for $k$  = $10$ \textbf{(Left)} and  $k$  = $50$ \textbf{(Right)} in  single-attribute setting  on \siftP\ dataset.}
\label{fig:AppendixResultsSingleAttriSIFTProb}
\end{figure}

\begin{figure}[t!]
\centering
\begin{tabular}{@{}c@{}c@{}c@{}c@{}}
\includegraphics[width=0.48\textwidth]{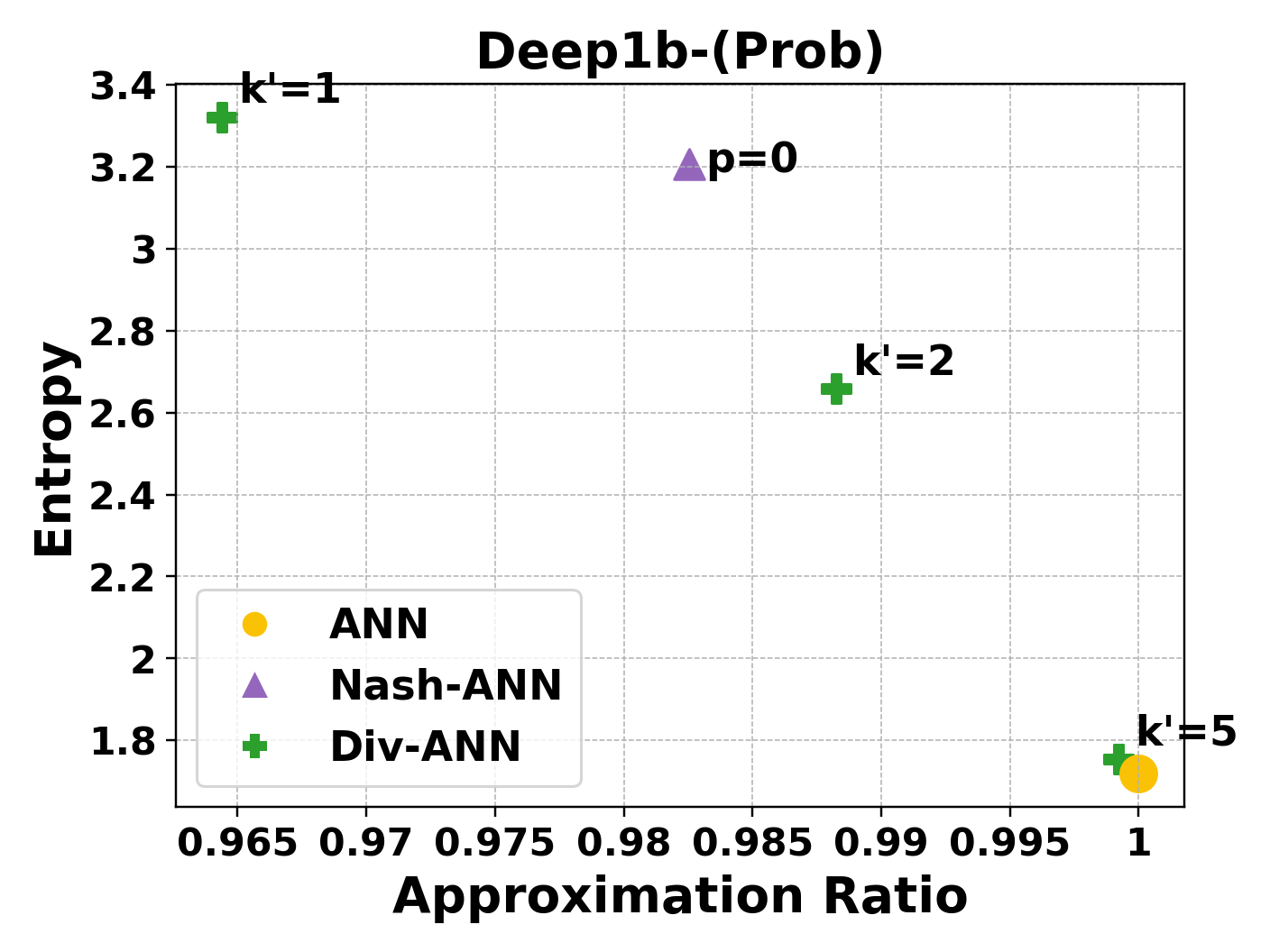} &
\includegraphics[width=0.48\textwidth]{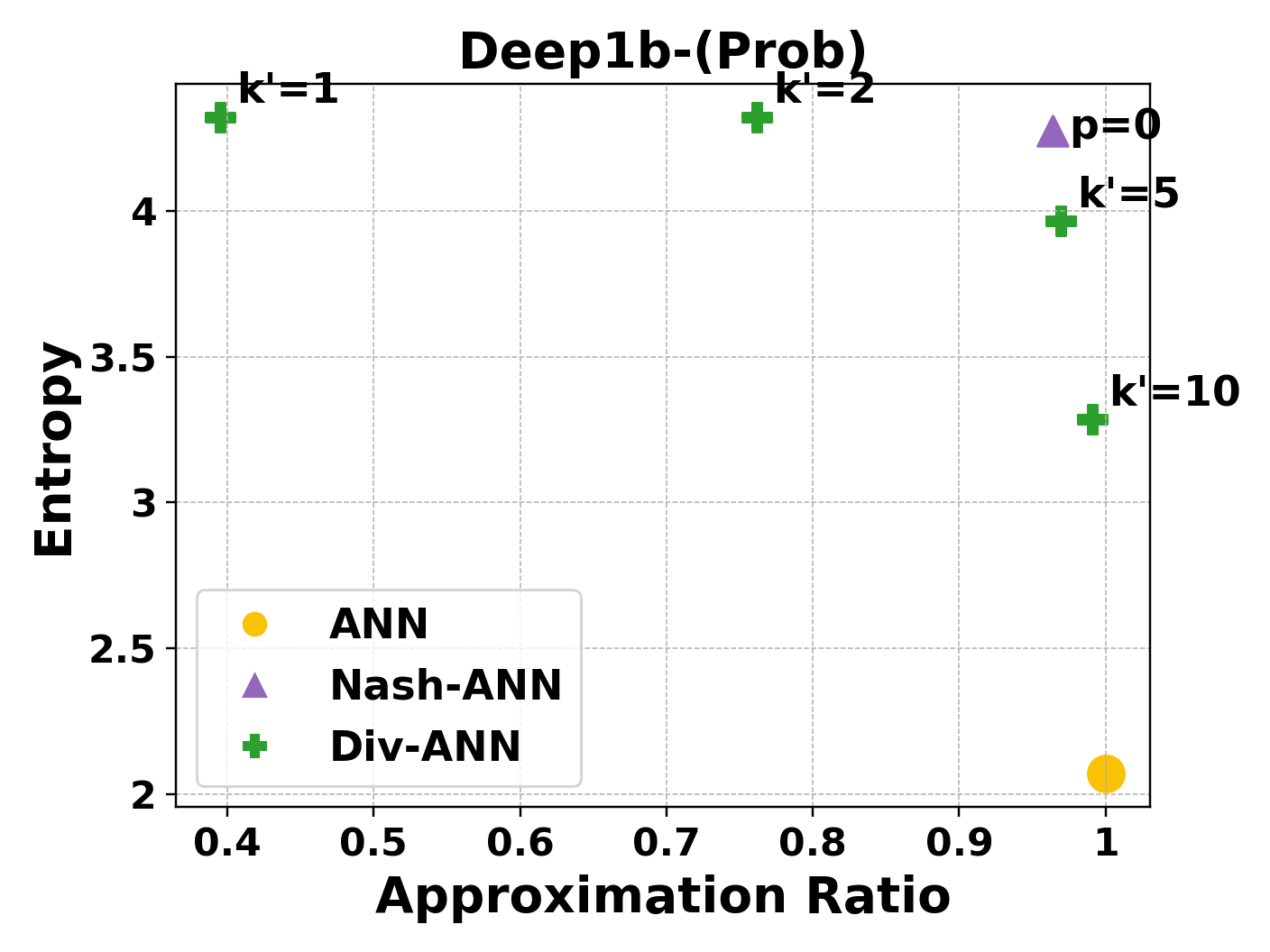} \\
\end{tabular}
\caption{The plots show approximation ratio versus entropy trade-offs for various algorithms for $k$  = $10$ \textbf{(Left)} and  $k$  = $50$ \textbf{(Right)} in  single-attribute setting  on \deepP\ dataset.}
\label{fig:AppendixResultsSingleAttriDeepProb}
\end{figure}

\begin{figure}[t!]
\centering
\begin{tabular}{@{}c@{}c@{}c@{}c@{}}
\includegraphics[width=0.48\textwidth]{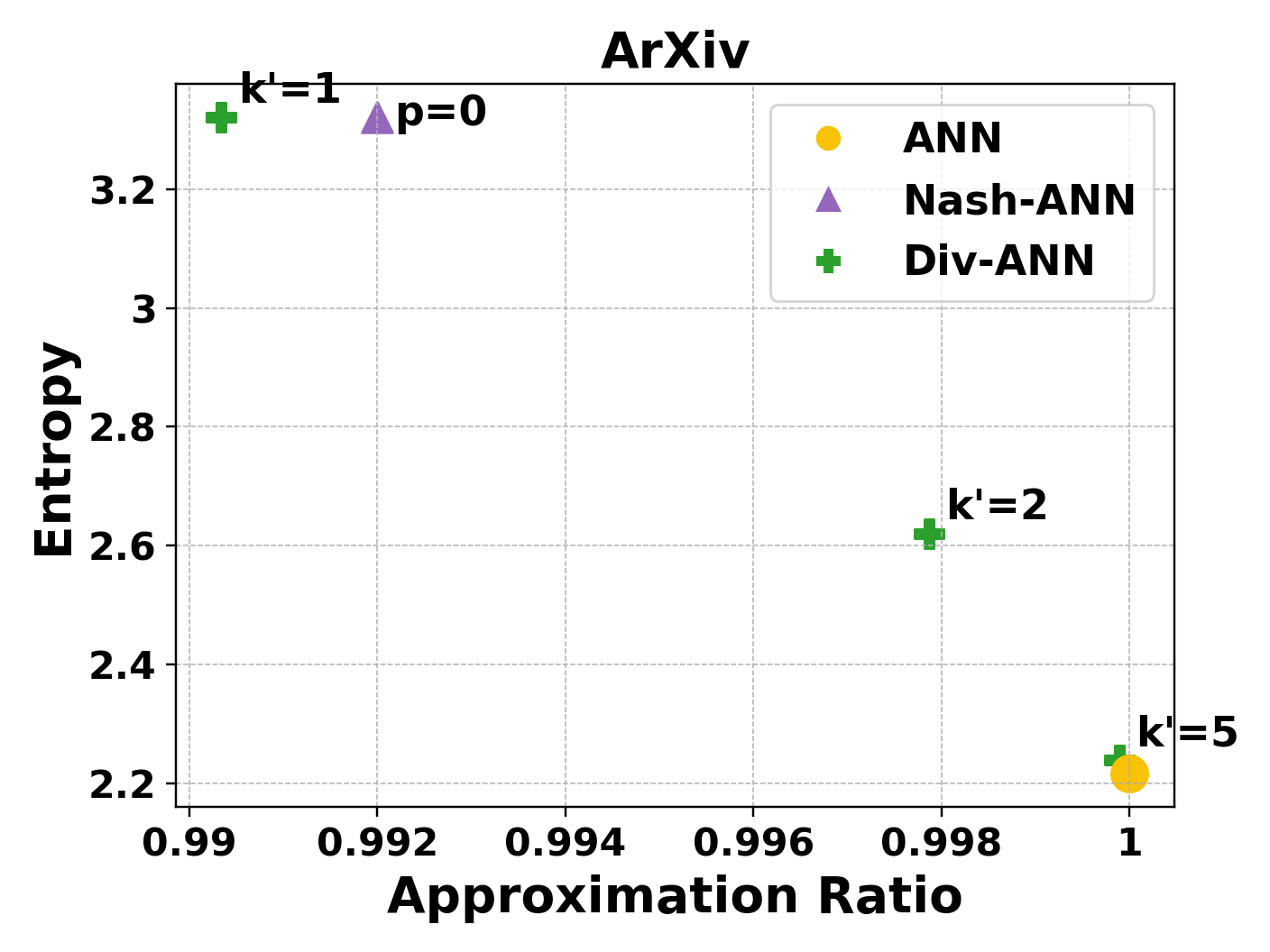} &
\includegraphics[width=0.48\textwidth]{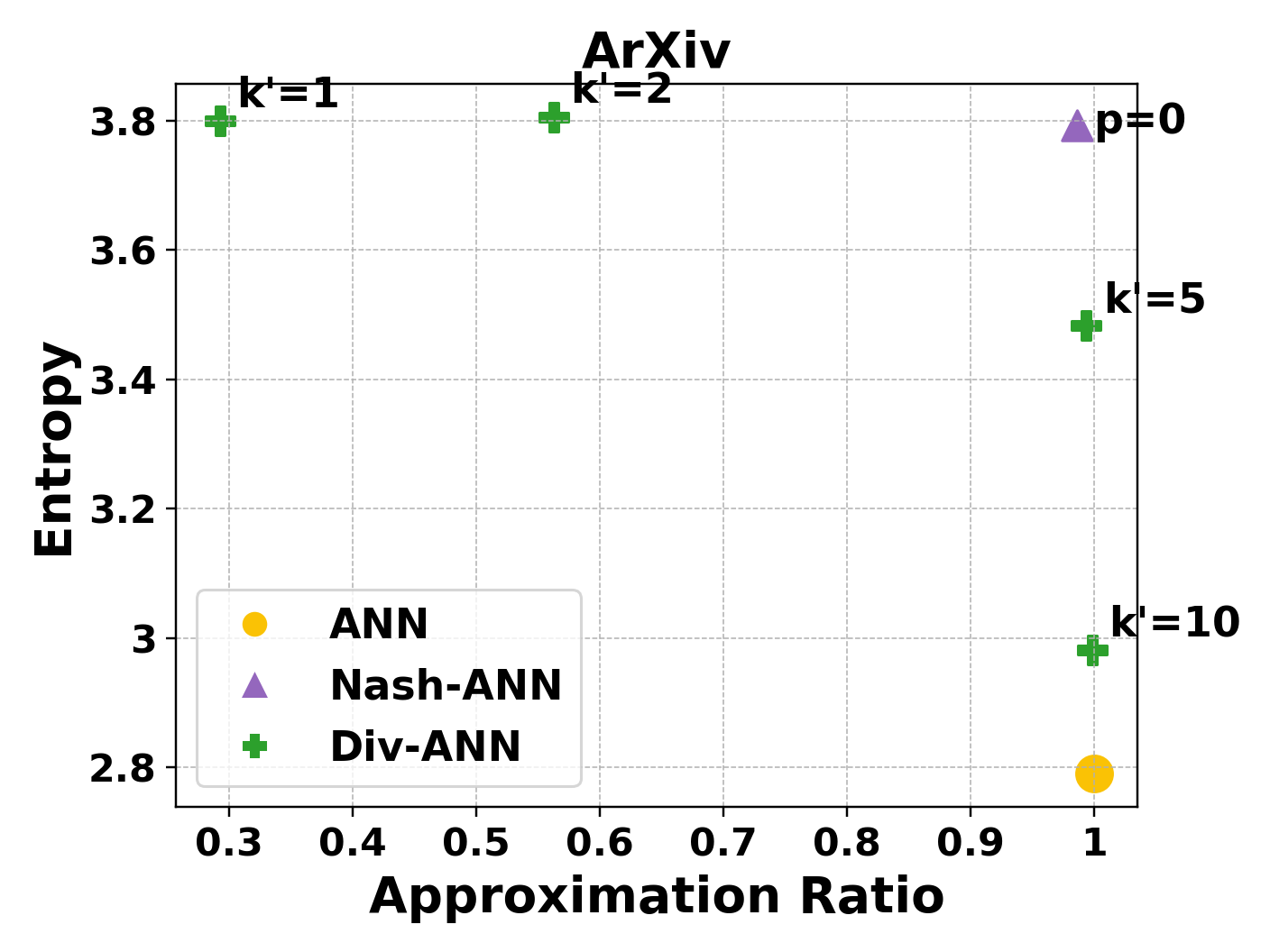} \\
\end{tabular}
\caption{The plots show approximation ratio versus entropy trade-offs for various algorithms for $k$  = $10$ \textbf{(Left)} and  $k$  = $50$ \textbf{(Right)} in  single-attribute setting on \arxiv\ dataset.}
\label{fig:AppendixResultsSingleAttriArxiv}
\end{figure}

\subsubsection{Performance of \pMeanANN}
In this set of experiments, we study the effect on trade-off between approximation ratio and entropy when the parameter $p$ in the $p$-NNS objective is varied over a range. Recall that the $p$-NNS problem with $p \to 0$ corresponds to the NaNNS problem, and with $p=1$, corresponds to the NNS problem. We experiment with values of $p \in \{-10, -1, -0.5, 0, 0.5, 1\}$ by running our algorithm \pMeanANN\ (\Cref{algo:p-mean-ann}) on the various datasets. 
The results are shown in Figures~\ref{fig:AppendixResultsSingleAttriSIFTCLusPTrend}, \ref{fig:AppendixResultsSingleAttriSIFTProbPTrend}, \ref{fig:AppendixResultsSingleAttriDeepProbPTrend}, and \ref{fig:AppendixResultsSingleAttriArxivPTrend}. We observe across all datasets for both $k=10$ and $k=50$ that, as $p$ decreases from $1$, the entropy increases but the approximation ratio decreases. This highlights the key intuition that as $p$ decreases, the behavior changes from utilitarian welfare ($p = 1$ aligns exactly with \diskann) to egalitarian welfare (more attribute-diverse). In other words, the parameter $p$ allows us to smoothly interpolate between complete relevance (the standard NNS with $p=1$) and complete diversity ($p \to -\infty$).

\begin{figure}[t!]
\centering
\begin{tabular}{@{}c@{}c@{}c@{}c@{}}
\includegraphics[width=0.48\textwidth]{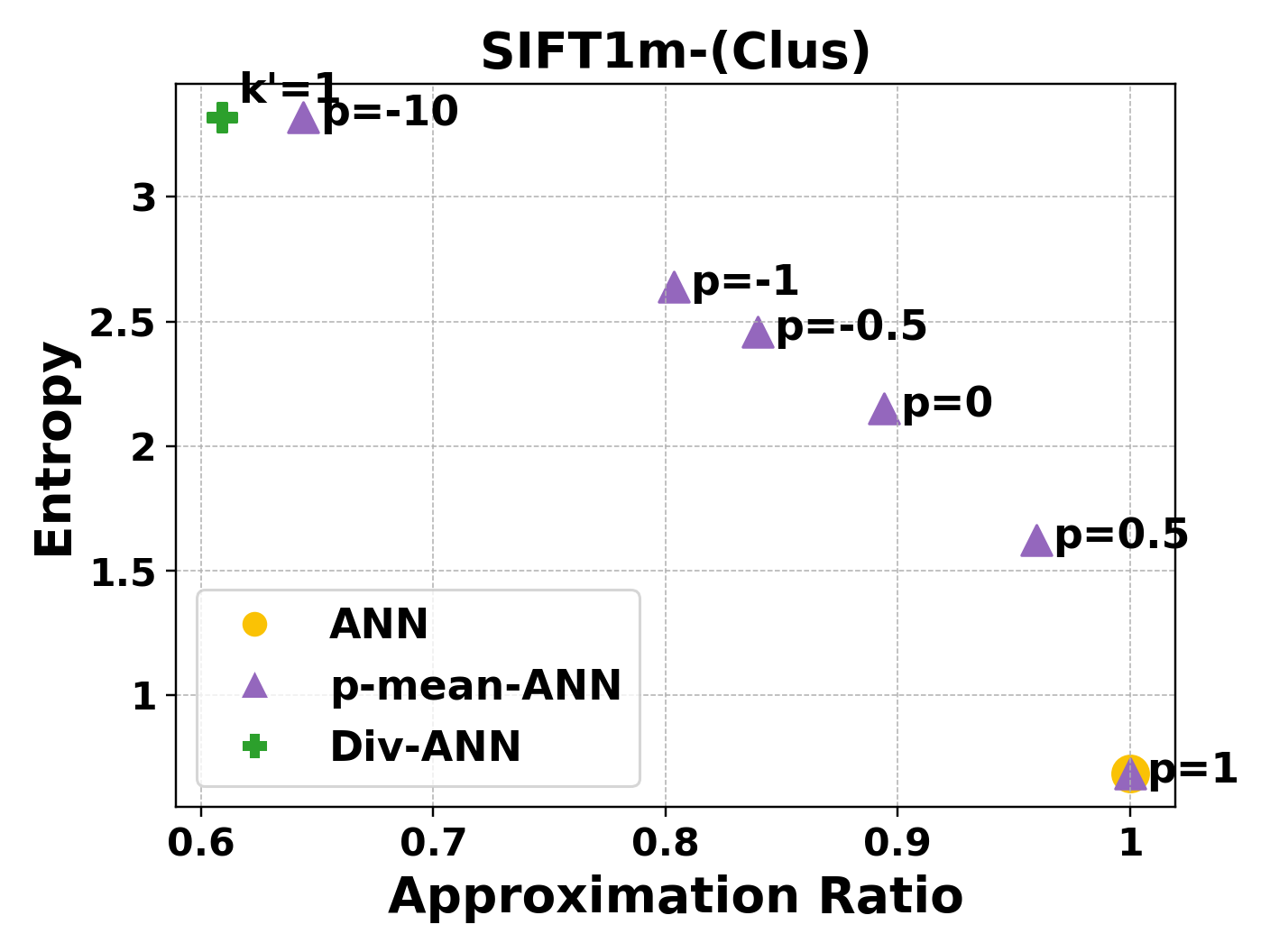} &
\includegraphics[width=0.48\textwidth]{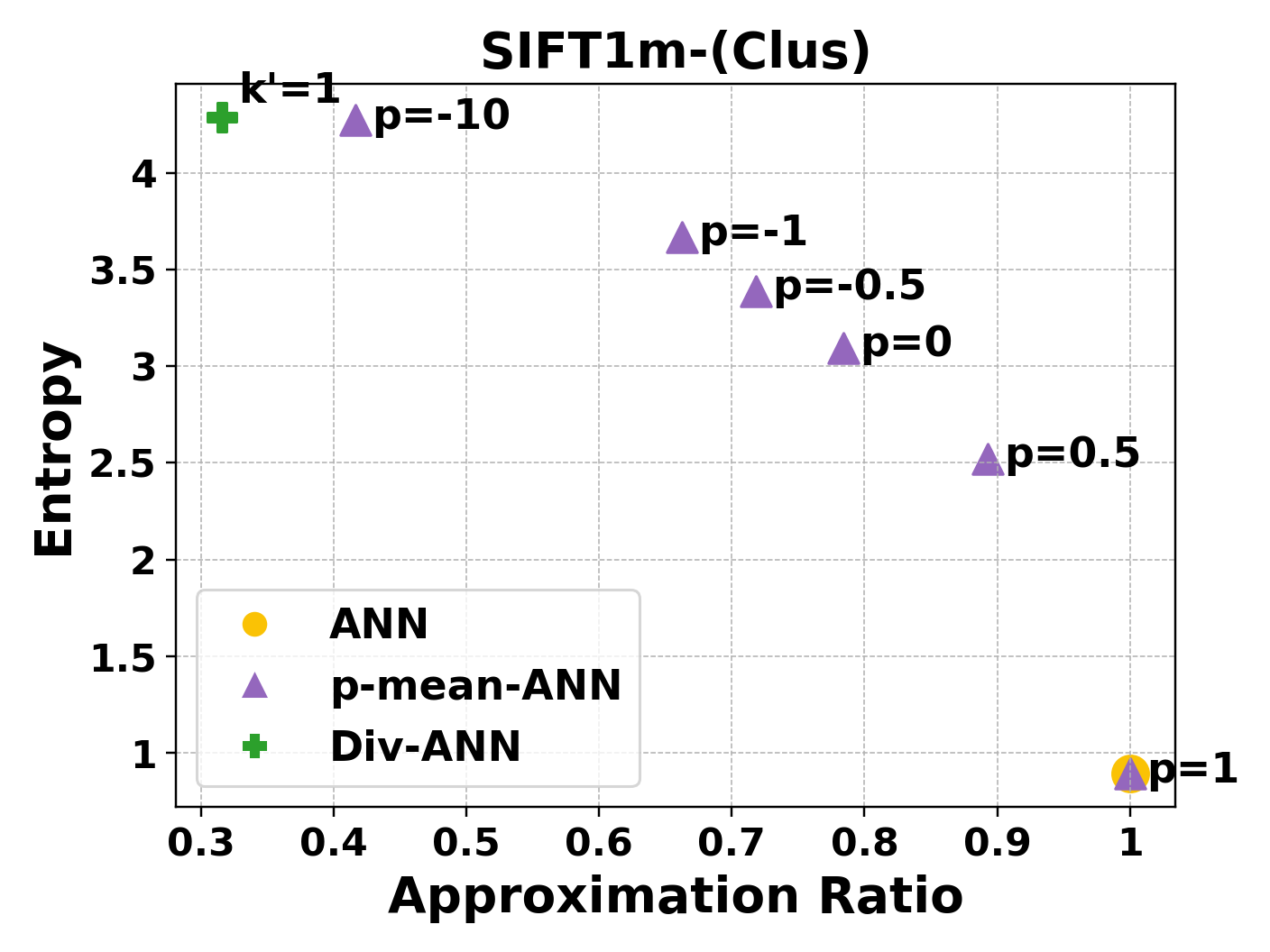} \\
\end{tabular}
\caption{The plots report the approximation ratio versus entropy trade-off of $p$-mean-ANN, as $p$ varies, for $k$ = $10$ \textbf{(Left)} and $k$ = $50$ \textbf{(Right)} on \siftC\ dataset in the single-attribute setting.}
\label{fig:AppendixResultsSingleAttriSIFTCLusPTrend}
\end{figure}

\begin{figure}[t!]
\centering
\begin{tabular}{@{}c@{}c@{}c@{}c@{}}
\includegraphics[width=0.48\textwidth]{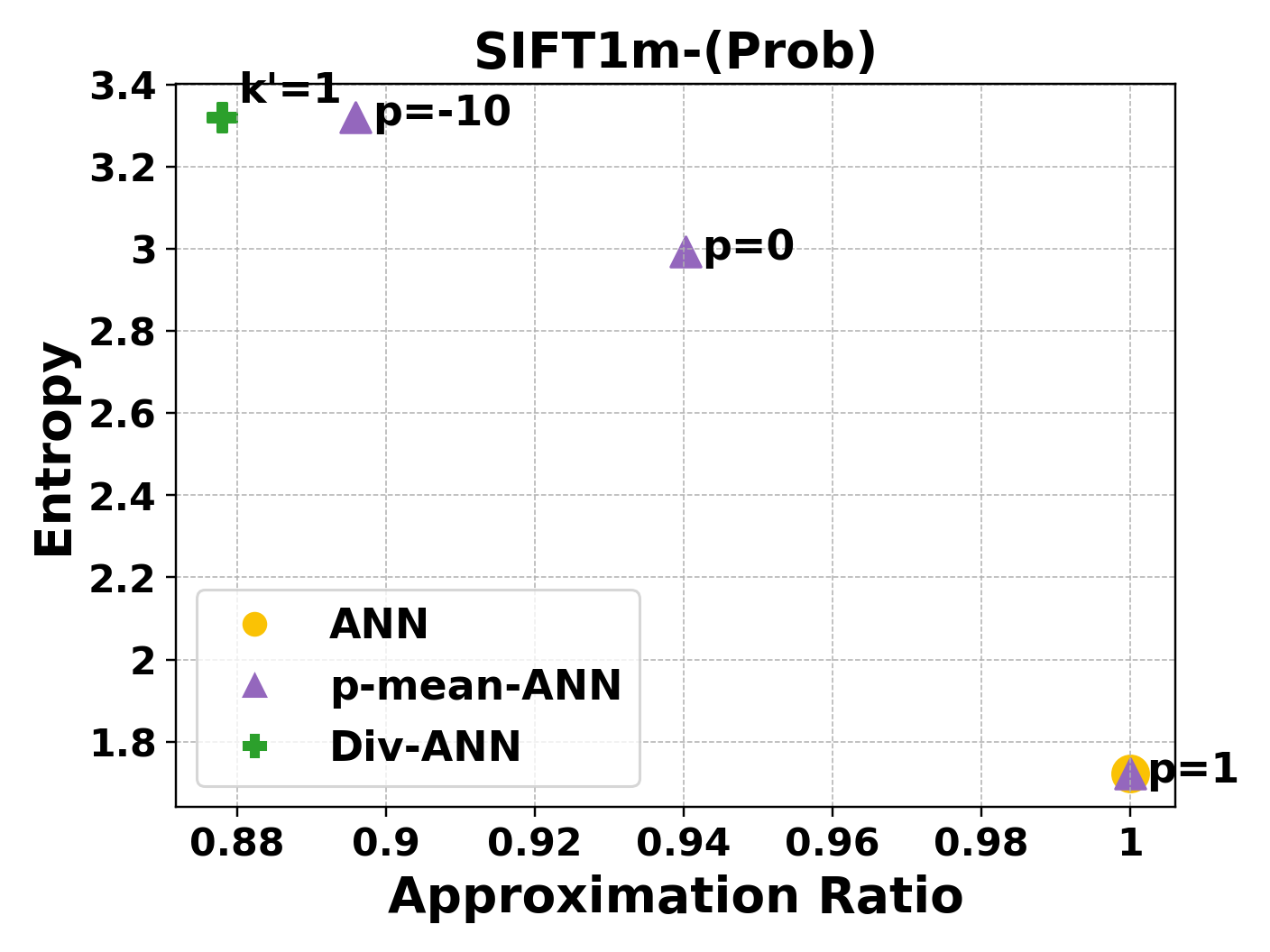} &
\includegraphics[width=0.48\textwidth]{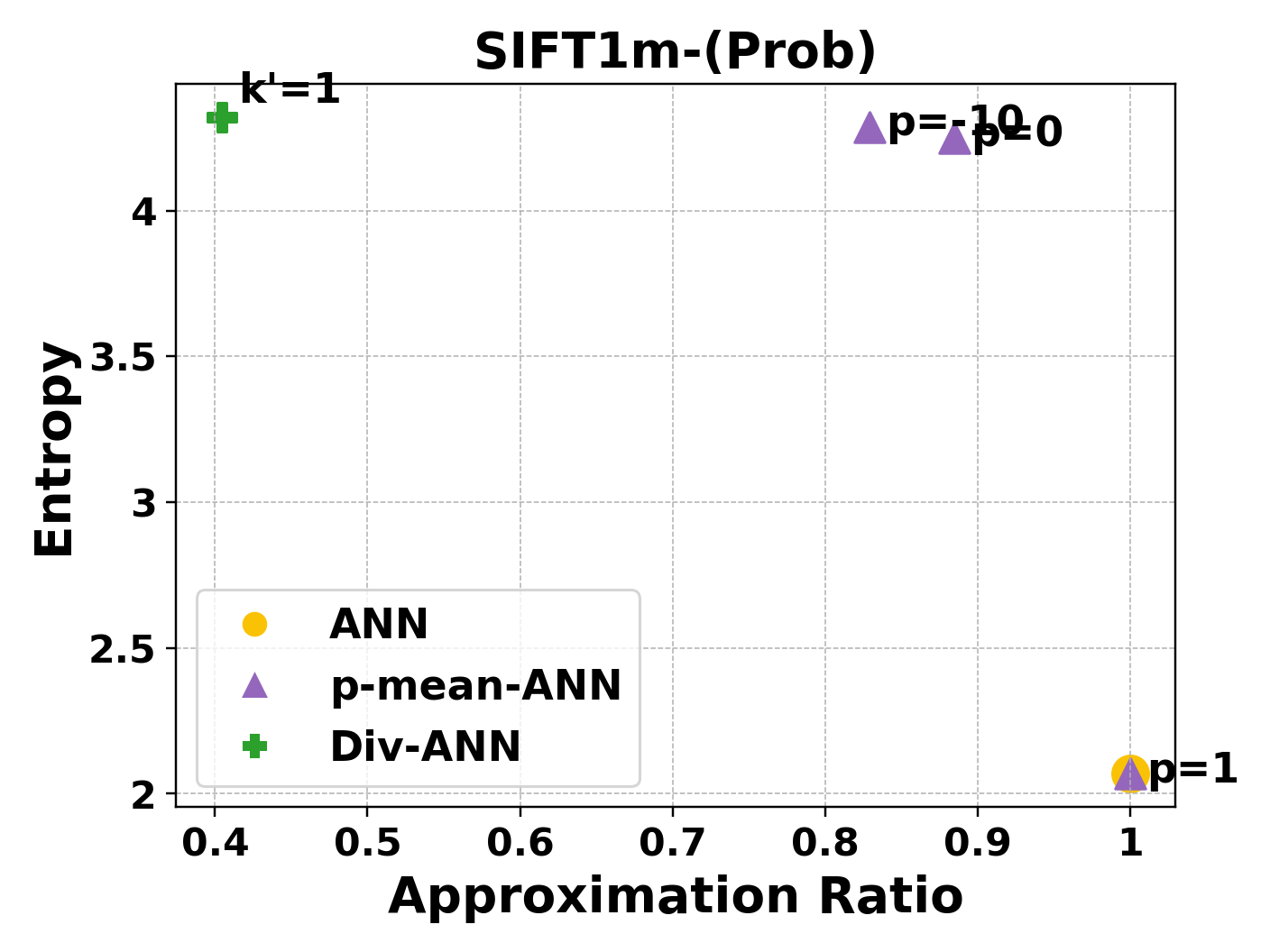} \\
\end{tabular}
\caption{The plots report the approximation ratio versus entropy trade-off of $p$-mean-ANN, as $p$ varies, for $k$ = $10$ \textbf{(Left)} and $k$ = $50$ \textbf{(Right)} on \siftP\ dataset in single-attribute setting. We omit points corresponding to $p \in \{-1, -0.5,  0.5\}$ since they were extremely close to the points $p = -10$ or $p=0$.}
\label{fig:AppendixResultsSingleAttriSIFTProbPTrend}
\end{figure}

\begin{figure}[t!]
\centering
\begin{tabular}{@{}c@{}c@{}c@{}c@{}}
\includegraphics[width=0.48\textwidth]{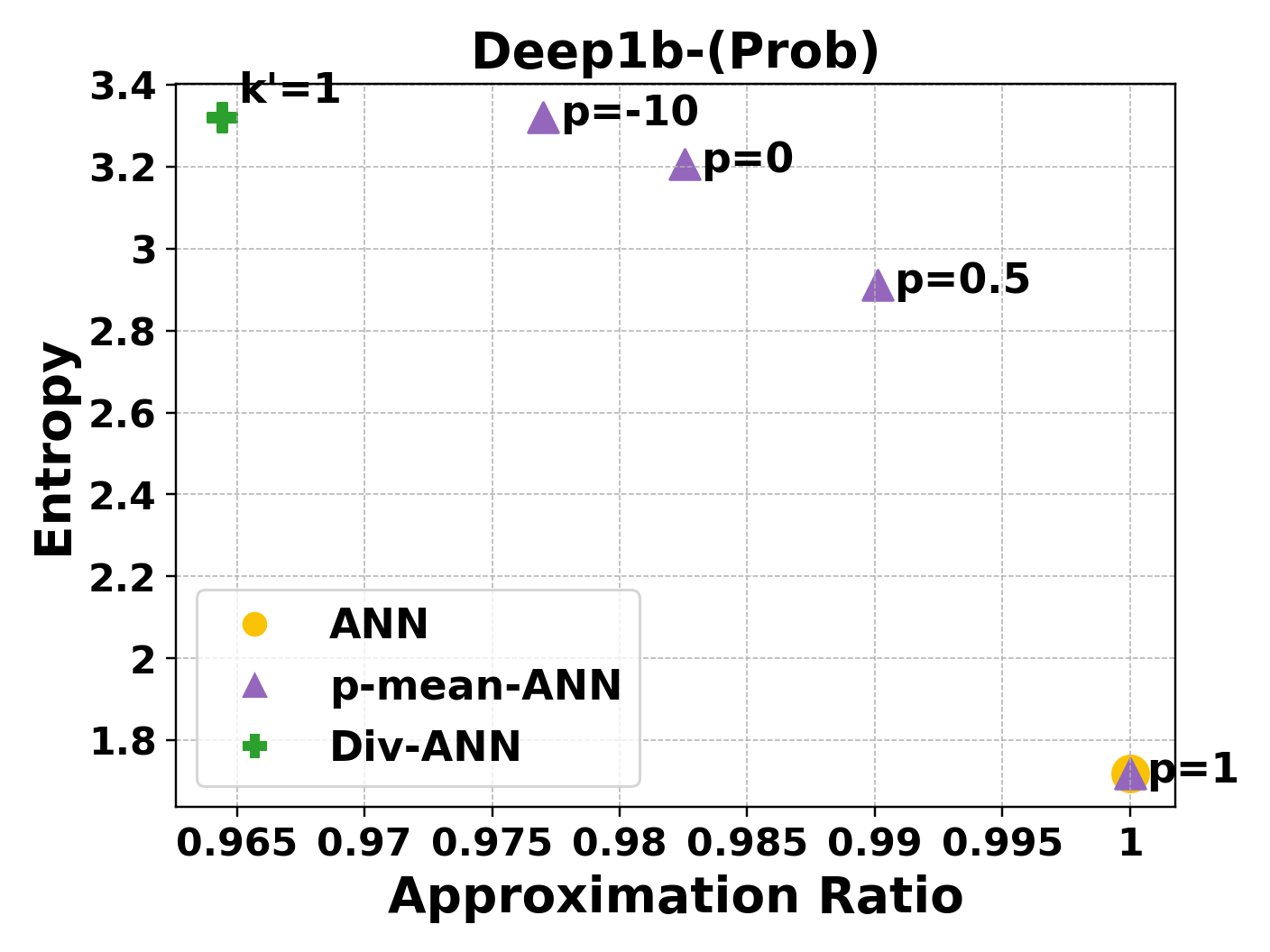} &
\includegraphics[width=0.48\textwidth]{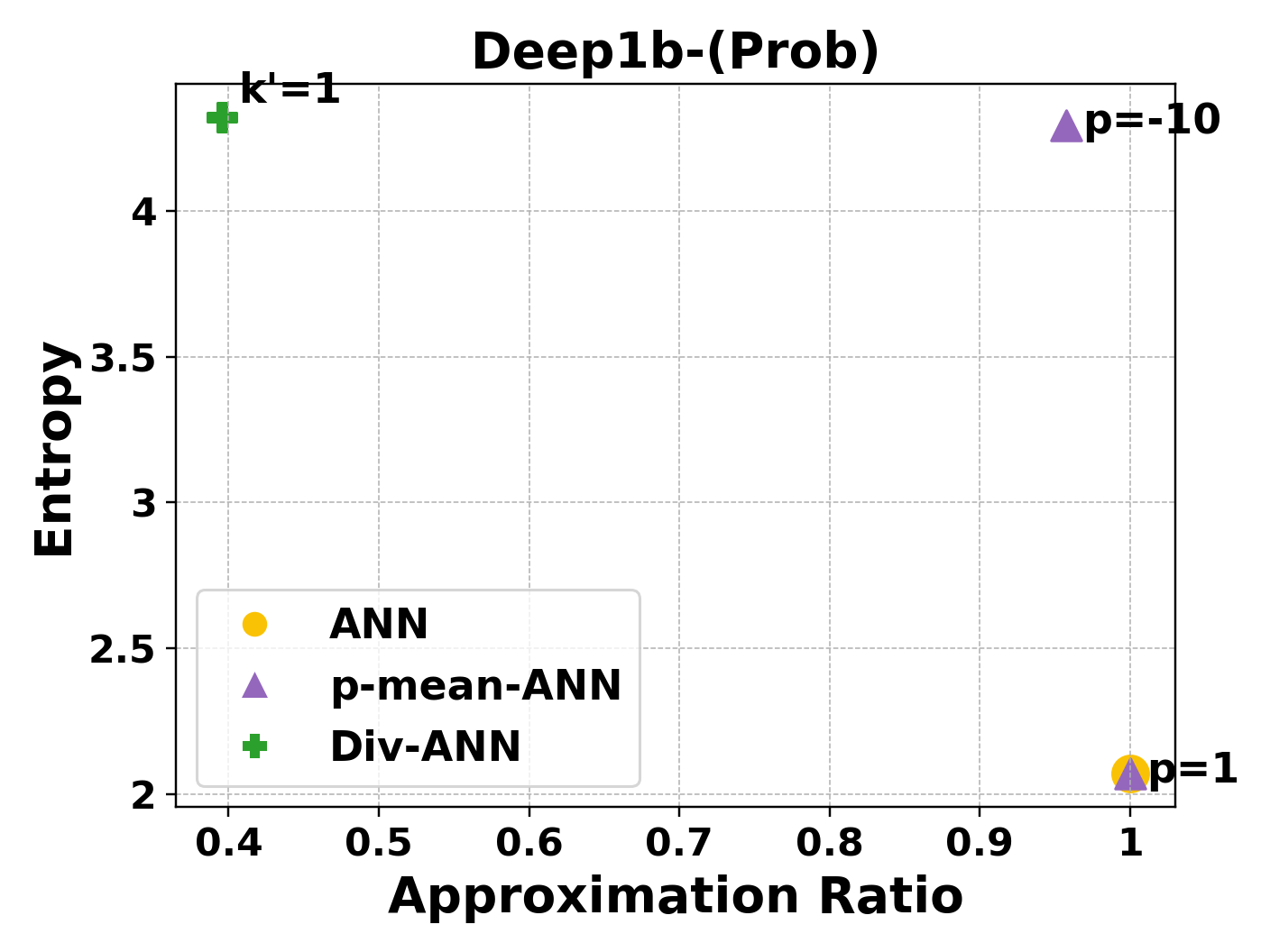} \\
\end{tabular}
\caption{The plots report the approximation ratio versus entropy trade-off of $p$-mean-ANN, as $p$ varies, for $k$ = $10$ \textbf{(Left)} and $k$ = $50$ \textbf{(Right)} on \deepP\ dataset in single-attribute setting. For $k=50$, we omit points corresponding to $p \in \{-1, -0.5, 0, 0.5\}$ since they were extremely close to the point $p = -10$. Due to the same reasons, we omit $p \in \{-1, -0.5\}$ for $k$ = $10$.}
\label{fig:AppendixResultsSingleAttriDeepProbPTrend}
\end{figure}

\begin{figure}[t!]
\centering
\begin{tabular}{@{}c@{}c@{}c@{}c@{}}
\includegraphics[width=0.48\textwidth]{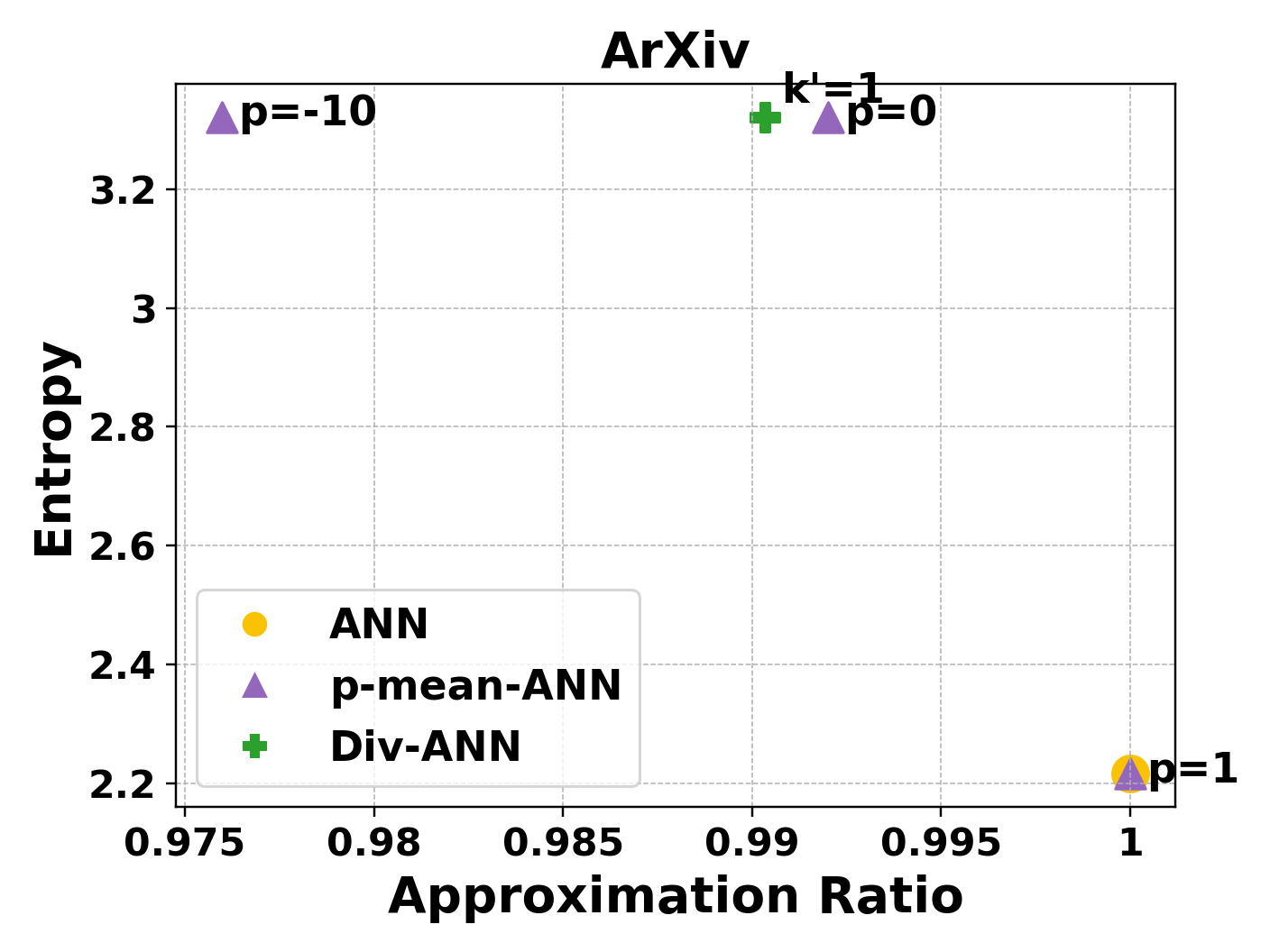} &
\includegraphics[width=0.48\textwidth]{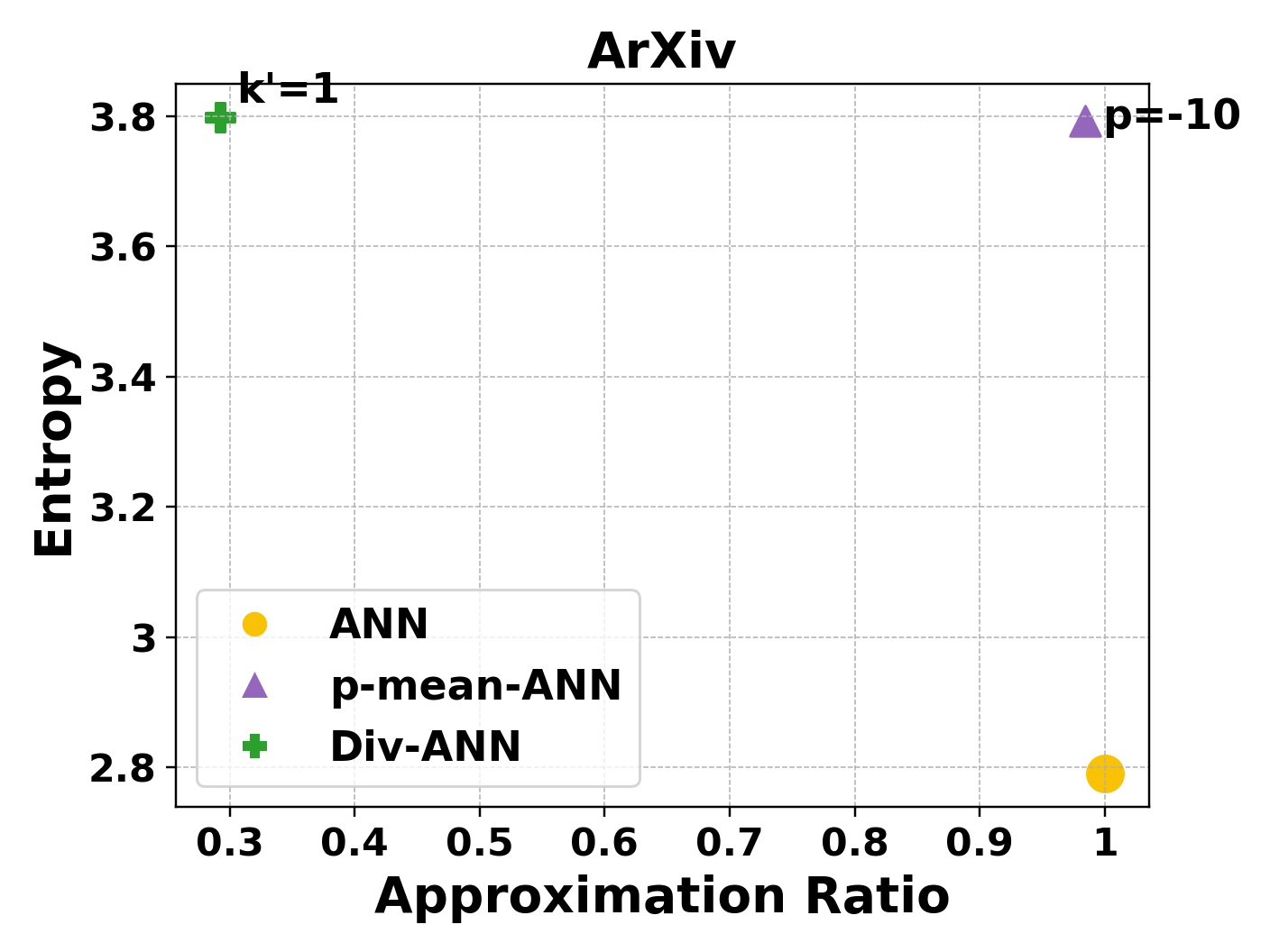} 
\end{tabular}
\caption{The plots report the approximation ratio versus entropy trade-off of $p$-mean-ANN, as $p$ varies, for $k$ = $10$ \textbf{(Left)} and $k$ = $50$ \textbf{(Right)} on \arxiv\ dataset in single-attribute setting. For $k=50$, we omit points corresponding to $p \in \{-1, -0.5, 0, 0.5\}$ since they were extremely close to the points $p = -10$. Due to similar reasons we omit $p \in \{-1, -0.5, 0.5\}$ for $k$ = $10$.}
\label{fig:AppendixResultsSingleAttriArxivPTrend}
\end{figure}

\subsubsection{Approximation Ratio Versus Inverse Simpson Index}
\label{appendix:paretoSingleAttriInvSimpsonAprxRatio}
We also report results (Figures \ref{fig:AppendixResultsSingleAttriSIFTCLusInvSimp}, \ref{fig:AppendixResultsSingleAttriSIFTProbInvSimp}, \ref{fig:AppendixResultsSingleAttriDeepProbInvSimp} and \ref{fig:AppendixResultsSingleAttriArxivInvSimp}) on approximation ratio versus inverse Simpson index for all the aforementioned datasets, comparing \ouralgo\ with \divann\ with various choices of constraint parameter $k'$. The trends are similar to those for approximation ratio versus entropy.

\begin{figure}[t!]
\centering
\begin{tabular}{@{}c@{}c@{}c@{}c@{}}
\includegraphics[width=0.48\textwidth]{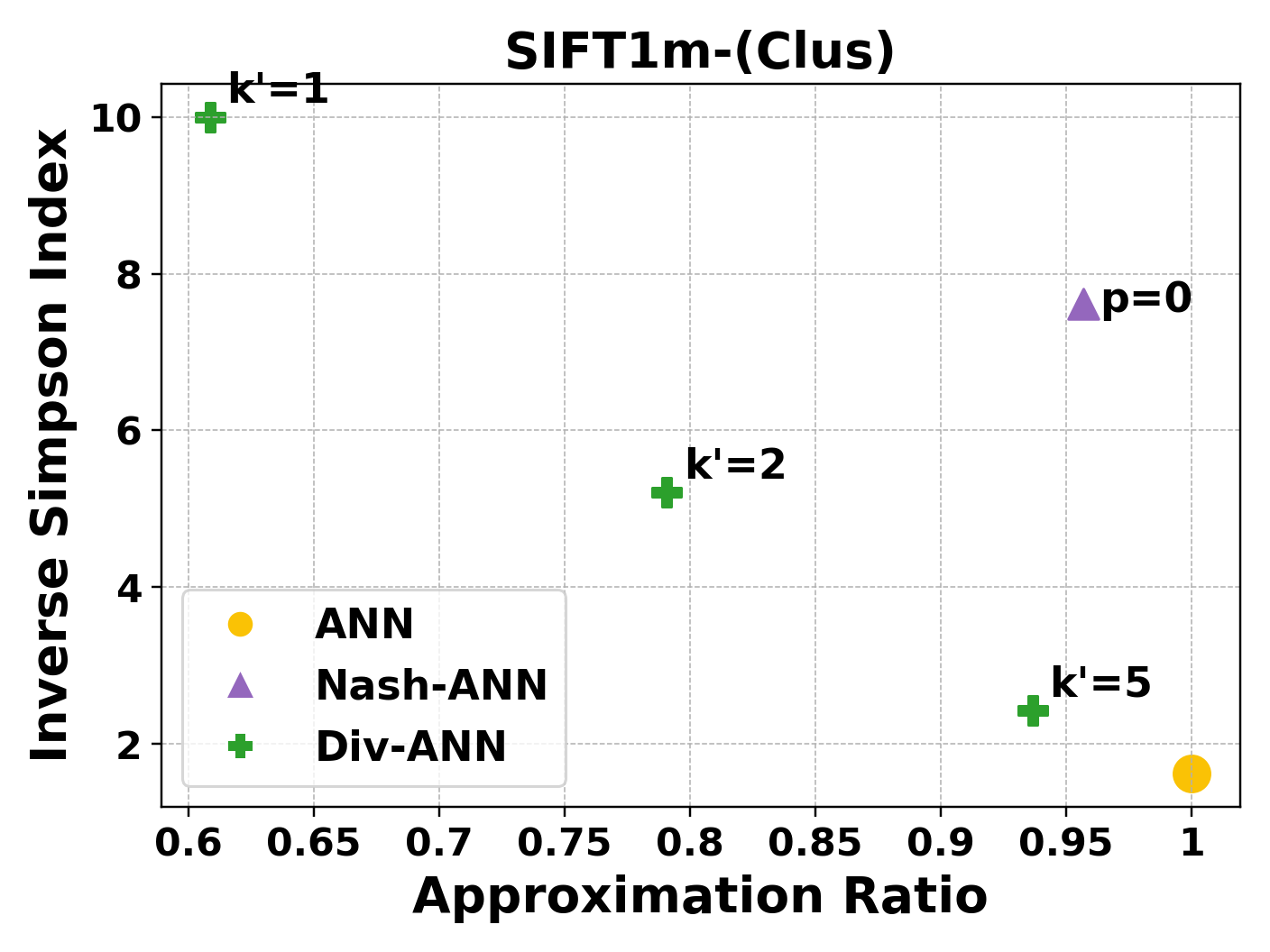} &
\includegraphics[width=0.48\textwidth]{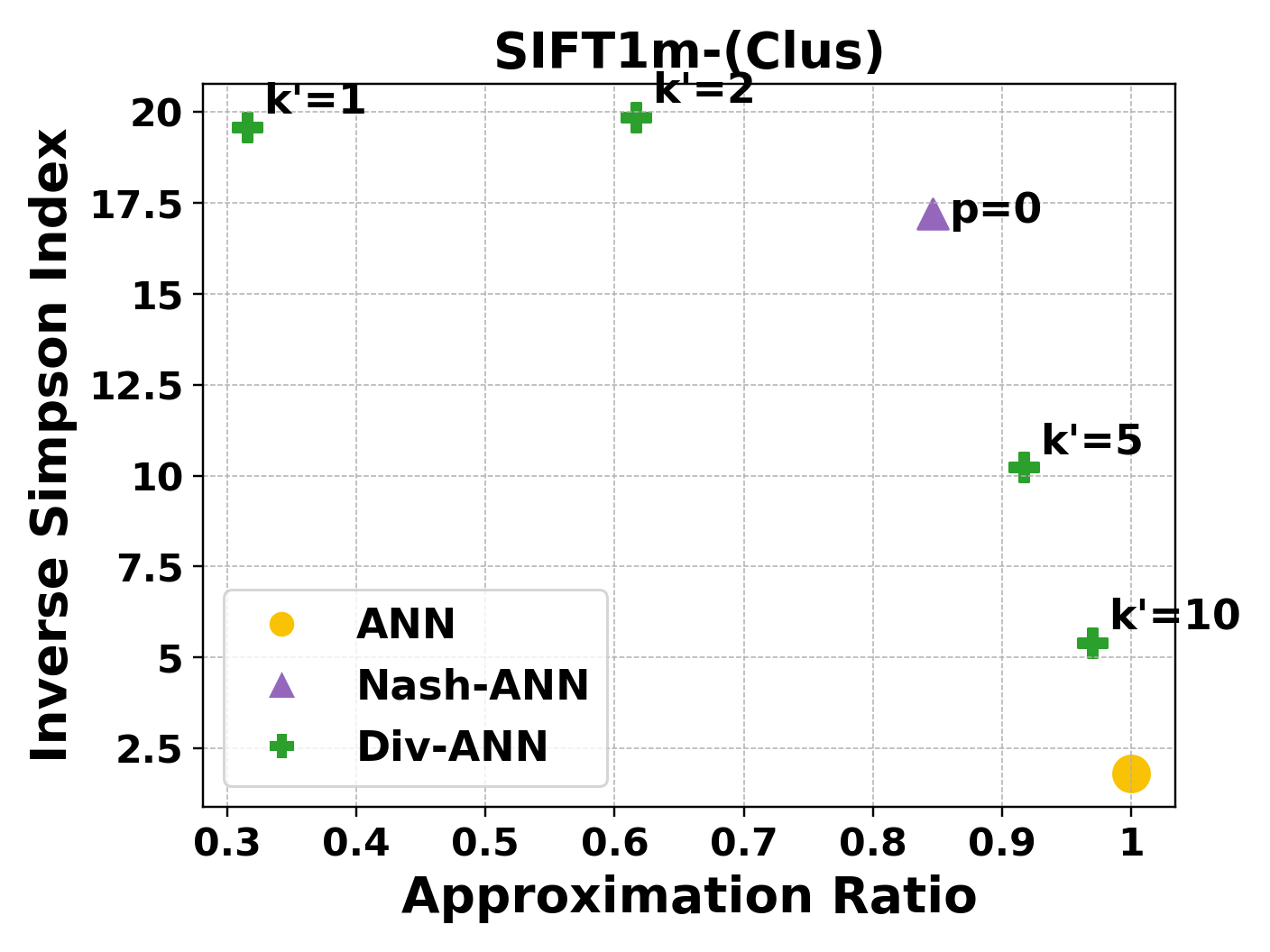} \\
\end{tabular}
\caption{The plots show approximation ratio versus inverse Simpson index trade-offs for various algorithms for $k$  = $10$ \textbf{(Left)} and  $k$  = $50$ \textbf{(Right)} in  single-attribute setting on \siftC\ dataset.}
\label{fig:AppendixResultsSingleAttriSIFTCLusInvSimp}
\end{figure}

\begin{figure}[t!]
\centering
\begin{tabular}{@{}c@{}c@{}c@{}c@{}}
\includegraphics[width=0.48\textwidth]{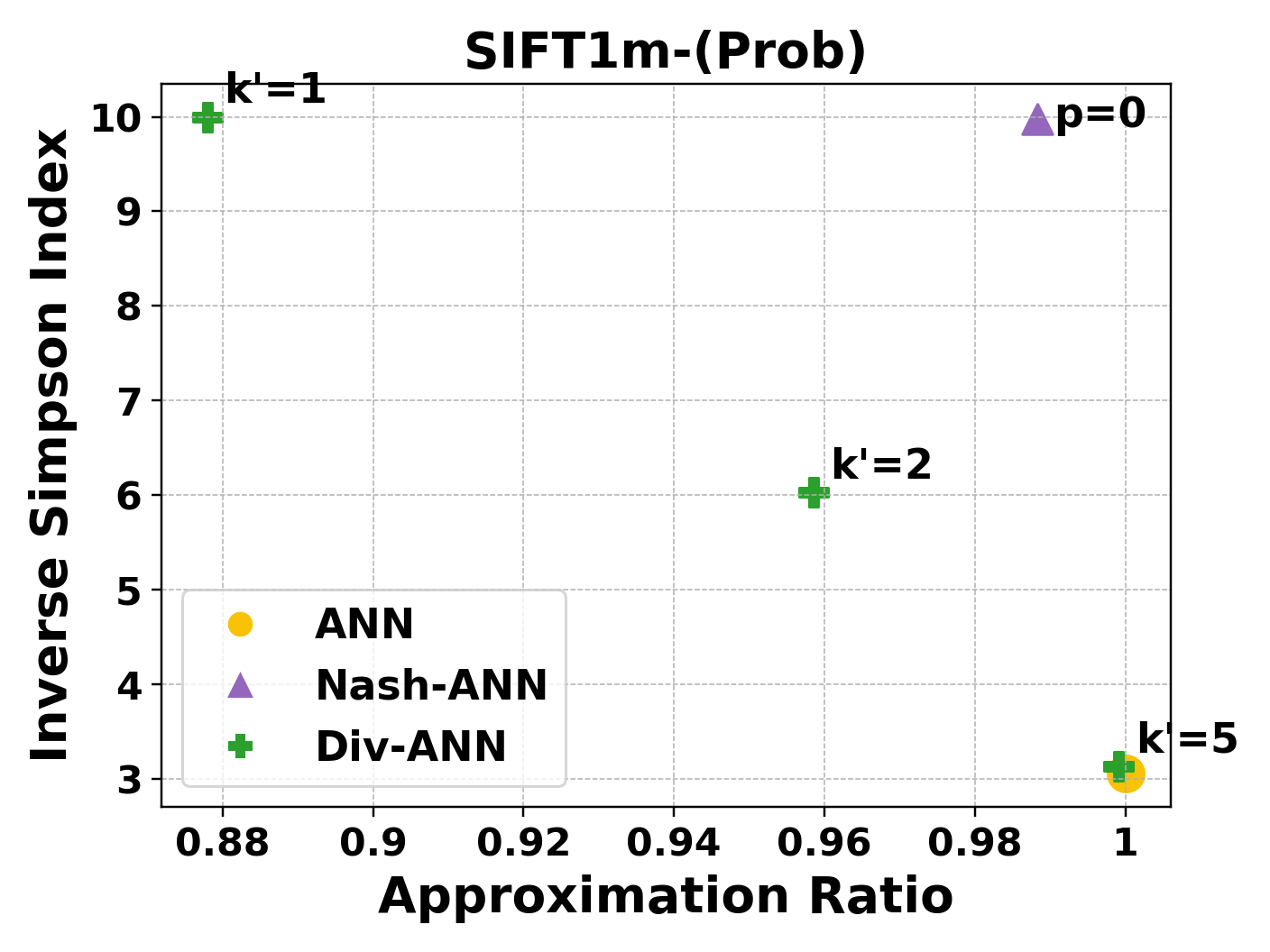} &
\includegraphics[width=0.48\textwidth]{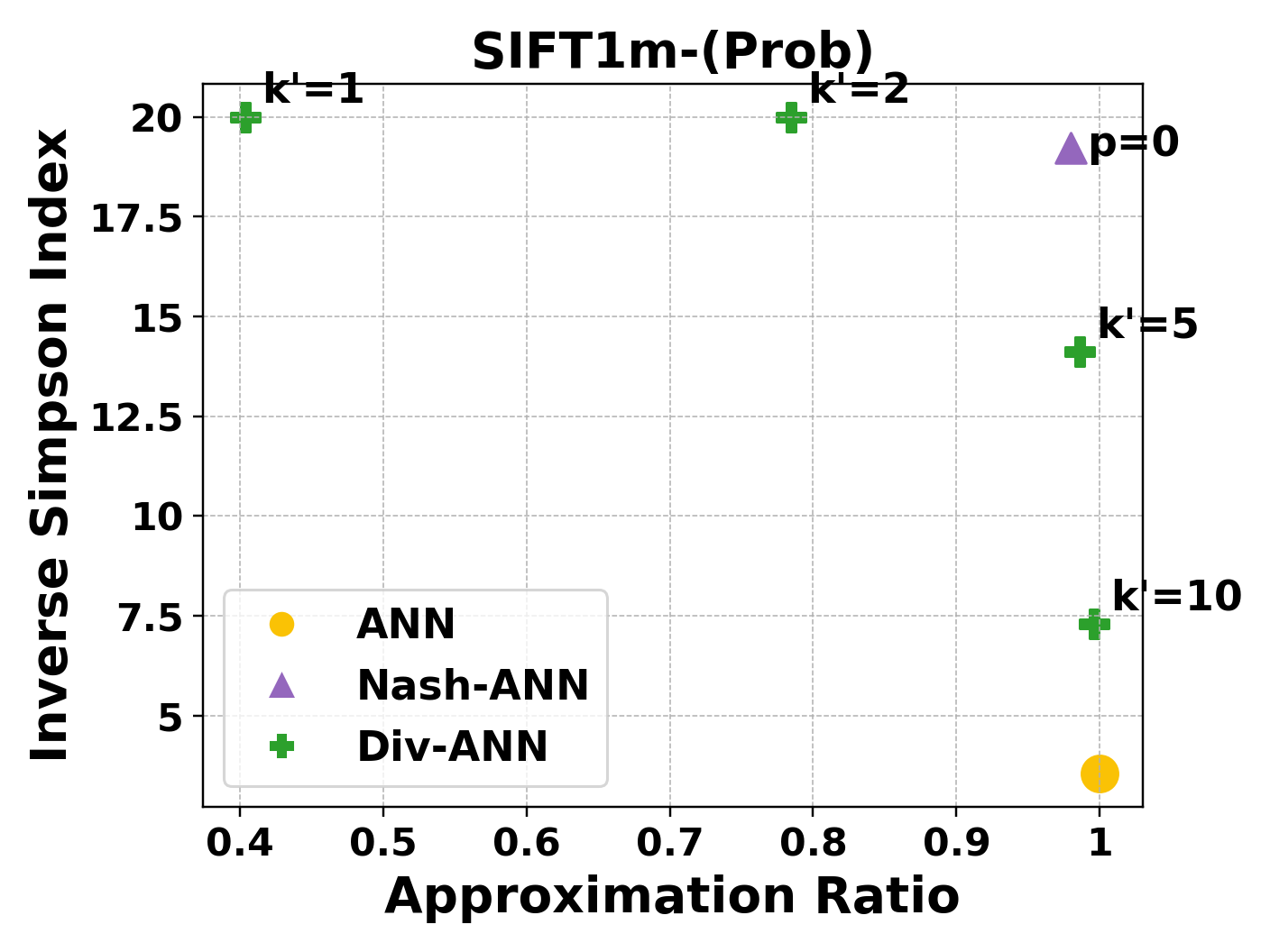} \\
\end{tabular}
\caption{The plots show approximation ratio versus inverse Simpson index trade-offs for various algorithms for $k$  = $10$ \textbf{(Left)} and  $k$  = $50$ \textbf{(Right)} in  single-attribute setting on \siftP\ dataset.}
\label{fig:AppendixResultsSingleAttriSIFTProbInvSimp}
\end{figure}

\begin{figure}[t!]
\centering
\begin{tabular}{@{}c@{}c@{}c@{}c@{}}
\includegraphics[width=0.48\textwidth]{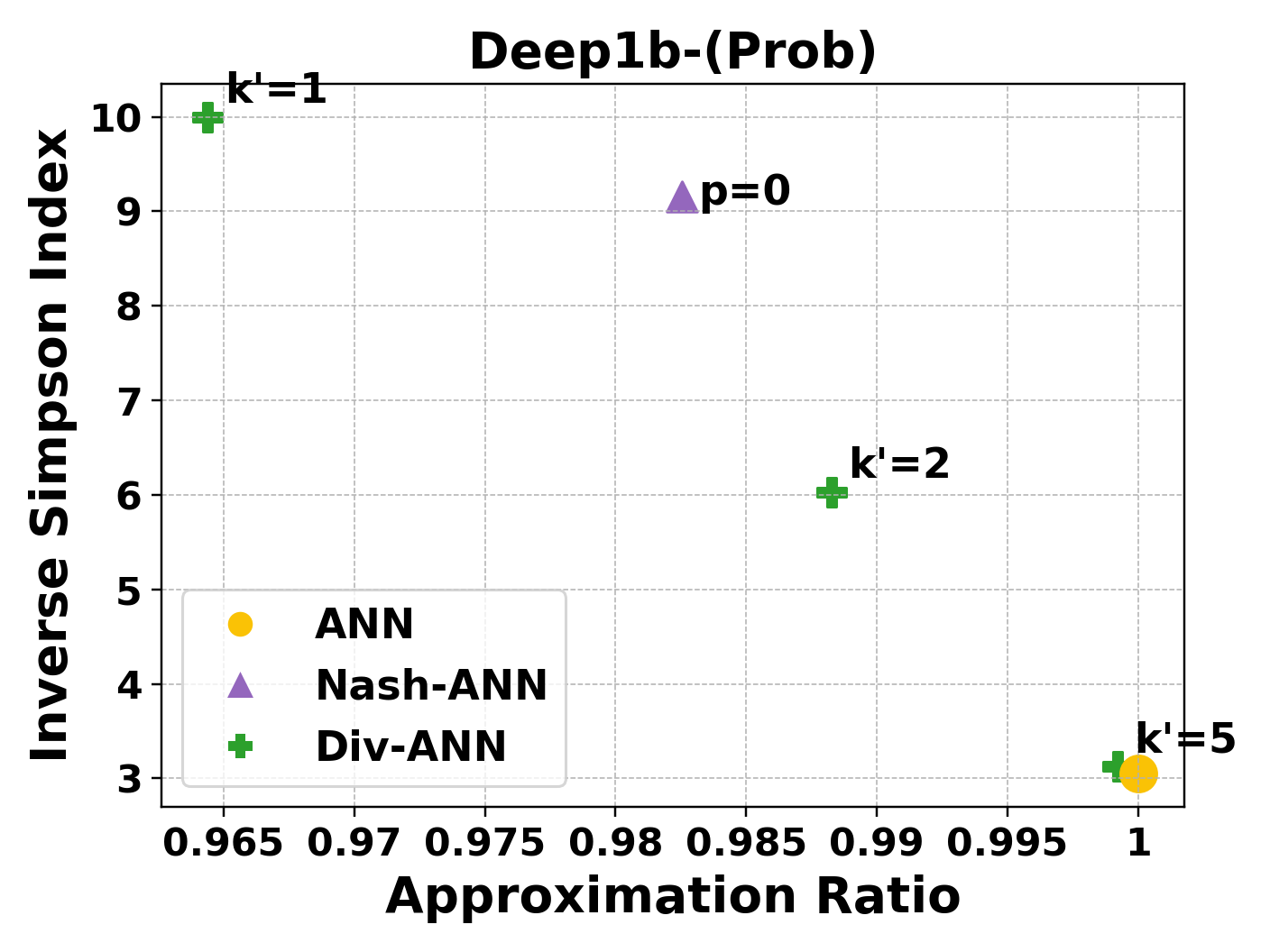} &
\includegraphics[width=0.48\textwidth]{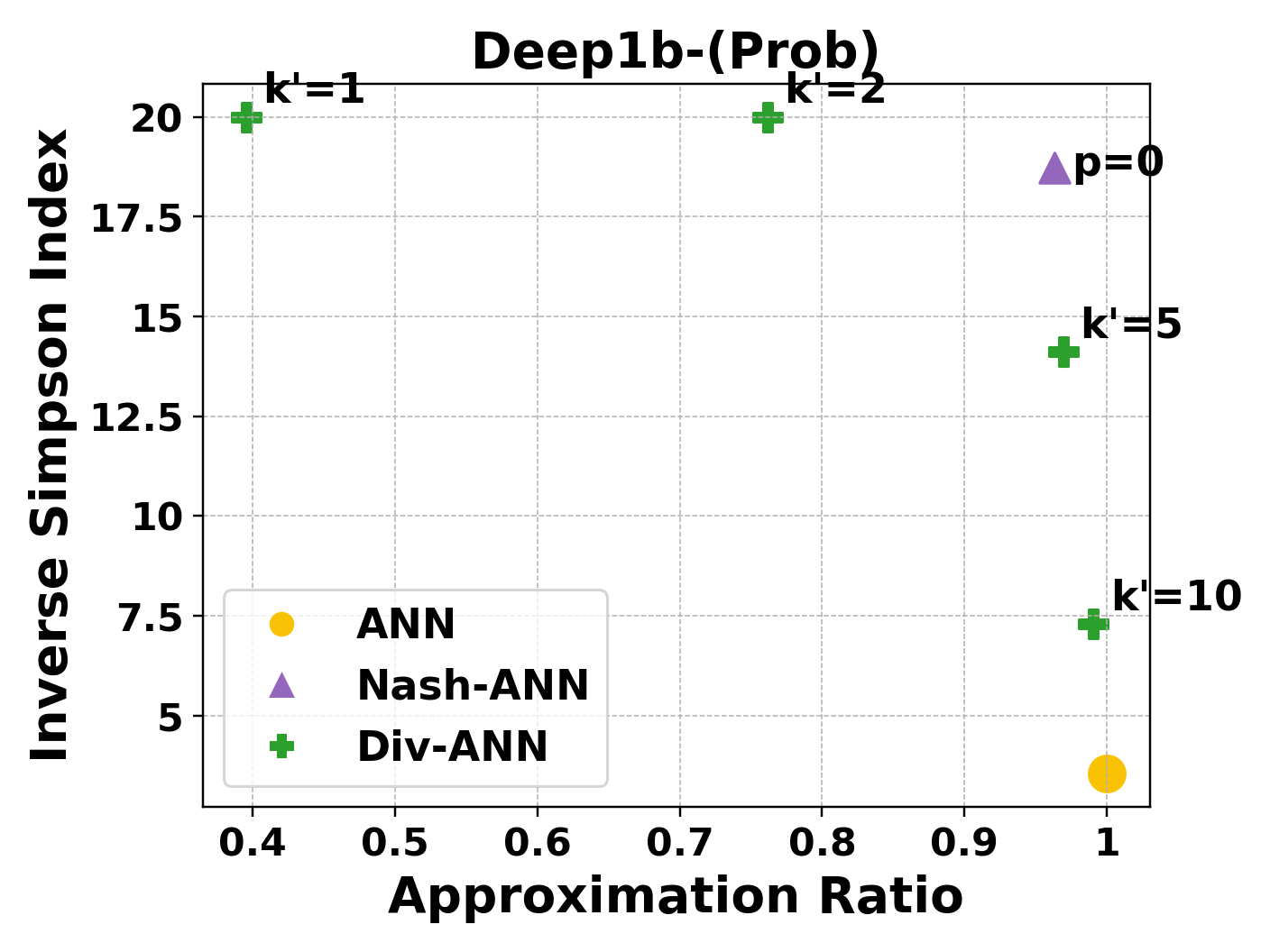} \\
\end{tabular}
\caption{The plots show approximation ratio versus inverse Simpson index trade-offs for various algorithms for $k$  = $10$ \textbf{(Left)} and  $k$  = $50$ \textbf{(Right)} in  single-attribute setting on \deepP\ dataset.}
\label{fig:AppendixResultsSingleAttriDeepProbInvSimp}
\end{figure}

\begin{figure}[t!]
\centering
\begin{tabular}{@{}c@{}c@{}c@{}c@{}}
\includegraphics[width=0.48\textwidth]{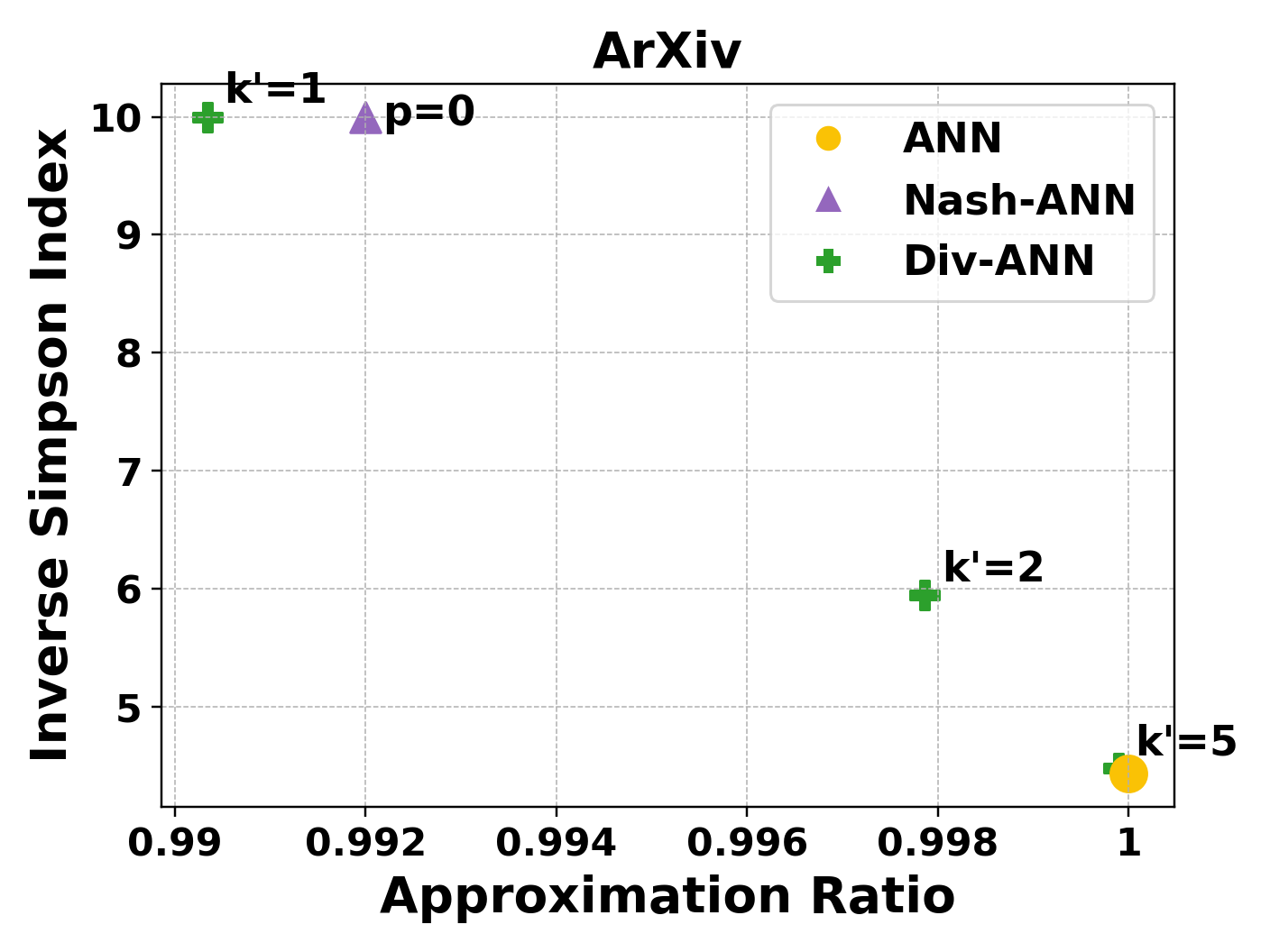} &
\includegraphics[width=0.48\textwidth]{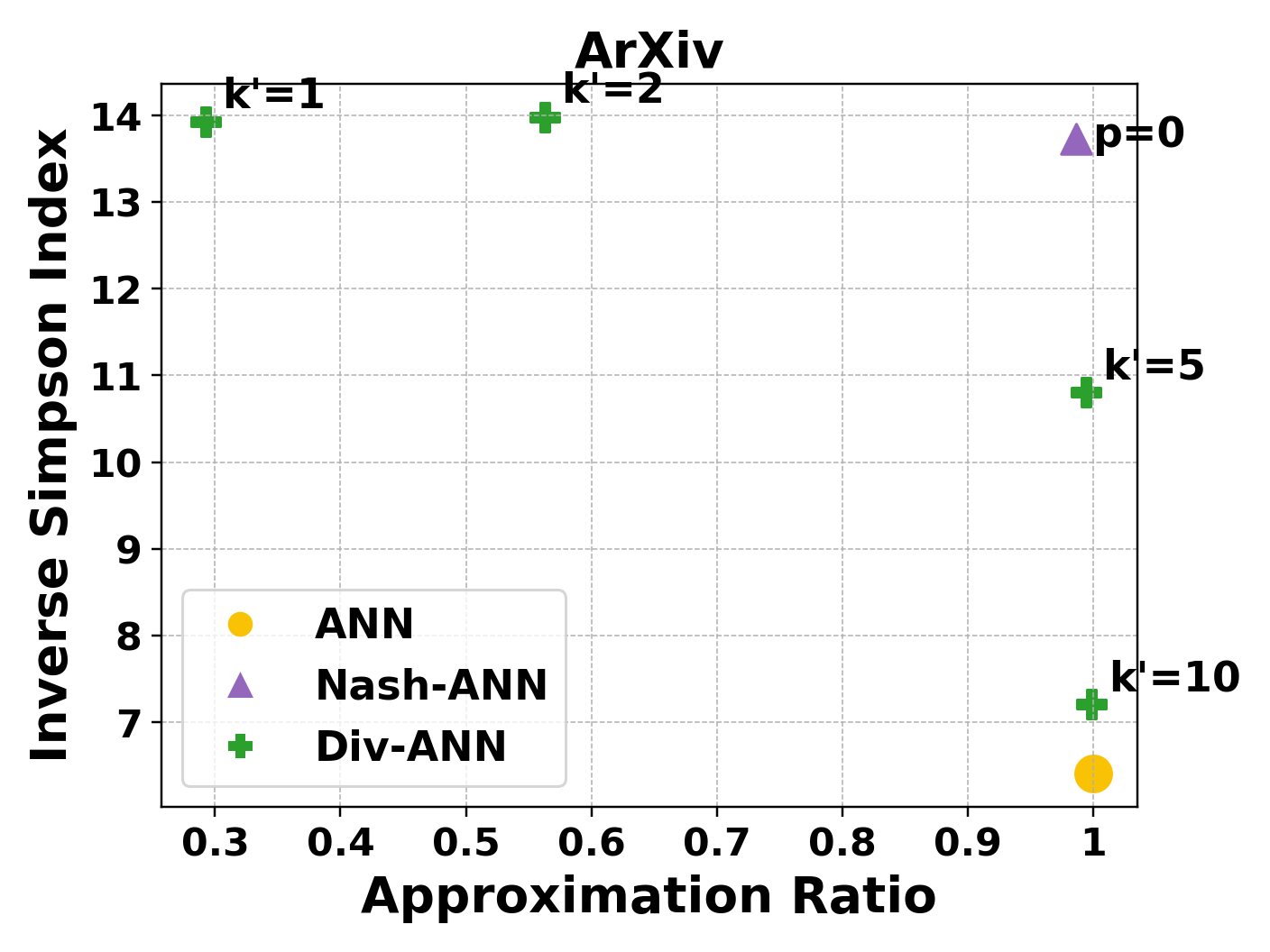} \\
\end{tabular}
\caption{The plots show approximation ratio versus inverse Simpson index trade-offs for various algorithms for $k$  = $10$ \textbf{(Left)} and  $k$  = $50$ \textbf{(Right)} in  single-attribute setting on \arxiv\ dataset.}
\label{fig:AppendixResultsSingleAttriArxivInvSimp}
\end{figure}

\subsubsection{Approximation Ratio Versus Distinct Attribute Count}
\label{appendix:paretoSingleAttriDistinctCountAprxRatio}
We also report the number of distinct attributes appearing in the set of vectors returned by different algorithms. Note that \divann\ by design always returns a set where the number of distinct attributes is at least $({k}/{k'})$. We plot approximation ratio versus number of distinct attributes and the results are shown in Figures~\ref{fig:AppendixResultsSingleAttriSIFTCLusDC}, \ref{fig:AppendixResultsSingleAttriSIFTProbDC}, \ref{fig:AppendixResultsSingleAttriDeepProbDC}, and \ref{fig:AppendixResultsSingleAttriArxivDC}. The results show that while \divann\ with $k'=1$ has high number of distinct attributes (by design), its approximation ratio is quite low. On the other hand, \ouralgo\ has almost equal or slightly lower number of distinct attributes but achieves very high approximation ratio.

\begin{figure}[t!]
\centering
\begin{tabular}{@{}c@{}c@{}c@{}c@{}}
\includegraphics[width=0.48\textwidth]{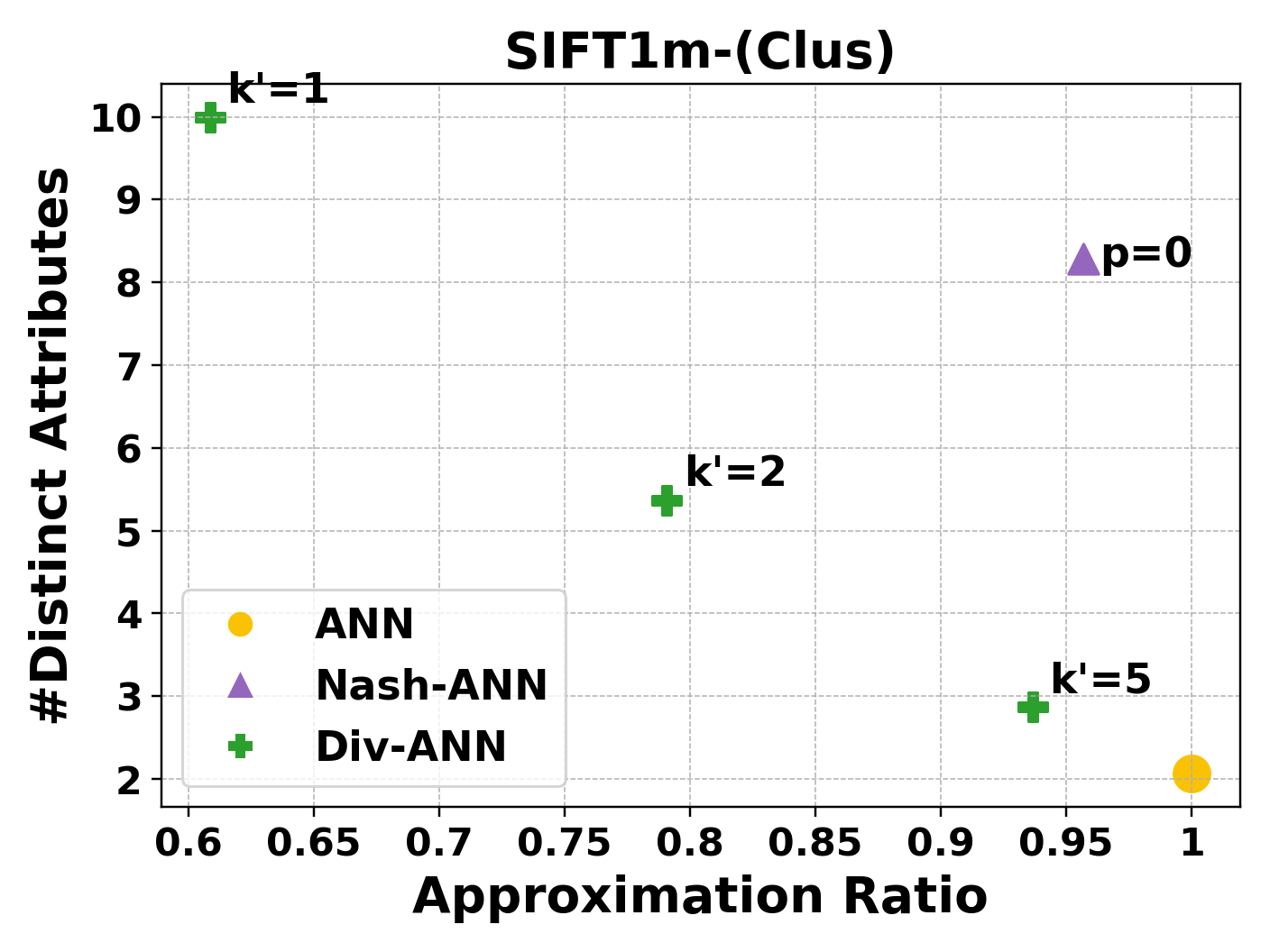} &
\includegraphics[width=0.48\textwidth]{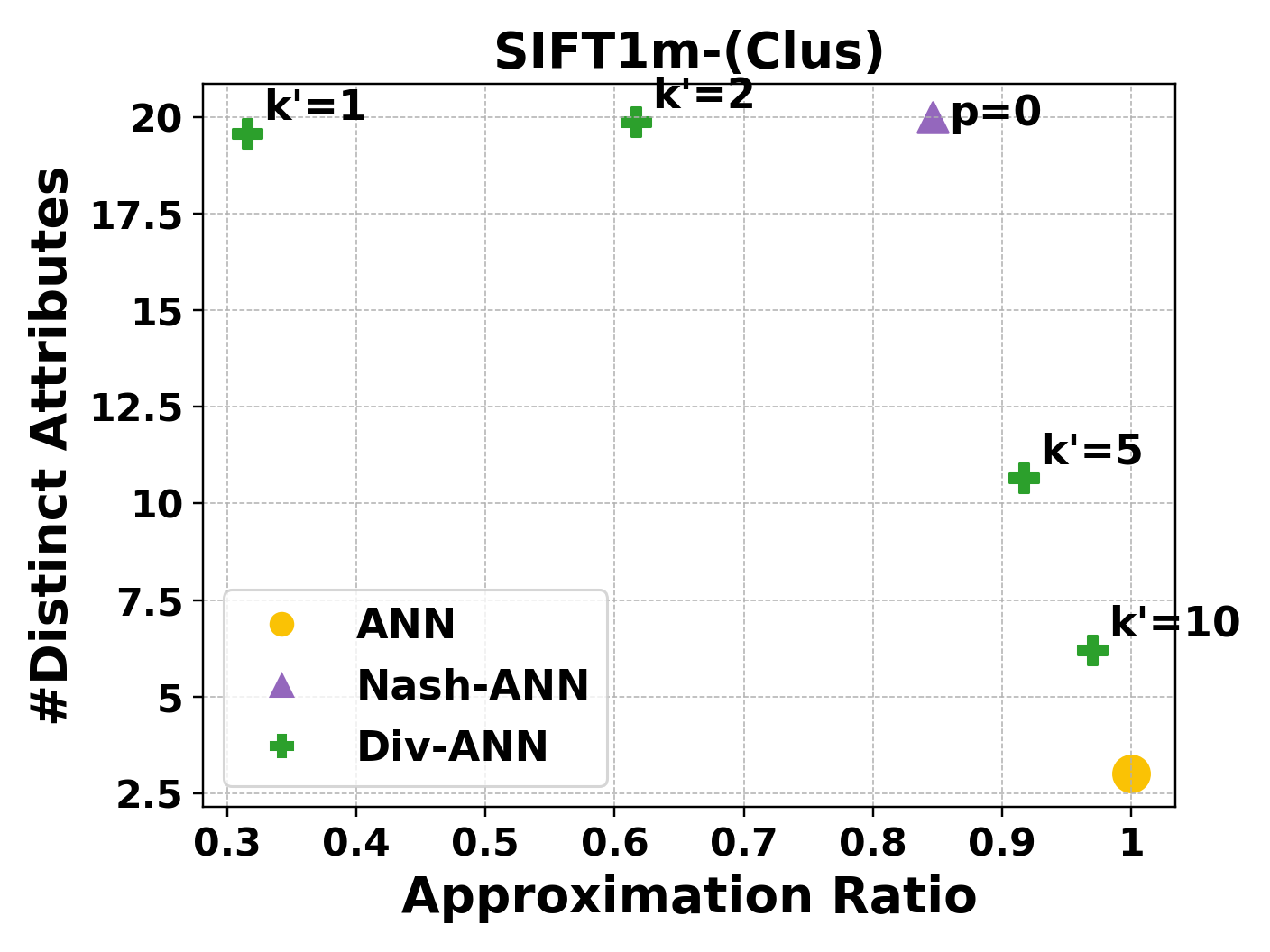} \\
\end{tabular}
\caption{The plots show approximation ratio versus distinct counts trade-offs for various algorithms for $k$  = $10$ \textbf{(Left)} and  $k$  = $50$ \textbf{(Right)} in  single-attribute setting on \siftC\ dataset.}
\label{fig:AppendixResultsSingleAttriSIFTCLusDC}
\vspace{-19pt}
\end{figure}

\begin{figure}[t!]
\centering
\begin{tabular}{@{}c@{}c@{}c@{}c@{}}
\includegraphics[width=0.48\textwidth]{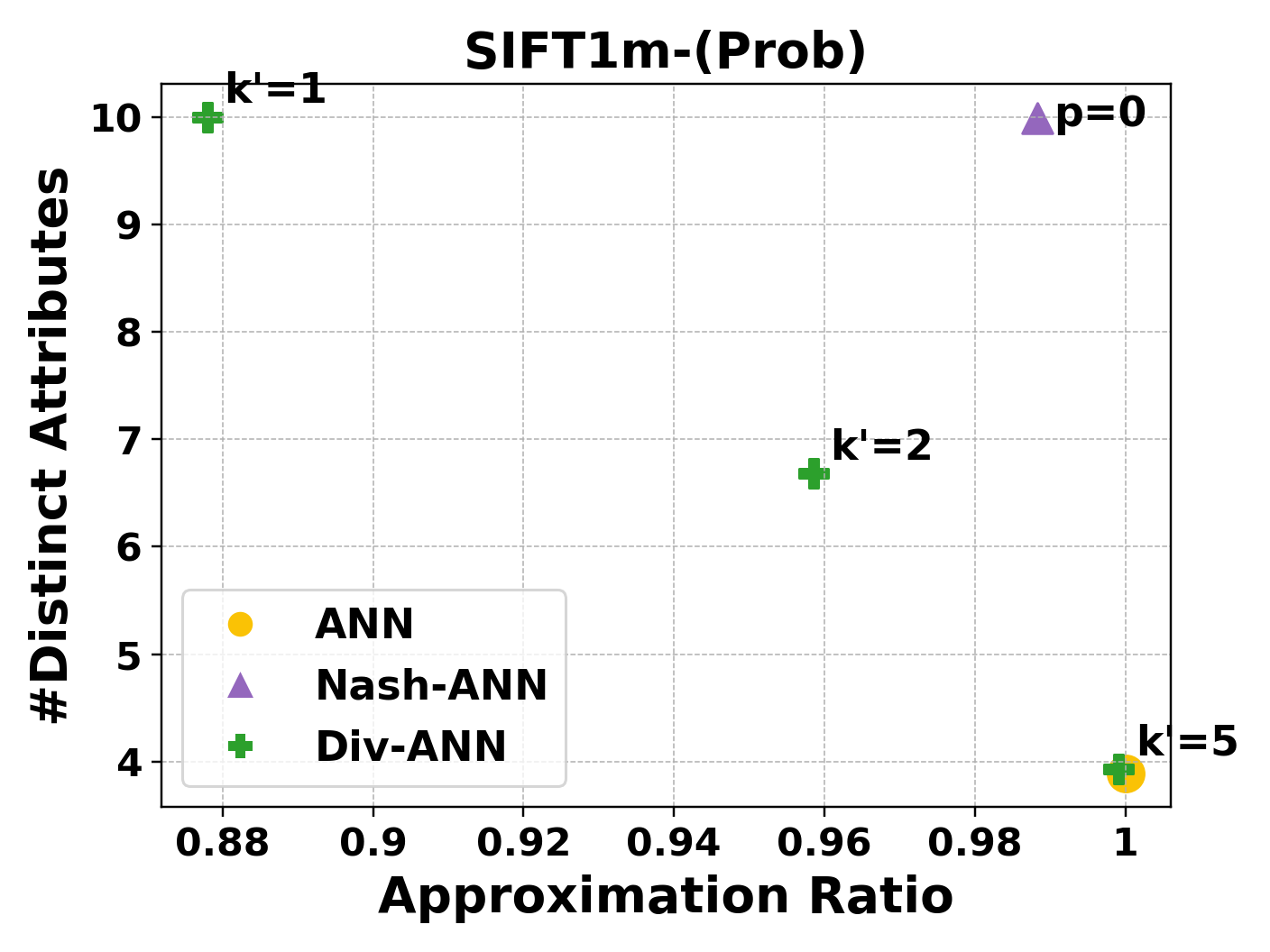} &
\includegraphics[width=0.48\textwidth]{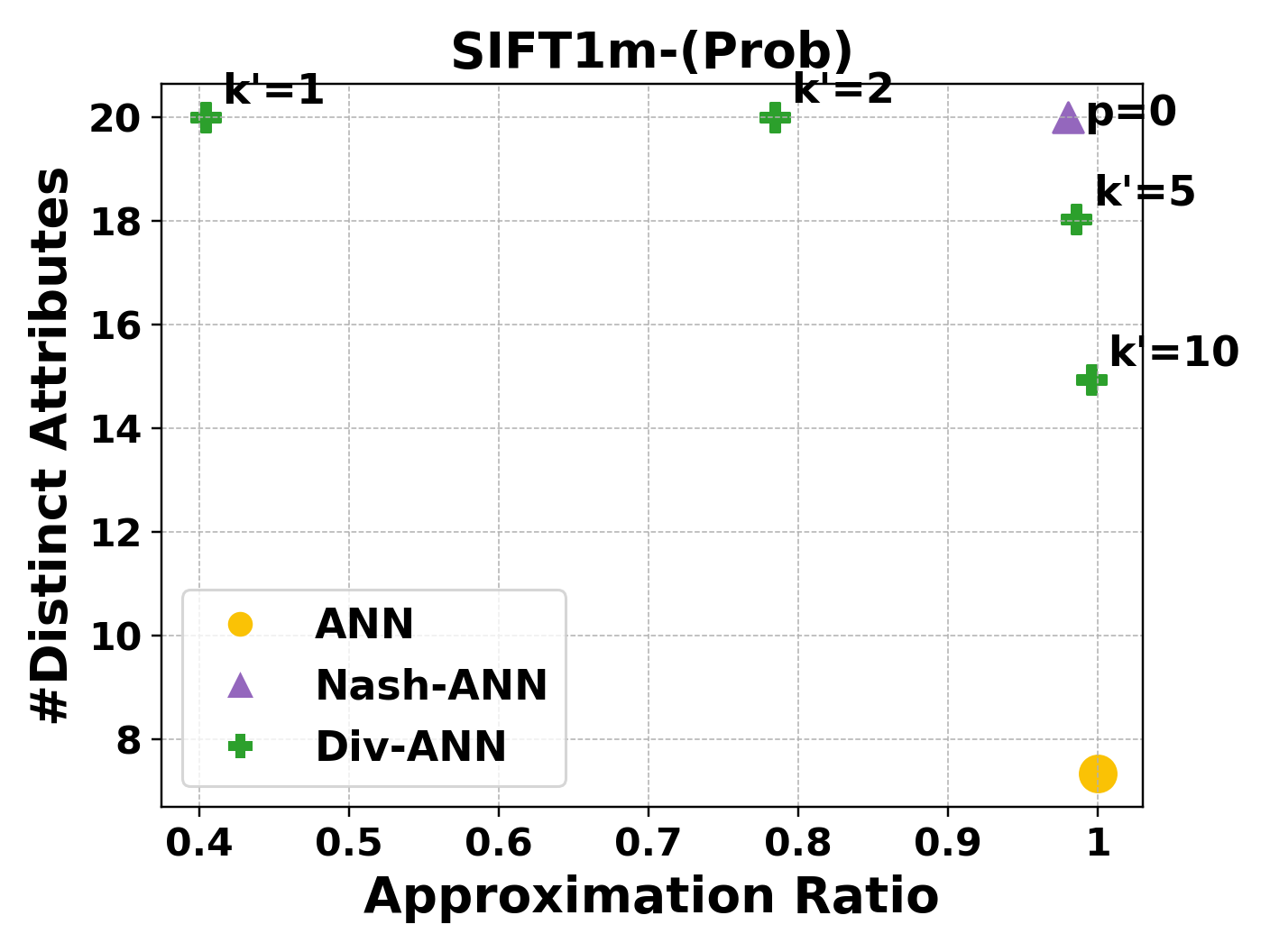} \\
\end{tabular}
\caption{The plots show approximation ratio versus distinct counts trade-offs for various algorithms for $k$  = $10$ \textbf{(Left)} and  $k$  = $50$ \textbf{(Right)} in  single-attribute setting on \siftP\ dataset.}
\label{fig:AppendixResultsSingleAttriSIFTProbDC}
\end{figure}

\begin{figure}[t!]
\centering
\begin{tabular}{@{}c@{}c@{}c@{}c@{}}
\includegraphics[width=0.48\textwidth]{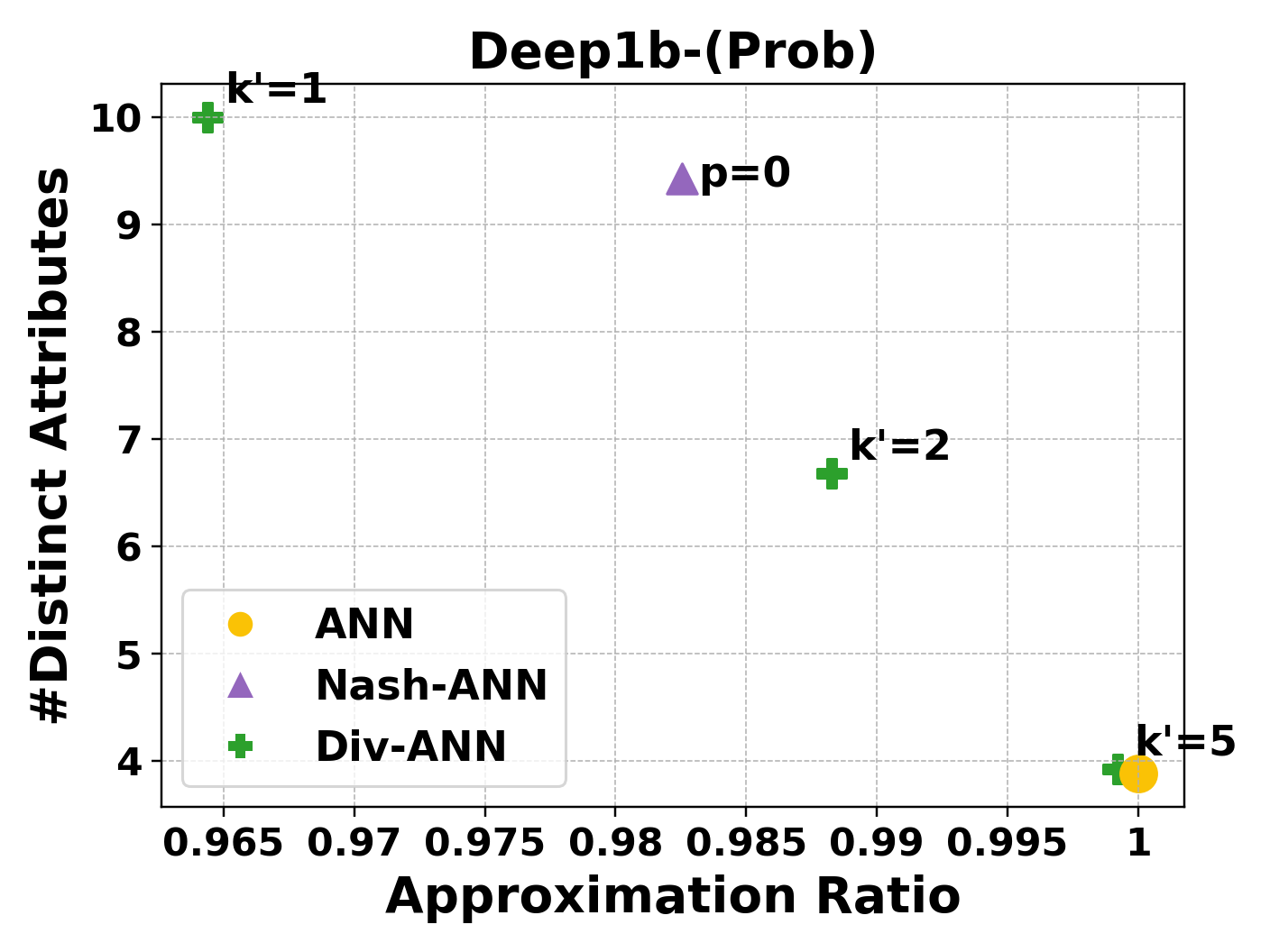} &
\includegraphics[width=0.48\textwidth]{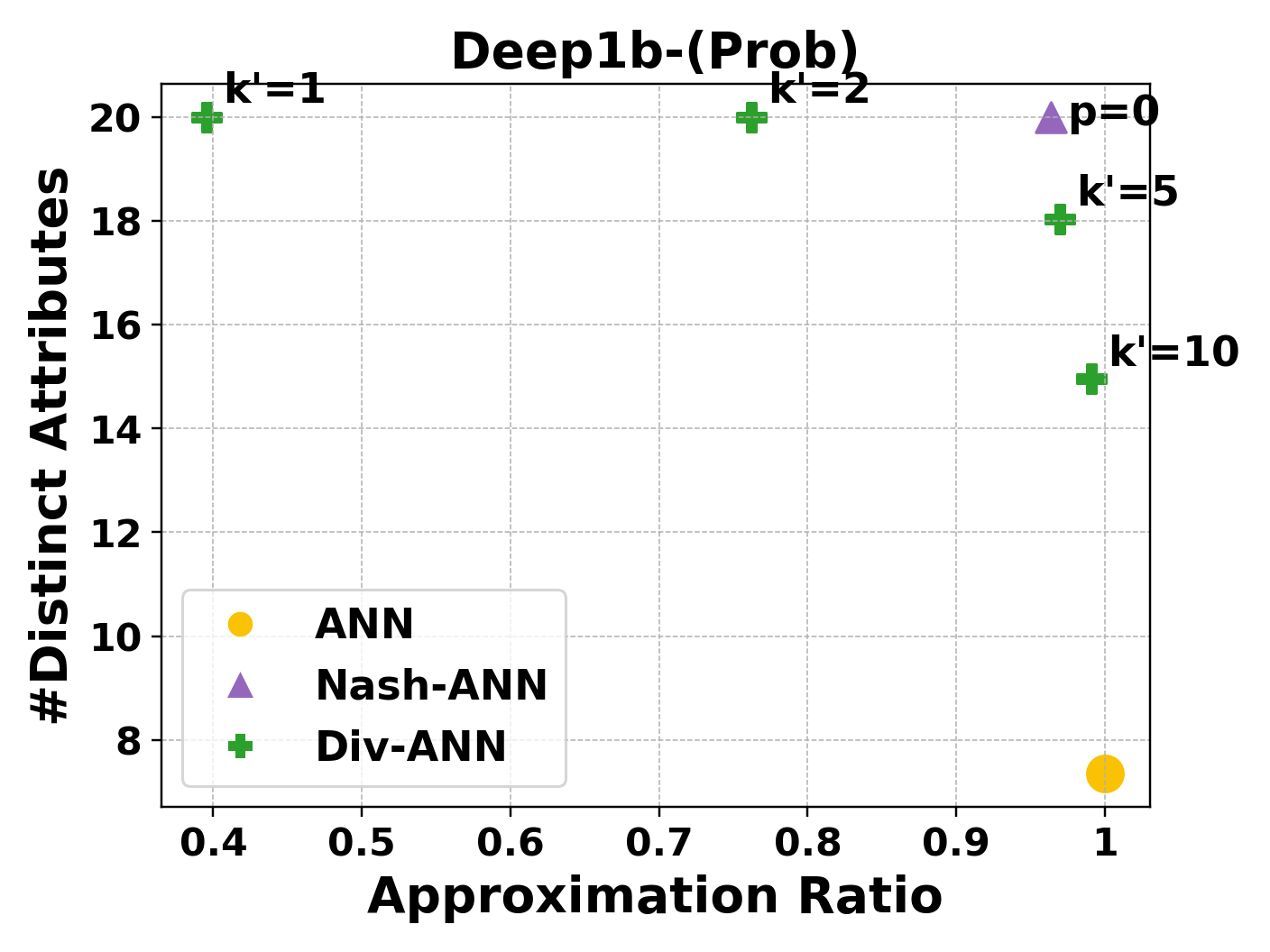} \\
\end{tabular}
\caption{The plots show approximation ratio versus distinct counts trade-offs for various algorithms for $k$  = $10$ \textbf{(Left)} and  $k$  = $50$ \textbf{(Right)} in  single-attribute setting on \deepP\ dataset.}
\label{fig:AppendixResultsSingleAttriDeepProbDC}
\end{figure}

\begin{figure}[t!]
\centering
\begin{tabular}{@{}c@{}c@{}c@{}c@{}}
\includegraphics[width=0.48\textwidth]{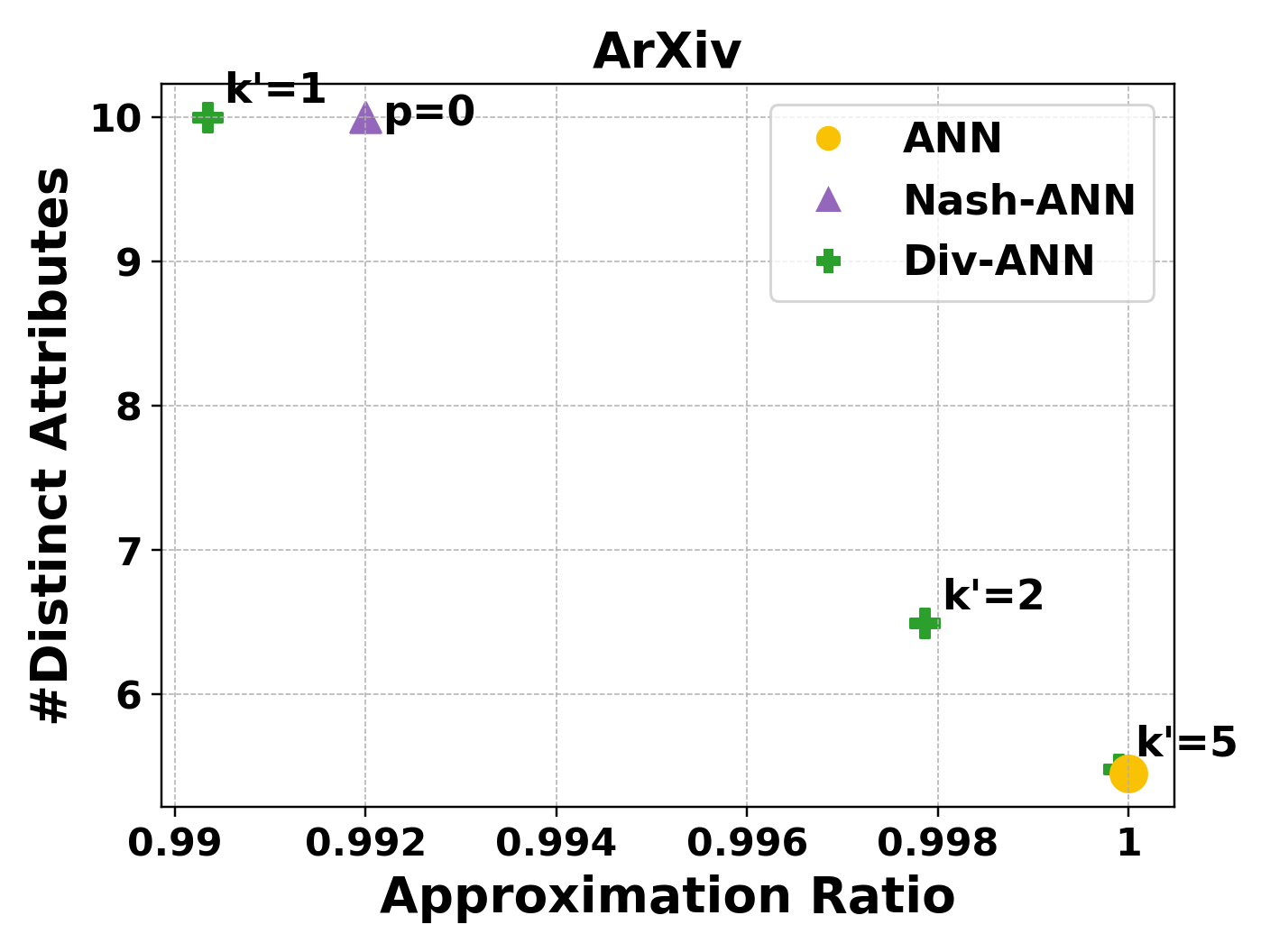} &
\includegraphics[width=0.48\textwidth]{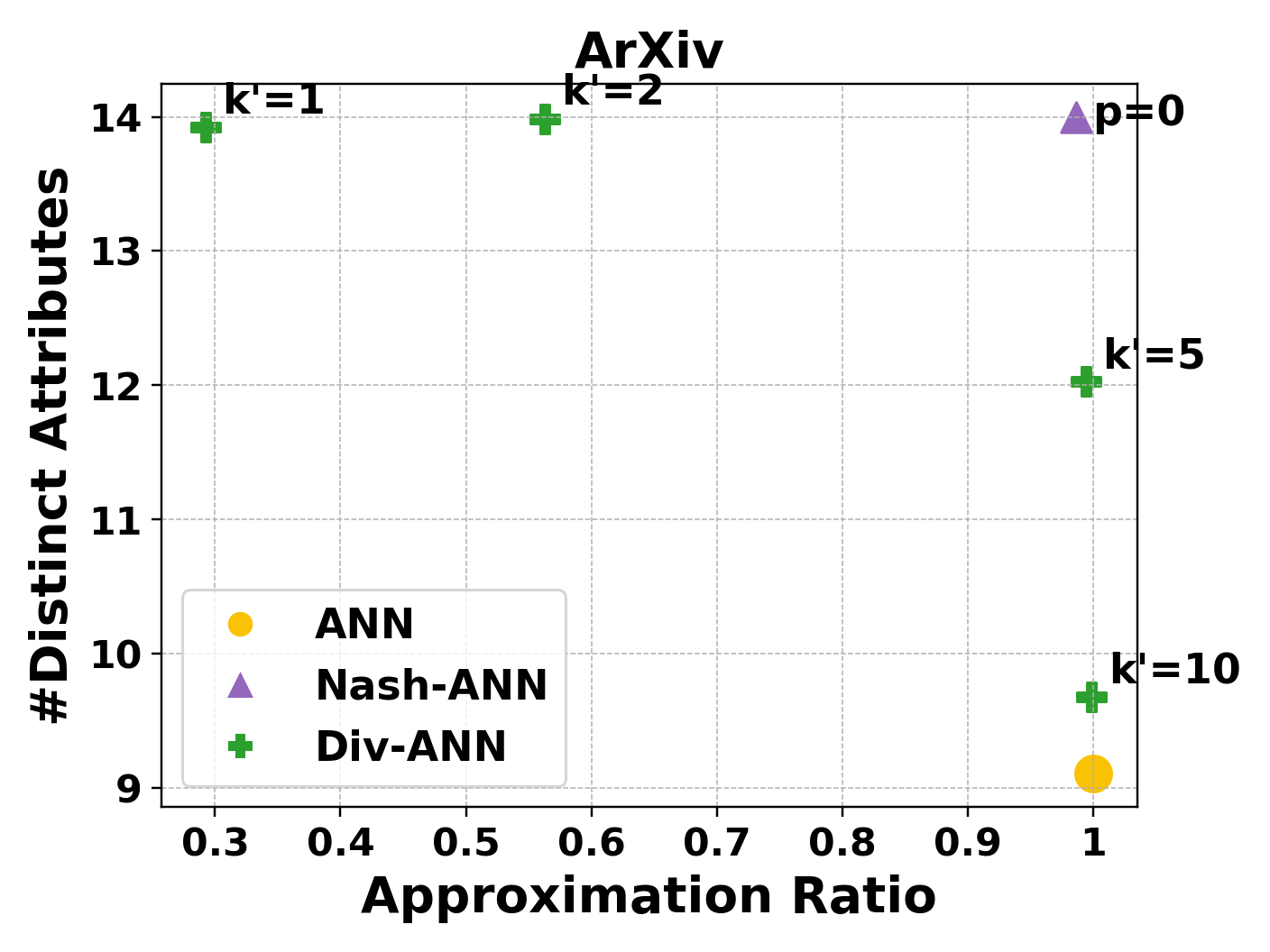} \\
\end{tabular}
\caption{The plots show approximation ratio versus distinct counts trade-offs for various algorithms for $k$  = $10$ \textbf{(Left)} and  $k$  = $50$ \textbf{(Right)} in  single-attribute setting on \arxiv\ dataset.}
\label{fig:AppendixResultsSingleAttriArxivDC}
\end{figure}

\subsubsection{Recall Versus Entropy}
\label{appendix:recallEntropy}
We also report results for another popular relevance metric in the nearest neighbor search literature, namely, recall. The results for different datasets are shown in Figures~\ref{fig:AppendixResultsSingleAttriAmazonRecall}, \ref{fig:AppendixResultsSingleAttriDeepClusRecall}, \ref{fig:AppendixResultsSingleAttriSIFTCLusRecall},  \ref{fig:AppendixResultsSingleAttriSIFTProbRecall}, \ref{fig:AppendixResultsSingleAttriDeepProbRecall}, and \ref{fig:AppendixResultsSingleAttriArxivRecall}. 

Note that as discussed earlier (\Cref{subsec:description-various-metrics--of-diversity-and-relevance}, \Cref{remark:approx-ratio-better-metric-than-recall}), recall can be a fragile metric when the goal is to balance between diversity and relevance. However, we still report recall to be consistent with prior literature and to demonstrate that \ouralgo\ does not perform poorly. In fact, it is evident from the plots that \ouralgo's recall value (relevance) surpasses that of \divann\ with $k' = 1$ (most attribute diverse solution) while achieving a similar entropy. As already noted, the approximation ratio for \ouralgo\ remains sufficiently high, indicating  that the retrieved set of neighbors lies within a reasonably good neighborhood of the true nearest neighbors of a given query.

\begin{figure}[!t]
\centering
\begin{tabular}{@{}c@{}c@{}c@{}c@{}}
\includegraphics[width=0.48\textwidth]{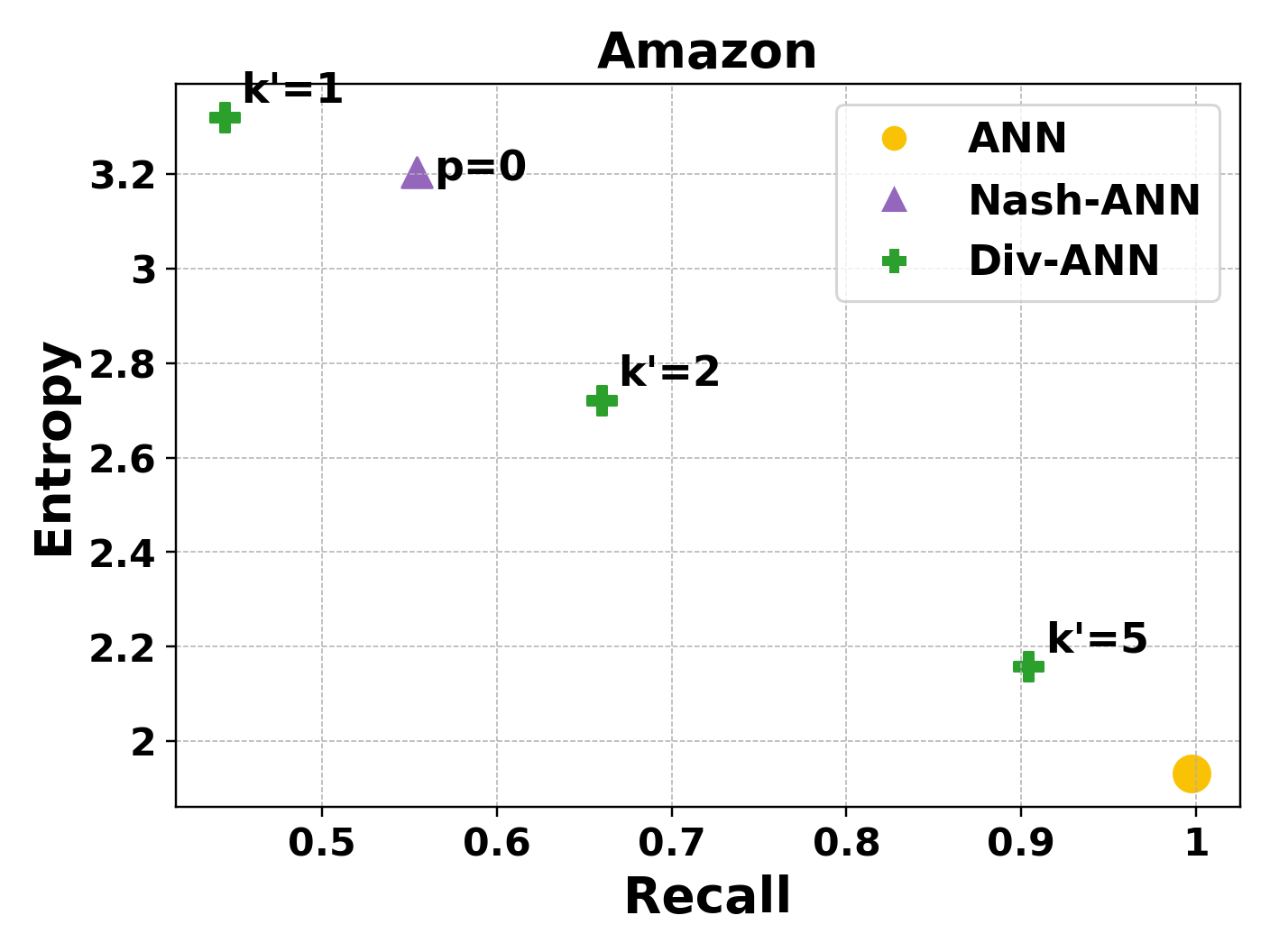} &
\includegraphics[width=0.48\textwidth]{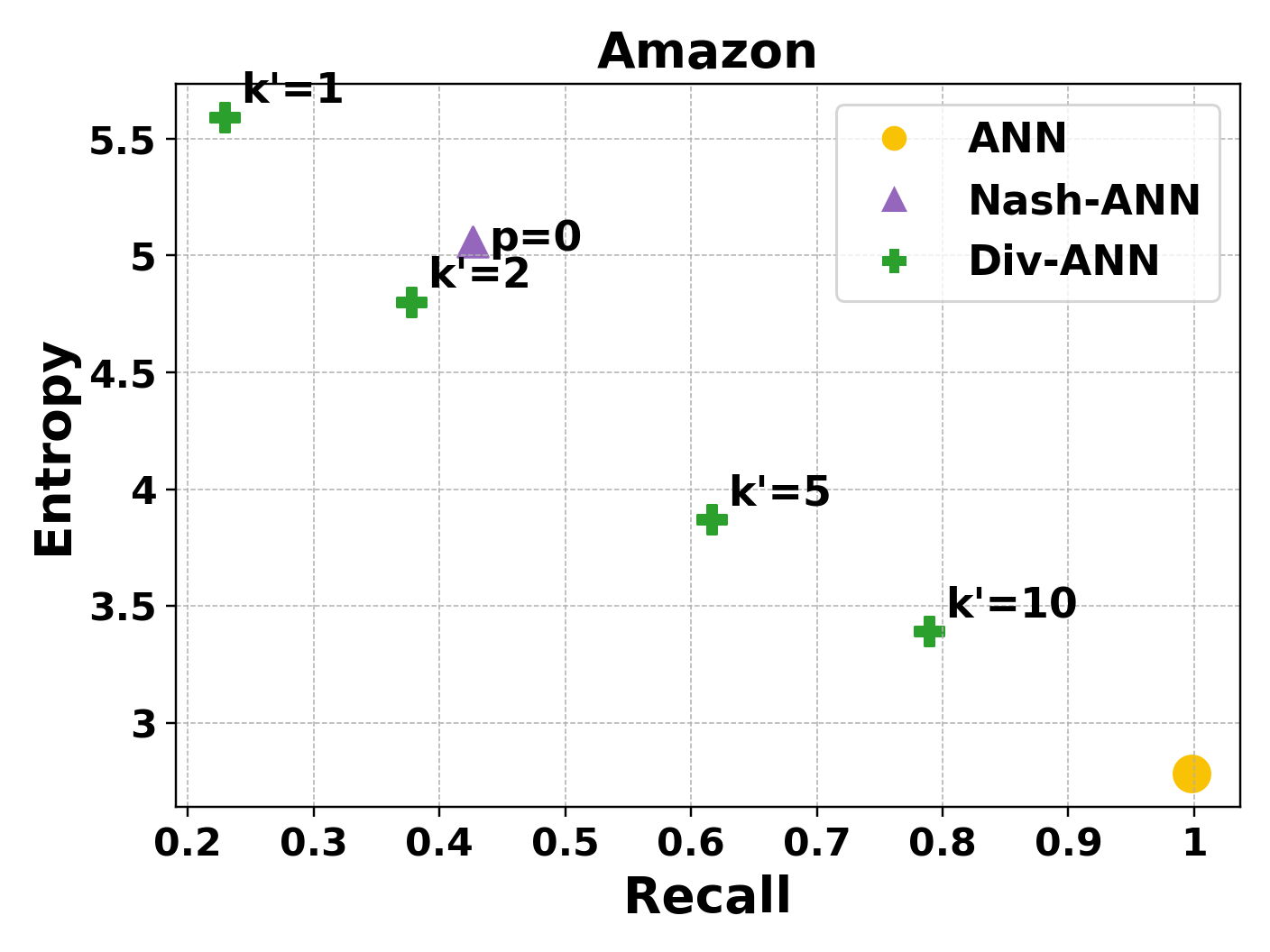} \\
\end{tabular}
\caption{The plots show recall versus entropy trade-offs for $k$  = $10$ \textbf{(Left)} and $k$  = $50$ \textbf{(Right)} in the single-attribute setting on \amazon\ dataset.}
\label{fig:AppendixResultsSingleAttriAmazonRecall}
\end{figure}

\begin{figure}[!t]
\centering
\begin{tabular}{@{}c@{}c@{}c@{}c@{}}
\includegraphics[width=0.48\textwidth]{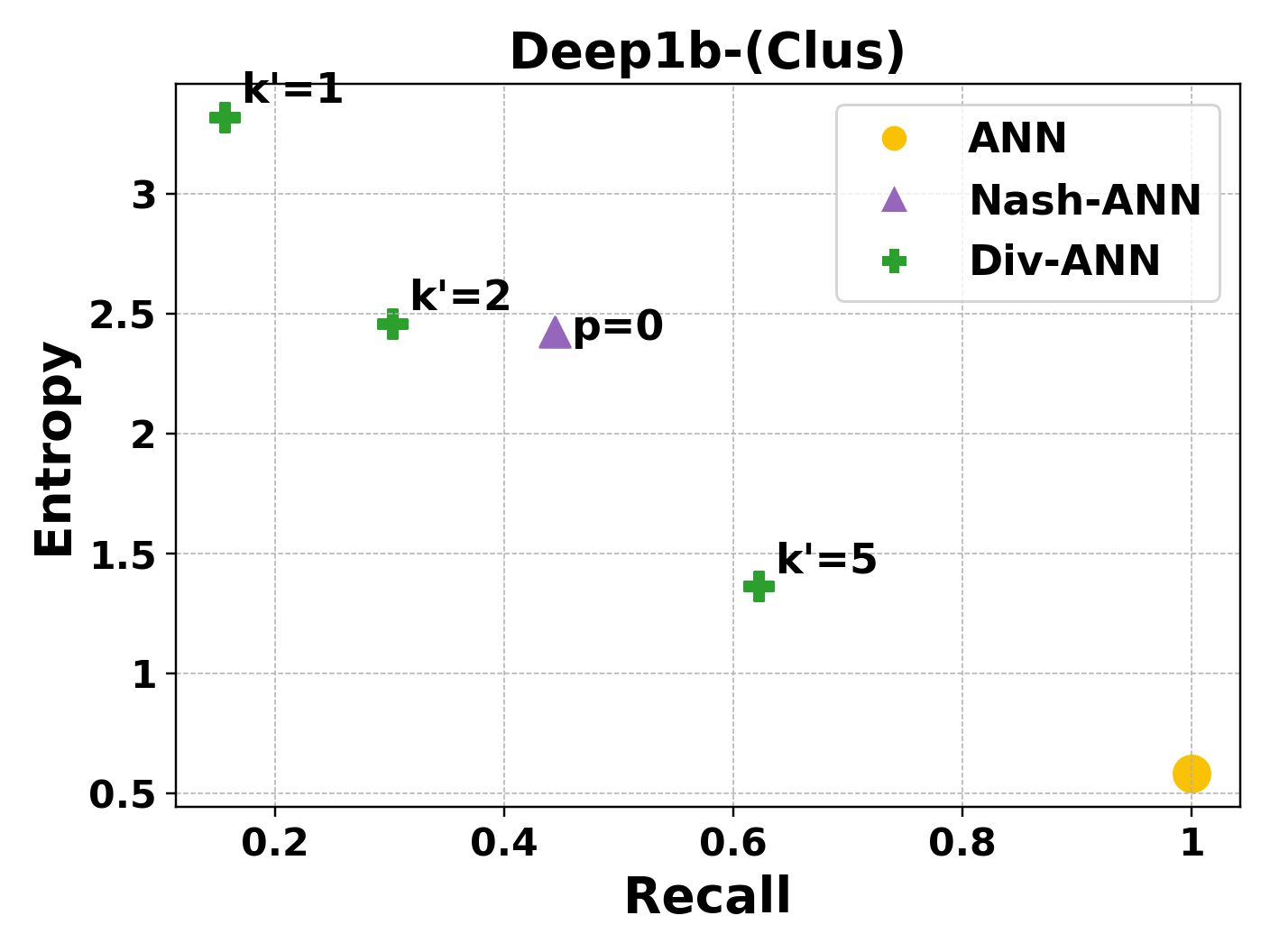} &
\includegraphics[width=0.48\textwidth]{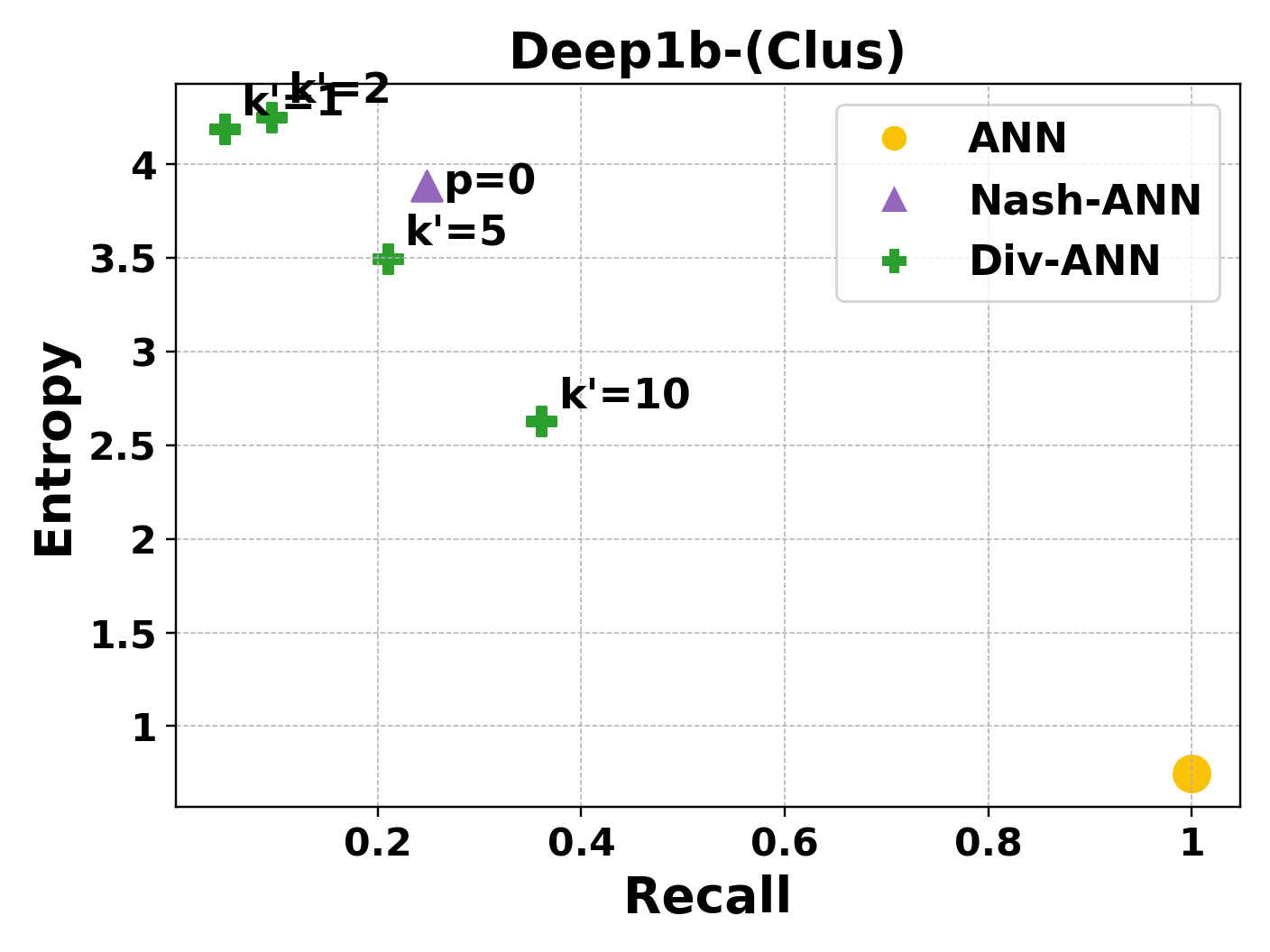} \\
\end{tabular}
\caption{The plots show recall versus entropy trade-offs for $k$  = $10$ \textbf{(Left)} and $k$  = $50$ \textbf{(Right)} in the single-attribute setting on \deepC\ dataset.}
\label{fig:AppendixResultsSingleAttriDeepClusRecall}
\end{figure}

\begin{figure}[t!]
\centering
\begin{tabular}{@{}c@{}c@{}c@{}c@{}}
\includegraphics[width=0.48\textwidth]{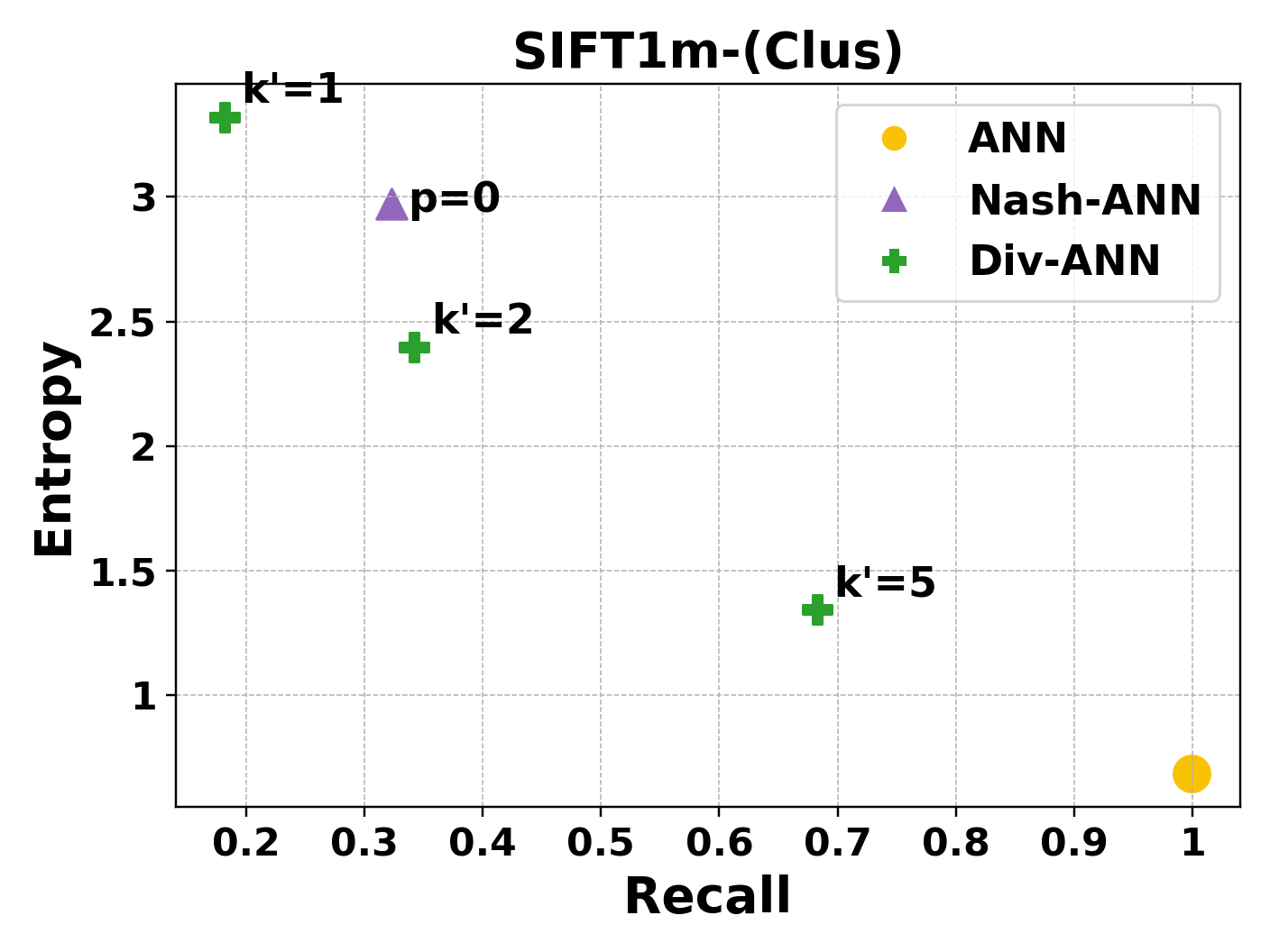} &
\includegraphics[width=0.48\textwidth]{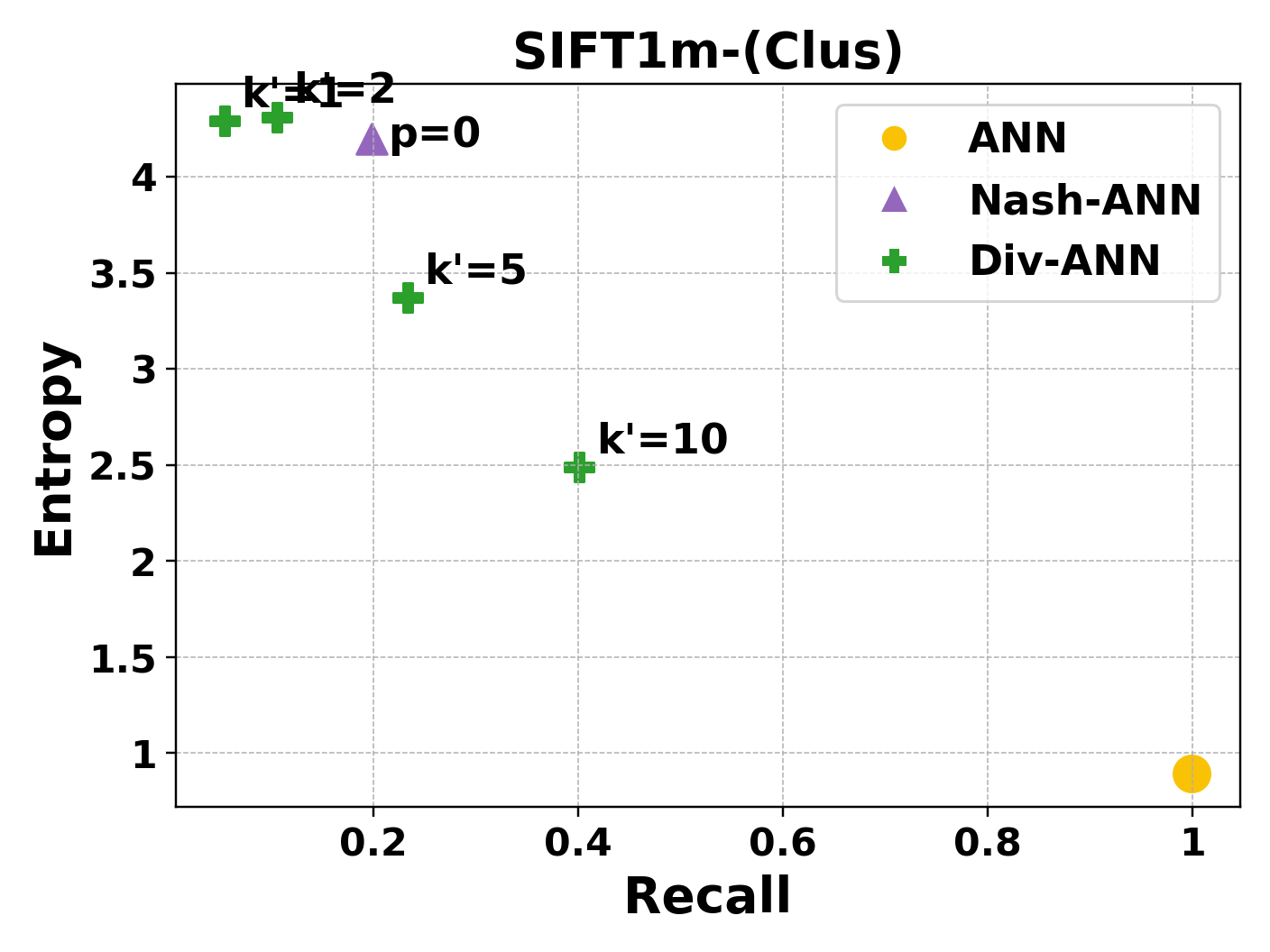} \\
\end{tabular}
\caption{The plots show recall versus entropy trade-offs for $k$  = $10$ \textbf{(Left)} and $k$  = $50$ \textbf{(Right)} in the single-attribute setting on \siftC\ dataset.}
\label{fig:AppendixResultsSingleAttriSIFTCLusRecall}
\end{figure}

\begin{figure}[t!]
\centering
\begin{tabular}{@{}c@{}c@{}c@{}c@{}}
\includegraphics[width=0.48\textwidth]{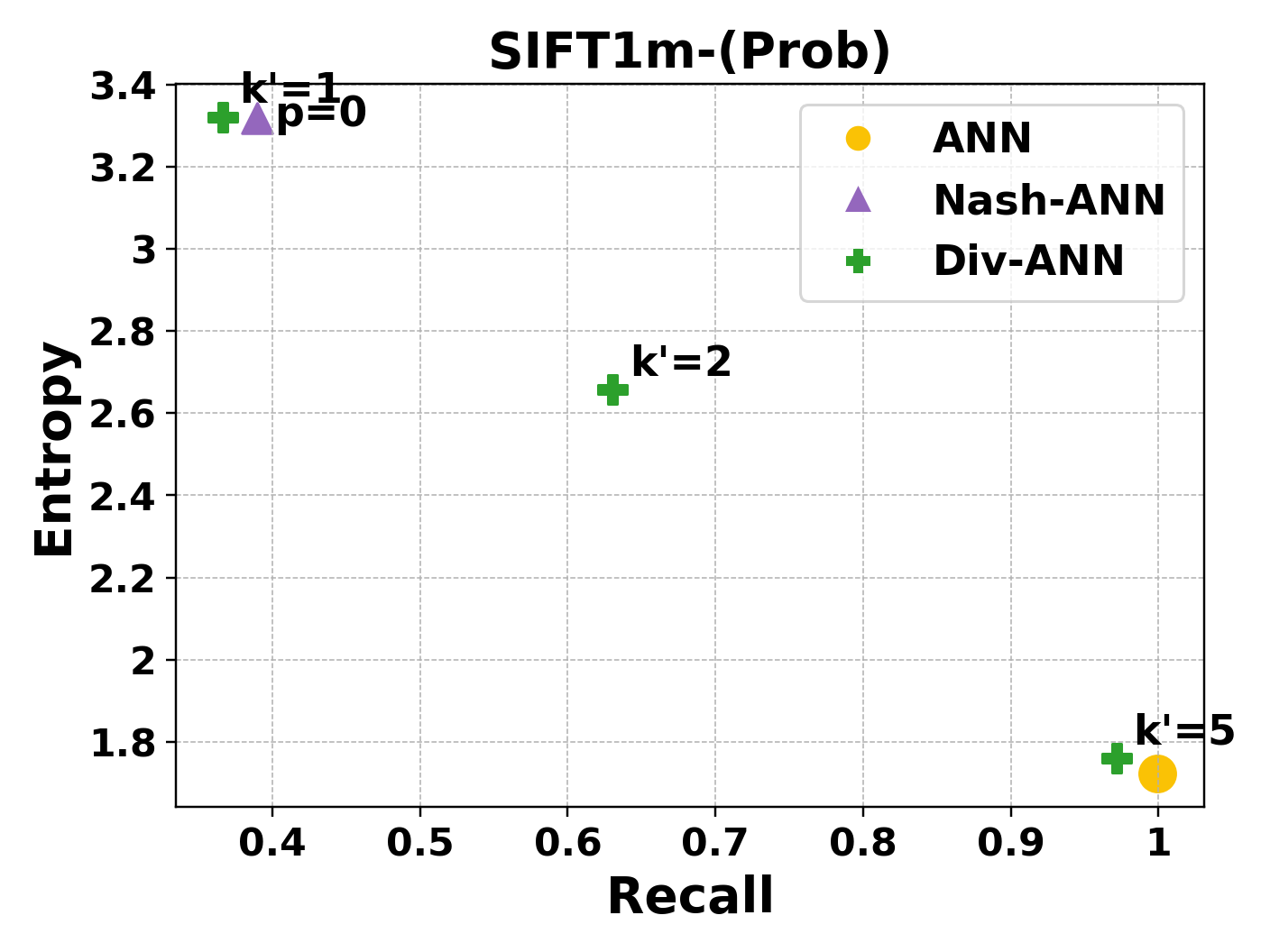} &
\includegraphics[width=0.48\textwidth]{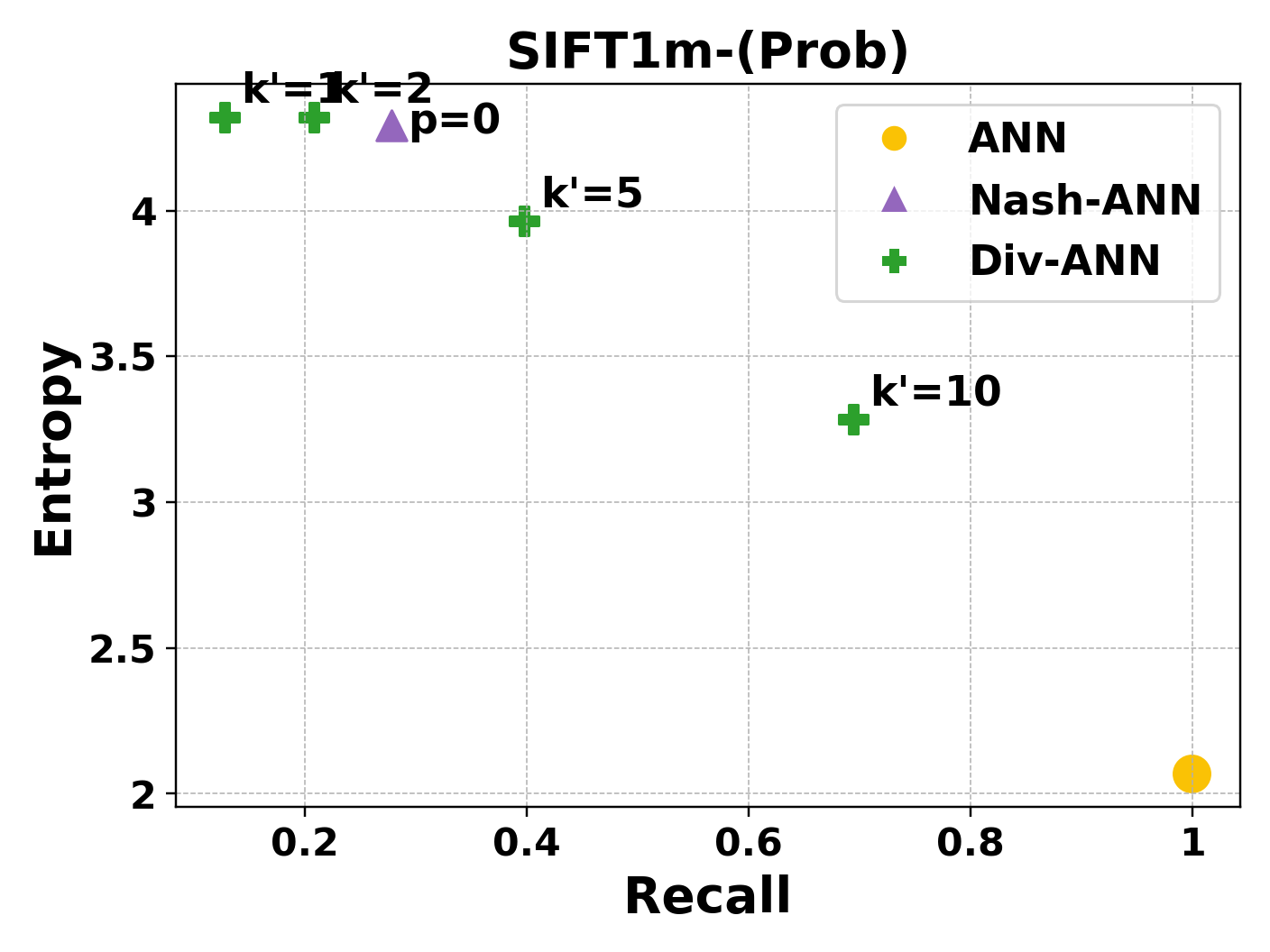} \\
\end{tabular}
\caption{The plots show recall versus entropy trade-offs for $k$  = $10$ \textbf{(Left)} and $k$  = $50$ \textbf{(Right)} in the single-attribute setting on \siftP\ dataset.}
\label{fig:AppendixResultsSingleAttriSIFTProbRecall}
\end{figure}

\begin{figure}[!t]
\centering
\begin{tabular}{@{}c@{}c@{}c@{}c@{}}
\includegraphics[width=0.48\textwidth]{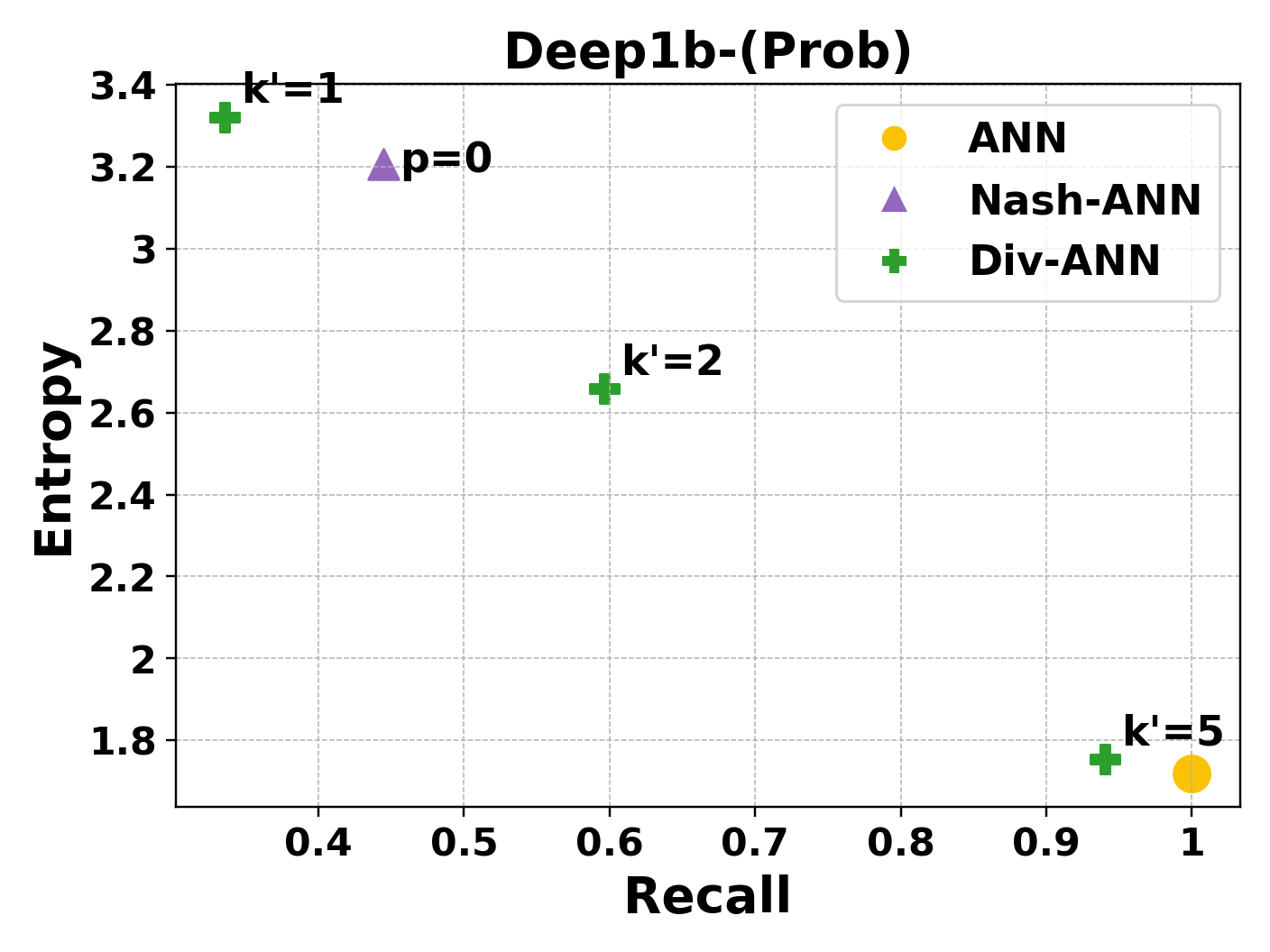} &
\includegraphics[width=0.48\textwidth]{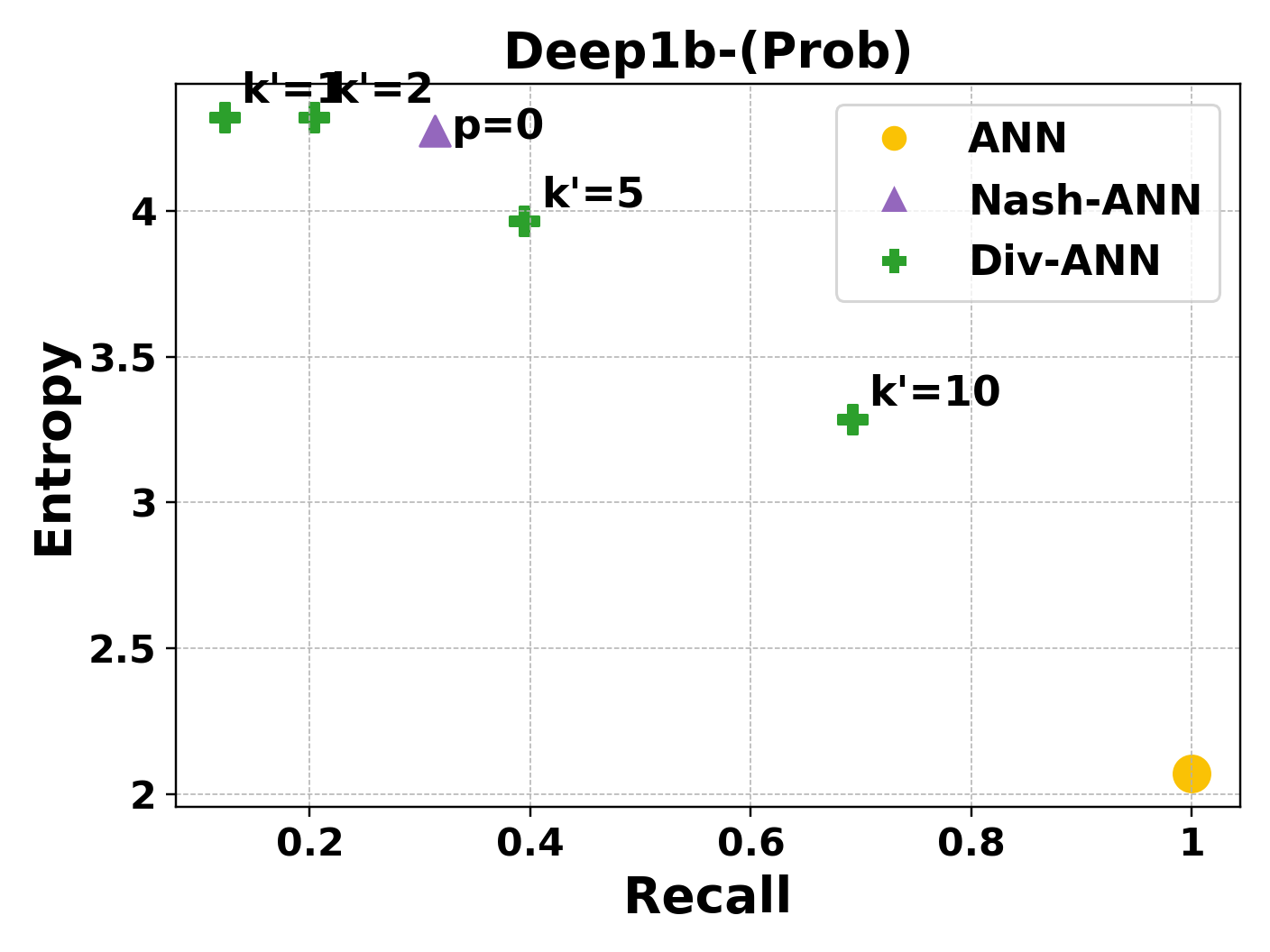} \\
\end{tabular}
\caption{The plots show recall versus entropy trade-offs for $k$  = $10$ \textbf{(Left)} and $k$  = $50$ \textbf{(Right)} in the single-attribute setting on \deepP\ dataset.}
\label{fig:AppendixResultsSingleAttriDeepProbRecall}
\end{figure}

\begin{figure}[!t]
\centering
\begin{tabular}{@{}c@{}c@{}c@{}c@{}}
\includegraphics[width=0.48\textwidth]{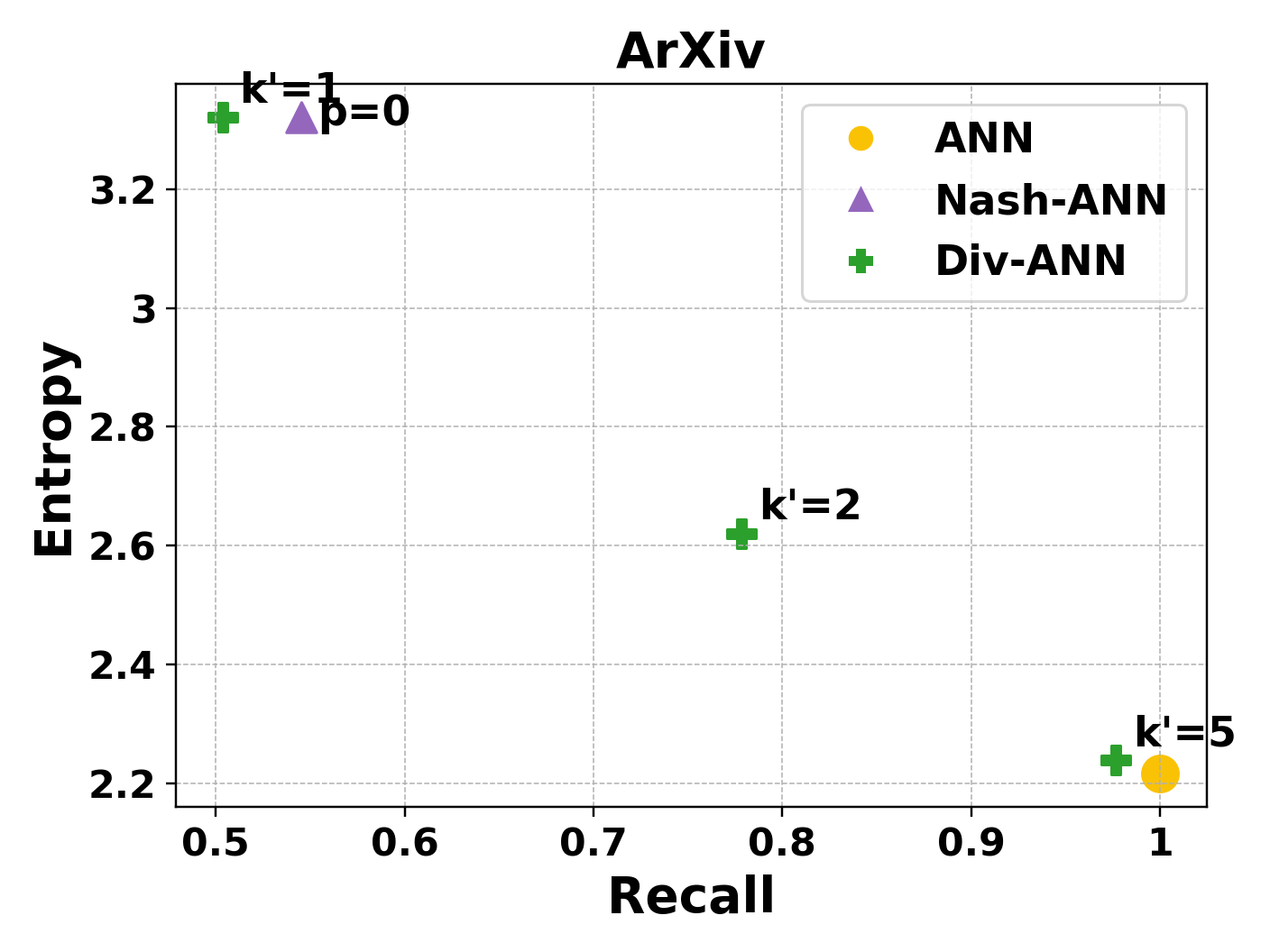} &
\includegraphics[width=0.48\textwidth]{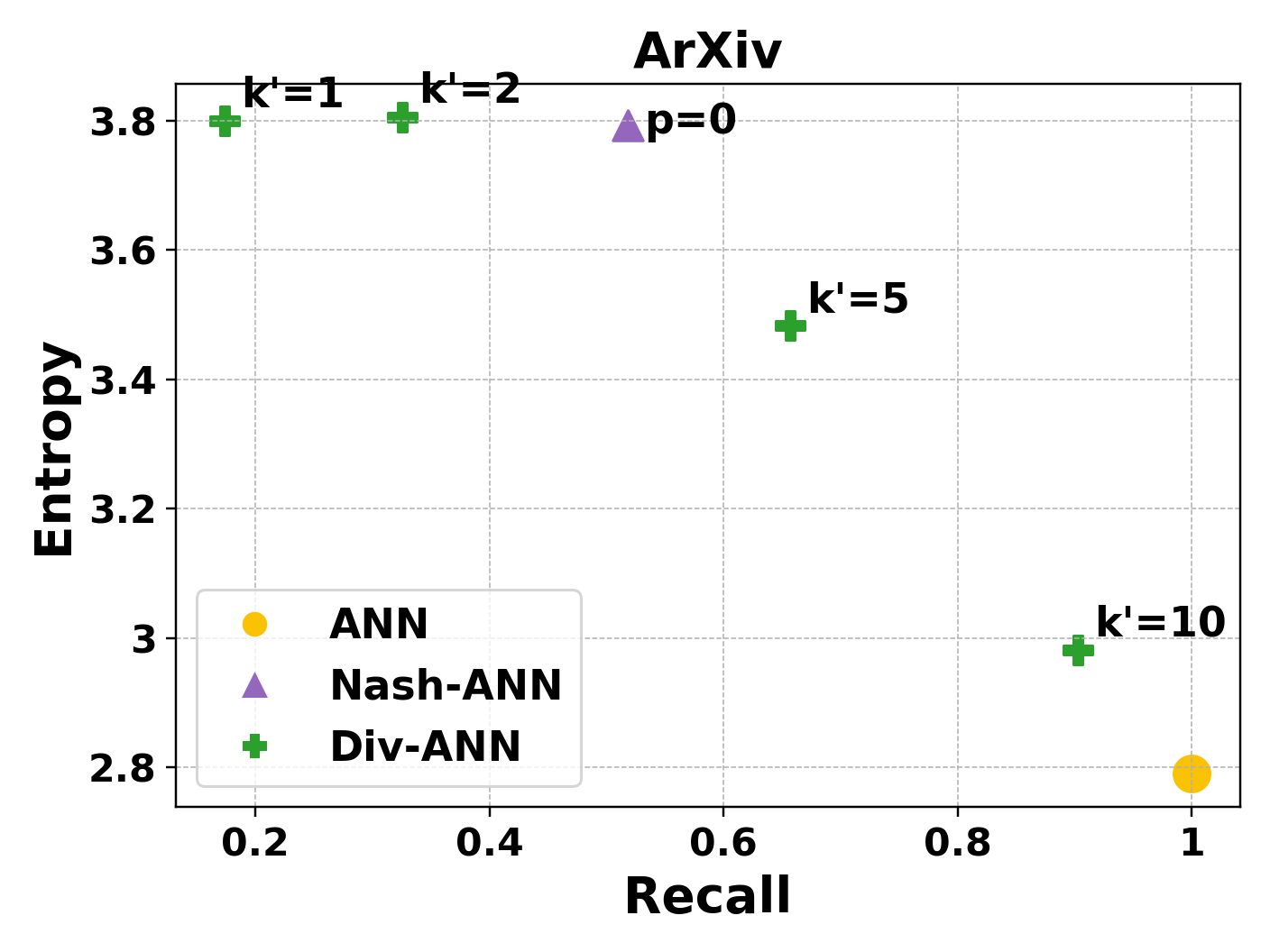} \\
\end{tabular}
\caption{The plots show recall versus entropy trade-offs for $k$  = $10$ \textbf{(Left)} and $k$  = $50$ \textbf{(Right)} in the single-attribute setting on \arxiv\ dataset.}
\label{fig:AppendixResultsSingleAttriArxivRecall}
\end{figure}

\subsection{Balancing Relevance and Diversity: Multi-attribute Setting}
\label{appendix:exp-multi-attribute}
Recall that our welfarist formulation seamlessly extends to the multi-attribute setting. In~\Cref{section:experiments}, we discussed the performance of \multinash\ and \multidivann\ on \siftC, where each input vector was associated with four attributes. In this section, we repeat the same set of experiments on one of the real-world datasets, namely \arxiv, which naturally contains two attribute classes ($m=2$; see~\Cref{subsec:description-various-metrics--of-diversity-and-relevance}, Diversity Metrics): update year ($|C_1| = 14$) and paper category ($|C_2| = 16$). Therefore, $c = | C_1 | + | C_2 | = 30$. The results for $k = 50$ are presented in Figure~\ref{fig:AppendixResultsMultiAttriArxiv}. Note that, in each plot, we restrict the entropy to one of the attribute classes ($C_1$ or $C_2$) so that the diversity within a class can be understood from these plots. The results indicate that \multinash\ achieves an approximation ratio very close to one while maintaining entropy levels comparable to \multidivann\ with $k' = 1$ or $2$ for both the attribute classes. In fact, \multinash\ Pareto dominates \multidivann\ with $k'=5$.

We also study the effect of varying $p$ in $p$-NNS problem in the multi-attribute setting. The results for performance of \texttt{Multi p-mean-ANN} (an analogue of \multinash, detailed in \Cref{subsec:algorithm-details-for-experiments}) for $p \in \{-10, -1, -0.5, 0, 0.5, 1\}$ are shown in Figures~\ref{fig:AppendixResultsMultiAttriArxivpTrend} and \ref{fig:mainPaperResultsMultiAttriB}. Interestingly, we observe that with decreasing $p$, the entropy (across $C_1$ or $C_2$) increases but the approximation ratio remains nearly the same and very close to $1$. On the other hand, \multidivann\ with $k'=1$ has very low approximation ratio. In fact, \texttt{Multi p-mean-ANN} with $p=-1$ and $-10$ Pareto dominates \multidivann\ with $k'=1$.

\begin{figure}[t!]
\centering
\begin{tabular}{@{}c@{}c@{}c@{}c@{}}
\includegraphics[width=0.48\textwidth]{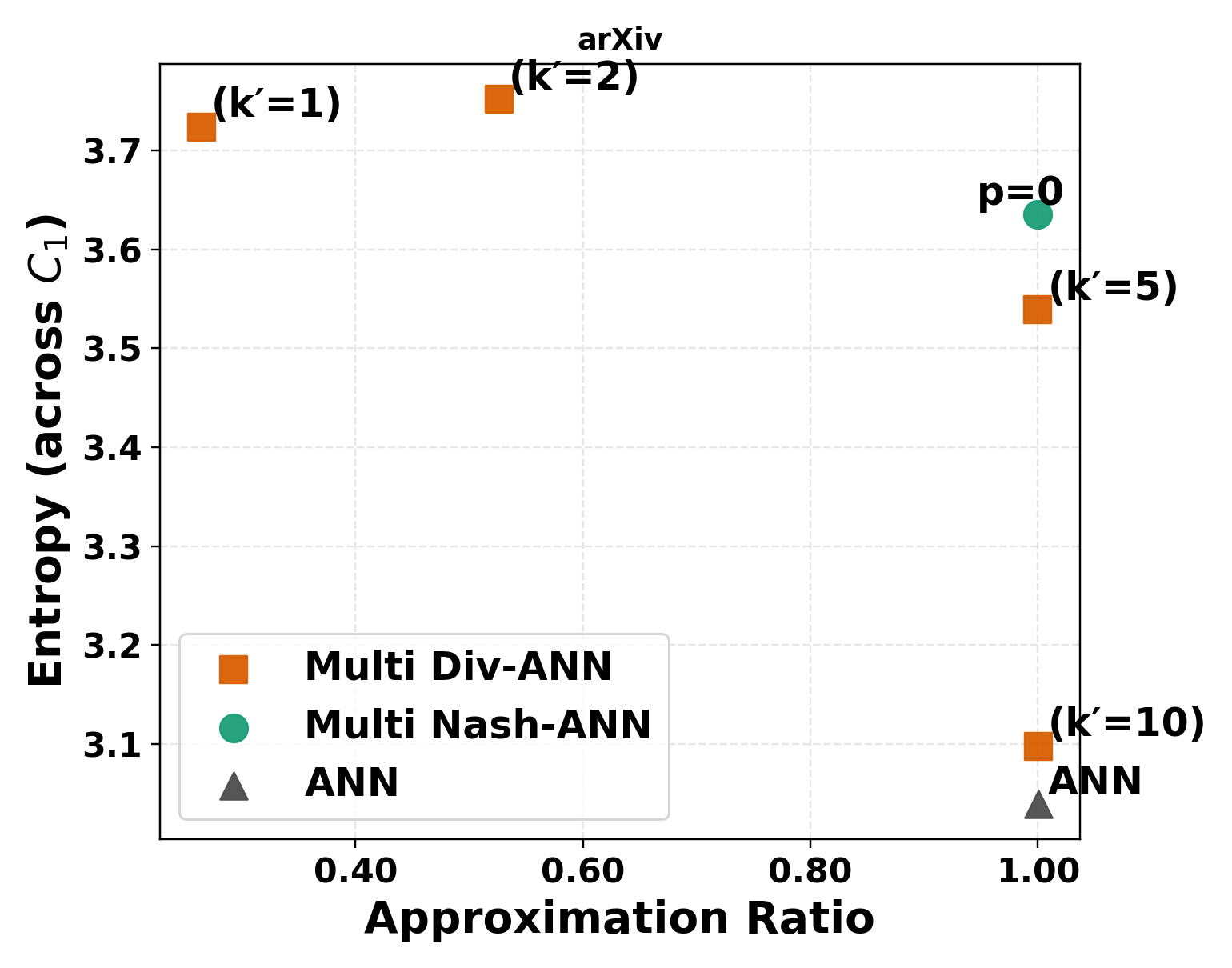} &
\includegraphics[width=0.48\textwidth]{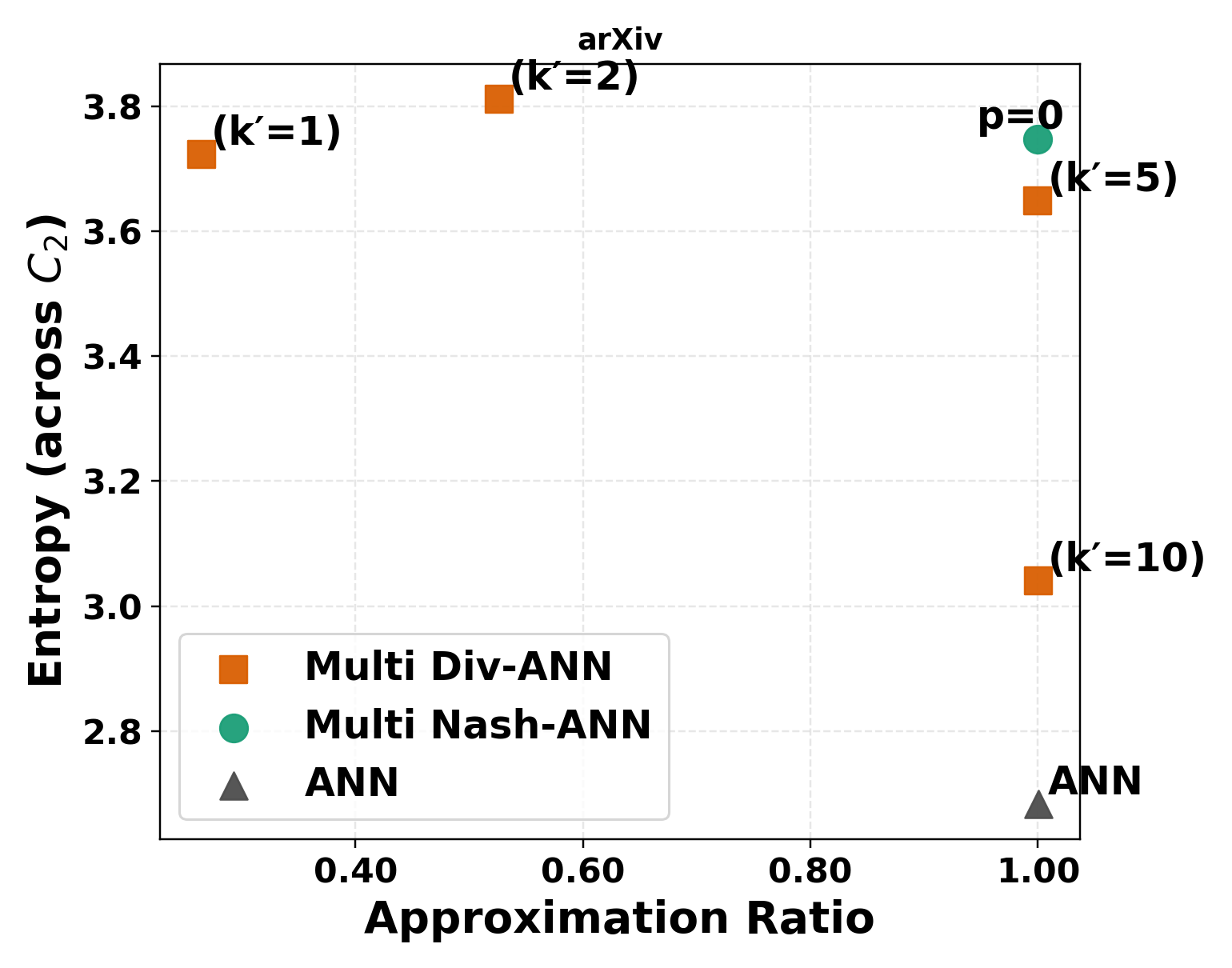} \\
\end{tabular}
\caption{The plots show approximation ratio versus entropy trade-offs for various algorithms on \arxiv\ dataset across attribute classes $C_1$ \textbf{(Left)} and $C_2$ \textbf{(Right)} in the multi-attribute setting for $k=50$.}
\label{fig:AppendixResultsMultiAttriArxiv}
\end{figure}

\begin{figure}[t!]
\centering
\begin{tabular}{@{}c@{}c@{}c@{}c@{}}
\includegraphics[width=0.48\textwidth]{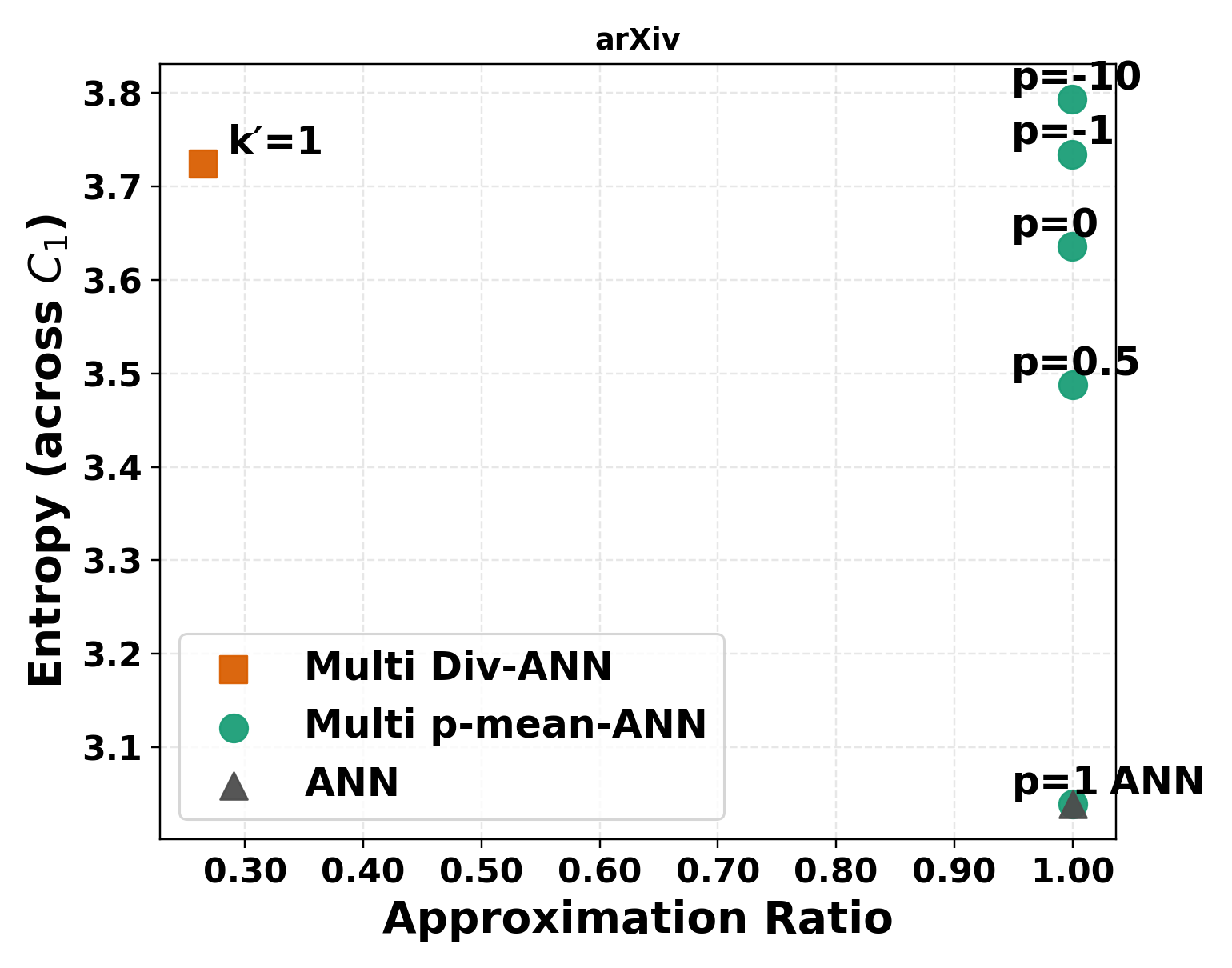} &
\includegraphics[width=0.48\textwidth]{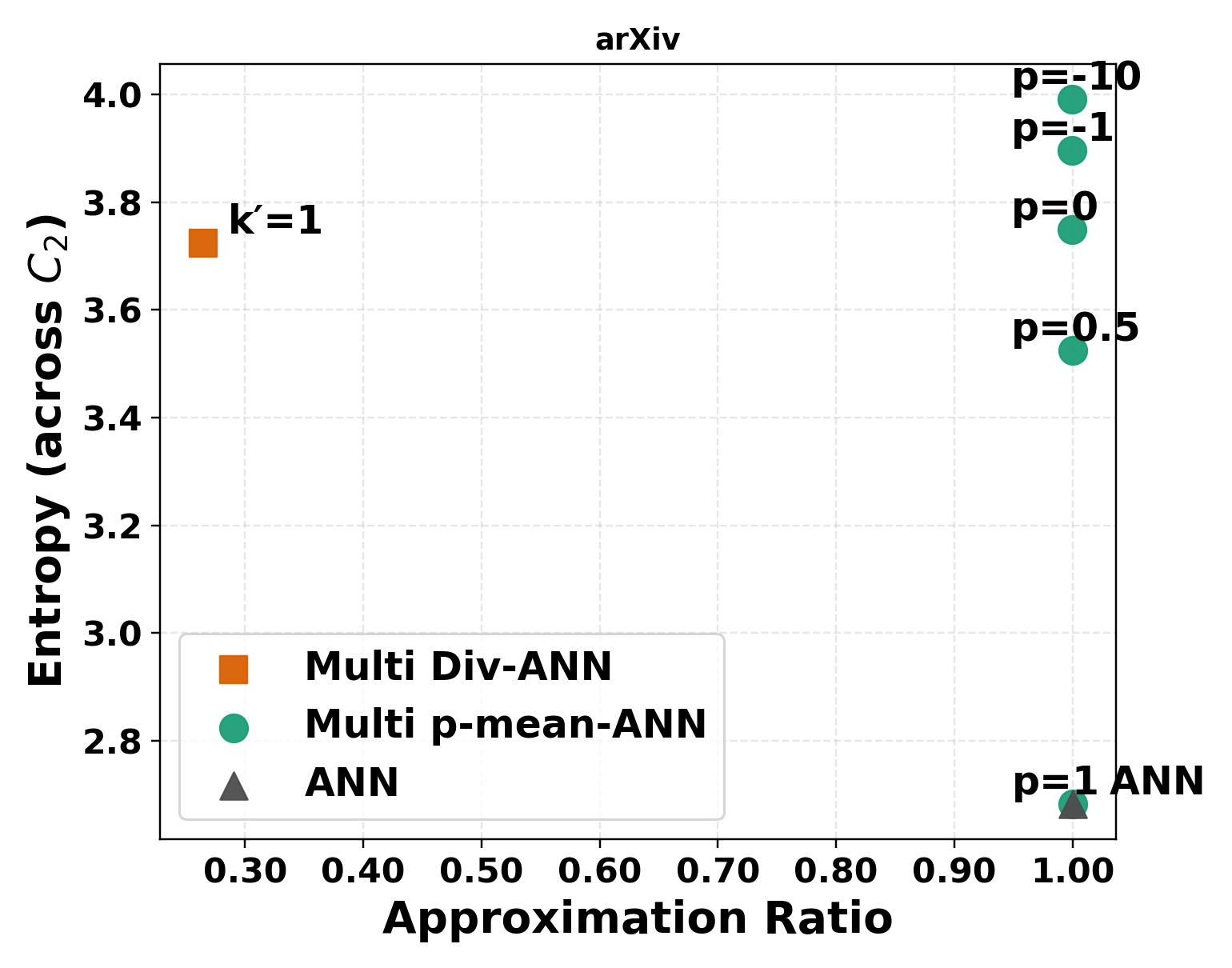} \\
\end{tabular}
\caption{The plots show approximation ratio versus entropy trade-offs for $p$-mean-ANN, as $p$ varies, for $k$ = $50$ across attribute classes $C_1$ \textbf{(Left)} and $C_2$ \textbf{(Right)} in the multi-attribute setting on \arxiv\ dataset.}
\label{fig:AppendixResultsMultiAttriArxivpTrend}
\end{figure}

\begin{figure}[t!]
\centering
\begin{tabular}{@{}c@{}c@{}c@{}c@{}}
\includegraphics[width=0.48\textwidth]{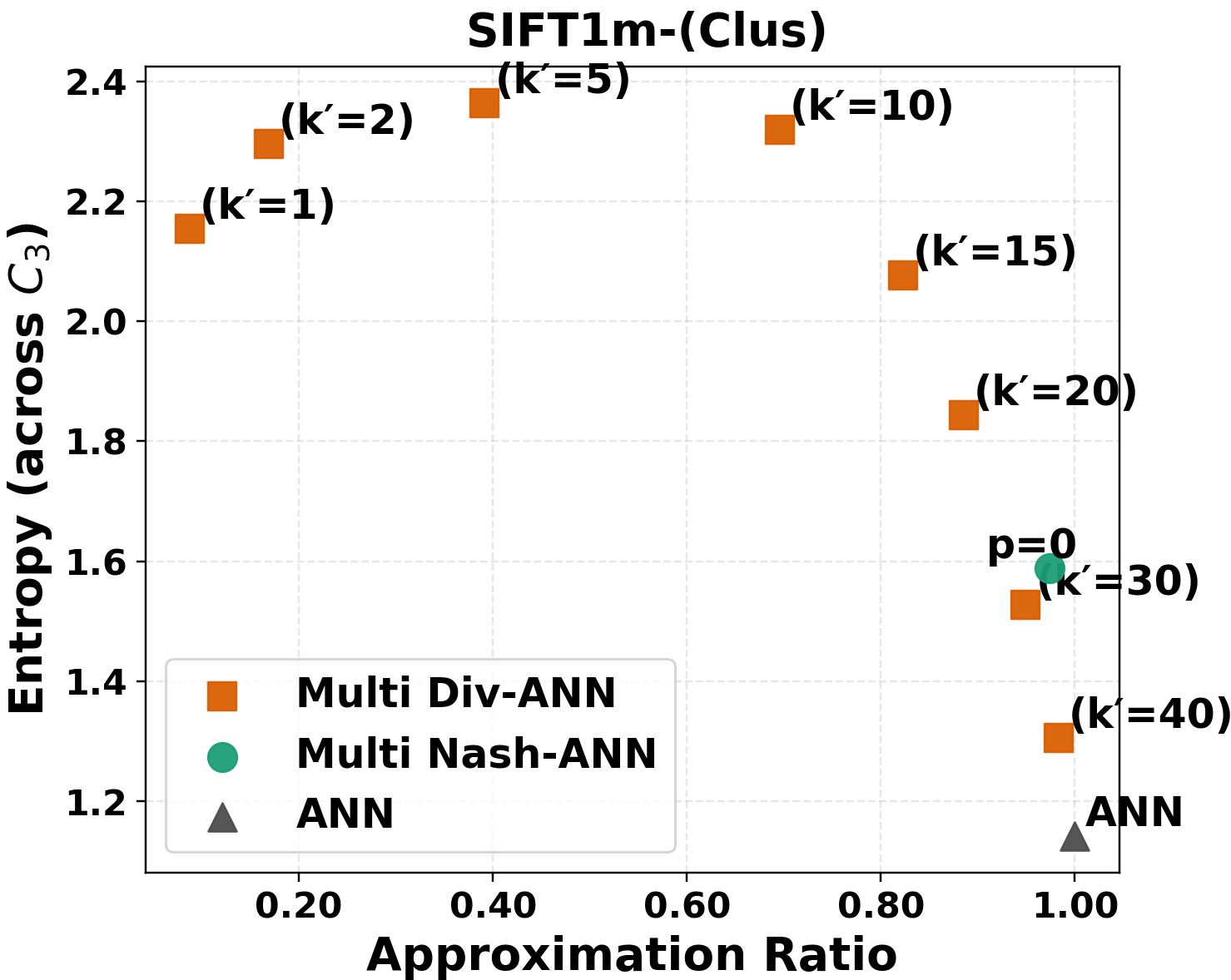} &
\includegraphics[width=0.48\textwidth]{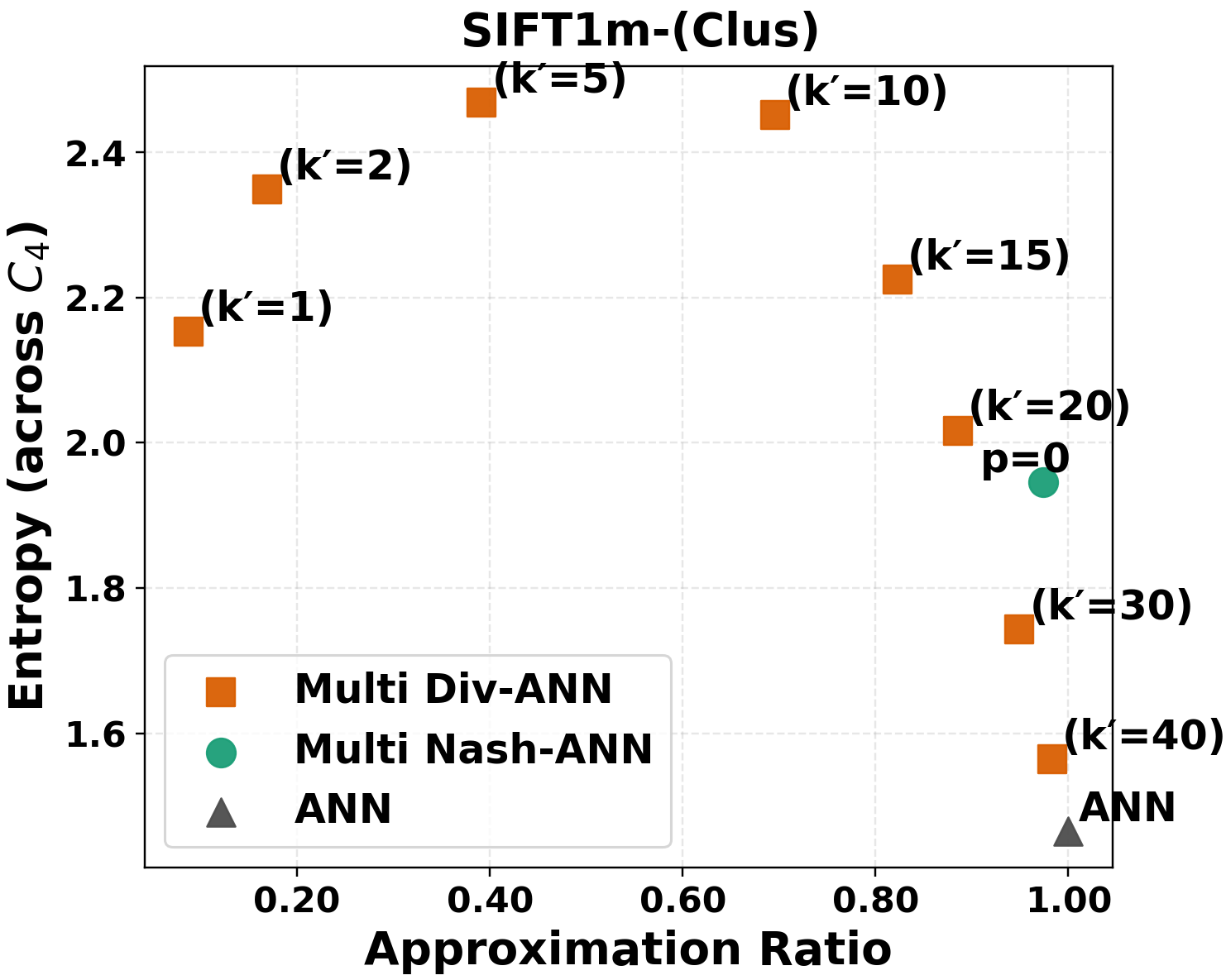} \\
\includegraphics[width=0.48\textwidth]{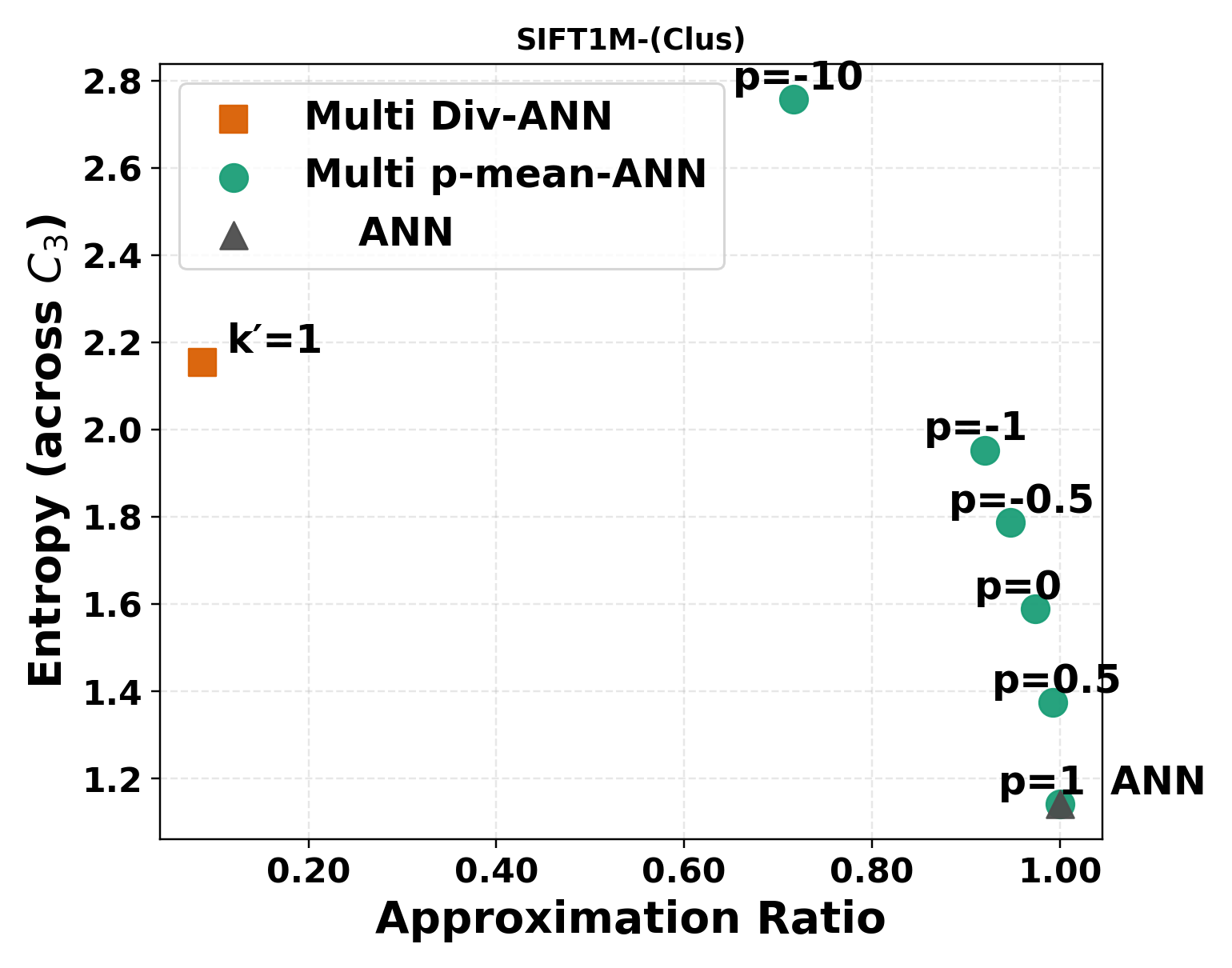} &
\includegraphics[width=0.48\textwidth]{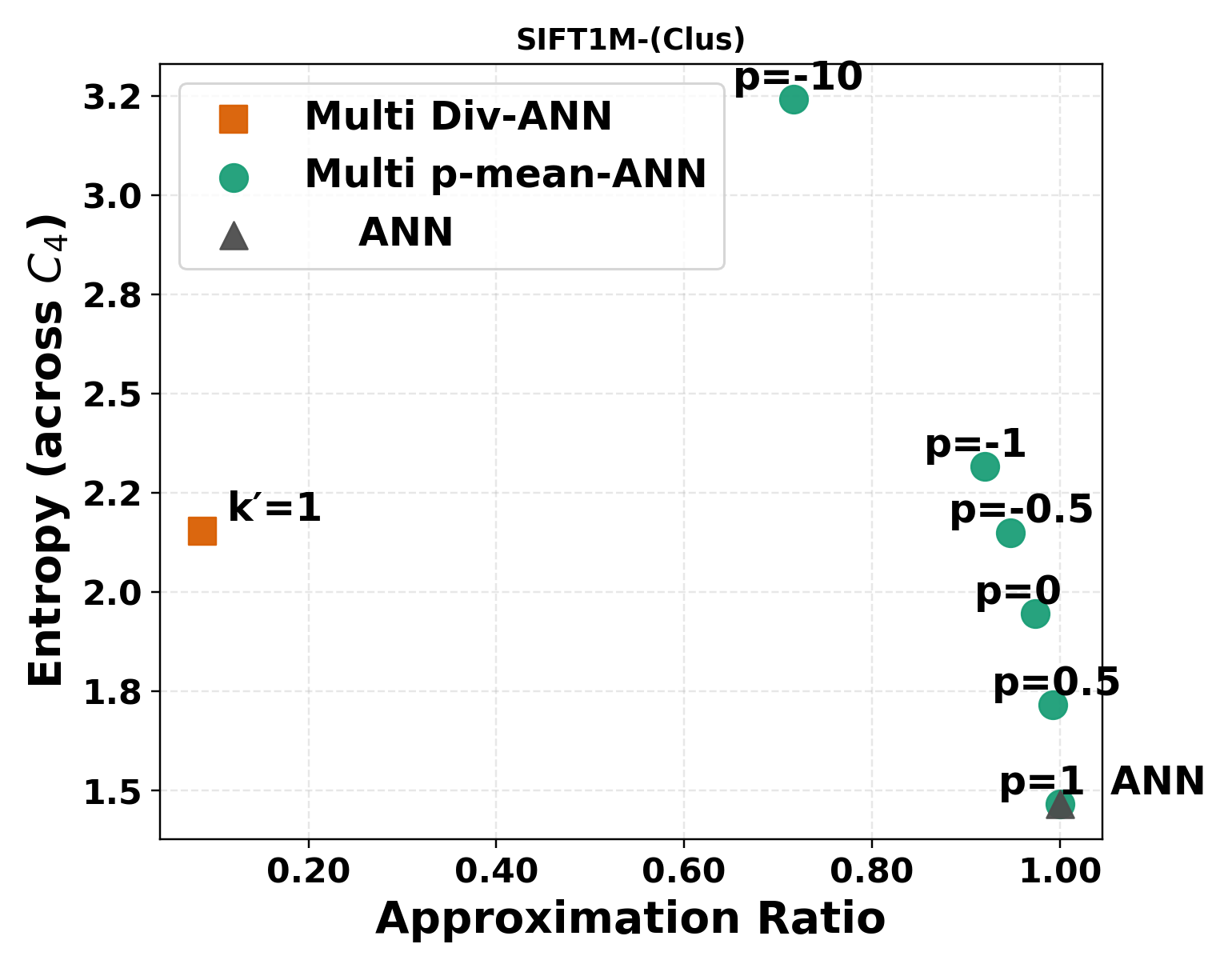} \\
\end{tabular}
\caption{\textbf{Top Row}: The plots show approximation ratio versus entropy trade-offs for \ouralgo\ against \divann\ with varying values of $k'$ for attribute class $C_3$ \textbf{(left)} and $C_4$ \textbf{(right)}. \textbf{Bottom Row}: The plots show approximation ratio versus entropy trade-offs for $p$-mean-ANN, as $p$ varies, across attribute classes $C_3$ \textbf{(left)} and $C_4$ \textbf{(right)}. Both the rows correspond to $k=50$ on \siftC\ dataset in the multi-attribute setting.}
\label{fig:mainPaperResultsMultiAttriB}
\end{figure}

\subsection{More Experiments for \texttt{$p$-FetchUnion-ANN}}
\label{appendix:FetchNashUnionNash}

In this section, we empirically study a faster heuristic algorithm for NSW and $p$-mean welfare formulations. Specifically, the heuristic---called \texttt{$p$-FetchUnion-ANN}---first fetches a sufficiently large candidate set of vectors (irrespective of their attributes) using the \diskann~algorithm. Then, it applies the Nash (or $p$-mean) selection (similar to Line~\ref{line:marginal} in~\Cref{algo:greedy-nash-ann} or Lines~\ref{line:marginal-positive-p}-\ref{line:marginal-negative-p} in~\Cref{algo:p-mean-ann}) within this set. That is, instead of starting out with $k$ neighbors for each $\ell \in [c]$ (as in Line \ref{line:hat-D} of \Cref{algo:greedy-nash-ann}), the alternative here is to work with sufficiently many neighbors from the set $\cup_{\ell=1}^c D_\ell$. 

We empirically show (in Tables~\ref{tab:amazon-k50} to~\ref{tab:sift20-uni-k50}) that this heuristic consistently achieves performance comparable to \texttt{$p$-Mean-ANN} across nearly all datasets and evaluation metrics. Since \texttt{$p$-FetchUnion-ANN} starts with a large pool of vectors (with high similarity to the query) retrieved by the~\diskann\ algorithm without diversity considerations, it achieves improved  approximation ratio over \texttt{$p$-Mean-ANN}. This trend is clearly evident in two datasets, namely \deepC\ and \siftC. However, the improvement in approximation ratio comes at the cost of reduced entropy, which can be explained by the fact that in restricting its search to an initially fetched large pool of vectors, \texttt{$p$-FetchUnion-ANN} may miss out on a more diverse solution that exists over the entire dataset. Another important aspect of \texttt{$p$-FetchUnion-ANN} is that, because it retrieves all neighbors from the union at once, the heuristic delivers substantially higher throughput (measured as queries answered per second, QPS) and therefore lower latency. The results validating these findings are reported in Tables~\ref{tab:QPSandLatency} and~\ref{tab:QPSandLatencyAmazon} for the \siftC\ and \amazon\ datasets, respectively. In particular, it serves almost \textbf{10$\times$} more queries on \siftC\ and \textbf{3$\times$} more queries on \amazon\ dataset. The latency values exhibit a similar trend with reductions of similar magnitude. In summary, these observations position the heuristic as a notably fast method for NaNNS and $p$-NNS, particularly when $c$ is large.

\begin{table}[!h]
\centering
\caption{Comparison of relevance and diversity of \pMeanANN, \texttt{$p$-FetchUnion-ANN} across different values of $p$ against \diskann\ and \divann\ ($k'=1$) for \amazon\ at $k=50$. }
\label{tab:amazon-k50}
\scalebox{0.75}{
\setlength{\tabcolsep}{6pt}
\renewcommand{\arraystretch}{1.2}
\begin{tabular}{llcccccc}
\toprule
\textbf{Metric} & \textbf{Algorithm} & $p=-10$ & $p=-1$ & $p=-0.5$ & $p=0$ & $p=0.5$ & $p=1$ \\
\midrule
\multirow{4}{*}{Approx.~Ratio}
  & \texttt{$p$-Mean-ANN} & 0.865$\pm$0.045 & 0.909$\pm$0.029 & 0.922$\pm$0.027 & 0.938$\pm$0.023 & 0.961$\pm$0.018 & 1.000$\pm$0.000 \\
  & \texttt{$p$-FetchUnion-ANN} & 0.907$\pm$0.033 & 0.912$\pm$0.030 & 0.921$\pm$0.027 & 0.935$\pm$0.024 & 0.958$\pm$0.019 & 1.000$\pm$0.000 \\
  & \diskann\  & \multicolumn{6}{c}{1.000$\pm$0.000} \\
  & \divann\  ($k'\!=\!1$) & \multicolumn{6}{c}{0.813$\pm$0.053} \\
\midrule
\multirow{4}{*}{Entropy}
  & \texttt{$p$-Mean-ANN} & 5.644$\pm$0.000 & 5.382$\pm$0.135 & 5.252$\pm$0.153 & 5.058$\pm$0.178 & 4.687$\pm$0.227 & 2.782$\pm$0.684 \\
  & \texttt{$p$-FetchUnion-ANN} & 5.364$\pm$0.156 & 5.333$\pm$0.149 & 5.261$\pm$0.150 & 5.099$\pm$0.171 & 4.736$\pm$0.221 & 2.782$\pm$0.684 \\
  & \diskann\  & \multicolumn{6}{c}{2.782$\pm$0.684} \\
  & \divann\  ($k'\!=\!1$) & \multicolumn{6}{c}{5.594$\pm$0.049} \\
\bottomrule
\end{tabular}
}%
\end{table}

\begin{table}[t]
\centering
\scalebox{0.75}{
\setlength{\tabcolsep}{6pt}
\renewcommand{\arraystretch}{1.2}
\begin{tabular}{llcccccc}
\toprule
\textbf{Metric} & \textbf{Algorithm} & $p=-10$ & $p=-1$ & $p=-0.5$ & $p=0$ & $p=0.5$ & $p=1$ \\
\midrule
\multirow{4}{*}{Approx.~Ratio}
  & \texttt{$p$-Mean-ANN} & 0.985$\pm$0.010 & 0.985$\pm$0.010 & 0.985$\pm$0.010 & 0.986$\pm$0.009 & 0.989$\pm$0.008 & 1.000$\pm$0.001 \\
  & \texttt{$p$-FetchUnion-ANN} & 0.989$\pm$0.007 & 0.989$\pm$0.007 & 0.989$\pm$0.007 & 0.990$\pm$0.006 & 0.991$\pm$0.006 & 1.000$\pm$0.001 \\
  & \diskann\  & \multicolumn{6}{c}{1.000$\pm$0.001} \\
  & \divann\  ($k'\!=\!1$) & \multicolumn{6}{c}{0.293$\pm$0.007} \\
\midrule
\multirow{4}{*}{Entropy}
  & \texttt{$p$-Mean-ANN} & 3.793$\pm$0.002 & 3.793$\pm$0.002 & 3.793$\pm$0.002 & 3.793$\pm$0.002 & 3.793$\pm$0.002 & 2.790$\pm$0.510 \\
  & \texttt{$p$-FetchUnion-ANN} & 3.704$\pm$0.167 & 3.704$\pm$0.166 & 3.704$\pm$0.166 & 3.704$\pm$0.166 & 3.704$\pm$0.166 & 2.790$\pm$0.510 \\
  & \diskann\  & \multicolumn{6}{c}{2.790$\pm$0.510} \\
  & \divann\  ($k'\!=\!1$) & \multicolumn{6}{c}{3.799$\pm$0.029} \\
\bottomrule
\end{tabular}
}%
\caption{Comparison of relevance and diversity of \pMeanANN, \texttt{$p$-FetchUnion-ANN} across different values of $p$ against \diskann\ and \divann\ ($k'=1$) for \arxiv\ at $k=50$. }
\label{tab:arxiv-k50}
\end{table}

\begin{table}[t]
\centering
\scalebox{0.75}{
\setlength{\tabcolsep}{6pt}
\renewcommand{\arraystretch}{1.2}
\begin{tabular}{llcccccc}
\toprule
\textbf{Metric} & \textbf{Algorithm} & $p=-10$ & $p=-1$ & $p=-0.5$ & $p=0$ & $p=0.5$ & $p=1$ \\
\midrule
\multirow{4}{*}{Approx.~Ratio}
  & \texttt{$p$-Mean-ANN} & 0.784$\pm$0.071 & 0.815$\pm$0.065 & 0.831$\pm$0.063 & 0.858$\pm$0.060 & 0.904$\pm$0.049 & 1.000$\pm$0.000 \\
  & \texttt{$p$-FetchUnion-ANN} & 0.958$\pm$0.033 & 0.961$\pm$0.030 & 0.962$\pm$0.029 & 0.963$\pm$0.028 & 0.968$\pm$0.024 & 1.000$\pm$0.000 \\
  & \diskann\  & \multicolumn{6}{c}{1.000$\pm$0.000} \\
  & \divann\  ($k'\!=\!1$) & \multicolumn{6}{c}{0.286$\pm$0.041} \\
\midrule
\multirow{4}{*}{Entropy}
  & \texttt{$p$-Mean-ANN} & 4.293$\pm$0.000 & 4.200$\pm$0.052 & 4.105$\pm$0.091 & 3.887$\pm$0.155 & 3.349$\pm$0.267 & 0.746$\pm$0.717 \\
  & \texttt{$p$-FetchUnion-ANN} & 2.101$\pm$1.214 & 2.101$\pm$1.214 & 2.099$\pm$1.212 & 2.095$\pm$1.207 & 2.068$\pm$1.179 & 0.746$\pm$0.717 \\
  & \diskann\  & \multicolumn{6}{c}{0.746$\pm$0.717} \\
  & \divann\  ($k'\!=\!1$) & \multicolumn{6}{c}{4.191$\pm$0.234} \\
\bottomrule
\end{tabular}
}%
\caption{Comparison of relevance and diversity of \pMeanANN, \texttt{$p$-FetchUnion-ANN} across different values of $p$ against \diskann\ and \divann\ ($k'=1$) for \deepC\ at $k=50$. }
\label{tab:deep20-k50}
\end{table}

\begin{table}[t]
\centering
\scalebox{0.75}{
\setlength{\tabcolsep}{6pt}
\renewcommand{\arraystretch}{1.2}
\begin{tabular}{llcccccc}
\toprule
\textbf{Metric} & \textbf{Algorithm} & $p=-10$ & $p=-1$ & $p=-0.5$ & $p=0$ & $p=0.5$ & $p=1$ \\
\midrule
\multirow{4}{*}{Approx.~Ratio}
  & \texttt{$p$-Mean-ANN} & 0.958$\pm$0.019 & 0.960$\pm$0.017 & 0.961$\pm$0.016 & 0.963$\pm$0.014 & 0.969$\pm$0.010 & 1.000$\pm$0.000 \\
  & \texttt{$p$-FetchUnion-ANN} & 0.958$\pm$0.019 & 0.960$\pm$0.017 & 0.961$\pm$0.016 & 0.963$\pm$0.014 & 0.969$\pm$0.010 & 1.000$\pm$0.000 \\
  & \diskann\  & \multicolumn{6}{c}{1.000$\pm$0.000} \\
  & \divann\  ($k'\!=\!1$) & \multicolumn{6}{c}{0.395$\pm$0.010} \\
\midrule
\multirow{4}{*}{Entropy}
  & \texttt{$p$-Mean-ANN} & 4.293$\pm$0.000 & 4.292$\pm$0.005 & 4.288$\pm$0.010 & 4.275$\pm$0.020 & 4.217$\pm$0.068 & 2.070$\pm$0.208 \\
  & \texttt{$p$-FetchUnion-ANN} & 4.293$\pm$0.001 & 4.292$\pm$0.005 & 4.288$\pm$0.010 & 4.275$\pm$0.020 & 4.217$\pm$0.068 & 2.070$\pm$0.207 \\
  & \diskann\  & \multicolumn{6}{c}{2.070$\pm$0.207} \\
  & \divann\  ($k'\!=\!1$) & \multicolumn{6}{c}{4.322$\pm$0.002} \\
\bottomrule
\end{tabular}
}%
\caption{Comparison of relevance and diversity of \pMeanANN, \texttt{$p$-FetchUnion-ANN} across different values of $p$ against \diskann\ and \divann\ ($k'=1$) for \deepP\ at $k=50$. }
\label{tab:deep-uni-k50}
\end{table}

\begin{table}[t]
\centering
\scalebox{0.75}{
\setlength{\tabcolsep}{6pt}
\renewcommand{\arraystretch}{1.2}
\begin{tabular}{llcccccc}
\toprule
\textbf{Metric} & \textbf{Algorithm} & $p=-10$ & $p=-1$ & $p=-0.5$ & $p=0$ & $p=0.5$ & $p=1$ \\
\midrule
\multirow{4}{*}{Approx.~Ratio}
  & \texttt{$p$-Mean-ANN} & 0.975$\pm$0.010 & 0.977$\pm$0.008 & 0.979$\pm$0.008 & 0.980$\pm$0.008 & 0.982$\pm$0.006 & 1.000$\pm$0.000 \\
  & \texttt{$p$-FetchUnion-ANN} & 0.975$\pm$0.010 & 0.977$\pm$0.008 & 0.979$\pm$0.008 & 0.980$\pm$0.008 & 0.982$\pm$0.006 & 1.000$\pm$0.000 \\
  & \diskann\  & \multicolumn{6}{c}{1.000$\pm$0.000} \\
  & \divann\  ($k'\!=\!1$) & \multicolumn{6}{c}{0.404$\pm$0.004} \\
\midrule
\multirow{4}{*}{Entropy}
  & \texttt{$p$-Mean-ANN} & 4.292$\pm$0.006 & 4.292$\pm$0.003 & 4.293$\pm$0.002 & 4.293$\pm$0.002 & 4.269$\pm$0.020 & 2.068$\pm$0.205 \\
  & \texttt{$p$-FetchUnion-ANN} & 4.292$\pm$0.006 & 4.292$\pm$0.003 & 4.293$\pm$0.002 & 4.293$\pm$0.003 & 4.269$\pm$0.020 & 2.068$\pm$0.205 \\
  & \diskann\  & \multicolumn{6}{c}{2.068$\pm$0.205} \\
  & \divann\  ($k'\!=\!1$) & \multicolumn{6}{c}{4.322$\pm$0.005} \\
\bottomrule
\end{tabular}
}%
\caption{Comparison of relevance and diversity of \pMeanANN, \texttt{$p$-FetchUnion-ANN} across different values of $p$ against \diskann\ and \divann\ ($k'=1$) for \siftP\ at $k=50$. }
\label{tab:sift20-uni-k50}
\end{table}

\begin{table}[!h]
\centering
\scalebox{0.75}{
\setlength{\tabcolsep}{6pt}
\renewcommand{\arraystretch}{1.2}
\begin{tabular}{llcccccc}
\toprule
\textbf{Metric} & \textbf{Algorithm} & $p=-10$ & $p=-1$ & $p=-0.5$ & $p=0$ & $p=0.5$ & $p=1$ \\
\midrule
\multirow{2}{*}{Query per Second}
  & \texttt{$p$-Mean-ANN} & 198.08  &  195.97  & 199.08 & 179.03 & 171.22 & 189.31 \\
  & \texttt{$p$-FetchUnion-ANN} &  620.27 & 610.62 & 551.02  & 608.76 & 572.57  & 591.76   \\
\midrule
\multirow{2}{*}{Latency ($\mu s$)}
  & \texttt{$p$-Mean-ANN} &  161385.00 & 163121.00 &  160503.00 & 178555.00  & 186780.00  & 168856.00 \\
  & \texttt{$p$-FetchUnion-ANN} &  51539.90 & 52362.30 & 58028.60 & 52521.60 & 55843.70 & 54030.80 \\
  \midrule
\multirow{2}{*}{$99.9$th percentile of Latency}
  & \texttt{$p$-Mean-ANN} & 433434.00   & 407151.00 &  418147.00 & 421725.00 &  475474.00 & 404477.00 \\
  & \texttt{$p$-FetchUnion-ANN} &  146632.00   & 144989.00  & 145620.00  & 145657.00  & 143627.00  & 146464.00 \\
\bottomrule
\end{tabular}
}%
\caption{Comparison of Queries per second and Latency of \pMeanANN, \texttt{$p$-FetchUnion-ANN} across different values of $p$ against \diskann\ and \divann\ ($k'=1$) for \amazon\ dataset for $k=50$. }
\label{tab:QPSandLatencyAmazon}
\end{table}

\end{document}